\documentclass[11pt,reqno]{amsproc}
\allowdisplaybreaks
\usepackage{cite}
\usepackage{amssymb}
\usepackage[usenames,dvipsnames]{xcolor}
\usepackage[most]{tcolorbox}
\tcbuselibrary{breakable}
\usepackage[stable]{footmisc}
\usepackage[dvipsnames]{xcolor}
\usepackage{float}
\usepackage{fullpage}
\usepackage{rotating}
\usepackage{stmaryrd}
\usepackage{subfigure}
\usepackage{enumerate}
\usepackage[small]{caption}
\usepackage{morefloats}
\usepackage{longtable}
\numberwithin{equation}{section}
\usepackage{mathrsfs}
\usepackage{graphicx}
\DeclareGraphicsExtensions{.eps}
\usepackage{epstopdf}
\usepackage{ucs}
\usepackage[utf8x]{inputenc}
\usepackage[semicolon,square,authoryear]{natbib}
\usepackage[framed,numbered]{mcode}
\usepackage[debug=false, colorlinks=true, pdfstartview=FitV, 
linkcolor=blue, citecolor=blue, urlcolor=blue]{hyperref}
\definecolor{orange}{RGB}{255,127,0}
\makeatletter
\def\maketag@@@#1{\hbox{\m@th\normalfont\normalsize#1}}
\makeatother
\newtheorem{theorem}{Theorem}[section]
\newtheorem{lemma}[theorem]{Lemma}

\newtheorem{corollary}[theorem]{Corollary}
\newtheorem{remark}{Remark}[section]

\usepackage{morefloats}
\newlength{\drop}
\definecolor{amethyst}{rgb}{0.6, 0.4, 0.8}
\definecolor{burgundy}{rgb}{0.5, 0.0, 0.13}

\title{\textbf{A stabilized mixed 
discontinuous Galerkin formulation 
for double porosity/permeability model}}

\author{\textbf{M.~S.~Joshaghani}, 
\textbf{S.~H.~S.~Joodat}
and 
\textbf{K.~B.~Nakshatrala} \\
  {\small Department of Civil and Environmental
    Engineering, University of Houston. \\
    \textbf{Correspondence to:}~\textsf{knakshatrala@uh.edu}}}


\begin{document}

\date{\today}

\begin{titlepage}
  \drop=0.1\textheight
  \centering
  \vspace*{\baselineskip}
  \rule{\textwidth}{1.6pt}\vspace*{-\baselineskip}\vspace*{2pt}
  \rule{\textwidth}{0.4pt}\\[\baselineskip]
       {\Large \textbf{\color{burgundy}
A stabilized mixed discontinuous Galerkin 
formulation \\[0.3\baselineskip]
for double porosity/permeability model}}\\[0.3\baselineskip]
       \rule{\textwidth}{0.4pt}\vspace*{-\baselineskip}\vspace{3.2pt}
       \rule{\textwidth}{1.6pt}\\[\baselineskip]
       \scshape
       An e-print of the paper is available on arXiv:~1805.01389. \par 
       \vspace*{0.8\baselineskip}
       Authored by \\[0.5\baselineskip]

  {\Large M.~S.~Joshaghani\par}
  {\itshape Graduate Student, University of Houston}\\[0.75\baselineskip]

  {\Large S.~H.~S.~Joodat\par}
  {\itshape Graduate Student, University of Houston}\\[0.75\baselineskip]
           
  {\Large K.~B.~Nakshatrala\par}
  {\itshape Department of Civil \& Environmental Engineering \\
  University of Houston, Houston, Texas 77204--4003 \\ 
  \textbf{phone:} +1-713-743-4418, \textbf{e-mail:} knakshatrala@uh.edu \\
  \textbf{website:} http://www.cive.uh.edu/faculty/nakshatrala}\\[0.25\baselineskip]
   \vspace{0.1mm} 
  \begin{figure}[h]
    \centering
\includegraphics[clip,width=0.54\linewidth]{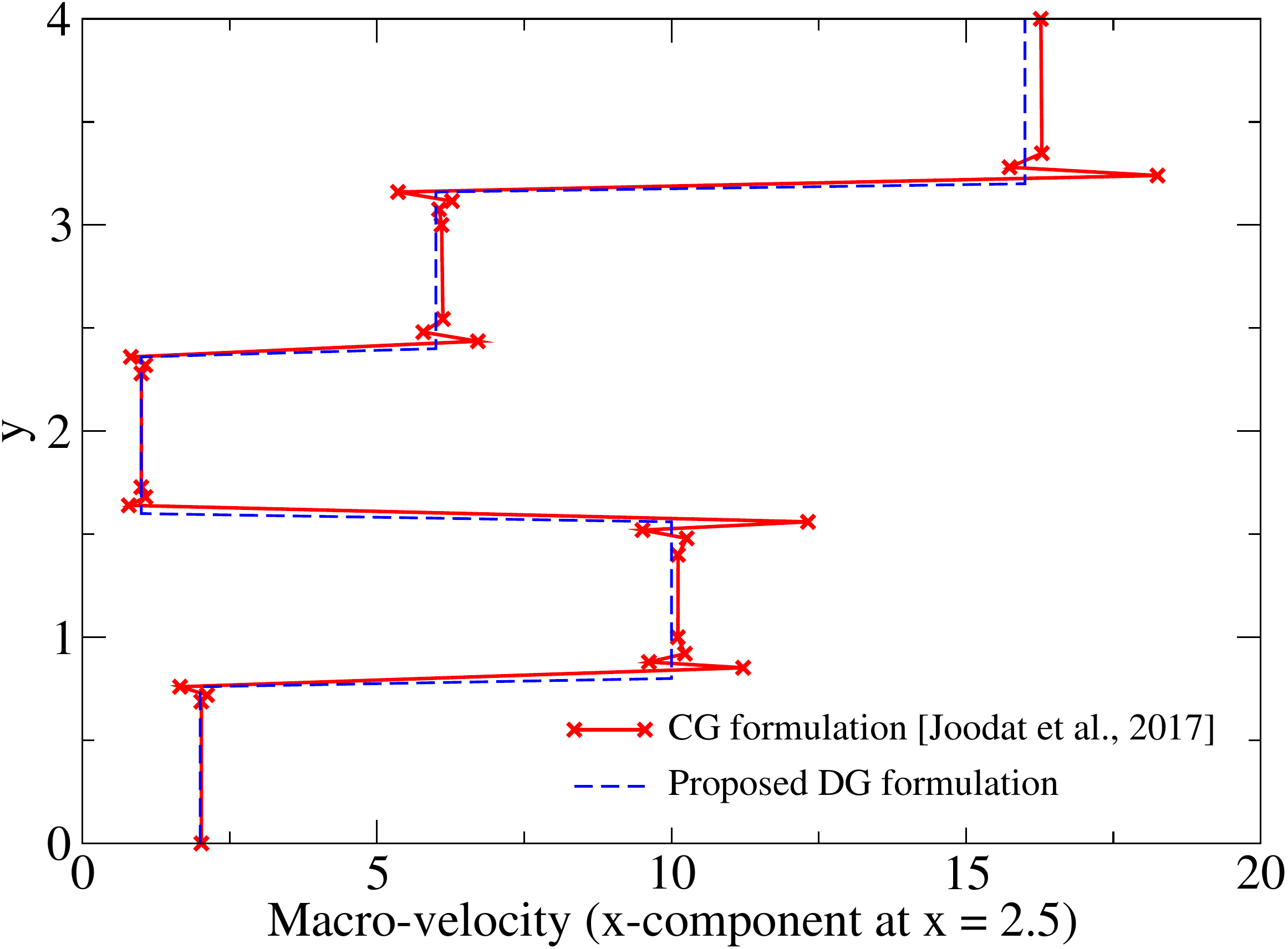}

  	\emph{{\small{Comparison of the profiles of 
	the macro-velocity under the stabilized mixed 
	continuous Galerkin (CG) formulation and 
	the proposed stabilized mixed discontinuous 
	Galerkin (DG) formulation. Under the CG 
	formulation, overshoots and undershoots 
	are observed at the interfaces of the layers. 
	On the other hand, the proposed DG 
	formulation is able to capture the physical 
	jumps across the interfaces.}}}
  \end{figure}
   \vfill
  {\scshape 2018} \\
  {\small Computational \& Applied Mechanics Laboratory} \par
\end{titlepage}

\begin{abstract}
  Modeling flow through porous media with multiple pore-networks has now become an active area of research due to recent technological endeavors like geological carbon sequestration and recovery of hydrocarbons from tight rock formations. Herein, we consider the double porosity/permeability (DPP) model, which describes the flow of a single-phase incompressible fluid through a porous medium exhibiting two dominant pore-networks with a possibility of mass transfer across them. We present a stable mixed discontinuous Galerkin (DG) formulation for the DPP model. The formulation enjoys several attractive features. 
These include: (i) Equal-order interpolation for all the field 
variables (which is computationally the most convenient) 
is stable under the proposed formulation. 
(ii) The stabilization terms are residual-based, and the stabilization parameters do not contain any mesh-dependent parameters.
(iii) The formulation is theoretically shown to be consistent, stable, and hence convergent. 
(iv) The formulation supports non-conforming discretization and distorted meshes. 
(v) The DG formulation has improved element-wise (local) mass balance compared to 
the corresponding continuous formulation. 
(vi) The proposed formulation can capture physical instabilities in coupled flow and transport problems under the DPP model.
\end{abstract}

\maketitle
\vspace{-0.4in}


\section*{A list of abbreviations and symbols} 
\label{Sec:S0_Notation}
\begin{longtable*}[c]{|p{.2\textwidth} || p{.75\textwidth}|}
  \hline
  \multicolumn{2}{| c |}{\emph{Abbreviations}} \\
  \hline
  CG  & Continuous Galerkin \\
  DG  & Discontinuous Galerkin \\
  DPP & Double porosity/permeability \\
  \hline
  \multicolumn{2}{|c|}{\emph{Symbols in the DPP model,
    $\mathsection \ref{Sec:S2_DG_GE}$}} \\
  \hline
  $\Omega$, $\overline{\Omega}$, $\partial \Omega$
  & Computational porous domain, its set closure,
  and its boundary\\
  $\mathbf{u}_1$, $\mathbf{u}_2$, $p_1$,
  $p_2$ & Velocity and pressure solution
  fields in the two pore-networks \\
  $k_1$, $k_2$ & Permeabilities
  in the two pore-networks \\
  $\gamma$, $\mu$ & True density
  and coefficient of viscosity of the fluid \\
  $\mathbf{b}_1$, $\mathbf{b}_2$ & Specific body forces
  in the pore-networks \\
  $\widehat{\mathbf{n}}(\mathbf{x})$ &
  Unit outward normal vector at $\mathbf{x}
  \in \partial \Omega$ \\
  $u_{n1}$, $u_{n2}$, $p_{01}$, 
  $p_{02}$ & Prescribed velocities and pressures,
  Eqs.~\eqref{Eqn:DG_GE_vBC_1}--\eqref{Eqn:DG_GE_Darcy_pBC_2} \\
  $\Gamma^u_1$, $\Gamma^u_2$, $\Gamma^p_1$, $\Gamma^p_2$
  & Velocity and pressure boundaries, 
  Eq.~\eqref{Eqn:well_posedness_BCs} \\
  $\chi$ & Mass exchange across the pore-networks,
  Eq.~\eqref{Eqn:DG_GE_Darcy_MT}\\
  $\beta$ & Parameter in the inter-pore mass transfer,
  Eq.~\eqref{Eqn:DG_GE_Darcy_MT} \\
  $\eta$ & Flow characterization parameter in the
  DPP model, Eq.~\eqref{Eqn:flow_characterization_param}\\
  \hline
  \multicolumn{2}{|c|}{\emph{Mesh-related quantities,
    $\mathsection \ref{Sec:S2_DG_GE}$}} \\
  \hline
  $Nele$ & Number of subdomains (elements) \\
  $\omega^{i}$, $\partial \omega^{i}$ & The $i$-th
  subdomain and its boundary $(i = 1, \cdots, Nele)$ \\
  $\widetilde{\Omega}$ & Union of all open subdomains,
  Eq.~\eqref{Eqn:DG_Omega_tilde} \\
  $\mathcal{E}$, $\mathcal{E}^{\mathrm{int}}$ &
  Sets of all and interior edges, respectively \\
  $\Upsilon$ & A typical edge (i.e., $\Upsilon \in \mathcal{E}$ or
  $\Upsilon \in \mathcal{E}^{\mathrm{int}}$) \\ 
  $\Gamma^{\mathrm{int}}$ & Union of internal boundaries  \\
  $h$ & Mesh-size, Eq.~\eqref{Eqn:DG_mesh_size} \\
  $h_{\omega}$ & Element diameter of $\omega$,
  $\mathsection \ref{Subsec:DG_mesh_related}$ 
  \& \textbf{Fig.~\ref{Fig:DPP_element_diameter}} \\
  $h_{\omega}^{\mathrm{inc}}$ & Diameter of
  the inscribed circle in $\omega$,
  $\mathsection \ref{Subsec:DG_mesh_related}$ 
  \& \textbf{Fig.~\ref{Fig:DPP_element_diameter}} \\
  $h_{\Upsilon}$ & Characteristic length of an edge,
  Eq.~\eqref{Eqn:DPP_h_Upsilon} \\
  $\mathcal{T}$, $\mathcal{T}_{h}$ & A mesh, and a mesh
  with mesh-size $h$ \\
  \hline
  \multicolumn{2}{|c|}{\emph{Symbols in the proposed
      DG formulation,
    $\mathsection \ref{Sec:S3_DG_Mixed}$}} \\
  \hline
  $(\cdot;\cdot)_{\mathcal{K}}$, $(\cdot;\cdot)$
  & $L_2$ inner-products over $\mathcal{K}$
  and $\widetilde{\Omega}$, respectively \\
  $\|\cdot\|_{\mathcal{K}}$, $\|\cdot\|$
  & $L_2$ norms over $\mathcal{K}$ and
  $\widetilde{\Omega}$, respectively \\
  $\mathbf{w}_1$, $\mathbf{w}_2$, $q_1$, 
  $q_2$ & Weighting functions for velocities and pressures \\
  ${\mathop{\mathbf{u}}^{\star}}_{1}$, 
  ${\mathop{\mathbf{u}}^{\star}}_{2}$, ${\mathop{p}^{\star}}_{1}$, ${\mathop{p}^{\star}}_{2}$ & Numerical fluxes,
  Eqs.~\eqref{Eqn:DPP2_NF_p1}--\eqref{Eqn:DPP2_NF_u2} \\
  $\{\!\!\{\cdot\}\!\!\}$, $\llbracket\cdot\rrbracket$ &
  Average and jump operators,
  Eqs.~\eqref{Eqn:DPP_jump_and_avg_scalar} \&
  \eqref{Eqn:DPP_restrictions_vector} \\
  $\eta_{u}$, $\eta_{p}$ & Stabilization parameters for
  jumps in velocities and pressures across interior
  edges, respectively;
  Eqs.~\eqref{Eqn:DG_stab_lambda_choices}--\eqref{sub:Stabilized_mixed} \\
  \hline
  \multicolumn{2}{|c|}{\emph{Constants in various estimates}} \\
  \hline
  $\mathcal{C}_{\mathrm{drag},1}$, 
  $\mathcal{C}_{\mathrm{drag},2}$ 
  &Bounds on drag coefficients, 
  Eq.~\eqref{Eqn:DG_drag_bounds} \\
  $\mathcal{C}_{\mathbf{e}_{\mathbf{u}_1}}$, 
  $\mathcal{C}_{\mathbf{e}_{\mathbf{u}_2}}$ 
  & Eqs.~\eqref{Eqn:DG_C_eu1} \& 
  \eqref{Eqn:DG_C_eu2} \\
  $\mathcal{C}_{\mathrm{int}}$ & Constant
  in standard estimate for interpolation error,
  Eq.~\eqref{Eqn:DG_standard_estimate_for_int_error} \\
  $\mathcal{C}_{\mathrm{inv}}$ & Constant
  in discrete inverse inequality, 
  Eq.~\eqref{Eqn:DDP_discrete_inverse_inequality} \\
  $\mathcal{C}_{\mathrm{lqu}}$ & Locally quasi-uniform
  coefficient,
  $\mathsection \ref{Subsec:DG_mesh_related}$ \& 
  Eq.~\eqref{Eqn:lqu_condition} \\
  $\mathcal{C}_{\mathrm{sp}}$ & Shape parameter,
  $\mathsection \ref{Subsec:DG_mesh_related}$ 
  \& Eq.~\eqref{Eqn:shape_parameter_condition} \\
  $\mathcal{C}_{\mathrm{trace}}$ & Constant
  in continuous trace inequality, 
  Eqs.~\eqref{Eqn:Trace_inequality_exp}
  \& \eqref{Eqn:Trace_inequality_exp2} \\
  \hline
  \multicolumn{2}{|c|}{\emph{Other symbols}} \\
  \hline
  $\mathscr{P}^{m}(\omega)$ & Set of all
  polynomials over $\omega$ up to and
  including $m$-th order,
  $\mathsection\ref{Subsec:DG_functional_analysis}$ \\
    $c$, $D$ & Concentration and diffusivity, 
  $\mathsection \ref{Sec:S8_DG_NR}$ \\
   $m(\omega)$ &  Net rate of volumetric flux from element $\omega$, 
   $\mathsection \ref{Sec:Element_wise_MB}$ \\
   $m^{\mathrm{out}}_{\mathrm{max}}$ &  Maximum element-wise mass outflow flux, 
   Eq.~\eqref{Eqn:Max_mass_outFlux}\\
   $m^{\mathrm{in}}_{\mathrm{max}}$ &  Maximum element-wise mass inflow flux, 
   Eq.~\eqref{Eqn:Max_mass_inFlux}\\
  \hline
\end{longtable*}

\section{INTRODUCTION AND MOTIVATION}
This paper presents a discontinuous Galerkin
version of the continuous stabilized mixed
formulation proposed recently by
\citep{Nakshatrala_Joodat_Ballarini_P2}
for the double porosity/permeability (DPP)
mathematical model. The DPP model describes
the flow of a single-phase incompressible
fluid in a rigid porous medium with two
distinct pore-networks with possible mass
transfer across them. A derivation of the
DPP model using the theory of interacting
continua and continuum thermomechanics along
with the mathematical properties that the
solutions of this model satisfy are presented
in \citep{Nakshatrala_Joodat_Ballarini}.

The motivation for this work is twofold. 
\emph{First}, due to the increasing interest
in geo-materials with multiple pore-scales
(e.g., vuggy carbonates and shales) and
the development of new synthetic
complex porous materials using advanced
manufacturing techniques, understanding
the flow of fluids in such porous materials
is currently an active area of research.
Darcy equations, which are commonly
used for modeling flow of fluids in
a porous medium with one single pore-network, are not
adequate to accurately describe
the flow dynamics in porous media
with multiple pore-networks.
Hence, it is required to develop
new mathematical models and computational
tools which can accurately capture the
flow characteristics in complex porous media
consisting of multiple pore-networks with
possible mass transfer across them.
For this purpose, \citep{Nakshatrala_Joodat_Ballarini}
have recently proposed a mathematical model,
which is capable of considering fluid flow
through two pore-networks. This mathematical
model will be referred to as the DPP model, which
forms a basis for the current paper. 

The \emph{second} motivation behind the
current paper is that the continuous
Galerkin (CG) based formulations suffer
from the so-called Gibbs phenomenon\footnote{
  Traditionally, the Gibbs phenomenon is the
  manifestation of overshoots and undershoots
  in the representation of a simple discontinuity
  using the Fourier series.
  This phenomenon was first observed by
  \citep{wilbraham1848certain}.
  A mathematical explanation was later provided
  in the papers \citep{gibbs1898fourier} and
  \citep{gibbs1899fourier}; the former paper had a mistake
  which was corrected in the later. 
  However, contrary to the traditional
  belief, one can observe undershoots and
  overshoots even when non-trigonometric
  functions are employed to approximate
  a simple discontinuous function in a
  least-squares sense. In particular,
  the ``Gibbs phenomenon'' can occur
  even under a piece-wise linear approximation
  \citep{foster1991gibbs}.  
} when
applied to problems with highly heterogeneous
medium properties such as layered media; which
manifests in the form of spurious oscillations
(overshoots and undershoots) at the interface
of a sharp change in medium properties (e.g.,
permeability).
\citep{Hughes_Masud_Wan_2006} have clearly
demonstrated that conventional continuous
finite element methods for \emph{Darcy
  equations} fall short in accurately
capturing jumps in the solution fields
at the location of material discontinuities. 
Since disparate
medium properties are frequently encountered 
in subsurface modeling, the stabilized mixed
four-field CG formulation recently proposed for DPP model
\citep{Nakshatrala_Joodat_Ballarini_P2}
will not be able to accurately capture
the velocity profiles in highly heterogeneous porous media and
will not suffice for realistic subsurface
modeling. This will be clearly demonstrated
using numerical simulations in a later
section of this paper.
We, therefore, develop a stabilized mixed
DG formulation for the DPP model, which
is robust, stable and capable of capturing
possible jumps in the solution fields due
to the existing disparate medium properties.  

It is important to mention that one can also
capture disparate medium properties and satisfy
the LBB \emph{inf-sup} stability condition
\citep{Brezzi_Fortin} by employing
an element from the H(div) family; which include
Raviart-Thomas spaces \citep{Raviart_Thomas_MAFEM_1977_p292},
N\'ed\'elec spaces \citep{nedelec1980mixed}
Brezzi-Douglas-Marini
(BDM) spaces \citep{Brezzi_Douglas_Marini_NumerMath_1985_v47_p217},
Brezzi-Douglas-Fortin-Marini (BDFM)
spaces \citep{Brezzi_Douglas_Fortin_Marini_MMNA_1987_v21_p581} and
Crouzeix-Raviart spaces \citep{crouzeix1973conforming}. 
Although there is an on-going debate on using
H(div) elements vs. DG methods, the later do
enjoy some unique desirable features.
DG methods combine the attractive features of 
both finite element and finite volume methods. 
Application of completely discontinuous basis 
functions in the form of piecewise polynomials
in DG methods provides them with the flexibility 
to support common non-conforming spaces (e.g., 
non-matching grids and hanging nodes, $h$-$p$ 
adaptivity, variable degrees of local interpolations) 
and handle jumps in the profiles of variables \citep{Riviere_Wheeler_2002,Cockburn_2003,Li_Riviere_2015numerical,
  Li_Riviere_2016numerical}.
DG methods also enjoy high parallel efficiency.
Unlike the conventional continuous formulations,
they are known to exhibit better local (or element-wise)
mass balance \citep{Hughes_Masud_Wan_2006,Riviere2008discontinuous}.

The origins of DG methods can be traced back to
\citep{Lions_1968} and \citep{Nitsche_1971_Uber}. 
One of the first successful applications of DG
formulation to solve a practical problem was by 
\citep{Reed_Hill_1973}, which addressed neutron
transport. 
Over the years, DG methods have been successfully
employed to solve hyperbolic PDEs 
\citep{Brezzi_Marini_Suli_2004,PalAbedi_Madhukar_Haber_2016}, 
elliptic PDEs \citep{Douglas_Dupont_1976,
Riviere_Wheeler_Girault_1999,Rusten_Vassilevski_Winther_1996,
Arnold_Brezzi_Cockburn_Marini_2002,Barrios_Bustinza_2007,Cockburn_Guzman_Wang_2009}, 
parabolic PDEs \citep{Douglas_Dupont_1976,Kulkarni_Rovas_Tortorelli_2007},
coupling algorithms \citep{nakshatrala2009time} 
and space-time finite elements \citep{palaniappan2004spacetime,abedi2006space}.
  Several variants of DG formulations have been
  developed over the years with varying merits
  for each variant. Some popular variants are
  Runge-Kutta DG \citep{cockburn2001runge},
  local DG \citep{castillo2000priori}, embedded
  DG \cite{guzey2007embedded}, compact DG
  \citep{peraire2008compact}, hybridizable
  DG \citep{cockburn2009unified} and
  adjoint-type variational multiscale
  DG \citep{Hughes_Masud_Wan_2006,badia2010stabilized}.
  Although these variants may look very
  different, a unified framework
  has been laid out by \citep{Arnold_Brezzi_Cockburn_Marini_2002},
  to derive DG methods systematically,
  and these methods differ in their choices of
  numerical fluxes. 
  However, to the best of authors' knowledge,
  there is no clear cut winner among
  these variants.

In this paper we employ the adjoint-type variational
multiscale approach to develop a stabilized mixed four-field DG formulation
for the DPP model. In order to circumvent the LBB
\emph{inf-sup} stability condition we add residual-based, adjoint-type stabilization terms
defined over the elements. In order to avoid Gibbs
phenomenon and at the same time maintain stability,
we choose appropriate and consistent numerical fluxes,
which are in the form of jumps and averages of the
medium properties and solution fields.
The resulting stabilized mixed DG formulation
enjoys several attractive features, which
include:
(i) The formulation is capable of eliminating the 
spurious numerical instabilities in the profiles 
of solutions and capturing the existing jumps in 
the material properties. (ii) Equal-order interpolation for all the
field variables, which is computationally preferred, is stable.
(iii) The formulation is mathematically shown
to be consistent, stable, and hence
convergent.
(iv) \emph{A priori} error estimation is
systematically obtained.
(v) The DG formulation exhibits improved
element-wise mass balance compared to its
continuous counterpart.
(vi) The formulation can be utilized to
capture physical instabilities in heterogeneous
porous media and to eliminate numerical
instabilities at the same time. 

The rest of this paper is organized as follows: 
Background material and preliminaries (including
the governing equations of the mathematical model)
are provided in Section \ref{Sec:S2_DG_GE}. 
The proposed stabilized mixed DG formulation is
presented in Section \ref{Sec:S3_DG_Mixed}.
A systematic convergence analysis and the error
estimation of the proposed DG formulation
are carried out in Section \ref{Sec:S4_DG_Error}.
Results of constant flow patch tests along with 
a sensitivity study on the stabilization parameters are presented 
in Section \ref{Sec:S5_DG_Patch_tests}.
Numerical convergence analysis and structure preserving properties are provided in Sections
\ref{Sec:S6_DG_Numerical_convergence} and \ref{Sec:S7_DG_Structure_preserving}, respectively.
In Section \ref{Sec:S8_DG_NR}, the proposed DG
formulation is implemented to study viscous-fingering-type
physical instabilities in heterogeneous porous
media with double pore-networks.
Finally, conclusions are drawn in 
Section \ref{Sec:S9_DG_CR}.

Throughout this paper, repeated indices do not imply summation.

\section{BACKGROUND MATERIAL AND PRELIMINARIES}
\label{Sec:S2_DG_GE}
\subsection{Governing
  equations\footnote{This subsection
    on the governing equations, which will
    be similar to our earlier papers
    \citep{Nakshatrala_Joodat_Ballarini,
    Nakshatrala_Joodat_Ballarini_P2},
    is provided to make this paper
    self-contained and for easy referencing.}}
The DPP model deals with the flow of
a single-phase incompressible fluid
through a rigid porous medium with two
pore-networks exhibiting different
hydromechanical properties.
We refer to these two pore-networks as
macro-pore and micro-pore networks,
which are denoted by subscripts $1$
and $2$, respectively.
We denote the porous domain by $\Omega
\subset \mathbb{R}^{nd}$, where ``$nd$''
represents the number of spatial dimensions.
For a precise mathematical treatment, we
assume that $\Omega$ is an open bounded
domain. The boundary $\partial \Omega =
\overline{\Omega} - \Omega$ is assumed
to be smooth, where the superposed bar
denotes the set closure.  
A spatial point is denoted by $\mathbf{x} \in \overline{\Omega}$.
The gradient operator with respect to $\mathbf{x}$
is denoted by $\mathrm{grad}[\cdot]$ and the corresponding 
divergence operator is denoted by $\mathrm{div}[\cdot]$. 
The unit outward normal to the boundary is denoted
by $\widehat{\mathbf{n}}(\mathbf{x})$.
The pressure and the discharge
(or Darcy) velocity fields in the macro-pore network
are, respectively, denoted by $p_{1}(\mathbf{x})$ and
$\mathbf{u}_{1}(\mathbf{x})$, and the corresponding
fields in the micro-pore network are denoted by
$p_{2}(\mathbf{x})$ and $\mathbf{u}_{2}(\mathbf{x})$.
We denote the viscosity and true density of the
fluid by $\mu$ and $\gamma$, respectively.

The abstract boundary value problem under
the DPP model takes the following form:
Find $\mathbf{u}_{1}(\mathbf{x})$,
$\mathbf{u}_{2}(\mathbf{x})$,
$p_{1}(\mathbf{x})$ and
$p_{2}(\mathbf{x})$ such that 
\begin{subequations}
  \label{Eqn:DPP_GE}
  \begin{alignat}{2}
    \label{Eqn:DG_GE_Darcy_BLM_1}
    &\mu k_{1}^{-1} \mathbf{u}_1(\mathbf{x})
    + \mathrm{grad}[p_1(\mathbf{x})]
    = \gamma \mathbf{b}(\mathbf{x})
    &&\quad \mathrm{in} \; \Omega \\
    \label{Eqn:DG_GE_Darcy_BLM_2}
    &\mu k_{2}^{-1} \mathbf{u}_2(\mathbf{x})
    + \mathrm{grad}[p_2(\mathbf{x})]
    = \gamma \mathbf{b}(\mathbf{x})
    &&\quad \mathrm{in} \; \Omega \\
    \label{Eqn:DG_GE_Darcy_mass_balance_1}
    &\mathrm{div}[\mathbf{u}_1(\mathbf{x})] = +\chi(\mathbf{x})
    &&\quad \mathrm{in} \; \Omega \\
    \label{Eqn:DG_GE_Darcy_mass_balance_2}
    &\mathrm{div}[\mathbf{u}_2(\mathbf{x})] = -\chi(\mathbf{x})
    &&\quad \mathrm{in} \; \Omega \\
    \label{Eqn:DG_GE_Darcy_MT}
    &\chi(\mathbf{x}) = -\frac{\beta}{\mu}
    (p_1(\mathbf{x}) - p_2(\mathbf{x}))
    &&\quad \mathrm{in} \; \Omega \\
    \label{Eqn:DG_GE_vBC_1}
    &\mathbf{u}_1(\mathbf{x}) \cdot
    \widehat{\mathbf{n}}(\mathbf{x})
    = u_{n1}(\mathbf{x})
    &&\quad \mathrm{on} \; \Gamma^{u}_{1} \\
    \label{Eqn:DG_GE_vBC_2}
    &\mathbf{u}_2(\mathbf{x}) \cdot
    \widehat{\mathbf{n}}(\mathbf{x})
    = u_{n2}(\mathbf{x}) 
    &&\quad \mathrm{on} \; \Gamma^{u}_{2} \\ 
    \label{Eqn:DG_GE_Darcy_pBC_1}
    &p_1(\mathbf{x}) = p_{01} (\mathbf{x})
    &&\quad \mathrm{on} \; \Gamma^{p}_{1} \\
    \label{Eqn:DG_GE_Darcy_pBC_2}
    &p_2(\mathbf{x}) = p_{02} (\mathbf{x})
    &&\quad \mathrm{on} \; \Gamma^{p}_{2} 
  \end{alignat}
\end{subequations}
where $k_1(\mathbf{x})$ and
$k_2(\mathbf{x})$, respectively,
denote the (isotropic) permeabilities
of the macro-pore and micro-pore
networks, $\mathbf{b}(\mathbf{x})$
denotes the specific body force,
and $\beta$ is a dimensionless
characteristic of the porous medium.
$\chi(\mathbf{x})$ accounts for the
mass exchange across the pore-networks
and is the rate of volume transfer of
the fluid between the two pore-networks per unit
volume of the porous medium. The
dimension of $\chi(\mathbf{x})$
is one over the time
[$\mathrm{M}^{0}\mathrm{L}^{0}\mathrm{T}^{-1}$].
$\Gamma_{i}^{u}$ denotes that part
of the boundary on which the normal
component of velocity is prescribed
in the macro-pore ($i=1$) and micro-pore
($i=2$) networks, and $u_{n1}(\mathbf{x})$
and $u_{n2}(\mathbf{x})$ denote the prescribed
normal components of the velocities on
$\Gamma^{u}_{1}$ and $\Gamma^{u}_{2}$,
respectively.  
$\Gamma_{i}^{p}$ is that part of the
boundary on which the pressure is
prescribed in the macro-pore ($i=1$)
and micro-pore ($i=2$) networks, and
$p_{01}(\mathbf{x})$ and $p_{02}(\mathbf{x})$
denote the prescribed pressures on
$\Gamma_{1}^{p}$ and $\Gamma_{2}^{p}$,
respectively.

For mathematical well-posedness, we
assume that
\begin{align}
  \Gamma_{1}^{u} \cup \Gamma_{1}^{p} =
  \partial \Omega, \quad
  \Gamma_{1}^{u} \cap \Gamma_{1}^{p} = \emptyset, \quad 
  \Gamma_{2}^{u} \cup \Gamma_{2}^{p} = \partial \Omega,  
  \quad \mathrm{and} \quad
  \Gamma_{2}^{u} \cap \Gamma_{2}^{p} = \emptyset
  \label{Eqn:well_posedness_BCs}
\end{align}
However, if $\Gamma_{1}^{p} = \emptyset$
and $\Gamma_{2}^{p} = \emptyset$ hold
simultaneously then one will be able
to find the pressure in each pore-network
only up to an arbitrary constant. 
We assume that the drag coefficients
in the two pore-networks, $\mu/k_1$
and $\mu/k_2$, are bounded below and
above. That is, 
\begin{align}
  0 < 
  \inf_{\mathbf{x} \in \Omega} \frac{\mu}{ k_i(\mathbf{x})} 
  \leq 
  \sup_{\mathbf{x} \in \Omega} \frac{\mu}{ k_i(\mathbf{x})}
  < +\infty
  \qquad i = 1, 2
\end{align}
This also means that there exist
two non-dimensional constants $1 \leq 
\mathcal{C}_{\mathrm{drag},1}, \;
\mathcal{C}_{\mathrm{drag},2} < +\infty$
where 
\begin{align}
  \label{Eqn:DG_drag_bounds}
  \mathcal{C}_{\mathrm{drag},1} := 
  \left(\sup_{\mathbf{x} \in \Omega} \frac{\mu}{ k_1(\mathbf{x})} \right)
  \left(\inf_{\mathbf{x} \in \Omega} \frac{\mu}{ k_1(\mathbf{x})} \right)^{-1}
  \quad \mathrm{and} \quad 
  \mathcal{C}_{\mathrm{drag},2} := 
  \left(\sup_{\mathbf{x} \in \Omega} \frac{\mu}{ k_2(\mathbf{x})} \right)
  \left(\inf_{\mathbf{x} \in \Omega} \frac{\mu}{ k_2(\mathbf{x})} \right)^{-1} 
\end{align}

\subsection{Geometrical definitions}
The domain is partitioned into ``$Nele$''
subdomains, which will be elements in the
context of the finite element method. These
elements form a mesh on the domain.
Mathematically, a mesh $\mathcal{T}$
on $\Omega$ is a finite collection
of disjoint polyhedra $\mathcal{T}
= \{\omega^{1}, \cdots, \omega^{Nele}\}$
such that
\begin{align}
  \label{Eqn:DG_FE_decomposition}
  \overline{\Omega} = \bigcup_{i = 1}^{Nele}
  \overline{\omega}^{i}
\end{align}
(Recall that an overline denotes the set closure.) 
We refer to $\omega^{i}$ as the
$i$-th subdomain (element). The
union of all open subdomains is
denoted by
\begin{align}
  \label{Eqn:DG_Omega_tilde}
  \widetilde{\Omega} =
  \bigcup_{i=1}^{Nele} \omega^{i} 
\end{align}
with the understanding that an
integration over $\widetilde{\Omega}$
is interpreted as follows:
\begin{align}
  \label{Eqn:DPP_broken_integral}
  \int_{\widetilde{\Omega}} (\cdot)
  \mathrm{d} \Omega =
  \sum_{i=1}^{Nele} \int_{\omega^{i}}
  (\cdot) \mathrm{d} \Omega 
\end{align}

The boundary of element $\omega^{i}$
is denoted by $\partial \omega^{i} :=
\overline{\omega}^i - \omega^i$.
The set of all edges{\footnote{For
    simplicity, we use ``edge'' to
    refer to a node in 1D, an edge
    in 2D and a face in 3D in the
    entire paper. The context will
    be clear from the particular
    discussion.}} in the mesh is
denoted by $\mathcal{E}$ and the
set of all interior edges is denoted
by $\mathcal{E}^{\mathrm{int}}$. 
The entire boundary of the skeleton
of the mesh (i.e, the union of all
the interior and exterior edges) is
denoted by
\begin{align}
  \Gamma =
  \bigcup_{\Upsilon \in \mathcal{E}} \Upsilon
  \equiv \bigcup_{i = 1}^{Nele} \partial \omega^{i} 
\end{align}
The entire interior boundary (i.e.,
the union of all the interior edges)
is denoted by 
\begin{align}
  \label{Eqn:DPP_Gamma_int}
  \Gamma^{\mathrm{int}} = 
  \bigcup_{\Upsilon \in \mathcal{E}^{\mathrm{int}}}
  \Upsilon
  \equiv \Gamma \setminus \partial \Omega
\end{align}
Similar to the broken integral
over $\widetilde{\Omega}$ (i.e.,
equation \eqref{Eqn:DPP_broken_integral}),
the integral over $\Gamma^{\mathrm{int}}$
should be interpreted as follows:
\begin{align}
  \int_{\Gamma^{\mathrm{int}}}(\cdot)\mathrm{d} \Gamma 
  = \sum_{\Upsilon \in \mathcal{E}^{\mathrm{int}}}
  \int_{\Upsilon} (\cdot) \mathrm{d} \Gamma 
\end{align}

\subsection{Average and jump operators}
Consider an interior edge $\Upsilon \in
\mathcal{E}^{\mathrm{int}}$. We denote the
elements that juxtapose $\Upsilon$ by
$\omega_{\Upsilon}^{+}$ and $\omega_{\Upsilon}^{-}$. 
The unit normal vectors on this interior
edge pointing outwards to $\omega^{+}_{\Upsilon}$
and $\omega^{-}_{\Upsilon}$ are, respectively,
denoted by $\widehat{\mathbf{n}}^{+}_{\Upsilon}$
and $\widehat{\mathbf{n}}^{-}_{\Upsilon}$ (see
\textbf{Fig.~\ref{Fig:DG_domain_decomposition}}). 
The average $\{\!\!\{\cdot\}\!\!\}$
and jump $\llbracket \cdot \rrbracket$
operators on $\Upsilon$ for a scalar
field $\varphi(\mathbf{x})$ are,
respectively, defined as follows: 
\begin{align}
  \label{Eqn:DPP_jump_and_avg_scalar}
  \{\!\!\{\varphi\}\!\!\} :=
  \frac{1}{2} \left(\varphi^{+}_{\Upsilon}(\mathbf{x}) 
  + \varphi^{-}_{\Upsilon}(\mathbf{x})\right)
  \quad \mathrm{and} \quad 
  \llbracket \varphi \rrbracket :=
  \varphi^{+}_{\Upsilon}(\mathbf{x})
  \widehat{\mathbf{n}}^{+}_{\Upsilon}(\mathbf{x}) 
  + \varphi^{-}_{\Upsilon}(\mathbf{x})
  \widehat{\mathbf{n}}^{-}_{\Upsilon}(\mathbf{x})
  \qquad \forall \mathbf{x} \in \Upsilon
\end{align}
where $\varphi^{+}_{\Upsilon}(\mathbf{x})$
and $\varphi^{-}_{\Upsilon}(\mathbf{x})$ are
the restrictions of $\varphi(\mathbf{x})$
onto the elements $\omega^{+}_{\Upsilon}$
and $\omega^{-}_{\Upsilon}$, respectively. 
Mathematically, 
\begin{align}
  \label{Eqn:DPP_restrictions_scalar}
  \varphi^{+}_{\Upsilon}(\mathbf{x}) :=
  \varphi(\mathbf{x})\big|_{\partial \omega^{+}_{\Upsilon}}
  \quad \mathrm{and} \quad
  \varphi^{-}_{\Upsilon}(\mathbf{x}) :=
  \varphi(\mathbf{x}) \big|_{\partial \omega^{-}_{\Upsilon}}
  \quad \forall \mathbf{x} \in \Upsilon
\end{align}
For a vector field $\boldsymbol{\tau}
(\mathbf{x})$, these operators on
$\Upsilon$ are defined as follows:
\begin{align}
  \label{Eqn:DPP_restrictions_vector}
  \{\!\!\{ \boldsymbol{\tau} \}\!\!\} :=
  \frac{1}{2} \left(\boldsymbol{\tau}^{+}_{\Upsilon}
  (\mathbf{x}) + \boldsymbol{\tau}^{-}_{\Upsilon}
  (\mathbf{x})\right)
  \quad \mathrm{and} \quad
  \llbracket \boldsymbol{\tau} \rrbracket :=
  \boldsymbol{\tau}^{+}_{\Upsilon}(\mathbf{x}) \cdot
  \widehat{\mathbf{n}}^{+}_{\Upsilon}(\mathbf{x})
  + \boldsymbol{\tau}^{-}_{\Upsilon}(\mathbf{x}) \cdot
  \widehat{\mathbf{n}}^{-}_{\Upsilon}(\mathbf{x})
  \qquad \forall \mathbf{x} \in \Upsilon
\end{align}
where $\boldsymbol{\tau}^{+}_{\Upsilon}(\mathbf{x})$
and $\boldsymbol{\tau}^{-}_{\Upsilon}(\mathbf{x})$
are defined similar to equation
\eqref{Eqn:DPP_restrictions_scalar}. 
It is important to note that the jump
operator acts on a scalar field to
produce a vector field and vice-versa. 
It is also important to note that the
above definitions are independent of
the ordering of the elements. 
The following identity will be
used in the rest of this paper:
\begin{align}
  \label{Eqn:DG_jump_avg_identity}
  \llbracket \varphi \boldsymbol{\tau} \rrbracket 
  = \llbracket \boldsymbol{\tau} \rrbracket
  \{\!\!\{ \varphi \}\!\!\} 
  + \{\!\!\{ \boldsymbol{\tau} \}\!\!\}
  \cdot \llbracket \varphi \rrbracket
\end{align}

\begin{figure}
  \includegraphics[scale=0.62]{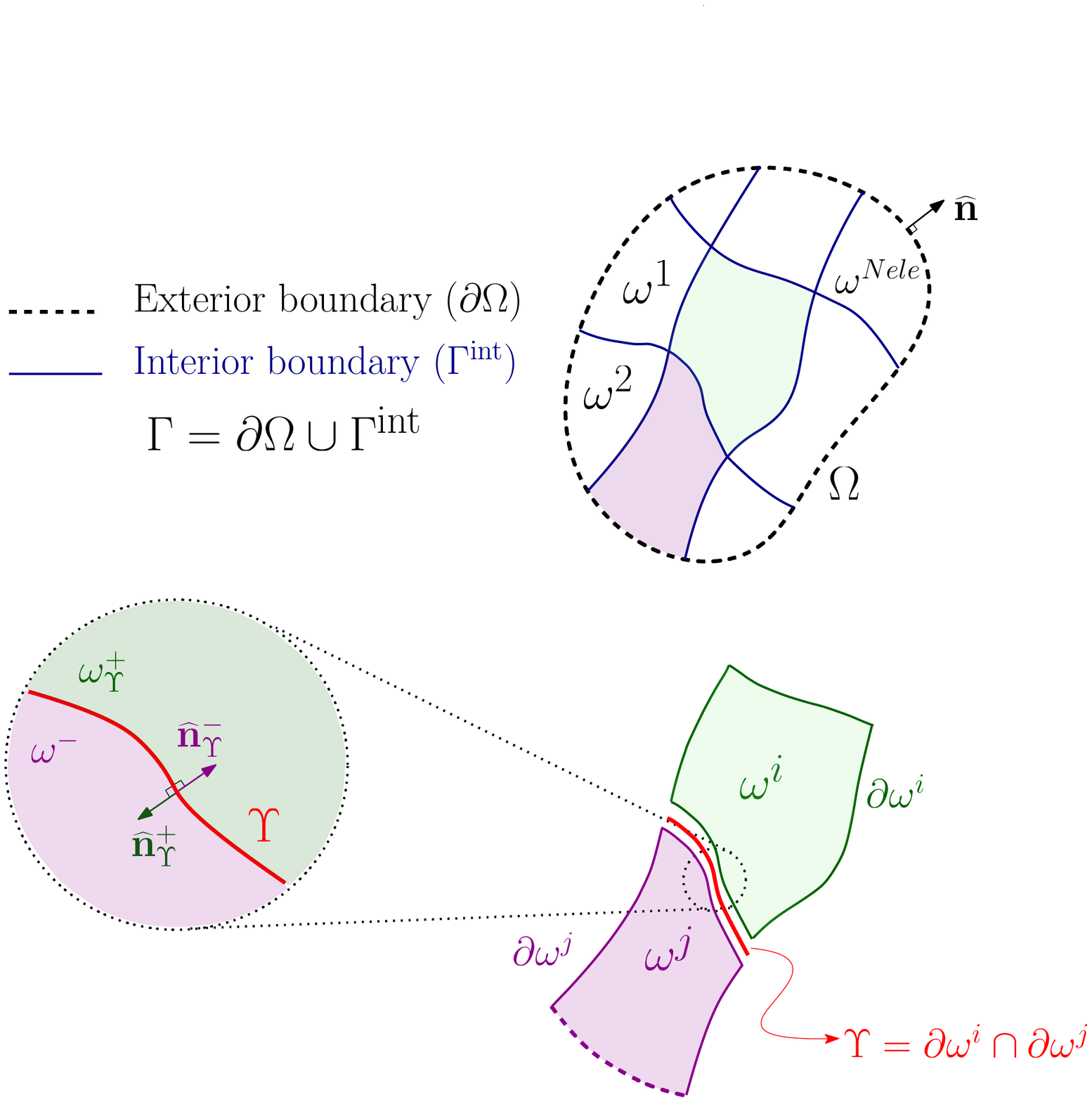}
  \caption{This figure shows the decomposition of
  the domain into subdomains (finite elements).
  External ($\partial \Omega$) and internal
  ($\Gamma^{\mathrm{int}}$) boundaries of the
  domain, the shared interface ($\Upsilon$)
  between two adjacent elements, as well as
  normal vectors to the boundaries are shown.}
  \label{Fig:DG_domain_decomposition}
\end{figure}

\subsection{Mesh-related quantities}
\label{Subsec:DG_mesh_related}
We denote the \emph{element diameter}
(i.e., the length of the largest edge)
of $\omega \in \mathcal{T}$ by $h_{\omega}$. 
The maximum element diameter in a given mesh
is referred to as \emph{the mesh-size} and
is denoted by:
\begin{align}
  \label{Eqn:DG_mesh_size}
  h := \max_{\omega \in \mathcal{T}} h_{\omega}
\end{align}
We denote the \emph{diameter of the inscribed
  circle} in $\omega \in \mathcal{T}$ by
$h_{\omega}^{\mathrm{inc}}$ (see
\textbf{Fig.~\ref{Fig:DPP_element_diameter}}). 
For an internal edge $\Upsilon \in
\mathcal{E}^{\mathrm{int}}$, shared by 
elements $\omega^{+}_{\Upsilon}$ and
$\omega^{-}_{\Upsilon}$, we define the
characteristic length $h_{\Upsilon}$
as follows:
\begin{align}
  \label{Eqn:DPP_h_Upsilon}
  h_{\Upsilon} = \frac{1}{2}
  \left(h_{\omega^{+}_{\Upsilon}} + 
  h_{\omega^{-}_{\Upsilon}}\right)
\end{align}
For an external edge $\Upsilon \in \mathcal{E}
\setminus \mathcal{E}^{\mathrm{int}}$, $h_{\Upsilon}$
is set to be equal to the element diameter
of the element containing the edge $\Upsilon$. 

\begin{figure}
  \includegraphics[scale=0.62]{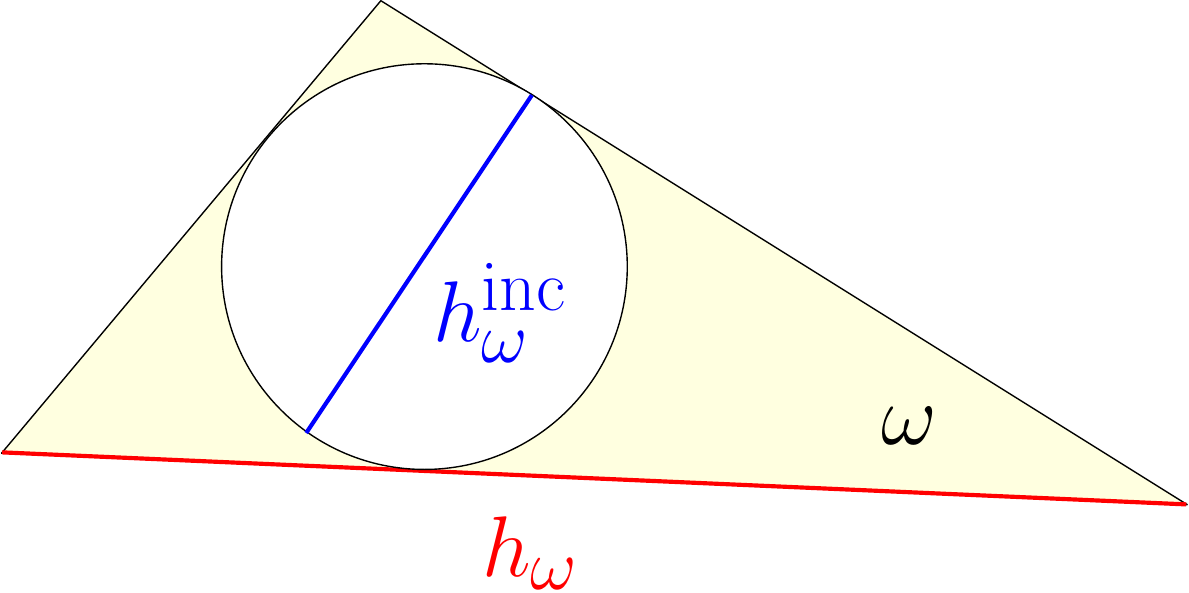}
  \caption{This figure illustrates the element diameter
    parameter $h_{\omega}$ and the diameter of the
    inscribed circle $h_{\omega}^{\mathrm{inc}}$ for
    a typical element $\omega \in \mathcal{T}$.}
  \label{Fig:DPP_element_diameter}
\end{figure}

We place two restrictions on a mesh, and
we refer to a mesh satisfying these
two restrictions as an \emph{admissible
  mesh}.
\begin{enumerate}[(i)]
\item The mesh is \emph{shape regular}
  \citep{braess2007finite}, which means
  that there exists a constant number
  $\mathcal{C}_{\mathrm{sp}}$ such that 
  \begin{align}
    \label{Eqn:shape_parameter_condition}
    \mathcal{C}_{\mathrm{sp}} h_{\omega}
    \leq h_{\omega}^{\mathrm{inc}} 
    \quad \forall \omega \in \mathcal{T}
  \end{align}
  The constant $\mathcal{C}_{\mathrm{sp}}$
  is commonly referred to as the
  \emph{shape parameter}. 
\item The mesh is 
  \emph{locally quasi-uniform}, which
  also goes by the name \emph{contact
    regularity} \citep{Dolejsi_Feistauer}.
  This condition requires that the element
  diameters of any two neighboring elements
  obey an equivalence relation.
  That is, there exists a constant
  number $\mathcal{C}_{\mathrm{lqu}} > 0$
  such that
  \begin{align}
    \label{Eqn:lqu_condition}
    \frac{1}{\mathcal{C}_{\mathrm{lqu}}}
    h_{\omega^{+}_{\Upsilon}} \leq
    h_{\omega^{-}_{\Upsilon}} \leq
    \mathcal{C}_{\mathrm{lqu}} h_{\omega^{+}_{\Upsilon}}
    \quad \forall \Upsilon
    \in \mathcal{E}^{\mathrm{int}} 
  \end{align}
  The ordering of the neighboring
  elements (i.e., which element is
  ``$+$'' and which one is ``$-$'') 
  in the above inequality is arbitrary.
  This means that the above inequality
  holds even if $\omega_{\Upsilon}^{+}$
  and $\omega_{\Upsilon}^{-}$ are
  interchanged.
  The locally quasi-uniform condition
  implies the following useful bound: 
  \begin{subequations}
    \label{Eqn:DPP_useful_bound_on_h}
    \begin{align}
    \frac{1}{2} \left(1 + \frac{1}{\mathcal{C}_{\mathrm{lqu}}}
    \right) h_{\omega^{+}_{\Upsilon}} 
    \leq h_{\Upsilon} \leq
    \frac{1}{2} \left(1 + \mathcal{C}_{\mathrm{lqu}}\right)
    h_{\omega^{+}_{\Upsilon}} \quad \forall \Upsilon
    \in \mathcal{E}^{\mathrm{int}} \\
    \frac{1}{2} \left(1 + \frac{1}{\mathcal{C}_{\mathrm{lqu}}}
    \right) h_{\omega^{-}_{\Upsilon}} 
    \leq h_{\Upsilon} \leq
    \frac{1}{2} \left(1 + \mathcal{C}_{\mathrm{lqu}}\right)
    h_{\omega^{-}_{\Upsilon}} \quad \forall \Upsilon
    \in \mathcal{E}^{\mathrm{int}} 
    \end{align}
  \end{subequations}
\end{enumerate}

A mesh $\mathcal{T}$ with mesh-size
$h$ will be denoted by $\mathcal{T}_{h}$.
A sequence of meshes will be denoted by
$\mathcal{T}_{\mathcal{H}}$, where $\mathcal{H}
= (0,\bar{h}$).
$\mathcal{T}_{\mathcal{H}}$ is said to be
an admissible sequence of meshes if
$\mathcal{T}_{h}$ is admissible for
every $h \in \mathcal{H}$. 

\begin{remark}
  There are other notions of characteristic
  mesh sizes which are employed for DG methods.
  For example, an element length scale has
  been employed in \citep{Hughes_Masud_Wan_2006},
  which takes the following
  form under our notation:
  \begin{align}
    \widehat{h} = \frac{\mathrm{meas}(\omega^{+}_{\Upsilon})
      + \mathrm{meas}
      (\omega^{-}_{\Upsilon})}{2\;\mathrm{meas}(\Upsilon)}
  \end{align}
  where $\mathrm{meas}(\cdot)$
  denotes the measure of a set. 
  A good discussion on various mesh-based
  characteristic lengths can be found in
  \citep{Dolejsi_Feistauer}.
\end{remark}

\subsection{Functional analysis aspects}
\label{Subsec:DG_functional_analysis}
We introduce the following broken Sobolev
spaces (which are piece-wise discontinuous
spaces): 
\begin{subequations}
  \label{Eqn:DG_function_spaces}
  \begin{align}
    \label{Eqn:DG_function_space_U}
  \mathcal{U} &:= \left\{\mathbf{u}(\mathbf{x}) \; \big| \;
  \mathbf{u}(\mathbf{x})\big|_{\omega^i} \in \left(L_{2}(\omega^i)
  \right)^{nd}; \; \mathrm{div}[\mathbf{u}] \in L_{2}(\omega^i); \;
  i = 1, \cdots, Nele \right\} \\
  \widetilde{\mathcal{P}} &:= \left\{p(\mathbf{x}) \; \big| \;
  p(\mathbf{x})\big|_{\omega^i} \in L_{2}(\omega^i); \;
  i = 1, \cdots, Nele \right\} \\
    \widetilde{\mathcal{Q}} &:= \left\{p(\mathbf{x}) \; \big| \;
    p(\mathbf{x})\big|_{\omega^i} \in H^{1}(\omega^i); \;
  i = 1, \cdots, Nele \right\} \\
  \mathcal{P} &:= \left\{(p_1(\mathbf{x}),p_2(\mathbf{x}))
  \in \widetilde{\mathcal{P}} \times \widetilde{\mathcal{P}} 
  \; \big| \;
  \left(\int_{\widetilde{\Omega}} p_{1}(\mathbf{x}) \mathrm{d} \Omega\right)
  \left(\int_{\widetilde{\Omega}} p_{2}(\mathbf{x}) \mathrm{d} \Omega\right)
  = 0 \right\} \\
  \mathcal{Q} &:= \left\{(p_1(\mathbf{x}),p_2(\mathbf{x}))
  \in \widetilde{\mathcal{Q}} \times \widetilde{\mathcal{Q}} 
  \; \big| \;
  \left(\int_{\widetilde{\Omega}} p_{1}(\mathbf{x}) \mathrm{d} \Omega\right)
  \left(\int_{\widetilde{\Omega}} p_{2}(\mathbf{x}) \mathrm{d} \Omega\right)
  = 0 \right\} \label{Eqn:Q}
\end{align}
\end{subequations}
where $L_2(\omega^{i})$ denotes the set of
all square-integrable functions defined on
$\omega^{i}$, and $H^1(\omega^{i})$ is a standard Sobolev
space \citep{Evans_PDE}.

\begin{remark}
  The following condition in $\mathcal{P}$
  and $\mathcal{Q}$ spaces (which is expressed
  in terms of the mean pressures in the two
  pore-networks): 
  \begin{align*}
    \left(\int_{\widetilde{\Omega}} p_1(\mathbf{x})
    \mathrm{d} \Omega \right)
    \left(\int_{\widetilde{\Omega}} p_2(\mathbf{x})
    \mathrm{d} \Omega \right) = 0 
  \end{align*}
  is one of the ways to fix the datum for
  the pressure. However, this condition is
  seldom employed in a numerical implementation.
  Alternatively, one can prescribe the pressure
  on a portion of the boundary in one of the
  pore-networks. For further details refer
  to \citep{Nakshatrala_Joodat_Ballarini_P2}.
\end{remark}
  
We denote the standard $L_2$ inner-product over
a set $\mathcal{K}$ by $(\cdot;\cdot)_{\mathcal{K}}$.
That is, 
\begin{align}
  (a;b)_{\mathcal{K}} := \int_{\mathcal{K}}
  a \cdot b \; \mathrm{d} \mathcal{K}
\end{align}
and the associated standard $L_2$ norm
is denoted by $\|\cdot\|_{\mathcal{K}}$
as follows: 
\begin{align}
  \|a\|_{\mathcal{K}} = \sqrt{(a;a)_{\mathcal{K}}}
\end{align}
The subscript in the $L_2$ inner-product
and the associated norm will be dropped
if $\mathcal{K} = \widetilde{\Omega}$. 

In a subsequent section on the interpolation error,
we employ a general order Sobolev semi-norm. To this
end, let $\alpha = (\alpha_1, \cdots, \alpha_{nd}) \in
\mathbb{N}^{nd}$ be a $nd$-tuple (i.e., multi-index),
the order of which is denoted by $|\alpha| :=
\sum_{i=1}^{nd} \alpha_i$.
We denote the multi-index (classical or
distributional) partial derivative by
$D^{\alpha}(\cdot)$.
For a scalar function $\varphi(\mathbf{x}) \in
C^{\infty}_{c}(\mathcal{K})$ (which is a set of
infinitely differentiable functions with compact
support in $\mathcal{K}$) \citep{Evans_PDE}, the
multi-index (classical) partial derivative with
respect to a given coordinate system $\mathbf{x}
= (x_1, \cdots, x_{nd})$ is defined as follows:
\begin{align}
  D^{\alpha} \varphi(\mathbf{x}) :=
  \frac{\partial^{|\alpha|} \varphi(\mathbf{x})}{\partial x_{1}^{\alpha_1}
    \partial x_{2}^{\alpha_2} 
    \cdots \partial x_{nd}^{\alpha_{nd}}}
\end{align}
Then, the multi-index distributional partial
derivative of a scalar field $a : \mathcal{K}
\rightarrow \mathbb{R}$ is defined as follows:
\begin{align}
  \left(D^{\alpha} a(\mathbf{x});\varphi(\mathbf{x})\right)_{\mathcal{K}}
  := (-1)^{|\alpha|}
  \left(a(\mathbf{x}); D^{\alpha}\varphi(\mathbf{x})\right)_{\mathcal{K}}
  \quad \forall \varphi(\mathbf{x}) \in C^{\infty}_{c}(\mathcal{K})
\end{align}
For a scalar field $a:\mathcal{K} \rightarrow
\mathbb{R}$, the $s$-th order Sobolev semi-norm
over $\mathcal{K}$ is defined as follows:
\begin{align}
  \label{Eqn:Sobolev_norm_Scalar}
  |a|_{H^{s}(\mathcal{K})} := \left( \sum_{|\alpha| = s}
  \|D^{\alpha}a(\mathbf{x}) \|^2_{\mathcal{K}} \right)^{1/2}
\end{align}
and for a vector field $\mathbf{a}:\mathcal{K}
\rightarrow \mathbb{R}^{nd}$ with scalar components
$a_i~(i=1, \cdots, nd)$, the corresponding semi-norm
is defined as follows:
\begin{align}
  |\mathbf{a}|_{H^{s}(\mathcal{K})} :=
  \left( \sum_{i = 1}^{nd} |a_i|^2_{H^s(\mathcal{K})}\right)^{1/2}
  \label{Eqn:Sobolev_norm_Vector}
\end{align}

\subsubsection{Inverse and trace
  inequalities\footnote{For these results we assume that the velocity
fields belong to $(H^{1}(\omega))^{nd}$
instead of $H(\mathrm{div},\omega)$, which
was the case in the function space
\eqref{Eqn:DG_function_space_U}. The reason
is that one has to deal with half-Sobolev
spaces and corresponding dual spaces (i.e.,
negative half-spaces) for trace inequalities
under $H(\mathrm{div})$; which makes the
convergence and error analyses more involved.
Moreover, the authors are not aware of any 
discrete trace inequalities available
in the mathematical analysis literature
that can be easily used under half-Sobolev
spaces.}}
The inequalities given below play a crucial role
in obtaining bounds on the error due to terms
defined on the element interface. Mathematical
proofs to these estimates can be found in
\citep{verfurth2013posteriori,arnold1982interior,
  Dolejsi_Feistauer,di2011mathematical}.

\begin{lemma}{(Continuous trace inequality.)}
  \label{Lemma:DDP_continuous_trace_inequality}
  For an admissible mesh $\mathcal{T}_{h}$, the
  following estimates hold $\forall \omega \in
  \mathcal{T}_{h}$:
  \begin{align}
    \label{Eqn:Trace_inequality_exp} 
    &\|v\|_{\partial \omega} \leq
    \mathcal{C}_{\mathrm{trace}} \left(\frac{1}{\sqrt{h}_{\omega}} \| v\|_{\omega}
    + \sqrt{h_{\omega}} \; \| \mathrm{grad}[v]\|_{\omega}\right)
    \quad \forall v(\mathbf{x}) \in H^1(\omega) \\
    \label{Eqn:Trace_inequality_exp2}
    &\| \mathbf{v}\|_{\partial \omega} \leq \mathcal{C}_{\mathrm{trace}}
    \left(\frac{1}{\sqrt{h_{\omega}}} \| \mathbf{v}\|_{\omega}
    + \sqrt{h_{\omega}}\;\|\mathrm{grad}[\mathbf{v}]\|_{\omega}\right) 
    \quad \forall \mathbf{v}(\mathbf{x}) \in
    (H^1(\omega))^{nd}
  \end{align}
  where the $\mathcal{C}_{\mathrm{trace}}$ depends on
  the shape parameter (i.e., $\mathcal{C}_{\mathrm{sp}}$)
  and the number of spatial dimensions ($nd$) but it is 
  \emph{not} dependent on $h_{\omega}$.
\end{lemma}

Let $\mathscr{P}^{m}(\omega)$ denote
the set of all polynomials up to and
including $m$-th order over $\omega
\in \mathcal{T}_{h}$. We then have
the following discrete inequalities. 

\begin{lemma}{(Discrete inverse inequality.)}
  \label{Lemma:DDP_discrete_inverse_inequality}
  Let $\mathcal{T}_{h}$ be an admissible mesh. Then
  the following estimates hold $\forall \omega \in
  \mathcal{T}_h$:
  \begin{align}
    \label{Eqn:DDP_discrete_inverse_inequality}
    &\|\mathrm{grad}[v^{h}]\|_{\omega} \leq
    \mathcal{C}_{\mathrm{inv}} h_{\omega}^{-1}
    \|v^{h}\|_{\omega}
    \quad \forall v^{h}(\mathbf{x}) \in
    H^{1}(\omega) \cap \mathscr{P}^{m}(\omega) \\
    &\|\mathrm{grad}[\mathbf{v}^{h}]\|_{\omega} \leq
    \mathcal{C}_{\mathrm{inv}} h_{\omega}^{-1}
    \|\mathbf{v}^{h}\|_{\omega}
    \quad \forall \mathbf{v}^{h}(\mathbf{x}) \in
    \left(H^{1}(\omega)\right)^{nd} \cap
    \left(\mathscr{P}^{m}(\omega)\right)^{nd}
  \end{align}
  where $\mathcal{C}_{\mathrm{inv}}$ is a constant
  dependent on the shape parameter
  ($\mathcal{C}_{\mathrm{sp}}$), the number of
  spatial dimensions ($nd$) and the polynomial
  order ($m$), but it does not depend on $h_{\omega}$
  or on the fields $v^{h}(\mathbf{x})$ and
  $\mathbf{v}^{h}(\mathbf{x})$.
\end{lemma}

\begin{lemma}{(Discrete trace inequality.)}
  \label{Lemma:DDP_discrete_trace_inequality}
  For an admissible mesh $\mathcal{T}_{h}$,
  the following estimates hold $\forall
  \omega \in \mathcal{T}_{h}$:
  \begin{align}
    \label{Eqn:DPP_dti_scalar} 
    &\|v^{h}\|_{\partial \omega} \leq
    \mathcal{C}_{\mathrm{trace}} \left(1 + \mathcal{C}_{\mathrm{inv}}\right)
    \frac{1}{\sqrt{h}_{\omega}} \|v^{h}\|_{\omega}
    \quad \forall v^{h}(\mathbf{x}) \in H^1(\omega)
    \cap \mathscr{P}^{m}(\omega) \\
    \label{Eqn:DPP_dti_vector} 
    &\|\mathbf{v}^{h}\|_{\partial \omega} \leq
    \mathcal{C}_{\mathrm{trace}} \left(1 + \mathcal{C}_{\mathrm{inv}}\right)
    \frac{1}{\sqrt{h}_{\omega}} \|\mathbf{v}^{h}\|_{\omega}
    \quad \forall \mathbf{v}^{h}(\mathbf{x}) \in
    \left(H^1(\omega)\right)^{nd} \cap
    \left(\mathscr{P}^{m}(\omega)\right)^{nd}
  \end{align}
\end{lemma}

\section{A STABILIZED MIXED DG FORMULATION}
\label{Sec:S3_DG_Mixed}
We propose a stabilized \emph{four-field} formulation
for the DPP model. The proposed formulation
draws its inspiration from the stabilized \emph{two-field}
formulations proposed by
\citep{Hughes_Masud_Wan_2006,
  badia2010stabilized} for Darcy equations,
which describe the flow of an incompressible
fluid through a porous medium with a single
pore-network.

\subsection{Weak form in terms of numerical fluxes}
Multiplying the governing equations
\eqref{Eqn:DG_GE_Darcy_BLM_1}--\eqref{Eqn:DG_GE_Darcy_mass_balance_2}
by weighting functions, integrating over an element
$\omega$, and using equation \eqref{Eqn:DG_GE_Darcy_MT} and
the divergence theorem, we
obtain the following:
\begin{align}
  &\left( \mathbf{w}_1 ; \mu k_{1}^{-1} \mathbf{u}_{1} \right)_{\omega}
  - \left(\mathrm{div}[\mathbf{w}_{1}] ; p_{1} \right)_{\omega}
  + \left(\mathbf{w}_{1} \cdot \widehat{\mathbf{n}};
  \overset{*}{p}_{1} \right)_{\partial \omega} 
  +\left( \mathbf{w}_2 ; \mu k_{2}^{-1} \mathbf{u}_{2} \right)_{\omega}
  - \left(\mathrm{div}[\mathbf{w}_{2}] ; p_{2} \right)_{\omega}
  + \left(\mathbf{w}_{2} \cdot \widehat{\mathbf{n}};
  \overset{*}{p}_{2} \right)_{\partial \omega} \nonumber \\
  &\qquad +\left(q_1;\mathrm{div}[\mathbf{u}_1]\right)_{\omega}
  + \left(q_1;\left(\overset{*}{\mathbf{u}}_{1}
  - \mathbf{u}_1\right) \cdot \widehat{\mathbf{n}}
  \right)_{\partial \omega}
  +\left(q_2;\mathrm{div}[\mathbf{u}_2]\right)_{\omega}
  + \left(q_2;\left(\overset{*}{\mathbf{u}}_{2}
  - \mathbf{u}_2\right) \cdot \widehat{\mathbf{n}}
  \right)_{\partial \omega} \nonumber \\
  &\qquad + \left(q_1 - q_2; \frac{\beta}{\mu}(p_{1} - p_{2})\right)_{\omega} 
  = \left( \mathbf{w}_{1} ; \gamma \mathbf{b}_1 \right)_{\omega}
  + \left( \mathbf{w}_{2} ; \gamma \mathbf{b}_2 \right)_{\omega}
  \label{Eqn:Weakform}
\end{align}
where $\overset{*}{p}_1$ and $\overset{*}{p}_2$
are the numerical fluxes for the pressures
and $\overset{*}{\mathbf{u}}_1$ and
$\overset{*}{\mathbf{u}}_2$ are the
numerical fluxes for the velocities.
Summing the above equation over all
the elements and using the identity
\eqref{Eqn:DG_jump_avg_identity}, we
obtain the following weak form in
terms of numerical fluxes:
{\small
\begin{align}
  &\left( \mathbf{w}_1 ; \mu k_{1}^{-1} \mathbf{u}_{1} \right)
  - \left(\mathrm{div}[\mathbf{w}_{1}] ; p_{1} \right)
  + \left(\{\!\!\{\mathbf{w}_{1} \}\!\!\};
  \llbracket\overset{*}{p}_{1}\rrbracket \right)_{\Gamma^{\mathrm{int}}}
    + \left(\llbracket\mathbf{w}_{1} \rrbracket;
    \{\!\!\{\overset{*}{p}_{1}\}\!\!\}
    \right)_{\Gamma^{\mathrm{int}}}
  + \left(\mathbf{w}_{1} \cdot \widehat{\mathbf{n}};
  \overset{*}{p}_{1} \right)_{\partial \Omega}
  \nonumber \\
  &+\left( \mathbf{w}_2 ; \mu k_{2}^{-1} \mathbf{u}_{2} \right)
  - \left(\mathrm{div}[\mathbf{w}_{2}] ; p_{2} \right)
    + \left(\{\!\!\{\mathbf{w}_{2} \}\!\!\};
  \llbracket\overset{*}{p}_{2}\rrbracket \right)_{\Gamma^{\mathrm{int}}}
    + \left(\llbracket\mathbf{w}_{2} \rrbracket;
    \{\!\!\{\overset{*}{p}_{2}\}\!\!\}
    \right)_{\Gamma^{\mathrm{int}}}
  + \left(\mathbf{w}_{2} \cdot \widehat{\mathbf{n}};
  \overset{*}{p}_{2} \right)_{\partial \Omega} \nonumber \\
  &+\left(q_1;\mathrm{div}[\mathbf{u}_1]\right)
  + \left(\{\!\!\{q_1\}\!\!\};
  \llbracket\overset{*}{\mathbf{u}}_{1}\rrbracket
  - \llbracket \mathbf{u}_1\rrbracket
  \right)_{\Gamma^{\mathrm{int}}} 
    + \left(\llbracket q_1\rrbracket;
  \{\!\!\{\overset{*}{\mathbf{u}}_{1}\}\!\!\}
  - \{\!\!\{ \mathbf{u}_1\}\!\!\}
  \right)_{\Gamma^{\mathrm{int}}} 
  + \left(q_1;\left(\overset{*}{\mathbf{u}}_{1}
  - \mathbf{u}_1\right) \cdot \widehat{\mathbf{n}}
  \right)_{\partial \Omega} \nonumber \\
  &+\left(q_2;\mathrm{div}[\mathbf{u}_2]\right)
  + \left(\{\!\!\{q_2\}\!\!\};
  \llbracket\overset{*}{\mathbf{u}}_{2}\rrbracket
  - \llbracket \mathbf{u}_2\rrbracket
  \right)_{\Gamma^{\mathrm{int}}} 
    + \left(\llbracket q_2\rrbracket;
  \{\!\!\{\overset{*}{\mathbf{u}}_{2}\}\!\!\}
  - \{\!\!\{ \mathbf{u}_2\}\!\!\}
  \right)_{\Gamma^{\mathrm{int}}} 
  + \left(q_2;\left(\overset{*}{\mathbf{u}}_{2}
  - \mathbf{u}_2\right) \cdot \widehat{\mathbf{n}}
  \right)_{\partial \Omega} \nonumber \\
  &\; + \left(q_1 - q_2; \frac{\beta}{\mu}(p_{1} - p_{2})\right)
  = \left( \mathbf{w}_{1} ; \gamma \mathbf{b}_1 \right)
  + \left( \mathbf{w}_{2} ; \gamma \mathbf{b}_2 \right)
\end{align}}

Physically, the jumps in pressures and the normal
component of velocities should vanish on any curve
which is entirely inside the domain, and in particular,
on any interior edge. That is,
\begin{align}
  \label{Eqn:DG_Physical_jumps}
  \llbracket p_1 \rrbracket = \mathbf{0}, \;
  \llbracket p_2 \rrbracket = \mathbf{0}, \;
  \llbracket \mathbf{u}_1 \rrbracket = 0
  \; \mathrm{and} \;
  \llbracket \mathbf{u}_2 \rrbracket = 0
  \quad \mathrm{on} \; \Gamma^{\mathrm{int}}
\end{align}
Numerical fluxes are important components of DG 
methods, which have to be selected carefully.
The choice of these numerical fluxes can greatly
affect the stability of a DG formulation. Herein,
we consider the following general expressions for
the numerical fluxes:
\begin{align}
  \label{Eqn:DPP2_NF_p1}
  {\mathop{p}^{\star}}_{1} &=
  \left\{\begin{array}{ll}
   \lambda_{1}^{(1)} \{\!\!\{p_1\}\!\!\}
  + \frac{\lambda_{1}^{(2)}}{2} \llbracket p_1 \rrbracket
  \cdot \widehat{\mathbf{n}}
  + \lambda_{1}^{(3)} \llbracket \mathbf{u}_1 \rrbracket
  &\mathrm{on} \; \Gamma^{\mathrm{int}} \\
  p_1 \quad \quad \quad \quad \quad \quad \quad \quad \quad \quad \quad \quad \quad \quad \quad \quad \quad \quad    & \mathrm{on} \; \Gamma^{u}_1 \\
  p_{01} \quad \quad \quad \quad \quad \quad \quad \quad \quad \quad \quad \quad \quad \quad \quad \quad \quad \quad & \mathrm{on} \; \Gamma^{p}_1 
  \end{array} \right. \\
  {\mathop{p}^{\star}}_{2} &=
  \left\{ \begin{array}{ll}
    \lambda_{2}^{(1)} \{\!\!\{p_2\}\!\!\}
    + \frac{\lambda_{2}^{(2)}}{2} \llbracket p_2 \rrbracket
  \cdot \widehat{\mathbf{n}}
  + \lambda_{2}^{(3)} \llbracket \mathbf{u}_2 \rrbracket
  & \mathrm{on} \; \Gamma^{\mathrm{int}} \\
  p_2 \quad \quad \quad \quad \quad \quad \quad \quad \quad \quad \quad \quad \quad \quad \quad \quad \quad \quad  & \mathrm{on} \; \Gamma^{u}_2 \\
  p_{02} \quad \quad \quad \quad \quad \quad \quad \quad \quad \quad \quad \quad \quad \quad \quad \quad \quad \quad & \mathrm{on} \; \Gamma^{p}_2 
  \end{array} \right.
\end{align}
\begin{subequations}
  \begin{align}
    &{\mathop{\mathbf{u}}^{\star}}_{1} =
    \Lambda_{1}^{(1)} \{\!\!\{ \mathbf{u}_1 \}\!\!\}
    + \frac{\Lambda_{1}^{(2)}}{2} \llbracket \mathbf{u}_{1} \rrbracket
    \widehat{\mathbf{n}} 
    + \Lambda_{1}^{(3)} \llbracket p_1 \rrbracket
    \quad \mathrm{on} \; \Gamma^{\mathrm{int}} \\
    &{\mathop{\mathbf{u}}^{\star}}_{1} \cdot \widehat{\mathbf{n}} =
    \left\{
    \begin{array}{ll}
      u_{n1} & \mathrm{on} \; \Gamma_{1}^{u} \\
      \mathbf{u}_{1} \cdot \widehat{\mathbf{n}} 
      & \mathrm{on} \; \Gamma_{1}^{p}
    \end{array} \right.
  \end{align}
\end{subequations}
\begin{subequations}
    \label{Eqn:DPP2_NF_u2}
  \begin{align}
    &{\mathop{\mathbf{u}}^{\star}}_{2} =
    \Lambda_{2}^{(1)} \{\!\!\{ \mathbf{u}_2 \}\!\!\}
    + \frac{\Lambda_{2}^{(2)}}{2} \llbracket \mathbf{u}_{2} \rrbracket
  \widehat{\mathbf{n}} 
  + \Lambda_{2}^{(3)} \llbracket p_2 \rrbracket
  \quad \mathrm{on} \; \Gamma^{\mathrm{int}} \\
  &{\mathop{\mathbf{u}}^{\star}}_{2} \cdot \widehat{\mathbf{n}} =
    \left\{
    \begin{array}{ll}
      u_{n2} & \mathrm{on} \; \Gamma_{2}^{u} \\
      \mathbf{u}_{2} \cdot \widehat{\mathbf{n}} 
      & \mathrm{on} \; \Gamma_{2}^{p}
    \end{array} \right.
  \end{align}
\end{subequations}
where $\lambda_{i}^{(j)}$ and $\Lambda_{i}^{(j)}$
($i, j = 1 \; \mathrm{or} \; 2$) are constants.
It is easy to check that these numerical
fluxes satisfy the following relations
on $\Gamma^{\mathrm{int}}$:
\begin{subequations}
  \begin{alignat}{2}
    &\{\!\!\{ \mathop{p_{1}}^{\star}\}\!\!\}
    = \lambda_{1}^{(1)} \{\!\!\{p_1\}\!\!\} + \lambda_{1}^{(3)}
    \llbracket \mathbf{u}_{1} \rrbracket 
    \quad \mathrm{and} \quad
    &&\llbracket {\mathop{p}^{\star}}_{1} \rrbracket
    = \lambda_{1}^{(2)} \llbracket p_1 \rrbracket \\
    &\{\!\!\{ \mathop{p_{2}}^{\star}\}\!\!\}
    = \lambda_{2}^{(1)} \{\!\!\{p_2\}\!\!\} + \lambda_{2}^{(3)}
    \llbracket \mathbf{u}_{2} \rrbracket 
    \quad \mathrm{and} \quad
    && \llbracket {\mathop{p}^{\star}}_{2} \rrbracket
    = \lambda_{2}^{(2)} \llbracket p_2 \rrbracket \\
    &\{\!\!\{ \mathop{\mathbf{u}_{1}}^{\star}\}\!\!\} =
    \Lambda_{1}^{(1)} \{\!\!\{\mathbf{u}_1\}\!\!\} + \Lambda_{1}^{(3)}
    \llbracket p_1 \rrbracket
    \quad \mathrm{and} \quad 
    &&\llbracket {\mathop{\mathbf{u}}^{\star}}_{1} \rrbracket
    = \Lambda_{1}^{(2)} \llbracket \mathbf{u}_1 \rrbracket \\
    &\{\!\!\{ \mathop{\mathbf{u}_{2}}^{\star} \}\!\!\}=
    \Lambda_{2}^{(1)} \{\!\!\{\mathbf{u}_2\}\!\!\}
    + \Lambda_{2}^{(3)} \llbracket
    p_2 \rrbracket
    \quad \mathrm{and} \quad 
    &&\llbracket {\mathop{\mathbf{u}}^{\star}}_{2} \rrbracket
    = \Lambda_{2}^{(2)} \llbracket \mathbf{u}_2 \rrbracket 
  \end{alignat}
\end{subequations}

\subsection{The classical mixed DG formulation}
This formulation is based on the Galerkin formalism
and can be obtained by making the following choices:
\begin{align}
  \lambda_{1}^{(1)} = \lambda_{2}^{(1)}
  = \Lambda_{1}^{(1)} = \Lambda_{2}^{(1)} = 1
\end{align}
and the other constants in equations
\eqref{Eqn:DPP2_NF_p1}--\eqref{Eqn:DPP2_NF_u2}
are taken to be zeros. 
The numerical fluxes on $\Gamma^{\mathrm{int}}$
under the classical mixed DG formulation
take the following form:
\begin{align}
  \overset{*}{p}_{1} = \{\!\!\{p_1\}\!\!\}, \;
  \overset{*}{p}_{2} = \{\!\!\{p_2\}\!\!\}, \;
  \overset{*}{\mathbf{u}}_{1} = \{\!\!\{\mathbf{u}_1\}\!\!\}
  \; \mathrm{and} \;
  \overset{*}{\mathbf{u}}_{2} = \{\!\!\{\mathbf{u}_2\}\!\!\}
\end{align}
The above numerical fluxes are similar to
the ones employed by \citep{Bassi_Rebay_1997},
which are known to be consistent but do not
result in a stable DG method
\citep{Arnold_Brezzi_Cockburn_Marini_2002}. 
The corresponding weak formulation reads:~Find
$\left(\mathbf{u}_1(\mathbf{x}),
\mathbf{u}_2(\mathbf{x})\right) \in
\mathcal{U} \times \mathcal{U} $,
$\left(p_1(\mathbf{x}), p_2(\mathbf{x})
\right) \in \mathcal{P}$ such that we
have
\begin{align}
  \label{Eqn:DG_mixed_formulation}
  \mathcal{B}^{\mathrm{DG}}_{\mathrm{Gal}}
  (\mathbf{w}_1,\mathbf{w}_2,q_1,q_2;
  \mathbf{u}_1,\mathbf{u}_2,p_1,p_2)
  = \mathcal{L}^{\mathrm{DG}}_{\mathrm{Gal}}
  (\mathbf{w}_1,\mathbf{w}_2,q_1,q_2) \nonumber \\
\quad \forall \left(\mathbf{w}_1(\mathbf{x}),
  \mathbf{w}_{2}(\mathbf{x})\right) \in
  \mathcal{U} \times \mathcal{U},~
  \left(q_1(\mathbf{x}),~q_2(\mathbf{x})\right) \in \mathcal{P} 
\end{align}
where the bilinear form and the linear
functional are, respectively, defined
as follows:
\begin{subequations}
\begin{align}
  \label{Eqn:DG_bilinear_form}
  \mathcal{B}^{\mathrm{DG}}_{\mathrm{Gal}}
  &:= \left( \mathbf{w}_1 ; \mu k_{1}^{-1} \mathbf{u}_{1} \right)
  - \left(\mathrm{div}[\mathbf{w}_{1}] ; p_{1} \right)
  + \left(q_1;\mathrm{div}[\mathbf{u}_{1}]\right) 
  + \left(\llbracket \mathbf{w}_{1} \rrbracket; \{\!\!\{ p_{1} \}\!\!\}
  \right)_{\Gamma^{\mathrm{int}}}
  - \left(\{\!\!\{q_{1} \}\!\!\} ; \llbracket \mathbf{u}_{1} \rrbracket
  \right)_{\Gamma^{\mathrm{int}}}  \nonumber \\
  &\; +\left( \mathbf{w}_2 ; \mu k_{2}^{-1} \mathbf{u}_{2} \right)
  - \left(\mathrm{div}[\mathbf{w}_{2}] ; p_{2} \right)
  + \left(q_2;\mathrm{div}[\mathbf{u}_{2}]\right) 
  + \left(\llbracket \mathbf{w}_{2} \rrbracket; \{\!\!\{p_{2} \}\!\!\}
  \right)_{\Gamma^{\mathrm{int}}}
  - \left(\{\!\!\{q_{2} \}\!\!\} ; \llbracket \mathbf{u}_{2} \rrbracket
  \right)_{\Gamma^{\mathrm{int}}}  \nonumber \\
  &\; + \left(q_1 - q_2; \frac{\beta}{\mu}(p_{1} - p_{2}) \right)
  + \left(\mathbf{w}_{1} \cdot \widehat{\mathbf{n}}; p_1
  \right)_{\Gamma_{1}^{u}}
  + \left(\mathbf{w}_{2} \cdot \widehat{\mathbf{n}}; p_2
  \right)_{\Gamma_{2}^{u}} 
  -\left(q_{1};\mathbf{u}_{1} \cdot \widehat{\mathbf{n}}\right)_{\Gamma^{u}_{1}}  
  -\left(q_{2};\mathbf{u}_{2} \cdot \widehat{\mathbf{n}}\right)_{\Gamma^{u}_{2}}  \\
  \label{Eqn:DG_linear_form}
  \mathcal{L}^{\mathrm{DG}}_{\mathrm{Gal}}
  &:= \left( \mathbf{w}_{1} ; \gamma \mathbf{b}_1 \right)
  + \left( \mathbf{w}_{2} ; \gamma \mathbf{b}_2 \right)
  -(\mathbf{w}_1 \cdot \widehat{\mathbf{n}}; p_{01})_{\Gamma_{1}^{p}} 
  -(\mathbf{w}_2 \cdot \widehat{\mathbf{n}}; p_{02})_{\Gamma_{2}^{p}} 
  - \left(q_{1};u_{n1} \right)_{\Gamma^{u}_{1}}
  - \left(q_{2};u_{n2}\right)_{\Gamma^{u}_{2}} 
\end{align}
\end{subequations}
The classical mixed DG formulation is not stable
under all combinations of interpolation functions
for the field variables, which is due to the
violation of the LBB \emph{inf-sup} stability
condition \citep{Brezzi_Fortin}. Specifically,
equal-order interpolation for all the field
variables is not stable under the classical
mixed DG formulation.
This numerical instability (due to the interpolation
functions) is different from the aforementioned
instability due to the numerical fluxes (i.e.,
Bassi-Rebay DG method). 
We develop a stabilized mixed DG formulation
which does not suffer from any of the aforementioned
instabilities. This is achieved by adding
adjoint-type, residual-based stabilization
terms (which are defined over the subdomains
and circumvent the LBB \emph{inf-sup} stability
condition) and by incorporating appropriate
numerical fluxes (which are consistent and stable
and are defined along the subdomain interfaces).

\subsection{Proposed stabilized mixed DG formulation} 
This formulation makes the following choices:
\begin{subequations}
\label{Eqn:DG_stab_lambda_choices}
  \begin{align}
    &\lambda_{1}^{(1)} = \lambda_{2}^{(1)} = 1, \;
    \lambda_{1}^{(3)} = \eta_u h_{\Upsilon} 
    \{\!\!\{ \mu k_1^{-1} \}\!\!\}
    \; \mathrm{and} \;
    \lambda_{2}^{(3)} = \eta_u h_{\Upsilon} 
    \{\!\!\{ \mu k_2^{-1} \}\!\!\} \\
    &\Lambda_{1}^{(1)} = \Lambda_{2}^{(1)} = 1, \;
    \Lambda_{1}^{(3)} = \frac{\eta_p}{h_{\Upsilon}}
    \{\!\!\{ \mu^{-1} k_1 \}\!\!\}
    \; \mathrm{and} \;
    \Lambda_{2}^{(3)} = \frac{\eta_p}{h_{\Upsilon}}
    \{\!\!\{ \mu^{-1} k_2 \}\!\!\}
  \end{align}
\end{subequations}
and the other constants in equations
\eqref{Eqn:DPP2_NF_p1}--\eqref{Eqn:DPP2_NF_u2} 
are taken to be zero. $\eta_{u}$ and $\eta_{p}$ are
non-negative, non-dimensional bounded constants.
The corresponding numerical fluxes 
on $\Gamma^{\mathrm{int}}$ take the
following form:
\begin{align}
\label{Eqn:DG_stab_numerical_flux_choices}
  &\overset{*}{p}_{1} = \{\!\!\{p_{1}\}\!\!\}
  + \eta_{u} h_{\Upsilon} \{\!\!\{\mu k_1^{-1}\}\!\!\}
  \llbracket \mathbf{u}_1 \rrbracket, \; 
  \overset{*}{p}_{2} = \{\!\!\{p_{2}\}\!\!\}
  + \eta_{u} h_{\Upsilon} \{\!\!\{\mu k_2^{-1}\}\!\!\}
  \llbracket \mathbf{u}_2
  \rrbracket, \; \nonumber \\
  &\qquad \overset{*}{\mathbf{u}}_{1}
  = \{\!\!\{\mathbf{u}_{1}\}\!\!\}
  + \frac{\eta_{p}}{h_{\Upsilon}}
  \{\!\!\{\mu^{-1} k_1\}\!\!\}
  \llbracket p_1 \rrbracket 
  \; \mathrm{and} \; 
  \overset{*}{\mathbf{u}}_{2} = \{\!\!\{\mathbf{u}_{2}\}\!\!\}
  + \frac{\eta_{p}}{h_{\Upsilon}} \{\!\!\{\mu^{-1} k_2\}\!\!\}
  \llbracket p_2 \rrbracket 
\end{align}

\begin{tcolorbox}[breakable]
  The mathematical statement of
  the proposed stabilized mixed
  DG formulation reads as follows:~Find
  $\left(\mathbf{u}_1(\mathbf{x}),\mathbf{u}_2(\mathbf{x})\right)
  \in \mathcal{U} \times \mathcal{U}$, $\left(p_1(\mathbf{x}),
  p_2(\mathbf{x})\right) \in \mathcal{Q}$ such that we have
  \begin{align}
    \mathcal{B}_{\mathrm{stab}}^{\mathrm{DG}}
    (\mathbf{w}_1,\mathbf{w}_2,q_1,q_2;
    \mathbf{u}_1,\mathbf{u}_2,p_1,p_2)
    = \mathcal{L}_{\mathrm{stab}}^{\mathrm{DG}}
    (\mathbf{w}_1,\mathbf{w}_2,q_1,q_2) \nonumber \\
    \quad \forall \left(\mathbf{w}_1(\mathbf{x}),
    \mathbf{w}_2(\mathbf{x})\right) \in
    \mathcal{U} \times \mathcal{U},~
    \left(q_1(\mathbf{x}),q_2(\mathbf{x})\right)
    \in \mathcal{Q}
    \label{Eqn:VMS_DG_Weak_Form} 
  \end{align}
  where the bilinear form and the linear
  functional are, respectively, defined
  as follows:
  \begin{subequations} 
\label{sub:Stabilized_mixed}
    \begin{alignat}{2}
      \mathcal{B}^{\mathrm{DG}}_{\mathrm{stab}}
      := \mathcal{B}^{\mathrm{DG}}_{\mathrm{Gal}} 
      &-\frac{1}{2} \left(\mu k_1^{-1}
      \mathbf{w}_1 - \mathrm{grad}[q_1];
      \mu^{-1} k_1 (\mu k^{-1}_1
      \mathbf{u}_1 + \mathrm{grad}[p_1])\right)
      \nonumber \\
      &-\frac{1}{2} \left(\mu k_2^{-1}
      \mathbf{w}_2 - \mathrm{grad}[q_2];
      \mu^{-1} k_2(\mu k^{-1}_2 \mathbf{u}_2 + \mathrm{grad}[p_2])\right) 
      \nonumber\\
      &+ \left(\eta_{u} h_{\Upsilon} \{\!\!\{\mu k_1^{-1}\}\!\!\}
      \llbracket \mathbf{w}_{1} \rrbracket ; \llbracket
      \mathbf{u}_1 \rrbracket
      \right)_{\Gamma^{\mathrm{int}}}
      + \left(\eta_{u} h_{\Upsilon} \{\!\!\{\mu k_2^{-1}\}\!\!\}
      \llbracket \mathbf{w}_{2} \rrbracket ; \llbracket
      \mathbf{u}_2 \rrbracket\right)_{\Gamma^{\mathrm{int}}} \nonumber \\
      &+ \left(\frac{\eta_{p}}{h_{\Upsilon}} \{\!\!\{\mu^{-1} k_1\}\!\!\}
      \llbracket q_{1} \rrbracket ; \llbracket p_1 \rrbracket
      \right)_{\Gamma^{\mathrm{int}}}
      + \left(\frac{\eta_{p}}{h_{\Upsilon}} \{\!\!\{\mu^{-1} k_2\}\!\!\}
      \llbracket q_{2} \rrbracket ; \llbracket p_2
      \rrbracket\right)_{\Gamma^{\mathrm{int}}}  \label{Eqn:VMS_DG_bilinear_form}  \\
      \mathcal{L}_{\mathrm{stab}}^{\mathrm{DG}} := \mathcal{L}^{\mathrm{DG}}_{\mathrm{Gal}}
      &-\frac{1}{2} \left(\mu k_1^{-1} \mathbf{w}_1 - \mathrm{grad}[q_1];
      \mu^{-1} k_1 \gamma \mathbf{b}_{1}\right)
      -\frac{1}{2} \left(\mu k_2^{-1} \mathbf{w}_2 - \mathrm{grad}[q_2];
      \mu^{-1} k_2 \gamma \mathbf{b}_{2}\right) \label{Eqn:VMS_DG_linear}
    \end{alignat}
  \end{subequations}

To completely define the formulation, the
parameters $\eta_u$ and $\eta_p$ have to
be prescribed. We make the following
recommendation, which is based on the
theoretical convergence analysis (see
$\mathsection\ref{Sec:S4_DG_Error}$)
and extensive numerical simulations
(see $\mathsection\ref{Sec:S5_DG_Patch_tests}$--$\mathsection\ref{Sec:S8_DG_NR}$):
\begin{enumerate}[(i)]
\item For conforming approximations, the
  parameters can be taken to be $\eta_u
  = \eta_p = 0$.
\item For non-conforming approximations,
  the parameters can be taken to be
  $\eta_{u}=\eta_{p}=10$ or $100$. (See
  $\mathsection~\ref{Sec5:2D_square}$).
\end{enumerate}
\end{tcolorbox}

A few remarks about the stabilized formulation
are in order. 
\begin{enumerate}[(a)]
\item The above stabilized formulation
  is an adjoint-type formulation. We
  have posed even
  the classical mixed formulation as an
  adjoint-type (see the bilinear
  form \eqref{Eqn:DG_bilinear_form}). 
  In addition, the stabilization terms within the
  elements (i.e., in $\widetilde{\Omega}$)
  are of adjoint-type, which look similar
  to the one proposed by
  \citep{Hughes_Masud_Wan_2006} for the
  {\color{red} two-field} Darcy equations.
\item Since the formulation is of adjoint-type,
  the formulation will not give rise to
  symmetric coefficient (``stiffness'') matrix.
  But the coefficient matrix will be positive
  definite, which can be inferred from 
  Lemma \ref{Lemma:DD2_stab_norm}.
  Alternatively, the above stabilized formulation
  can be posed as an equivalent symmetric
  formulation by replacing $q_1$ and $q_2$ with
  $-q_1$ and $-q_2$, respectively; which is
  justified as $q_1$ and $q_2$ are arbitrary
  weighting functions. In this case, the
  resulting symmetric
  formulation will not result in positive-definite
  coefficient matrix. 
\item  In order to minimize the drift in the
  solution fields, especially in the case of
  non-conforming discretization,
  additional stabilization terms on the interior
  boundaries (i.e.,
  terms containing $\eta_{u}$ and $\eta_{p}$)
  are required in both networks. The necessity
  of employing such stabilization parameters
  has been addressed by \citep{badia2010stabilized}
  for the case of Darcy equations. It is noteworthy
  that $\eta_{u}$ parameter was not included
  in the formulation proposed by
  \citep{Hughes_Masud_Wan_2006}, as they did not
  consider non-conforming approximations.
  \item Due to the presence of the terms
  containing $\eta_{u}$ and $\eta_{p}$, the
  above numerical fluxes are no longer
  similar to the ones proposed by
  \citep{Bassi_Rebay_1997}. The numerical
  fluxes employed in the proposed formulation
  are not the same as any of the DG methods
  discussed in the review paper
  \citep{Arnold_Brezzi_Cockburn_Marini_2002}.
\item In the case of Darcy equations,
  a stabilized formulation without
  edge stabilization terms has been
  developed and its convergence has
  been established by utilizing a
  lifting operator \citep{Brezzi_Hughes_Marini_Masud_SIAMJSC_2005_v22_p119}.
  The question about whether such an approach
  can be extended to the DPP model is worthy
  of an investigation, but is beyond the scope
  of this paper.
\end{enumerate}

\section{A THEORETICAL ANALYSIS OF THE PROPOSED DG FORMULATION}
\label{Sec:S4_DG_Error}
We start by grouping the field
variables and their corresponding
weighting functions as follows:
\begin{subequations}
  \begin{align}
    \label{Eqn:VMS_DG_Uexact}
    \mathbf{U} &= (\mathbf{u}_1(\mathbf{x}),
    \mathbf{u}_2(\mathbf{x}),p_1(\mathbf{x}),
    p_2(\mathbf{x})) \in \mathbb{U} \\
    \label{Eqn:VMS_DG_Wexact}
    \mathbf{W} &= (\mathbf{w}_1(\mathbf{x}),
    \mathbf{w}_2(\mathbf{x}),q_1(\mathbf{x}),
    q_2(\mathbf{x})) \in \mathbb{U}
  \end{align}
\end{subequations}
where the product space $\mathbb{U}$
is defined as follows:  
\begin{align}
  \mathbb{U} = \mathcal{U} \times
  \mathcal{U} \times \mathcal{Q} 
\end{align}
The proposed stabilized mixed DG formulation
\eqref{Eqn:VMS_DG_Weak_Form} can then be
compactly written as follows:~Find $\mathbf{U}
\in \mathbb{U}$ such that we have
\begin{align}
  \label{Eqn:VMS_DG_Weak_Form_compact} 
  \mathcal{B}_{\mathrm{stab}}^{\mathrm{DG}}(\mathbf{W};
  \mathbf{U})
  = \mathcal{L}_{\mathrm{stab}}^{\mathrm{DG}}(\mathbf{W})
  \quad \forall \mathbf{W} \in \mathbb{U} 
\end{align}
The stability of the proposed weak formulation 
will be established under the following norm:
{\small
  \begin{align}
    \label{Eqn:DG_stability_norm}
    \left(\|\mathbf{W}\|_{\mathrm{stab}}^{\mathrm{DG}}\right)^{2}
    := \mathcal{B}_{\mathrm{stab}}^{\mathrm{DG}}(\mathbf{W};\mathbf{W})
    &= \frac{1}{2}\left\|\sqrt{\frac{\mu}{k_1}}
    \mathbf{w}_{1}\right\|^2
    + \frac{1}{2}\left\|\sqrt{\frac{k_1}{\mu}}
    \mathrm{grad}[q_1]\right\|^2 \nonumber \\
    & + \frac{1}{2}\left\|\sqrt{\frac{\mu}{k_2}}
    \mathbf{w}_2 \right\|^2 
    + \frac{1}{2}\left\|\sqrt{\frac{k_2}{\mu}}
    \mathrm{grad}[q_2]\right\|^2 
    + \left\|\sqrt{\frac{\beta}{\mu}}
    (q_1 - q_2)\right\|^2 \nonumber \\
    & + \left\| \sqrt{\eta_u h_{\Upsilon}
      \{\!\!\{\mu~ k_1^{-1}\}\!\!\}}
    \; \llbracket \mathbf{w}_{1} \rrbracket
    \right\|_{\Gamma^{\mathrm{int}}}^{2}
    + \left\| \sqrt{\frac{\eta_p}{h_{\Upsilon}}
      \{\!\!\{\mu^{-1} k_1\}\!\!\}}
    \; \llbracket q_{1} \rrbracket
    \right\|_{\Gamma^{\mathrm{int}}}^{2} \nonumber \\
    & + \left\| \sqrt{\eta_u h_{\Upsilon}
      \{\!\!\{\mu~ k_2^{-1}\}\!\!\}}
    \; \llbracket \mathbf{w}_{2} \rrbracket
    \right\|_{\Gamma^{\mathrm{int}}}^{2}
    + \left\| \sqrt{\frac{\eta_p}{h_{\Upsilon}}
      \{\!\!\{\mu^{-1} k_2\}\!\!\}} \; \llbracket
    q_{2} \rrbracket \right\|_{\Gamma^{\mathrm{int}}}^{2}
    \quad \forall \mathbf{W} \in \mathbb{U}
\end{align}
}

\begin{lemma}{(Stability norm)}
  \label{Lemma:DD2_stab_norm}
  $\|\cdot\|_{\mathrm{stab}}^{\mathrm{DG}}$ is 
  a norm on $\mathbb{U}$. 
\end{lemma}
\begin{proof}
  The mathematical proof is similar to that of
  the continuous formulation, which is provided
  in \citep{Nakshatrala_Joodat_Ballarini_P2}.
\end{proof}

\subsection{Convergence theorem and error analysis}
In order to perform the error analysis
of the proposed stabilized mixed DG
formulation, we need to define the
finite element solution $\mathbf{U}^{h}$
and the corresponding weighting function as
\begin{subequations}
\begin{alignat}{2}
  &\mathbf{U}^{h} = (\mathbf{u}_{1}^{h}(\mathbf{x}),\mathbf{u}_{2}^{h}(\mathbf{x}),
  p_{1}^{h}(\mathbf{x}),
  p_{2}^{h}(\mathbf{x})) \in \mathbb{U}^{h} \\
  &\mathbf{W}^{h} = (\mathbf{w}_{1}^{h}(\mathbf{x}),\mathbf{w}_{2}^{h}(\mathbf{x}),
q_{1}^{h}(\mathbf{x}),
q_{2}^{h}(\mathbf{x})) \in \mathbb{U}^{h}
\end{alignat}
\end{subequations}
$\mathbb{U}^{h}$ is the closed linear subspace
of $\mathbb{U}$ and is defined as follows:
\begin{align}
    \mathbb{U}^{h} = \mathcal{U}^{h} \times \mathcal{U}^{h} \times \mathcal{Q}^{h}
\end{align}
where
\begin{subequations}
  \begin{align}
    \mathcal{U}^{h} &:= \left\{\mathbf{u}^{h}(\mathbf{x})
    \in \mathcal{U} \; \Big| \; \mathbf{u}^{h}(\mathbf{x})
    \in \left(C^{0}(\overline{\omega}^i)\right)^{nd};
    \mathbf{u}^{h}(\mathbf{x})|_{\omega^i} \in
    \left(\mathscr{P}^{k}(\omega^i)\right)^{nd}; 
    i = 1, \cdots, Nele  \right\} \\
    \mathcal{Q}^{h} &:= \left\{
    \left(p_1^{h},p_2^{h}\right) \in
    \mathcal{Q} \; \Big| \; p_{1}^{h}(\mathbf{x}), 
    p_{2}^{h}(\mathbf{x}) \in
    C^{0}(\overline{\omega}^i); 
    p_{1}^{h}(\mathbf{x}), p_{2}^{h}(\mathbf{x})|_{\omega^i} \in
    \mathscr{P}^{l}(\omega^i); 
    i = 1, \cdots, Nele  \right\}
  \end{align}
\end{subequations}
and $C^{0}(\overline{\omega}^i)$ is the
set of all continuous functions defined
on $\overline{\omega}^i$ (which is the
set closure of $\omega^i$).

The finite element formulation corresponding to
the proposed stabilized mixed DG formulation is defined as follows:~Find
$\mathbf{U}^{h} \in \mathbb{U}^{h}$ such that
we have
\begin{align}
  \mathcal{B}_{\mathrm{stab}}^{\mathrm{DG}}(\mathbf{W}^{h};\mathbf{U}^{h})
  = \mathcal{L}_{\mathrm{stab}}^{\mathrm{DG}}(\mathbf{W}^{h})
  \quad \forall \mathbf{W}^{h} \in \mathbb{U}^{h}
  \label{Eqn:DG_Galerkin_Weak_form_h}
\end{align}

The error in the finite element solution $\mathbf{E}$ is defined as the difference between the finite element solution and the exact solution.
If we define $\widetilde{\mathbf{U}}^{h}$ as an ``interpolate'' of $\mathbf{U}$ onto $\mathbb{U}^{h}$ \citep{Brenner_Scott}, decomposition of the error can be performed as follows:
\begin{align}
  \label{Eqn:DG_Error_decomposition}
  \mathbf{E} := \mathbf{U}^{h} - \mathbf{U}
  = \mathbf{E}^h + \mathbf{H} 
\end{align}
where $\mathbf{E}^h = \mathbf{U}^{h} - \widetilde{\mathbf{U}}^{h}$ is the approximation error and $\mathbf{H} = \widetilde{\mathbf{U}}^{h} - \mathbf{U}$ is the interpolation error. The components of $\mathbf{E}$ and $\mathbf{H}$ are as follows:
\begin{align}
  \mathbf{E} = \left\{\mathbf{e}_{\mathbf{u}_{1}},\mathbf{e}_{\mathbf{u}_{2}},e_{p_{1}},e_{p_{2}} \right\}
  \quad \mathrm{and} \quad 
  \mathbf{H} =
  \left\{ \boldsymbol{\eta}_{\mathbf{u}_{1}},
  \boldsymbol{\eta}_{\mathbf{u}_{2}},\eta_{p_{1}},\eta_{p_{2}}\right\}
  \label{Eqn:DG_components_of_E_and_H}
\end{align}

\begin{lemma}{(Estimates for approximation errors
    on $\Gamma^{\mathrm{int}}$.)}
  \label{Lemma:DG_approx_error_Gamma_int}
  On a sequence of admissible meshes,
  the following estimates hold: 
  \begin{align}
    \label{Eqn:Trace5}
    \left\| \sqrt{h_{\Upsilon} \{\!\!\{ \mu
    k_1^{-1} \}\!\!\}} \; \{\!\!\{
    \mathbf{e}_{\mathbf{u}_1} \}\!\!\}
    \right\|^2_{\Gamma^{\mathrm{int}}}
    \leq \mathcal{C}_{\mathbf{e}_{\mathbf{u}_1}} 
    \left\|\sqrt{\mu k_1^{-1}} \; 
    \mathbf{e}_{\mathbf{u}_1}\right\|^2
  \end{align}
  \begin{align}
    \label{Eqn:Trace6}
    \left\| \sqrt{h_{\Upsilon} \{\!\!\{ \mu k_2^{-1} \}\!\!\}} \; 
    \{\!\!\{ \mathbf{e}_{\mathbf{u}_2} \}\!\!\}
    \right\|^2_{\Gamma^{\mathrm{int}}}
    \leq \mathcal{C}_{\mathbf{e}_{\mathbf{u}_2}} 
    \left\|\sqrt{\mu k_2^{-1}} \;
    \mathbf{e}_{\mathbf{u}_2}\right\|^2
  \end{align}
\end{lemma}
\begin{proof}
  We first note that
  \begin{align}
    \left\|\sqrt{h_{\Upsilon} \{\!\!\{ \mu
    k_1^{-1} \}\!\!\}} \; \{\!\!\{
    \mathbf{e}_{\mathbf{u}_1} \}\!\!\}
    \right\|^2_{\Gamma^{\mathrm{int}}} = 
    \sum_{\Upsilon \in \mathcal{E}^{\mathrm{int}}}
    \left\|\sqrt{h_{\Upsilon} \{\!\!\{ \mu
    k_1^{-1} \}\!\!\}} \; \{\!\!\{
    \mathbf{e}_{\mathbf{u}_1} \}\!\!\}
    \right\|^2_{\Upsilon}
  \end{align}
  We now bound the approximation error
  of $\mathbf{u}_1$ on an interior edge
  $\Upsilon \in \mathcal{E}^{\mathrm{int}}$.
  The Cauchy-Schwarz inequality implies
  the following: 
  \begin{align}
    \left\|\sqrt{h_{\Upsilon} \{\!\!\{ \mu
    k_1^{-1} \}\!\!\}} \; \{\!\!\{
    \mathbf{e}_{\mathbf{u}_1} \}\!\!\}
    \right\|^2_{\Upsilon}
    &\leq \frac{1}{2} \left(
    \left\|\sqrt{h_{\Upsilon}
      \{\!\!\{ \mu k_1^{-1} \}\!\!\}} 
    \; \mathbf{e}_{\mathbf{u}_1}
    \right\|^2_{\partial \omega^{+}_{\Upsilon} \cap \Upsilon}
    + \left\|\sqrt{h_{\Upsilon} \{\!\!\{ \mu k_1^{-1}
      \}\!\!\}} \; \mathbf{e}_{\mathbf{u}_1}
    \right\|^2_{\partial \omega^{-}_{\Upsilon} \cap \Upsilon}
    \right) 
  \end{align}
  Noting the boundedness of the drag
  coefficients (i.e., equation
  \eqref{Eqn:DG_drag_bounds}), we
  obtain the following:
  \begin{align}
    \left\|\sqrt{h_{\Upsilon} \{\!\!\{ \mu
    k_1^{-1} \}\!\!\}} \; \{\!\!\{
    \mathbf{e}_{\mathbf{u}_1} \}\!\!\}
    \right\|^2_{\Upsilon}
    &\leq \frac{1}{2} \mathcal{C}_{\mathrm{drag},1}
    \left(\left\|\sqrt{h_{\Upsilon} \mu k_1^{-1}} \;  
    \mathbf{e}_{\mathbf{u}_1}
    \right\|^2_{\partial \omega^{+}_{\Upsilon} \cap \Upsilon}
    + \left\|\sqrt{h_{\Upsilon} \mu k_1^{-1}} \; 
    \mathbf{e}_{\mathbf{u}_1}
    \right\|^2_{\partial \omega^{-}_{\Upsilon} \cap \Upsilon}
    \right)
  \end{align}
  Using the bound based on the locally
  quasi-uniform condition (i.e., inequality
  \eqref{Eqn:DPP_useful_bound_on_h}) we
  obtain the following:
  {\small
  \begin{align}
    \left\|\sqrt{h_{\Upsilon} \{\!\!\{ \mu
    k_1^{-1} \}\!\!\}} \{\!\!\{
    \mathbf{e}_{\mathbf{u}_1} \}\!\!\}
    \right\|^2_{\Upsilon}
    &\leq \frac{1}{4} \mathcal{C}_{\mathrm{drag},1}
    (1 + \mathcal{C}_{\mathrm{lqu}})
    \left(\left\|\sqrt{h_{\omega^{+}_{\Upsilon}} \mu k_1^{-1}} \; 
    \mathbf{e}_{\mathbf{u}_1}
    \right\|^2_{\partial \omega^{+}_{\Upsilon} \cap \Upsilon}
    + \left\|\sqrt{h_{\omega^{-}_{\Upsilon}} \mu k_1^{-1}} \; 
    \mathbf{e}_{\mathbf{u}_1}
    \right\|^2_{\partial \omega^{-}_{\Upsilon} \cap \Upsilon}
    \right)
  \end{align}
  }
  By summing over all the interior
  edges we obtain the following: 
  \begin{align}
    \sum_{\Upsilon \in \mathcal{E}^{\mathrm{int}}}
    \left\|\sqrt{h_{\Upsilon} \{\!\!\{ \mu
    k_1^{-1} \}\!\!\}} \; \{\!\!\{
    \mathbf{e}_{\mathbf{u}_1} \}\!\!\}
    \right\|^2_{\Upsilon}
    &\leq \frac{1}{4} \mathcal{C}_{\mathrm{drag},1}
    (1 + \mathcal{C}_{\mathrm{lqu}})
    \sum_{\omega \in \mathcal{T}}
    \left\|\sqrt{h_{\omega} \mu k_1^{-1}} \; 
    \mathbf{e}_{\mathbf{u}_1}
    \right\|^2_{\partial \omega \setminus \partial \Omega}
    \nonumber \\
    &\leq \frac{1}{4} \mathcal{C}_{\mathrm{drag},1}
    (1 + \mathcal{C}_{\mathrm{lqu}})
    \sum_{\omega \in \mathcal{T}}
    \left\|\sqrt{h_{\omega} \mu k_1^{-1}} \;  
    \mathbf{e}_{\mathbf{u}_1}\right\|^2_{\partial \omega}
  \end{align}
  By invoking the discrete trace inequality
  \eqref{Eqn:DPP_dti_vector} we obtain the
  following:
  \begin{align}
    \sum_{\Upsilon \in \mathcal{E}^{\mathrm{int}}}
    \left\|\sqrt{h_{\Upsilon} \{\!\!\{ \mu
    k_1^{-1} \}\!\!\}} \; \{\!\!\{
    \mathbf{e}_{\mathbf{u}_1} \}\!\!\}
    \right\|^2_{\Upsilon}
    &\leq \frac{1}{4} \mathcal{C}_{\mathrm{drag},1}
    (1 + \mathcal{C}_{\mathrm{lqu}})
    \mathcal{C}^2_{\mathrm{trace}} (1 + \mathcal{C}_{\mathrm{inv}})^2
    \sum_{\omega \in \mathcal{T}}
    \left\|\sqrt{\mu k_1^{-1}} \; 
    \mathbf{e}_{\mathbf{u}_1}
    \right\|^2_{\omega} \nonumber \\
    &\leq \frac{1}{4} \mathcal{C}_{\mathrm{drag},1}
    \mathcal{C}^2_{\mathrm{trace}}
    (1 + \mathcal{C}_{\mathrm{inv}})^2
    (1 + \mathcal{C}_{\mathrm{lqu}})
    \left\|\sqrt{\mu k_1^{-1}} \;
    \mathbf{e}_{\mathbf{u}_1} \right\|^2   
  \end{align}
  (Recall that the subscript will be dropped if
  the $L_2$ norm is over $\widetilde{\Omega} :=
  \cup_{\omega \in \mathcal{T}} \; \omega$.) Thus,
  \begin{align}
    \label{Eqn:DG_C_eu1}
    \mathcal{C}_{\mathbf{e}_{\mathbf{u}_1}} := \frac{1}{4}
    \mathcal{C}_{\mathrm{drag},1}
    \mathcal{C}^2_{\mathrm{trace}}
    (1 + \mathcal{C}_{\mathrm{inv}})^2
    (1 + \mathcal{C}_{\mathrm{lqu}})   
  \end{align}
  On similar lines, one can establish
  the estimate \eqref{Eqn:Trace6} with
  \begin{align}
    \label{Eqn:DG_C_eu2}
    \mathcal{C}_{\mathbf{e}_{\mathbf{u}_2}} := \frac{1}{4}
    \mathcal{C}_{\mathrm{drag},2}
    \mathcal{C}^2_{\mathrm{trace}}
    (1 + \mathcal{C}_{\mathrm{inv}})^2
    (1 + \mathcal{C}_{\mathrm{lqu}})   
  \end{align}
\end{proof}

If a $p$-th order polynomial is employed for
a field variable $f(\mathbf{x})$ on an element
$\omega \in \mathcal{T}$ and the corresponding
interpolate denoted by $\widetilde{f}^{h}$, the
following estimate holds for the interpolation
error \citep{Brezzi_Fortin}:
\begin{align}
  \label{Eqn:DG_standard_estimate_for_int_error}
  \|f - \widetilde{f}^{h}\|_{\omega} \leq
  \mathcal{C}_{\mathrm{int}} h_{w}^{p+1}
  |f|_{H^{p+1}(\omega)} 
\end{align}
where $h_{\omega}$ is the element diameter of $\omega$,
$\mathcal{C}_{\mathrm{int}}$ is a non-dimensional constant
independent of $h_{\omega}$ and $f$, and
$|\cdot|_{H^{p+1}(\omega)}$ is a Sobolev
semi-norm, which is defined in equation
\eqref{Eqn:Sobolev_norm_Scalar}.

To avoid further introduction of constants,
we employ the notation $A \lesssim B$ to
denote that there exits a constant
$\mathcal{C}$, independent of the mesh size,
such that $A \leq \mathcal{C} B$. A similar
definition holds for $A \gtrsim B$. The
notation $A \sim B$ denotes the case when
$A \lesssim B$ and $A \gtrsim B$ hold
simultaneously. 

\begin{lemma}{(Estimates for interpolation
    errors on $\Gamma^{\mathrm{int}}$.)}
  \label{Lemma:DG_int_error_Gamma_int}
  If polynomial orders used for interpolation of
  $\mathbf{u}_1$, $\mathbf{u}_2$, $p_1$ and $p_2$
  are, respectively, $p$, $q$, $r$ and $s$ then
  the following estimates hold for the interpolation
  errors on $\Gamma^{\mathrm{int}}$:
  \begin{align}
    \label{Eqn:int_error_avg_eta_u1}
    \left\| \sqrt{\frac{h_{\Upsilon}}{\eta_p}
      \{\!\!\{ \mu k_1^{-1} \}\!\!\}} \;
    \{\!\!\{ \boldsymbol{\eta}_{\mathbf{u}_1} \}\!\!\}
    \right\|^2_{\Gamma^{\mathrm{int}}}
    \lesssim
    \sum_{\omega \in \mathcal{T}_{h}}
    h_{\omega}^{2(p+1)} |\mathbf{u}_1|^2_{H^{p+1}(\omega)}
  \end{align}
  \begin{align}
    \label{Eqn:int_error_avg_eta_u2}
    \left\| \sqrt{\frac{h_{\Upsilon}}{\eta_p}
      \{\!\!\{ \mu k_2^{-1} \}\!\!\}} \;
    \{\!\!\{ \boldsymbol{\eta}_{\mathbf{u}_2} \}\!\!\}
    \right\|^2_{\Gamma^{\mathrm{int}}}
    \lesssim
    \sum_{\omega \in \mathcal{T}_{h}}
    h_{\omega}^{2(q+1)} |\mathbf{u}_2|^2_{H^{q+1}(\omega)}
  \end{align}
  \begin{align}
    \label{Eqn:int_error_jump_eta_p1}
      \left\| \sqrt{h_{\Upsilon}^{-1} \{\!\!\{ \mu^{-1} k_1 \}\!\!\}} \; 
      \llbracket \eta_{p_1} \rrbracket \right\|^2_{\Gamma^{\mathrm{int}}}
      \lesssim
      \sum_{\omega \in \mathcal{T}_{h}}
      h_{\omega}^{2r}|p_1|^2_{H^{r+1}(\omega)}
    \end{align}
  \begin{align}
    \label{Eqn:int_error_jump_eta_p2}
      \left\| \sqrt{h_{\Upsilon}^{-1}
        \{\!\!\{ \mu^{-1} k_2 \}\!\!\}} \; 
      \llbracket \eta_{p_2} \rrbracket
      \right\|^2_{\Gamma^{\mathrm{int}}}
      \lesssim
      \sum_{\omega \in \mathcal{T}_{h}}
      h_{\omega}^{2s}|p_2|^2_{H^{s+1}(\omega)}
    \end{align}
\end{lemma}
\begin{proof}
  We first establish the estimate
  \eqref{Eqn:int_error_avg_eta_u1}.
  The boundedness of the drag coefficient
  $\mu/k_1(\mathbf{x})$ and the linearity
  of a norm imply the following: 
  {\small
    \begin{align}
      \left\| \sqrt{\frac{h_{\Upsilon}}{\eta_p}
        \{\!\!\{ \mu k_1^{-1} \}\!\!\}} \;
      \{\!\!\{ \boldsymbol{\eta}_{\mathbf{u}_1} \}\!\!\}
      \right\|^2_{\Upsilon}
      &\leq \frac{1}{\eta_p}
      \left(\sup_{\mathbf{x} \in \Omega}
      \frac{\mu}{k_1(\mathbf{x})}\right)
      \left\| \sqrt{h_{\Upsilon}} \;
      \{\!\!\{ \boldsymbol{\eta}_{\mathbf{u}_1} \}\!\!\}
      \right\|^2_{\Upsilon}
      \qquad \forall \Upsilon \in \mathcal{E}^{\mathrm{int}}
    \end{align}
  }
  Using the triangle inequality and the bound
  from the locally quasi-uniform condition
  \eqref{Eqn:DPP_useful_bound_on_h}, we obtain
  the following: 
  {\small
    \begin{align}
      \left\| \sqrt{\frac{h_{\Upsilon}}{\eta_p}
        \{\!\!\{ \mu k_1^{-1} \}\!\!\}} \;
      \{\!\!\{ \boldsymbol{\eta}_{\mathbf{u}_1} \}\!\!\}
      \right\|^2_{\Upsilon}
      &\lesssim
      \frac{1}{4} \left(1 + \mathcal{C}_{\mathrm{lqu}}\right)
      \left(\left\| \sqrt{h_{\omega_{\Upsilon}^{+}}} \;
      \boldsymbol{\eta}_{\mathbf{u}_1}
      \right\|^2_{\partial \omega_{\Upsilon}^{+} \cap \Upsilon}
      + \left\| \sqrt{h_{\omega_{\Upsilon}^{-}}} \;
      \boldsymbol{\eta}_{\mathbf{u}_1} 
      \right\|^2_{\partial \omega_{\Upsilon}^{-} \cap \Upsilon}
      \right)
      \quad \forall \Upsilon \in \mathcal{E}^{\mathrm{int}}
    \end{align}
  }
  By summing over all the interior edges
  and noting the linearity of a norm,
  we obtain the following: 
  {\small
    \begin{align}
      \left\| \sqrt{\frac{h_{\Upsilon}}{\eta_p}
        \{\!\!\{ \mu k_1^{-1} \}\!\!\}} \;
      \{\!\!\{ \boldsymbol{\eta}_{\mathbf{u}_1} \}\!\!\}
      \right\|^2_{\Gamma^{\mathrm{int}}}
      = \sum_{\Upsilon \in \mathcal{E}^{\mathrm{int}}}
      \left\| \sqrt{\frac{h_{\Upsilon}}{\eta_p}
        \{\!\!\{ \mu k_1^{-1} \}\!\!\}} \;
      \{\!\!\{ \boldsymbol{\eta}_{\mathbf{u}_1} \}\!\!\}
      \right\|^2_{\Upsilon}
      &\lesssim
      \sum_{\omega \in \mathcal{T}_{h}}
      \left(h_{\omega} 
      \left\|\boldsymbol{\eta}_{\mathbf{u}_1}
      \right\|^2_{\partial \omega}
      \right) 
    \end{align}
  }
  By invoking the discrete trace inequality
  \eqref{Eqn:DPP_dti_vector}, we obtain the
  following inequality: 
  {\small
    \begin{align}
      \left\| \sqrt{\frac{h_{\Upsilon}}{\eta_p}
        \{\!\!\{ \mu k_1^{-1} \}\!\!\}} \;
      \{\!\!\{ \boldsymbol{\eta}_{\mathbf{u}_1} \}\!\!\}
      \right\|^2_{\Gamma^{\mathrm{int}}}
      &\lesssim
      \sum_{\omega \in \mathcal{T}_{h}} 
      \left\|\boldsymbol{\eta}_{\mathbf{u}_1}
      \right\|^2_{\omega}
    \end{align}
  }
  If a polynomial of order $p$ is employed
  for approximating $\mathbf{u}_1$, then the
  standard estimate for the interpolation error
  \eqref{Eqn:DG_standard_estimate_for_int_error}
  provides the following: 
  {\small
    \begin{align}
      \left\| \sqrt{\frac{h_{\Upsilon}}{\eta_p}
        \{\!\!\{ \mu k_1^{-1} \}\!\!\}} \;
      \{\!\!\{ \boldsymbol{\eta}_{\mathbf{u}_1} \}\!\!\}
      \right\|^2_{\Gamma^{\mathrm{int}}}
      &\lesssim \sum_{\omega \in \mathcal{T}_{h}}
      h_{\omega}^{2(p+1)} |\mathbf{u}_1|^2_{H^{p+1}(\omega)}
    \end{align}
  }
  which is the estimate \eqref{Eqn:int_error_avg_eta_u1}. 
  By reasoning out on similar lines, one can establish
  the estimate \eqref{Eqn:int_error_avg_eta_u2}.

  We now establish the estimate
  \eqref{Eqn:int_error_jump_eta_p1}.
  The boundedness of the drag coefficient
  $\mu/k_1(\mathbf{x})$ and the linearity
  of a norm imply the following: 
  {\small
    \begin{align}
      \left\| \sqrt{h_{\Upsilon}^{-1} 
        \{\!\!\{ \mu^{-1} k_1\}\!\!\}} \;
      \llbracket \eta_{p_1} \rrbracket
      \right\|^2_{\Upsilon}
      &\leq 
      \left(\inf_{\mathbf{x} \in \Omega}
      \frac{\mu}{k_1(\mathbf{x})}\right)
      \left\| \sqrt{h_{\Upsilon}^{-1}} \;
      \llbracket \eta_{p_1} \rrbracket
      \right\|^2_{\Upsilon}
      \qquad \forall \Upsilon \in \mathcal{E}^{\mathrm{int}}
    \end{align}
  }
  Using the triangle inequality and the bound
  from the locally quasi-uniform condition
  \eqref{Eqn:DPP_useful_bound_on_h}, we obtain
  the following: 
  {\small
    \begin{align}
      \left\| \sqrt{h_{\Upsilon}^{-1}
        \{\!\!\{\mu^{-1} k_1 \}\!\!\}} \;
      \llbracket \eta_{p_1} \rrbracket 
      \right\|^2_{\Upsilon}
      &\lesssim 4 
      \left(1 + \frac{1}{\mathcal{C}_{\mathrm{lqu}}}\right)^{-1}
      \left(\left\| \sqrt{h_{\omega_{\Upsilon}^{+}}^{-1}} \;
      \eta_{p_1}
      \right\|^2_{\partial \omega_{\Upsilon}^{+} \cap \Upsilon}
      + \left\| \sqrt{h_{\omega_{\Upsilon}^{-}}^{-1}} \;
      \eta_{p_1} 
      \right\|^2_{\partial \omega_{\Upsilon}^{-} \cap \Upsilon}
      \right)
      \quad \forall \Upsilon \in \mathcal{E}^{\mathrm{int}}
    \end{align}
  }
  By summing over all the interior edges
  and noting the linearity of a norm,
  we obtain the following: 
  {\small
    \begin{align}
      \left\| \sqrt{h_{\Upsilon}^{-1}
        \{\!\!\{ \mu^{-1} k_1 \}\!\!\}} \;
      \llbracket \eta_{p_1} \rrbracket 
      \right\|^2_{\Gamma^{\mathrm{int}}}
      = \sum_{\Upsilon \in \mathcal{E}^{\mathrm{int}}}
      \left\| \sqrt{h_{\Upsilon}^{-1} 
        \{\!\!\{ \mu^{-1} k_1 \}\!\!\}} \;
      \llbracket \eta_{p_1} \rrbracket 
      \right\|^2_{\Upsilon}
      &\lesssim
      \sum_{\omega \in \mathcal{T}_{h}}
      \left(h_{\omega}^{-1}  
      \left\|\eta_{p_1}
      \right\|^2_{\partial \omega}
      \right) 
    \end{align}
  }
  By invoking the discrete trace inequality
  \eqref{Eqn:DPP_dti_scalar}, we obtain the
  following inequality: 
  {\small
    \begin{align}
      \left\| \sqrt{h_{\Upsilon}^{-1}
        \{\!\!\{ \mu^{-1} k_1 \}\!\!\}} \;
      \llbracket \eta_{p_1} \rrbracket 
      \right\|^2_{\Gamma^{\mathrm{int}}}
      &\lesssim
      \sum_{\omega \in \mathcal{T}_{h}}
      \left(h_{\omega}^{-2} 
      \left\|\eta_{p_1}
      \right\|^2_{\omega}
      \right) 
    \end{align}
  }
  If a polynomial of order $r$ is employed
  for approximating $p_1$, then the standard
  estimate for the interpolation error
  \eqref{Eqn:DG_standard_estimate_for_int_error}
  provides the following: 
  {\small
    \begin{align}
      \left\| \sqrt{h_{\Upsilon}^{-1}
        \{\!\!\{ \mu^{-1} k_1\}\!\!\}} \;
      \llbracket \eta_{p_1} \rrbracket 
      \right\|^2_{\Gamma^{\mathrm{int}}}
      &\lesssim \sum_{\omega \in \mathcal{T}_{h}}
      h_{\omega}^{2r} |p_1|^2_{H^{r+1}(\omega)}
    \end{align}
  }
  which is the estimate \eqref{Eqn:int_error_jump_eta_p1}. 
  By reasoning out on similar lines, one can establish
  the estimate \eqref{Eqn:int_error_jump_eta_p2}. 
\end{proof}

\begin{lemma}{(Estimate for $\mathbf{H}$ under the stability norm.)}
  \label{Lemma:DG_estimate_for_H_in_terms_of_stab_norm}
  If polynomial orders used for interpolation of
  $\mathbf{u}_1$, $\mathbf{u}_2$, $p_1$ and $p_2$
  are, respectively, $p$, $q$, $r$ and $s$ then
  the following estimate holds:
  {\small
    \begin{align}
      \label{Eqn:H_stab}
      \left(\|\mathbf{H}\|_{\mathrm{stab}}^{\mathrm{DG}}\right)^{2}
      \lesssim \sum_{\omega \in \mathcal{T}_h} \left(
      h^{2(p+1)}_{\omega} |\mathbf{u}_{1}|^2_{H^{p+1}(\omega)}
      + h^{2(q+1)}_{\omega} |\mathbf{u}_{2}|^2_{H^{q+1}(\omega)}
      + \left(1 + h_{\omega}^{2}\right) h^{2r}_{\omega} |p_{1}|^2_{H^{r+1}(\omega)}
      + \left(1 + h^{2}_{\omega}\right) h^{2s}_{\omega} |p_{2}|^2_{H^{s+1}(\omega)}
      \right)
    \end{align}
  }
  where the constant in the estimate is
  independent of the characteristic mesh
  length ($h$ or $h_{\omega}$) and the
  solution fields ($\mathbf{u}_{1}$,
  $\mathbf{u}_2$, $p_1$ and $p_2$).
\end{lemma}
\begin{proof}
  The definition of the stability
  norm \eqref{Eqn:DG_stability_norm}
  and the components of $\mathbf{H}$
  \eqref{Eqn:DG_components_of_E_and_H}
  imply the following: 
{\small
  \begin{align}
    \left(\|\mathbf{H}\|_{\mathrm{stab}}^{\mathrm{DG}}\right)^{2}
    &= \frac{1}{2}\left\|\sqrt{\frac{\mu}{k_1}}
    \boldsymbol{\eta}_{\mathbf{u}_1}\right\|^2
    + \frac{1}{2}\left\|\sqrt{\frac{k_1}{\mu}}
    \mathrm{grad}[\eta_{p_1}]\right\|^2 
    + \frac{1}{2}\left\|\sqrt{\frac{\mu}{k_2}}
    \boldsymbol{\eta}_{\mathbf{u}_2} \right\|^2 
    + \frac{1}{2}\left\|\sqrt{\frac{k_2}{\mu}}
    \mathrm{grad}[\eta_{p_2}]\right\|^2 \nonumber \\
    &+ \left\|\sqrt{\frac{\beta}{\mu}}
    (\eta_{p_1} - \eta_{p_2})\right\|^2 
    + \left\| \sqrt{\eta_u h_{\Upsilon}
      \{\!\!\{\mu~ k_1^{-1}\}\!\!\}}
    \; \llbracket \boldsymbol{\eta}_{\mathbf{u}_{1}}
    \rrbracket
    \right\|_{\Gamma^{\mathrm{int}}}^{2}
    + \left\| \sqrt{\frac{\eta_p}{h_{\Upsilon}}
      \{\!\!\{\mu^{-1} k_1\}\!\!\}}
    \; \llbracket \eta_{p_{1}} \rrbracket
    \right\|_{\Gamma^{\mathrm{int}}}^{2} \nonumber \\
    & + \left\| \sqrt{\eta_u h_{\Upsilon}
      \{\!\!\{\mu~ k_2^{-1}\}\!\!\}}
    \; \llbracket \boldsymbol{\eta}_{\mathbf{u}_{2}} \rrbracket
    \right\|_{\Gamma^{\mathrm{int}}}^{2}
    + \left\| \sqrt{\frac{\eta_p}{h_{\Upsilon}}
      \{\!\!\{\mu^{-1} k_2\}\!\!\}} \; \llbracket
    \eta_{p_{2}} \rrbracket \right\|_{\Gamma^{\mathrm{int}}}^{2}
\end{align}
}
Using the boundedness of the drag
coefficient of the first pore-network,
linearity of a norm and the standard
estimate for the interpolation error
\eqref{Eqn:DG_standard_estimate_for_int_error},
and noting that the polynomial order of
approximation for $\mathbf{u}_1$ is $p$,
we obtain the following:
\begin{align}
  \frac{1}{2}\left\|\sqrt{\frac{\mu}{k_1}}
  \boldsymbol{\eta}_{\mathbf{u}_1}\right\|^2
  &\leq \frac{1}{2} \sup_{\mathbf{x} \in \Omega}
  \frac{\mu}{k_1(\mathbf{x})}
  \sum_{\omega \in \mathcal{T}_{h}}
  \left\|\boldsymbol{\eta}_{\mathbf{u}_1}\right\|^2_{\omega}
  \lesssim \sum_{\omega \in \mathcal{T}_h}
  h_{\omega}^{2(p+1)}
  \left|\mathbf{u}_1\right|^2_{H^{p+1}(\omega)}
\end{align}
Similarly, 
\begin{align}
  \frac{1}{2}\left\|\sqrt{\frac{\mu}{k_2}}
  \boldsymbol{\eta}_{\mathbf{u}_2}\right\|^2
  &\leq \frac{1}{2} \sup_{\mathbf{x} \in \Omega}
  \frac{\mu}{k_2(\mathbf{x})}
  \sum_{\omega \in \mathcal{T}_{h}}
  \left|\boldsymbol{\eta}_{\mathbf{u}_2}
  \right|^2_{\omega}
  \lesssim \sum_{\omega \in \mathcal{T}_h}
  h_{\omega}^{2(q+1)}
  \left|\mathbf{u}_2\right|^2_{H^{q+1}(\omega)}
\end{align}

For the \emph{second} term,
we proceed as follows by
first noting the boundedness
of the drag coefficient in
the first pore-network: 
\begin{align}
  \frac{1}{2}\left\|\sqrt{\frac{k_1}{\mu}}
  \mathrm{grad}[\eta_{p_1}]\right\|^2
  &\leq \frac{1}{2}\inf_{\mathbf{x} \in \Omega}
  \frac{\mu}{k_1(\mathbf{x})}
  \sum_{\omega \in \mathcal{T}_h}
  \left\|\mathrm{grad}[\eta_{p_1}]\right\|^2_{\omega}
  \nonumber \\
  &\leq \frac{1}{2}\inf_{\mathbf{x} \in \Omega}
  \frac{\mu}{k_1(\mathbf{x})}
  C_{\mathrm{inv}}^{2}
  \sum_{\omega \in \mathcal{T}_h} h_{\omega}^{-2}
  \left\|\eta_{p_1}\right\|^2_{\omega}
  &&\; \mbox{[inverse estimate
    \eqref{Eqn:DDP_discrete_inverse_inequality}]}
  \nonumber \\
  &\lesssim \sum_{\omega \in \mathcal{T}_h} h_{\omega}^{2r}
  \left|p_1\right|^2_{H^{r+1}(\omega)}
  &&\; \mbox{[interpolation estimate
    \eqref{Eqn:DG_standard_estimate_for_int_error}]}
\end{align}
Similarly, one can derive the following
estimate for the \emph{fourth} term:
\begin{align}
  \frac{1}{2}\left\|\sqrt{\frac{k_2}{\mu}}
  \mathrm{grad}[\eta_{p_2}]\right\|^2
  &\lesssim \sum_{\omega \in \mathcal{T}_h} h_{\omega}^{2s}
  \left|p_2\right|^2_{H^{s+1}(\omega)}
\end{align}

The estimate for the \emph{fifth}
term utilizes the triangle inequality
and the interpolation estimate
\eqref{Eqn:DG_standard_estimate_for_int_error}
and it can be obtained as follows: 
{\small
\begin{align}
  \left\|\sqrt{\frac{\beta}{\mu}}
  (\eta_{p_1} - \eta_{p_2})\right\|^2
  \leq \frac{\beta}{\mu}
  \sum_{\omega \in \mathcal{T}_{h}}
  \left(\left\|\eta_{p_1}\right\|^{2}_{\omega}
  + \left\|\eta_{p_2}\right\|^{2}_{\omega}\right)
  \lesssim \sum_{\omega \in \mathcal{T}_{h}}
  \left(h_{\omega}^{2(r+1)}
  \left|p_1\right|^{2}_{H^{r+1}(\omega)}
  + h_{\omega}^{2(s+1)}
  \left|p_2\right|^{2}_{H^{s+1}(\omega)}\right)
\end{align}}

Using the boundedness of $\eta_u$ and
the drag coefficient of the first
pore-network and noting the linearity
of a norm, we obtain the following
estimate for the \emph{sixth} term: 
\begin{align}
  \left\| \sqrt{\eta_u h_{\Upsilon}
    \{\!\!\{\mu~ k_1^{-1}\}\!\!\}}
  \; \llbracket \boldsymbol{\eta}_{\mathbf{u}_{1}}
  \rrbracket
  \right\|_{\Gamma^{\mathrm{int}}}^{2}
  &\lesssim
  \sum_{\Upsilon \in \mathcal{E}^{\mathrm{int}}}
  \left\|\sqrt{h_{\Upsilon}} \;
  \llbracket \boldsymbol{\eta}_{\mathbf{u}_1}
  \rrbracket \; \right\|^2_{\Upsilon}
\end{align}
Using the bound based on the
locally quasi-uniform condition
\eqref{Eqn:DPP_useful_bound_on_h}
and the triangle inequality, we
obtain the following:
\begin{align}
  \left\| \sqrt{\eta_u h_{\Upsilon}
    \{\!\!\{\mu~ k_1^{-1}\}\!\!\}}
  \; \llbracket \boldsymbol{\eta}_{\mathbf{u}_{1}}
  \rrbracket
  \right\|_{\Gamma^{\mathrm{int}}}^{2}
  &\lesssim 
  \sum_{\Upsilon \in \mathcal{E}^{\mathrm{int}}}
  \left(
  \left\|\sqrt{h_{\omega_{\Upsilon}^{+}}} \;
  \boldsymbol{\eta}_{\mathbf{u}_1}
  \; \right\|^2_{\partial \omega_{\Upsilon}^{+} \cap \Upsilon}
  + \left\|\sqrt{h_{\omega_{\Upsilon}^{-}}} \;
  \boldsymbol{\eta}_{\mathbf{u}_1}
  \; \right\|^2_{\partial \omega_{\Upsilon}^{-} \cap \Upsilon}
  \right) \nonumber \\
  &\lesssim
  \sum_{\omega \in \mathcal{T}_h}
  h_{\omega}\left\|
  \boldsymbol{\eta}_{\mathbf{u}_1}
  \right\|^2_{\partial \omega}
\end{align}
Using the discrete trace inequality
\eqref{Eqn:DPP_dti_vector} and the
standard interpolation estimate
\eqref{Eqn:DG_standard_estimate_for_int_error},
we obtain the following:
\begin{align}
  \left\| \sqrt{\eta_u h_{\Upsilon}
    \{\!\!\{\mu~ k_1^{-1}\}\!\!\}}
  \; \llbracket \boldsymbol{\eta}_{\mathbf{u}_{1}}
  \rrbracket
  \right\|_{\Gamma^{\mathrm{int}}}^{2}
  &\lesssim
  \sum_{\omega \in \mathcal{T}_h}
  \left\|\boldsymbol{\eta}_{\mathbf{u}_1}
  \right\|^2_{\omega} 
  \lesssim
  \sum_{\omega \in \mathcal{T}_h}
  h_{\omega}^{2(p+1)}
  \left|\mathbf{u}_1\right|^2_{H^{p+1}(\omega)}
\end{align} 
A similar argument gives rise
to the following estimate for
the \emph{eighth} term:
\begin{align}
  \left\| \sqrt{\eta_u h_{\Upsilon}
    \{\!\!\{\mu~ k_2^{-1}\}\!\!\}}
  \; \llbracket \boldsymbol{\eta}_{\mathbf{u}_{2}}
  \rrbracket
  \right\|_{\Gamma^{\mathrm{int}}}^{2}
  \lesssim
  \sum_{\omega \in \mathcal{T}_h}
  h_{\omega}^{2(q+1)}
  \left|\mathbf{u}_2\right|^2_{H^{q+1}(\omega)}
\end{align} 

Noting that $\eta_p$ is a bounded constant,
estimates \eqref{Eqn:int_error_jump_eta_p1}
and \eqref{Eqn:int_error_jump_eta_p2}
immediately imply the following estimates
for the \emph{seventh} and \emph{ninth}
terms: 
\begin{align}
  \left\| \sqrt{\frac{\eta_p}{h_{\Upsilon}}
    \{\!\!\{\mu^{-1} k_1\}\!\!\}}
  \; \llbracket \eta_{p_{1}} \rrbracket
  \right\|_{\Gamma^{\mathrm{int}}}^{2}
  \lesssim
  \sum_{\omega \in \mathcal{T}_h} h_{\omega}^{2r}
  \left| p_1\right|^2_{H^{r+1}(\omega)} \\
  \left\| \sqrt{\frac{\eta_p}{h_{\Upsilon}}
    \{\!\!\{\mu^{-1} k_2\}\!\!\}}
  \; \llbracket \eta_{p_2} \rrbracket
  \right\|_{\Gamma^{\mathrm{int}}}^{2}
  \lesssim
  \sum_{\omega \in \mathcal{T}_h} h_{\omega}^{2s}
  \left| p_2\right|^2_{H^{s+1}(\omega)}
\end{align}
By adding up the individual
estimates for all the terms,
we obtain the desired result. 
\end{proof}

\begin{theorem}{(Consistency)}
  \label{Thm:DG_Consistency}
  The error in the finite element solution satisfies
  \begin{align}
    \mathcal{B}_{\mathrm{stab}}^{\mathrm{DG}}(\mathbf{W}^{h};
    \mathbf{E}) = 0 \quad 
    \forall \mathbf{W}^{h} \in \mathbb{U}^{h}  
    \subset \mathbb{U}
    \label{Eqn:DG_solution_Error}
  \end{align}
\end{theorem}
\begin{proof}
  The proof follows a standard procedure
  employed in the literature.
  Equation \eqref{Eqn:VMS_DG_Weak_Form_compact} implies
  that for all $\mathbf{W}^{h} \in \mathbb{U}^{h} \subset
  \mathbb{U}$ we have the following:
  \begin{subequations}
    \begin{align}
      \mathcal{B}_{\mathrm{stab}}^{\mathrm{DG}}(\mathbf{W}^{h};
      \mathbf{U}^{h})
      &= \mathcal{L}_{\mathrm{stab}}^{\mathrm{DG}}(\mathbf{W}^{h}) \\
      \mathcal{B}_{\mathrm{stab}}^{\mathrm{DG}}(\mathbf{W}^{h};
      \mathbf{U})
      &= \mathcal{L}_{\mathrm{stab}}^{\mathrm{DG}}(\mathbf{W}^{h}) 
    \end{align}
  \end{subequations}
  By subtracting the above two equations,
  invoking the linearity in the second slot
  of $\mathcal{B}_{\mathrm{stab}}^{\mathrm{DG}}(\cdot;\cdot)$
  and noting the definition of $\mathbf{E}$ given by
  \eqref{Eqn:DG_Error_decomposition}, we obtain the
  desired result.
\end{proof}

\begin{theorem}{(Convergence)}
  \label{Thm:DG_convergence_theorem}
  Under a sequence of admissible meshes,
  the finite element solution $\mathbf{U}^{h}
  \in \mathbb{U}^{h}$ tends to the exact
  solution $\mathbf{U} \in \mathbb{U}$
  almost everywhere\footnote{Two quantities
    that are the same except on a set of
    measure zero are said to be equal
    almost everywhere \citep{Evans_PDE}.}
  as the mesh-size $h \rightarrow 0$. 
\end{theorem}
\begin{proof}
  The error with respect to the
  stability norm can be rewritten
  as follows:
  \begin{alignat}{2}
    \left(\|\mathbf{E}\|_{\mathrm{stab}}^{\mathrm{DG}}\right)^{2}
    &= \mathcal{B}_{\mathrm{stab}}^{\mathrm{DG}}
    (\mathbf{E};\mathbf{E})
    = \mathcal{B}_{\mathrm{stab}}^{\mathrm{DG}}
    (\mathbf{E}^{h} + \mathbf{H};\mathbf{E})
    = \mathcal{B}_{\mathrm{stab}}^{\mathrm{DG}}(\mathbf{E}^{h};\mathbf{E})
    + \mathcal{B}_{\mathrm{stab}}^{\mathrm{DG}}(\mathbf{H};\mathbf{E})
    = \mathcal{B}_{\mathrm{stab}}^{\mathrm{DG}}
    (\mathbf{H};\mathbf{E})
    \label{Eqn:DG_stab_theoerem_Estab}
  \end{alignat}
  We invoked the definition of $\|\cdot\|_{\mathrm{stab}}^{\mathrm{DG}}$
  norm (i.e., Eq.~\eqref{Eqn:DG_stability_norm}) for establishing
  the first equality,
  the decomposition of the error (i.e.,
  Eq.~\eqref{Eqn:DG_Error_decomposition})
  for the second equality, linearity in the
  first slot of $\mathcal{B}_{\mathrm{stab}}^{\mathrm{DG}}
  (\cdot;\cdot)$ for the third equality, and consistency
  (i.e., Theorem \ref{Thm:DG_Consistency}) for the fourth
  equality.
  We now expand $\mathcal{B}_{\mathrm{stab}}^{\mathrm{DG}}
  (\mathbf{H};\mathbf{E})$ as follows:
  \begin{align}
    \mathcal{B}_{\mathrm{stab}}^{\mathrm{DG}}(\mathbf{H};\mathbf{E}) &=
    \mathcal{B}_{\mathrm{stab}}^{\mathrm{DG}}(
    \boldsymbol{\eta}_{\mathbf{u}_{1}},\boldsymbol{\eta}_{\mathbf{u}_{2}},
    \eta_{p_1},\eta_{p_{2}};\mathbf{e}_{\mathbf{u}_{1}},
    \mathbf{e}_{\mathbf{u}_{2}},e_{p_1},e_{p_{2}}) \nonumber \\
    & = \frac{1}{2}(\boldsymbol{\eta}_{\mathbf{u}_{1}};
    \mu k_{1}^{-1}\mathbf{e}_{\mathbf{u}_{1}})
    + \frac{1}{2} (\boldsymbol{\eta}_{\mathbf{u}_{1}};\mathrm{grad}[e_{p_{1}}])
    - \frac{1}{2} (\mathrm{grad}[\eta_{p_{1}}];\mathbf{e}_{\mathbf{u}_{1}})
    + \frac{1}{2} \left(\mathrm{grad}[\eta_{p_{1}}];\frac{k_1}{\mu}
    \mathrm{grad}[e_{p_{1}}]\right) \nonumber \\
    & + \frac{1}{2}(\boldsymbol{\eta}_{\mathbf{u}_{2}};\mu k_{2}^{-1}\mathbf{e}_{\mathbf{u}_{2}})
    + \frac{1}{2} (\boldsymbol{\eta}_{\mathbf{u}_{2}};\mathrm{grad}[e_{p_{2}}])
    - \frac{1}{2} (\mathrm{grad}[\eta_{p_{2}}];\mathbf{e}_{\mathbf{u}_{2}}) + \frac{1}{2} \left(\mathrm{grad}[\eta_{p_{2}}];\frac{k_2}{\mu}
    \mathrm{grad}[e_{p_{2}}]\right) \nonumber \\
    & + \left((\eta_{p_{1}}- \eta_{p_{2}});
    \frac{\beta}{\mu}(e_{p_{1}} - e_{p_{2}})\right)
    \nonumber \\
    & - \left(\{\!\!\{ \boldsymbol{\eta}_{\mathbf{u}_{1}} \}\!\!\}; \llbracket e_{p_1} \rrbracket
    \right)_{\Gamma^{\mathrm{int}}}
    + \left(\llbracket \eta_{p_1} \rrbracket ; \{\!\!\{ \mathbf{e}_{\mathbf{u}_{1}} \}\!\!\}
    \right)_{\Gamma^{\mathrm{int}}} 
    - \left(\{\!\!\{ \boldsymbol{\eta}_{\mathbf{u}_{2}} \}\!\!\}; \llbracket e_{p_2} \rrbracket
    \right)_{\Gamma^{\mathrm{int}}}
    + \left(\llbracket \eta_{p_2} \rrbracket ; \{\!\!\{ \mathbf{e}_{\mathbf{u}_{2}} \}\!\!\}
    \right)_{\Gamma^{\mathrm{int}}} 
  \nonumber\\
  &+ \left(\eta_{u} h_{\Upsilon} \{\!\!\{\mu k_1^{-1}\}\!\!\}
  \llbracket \boldsymbol{\eta}_{\mathbf{u}_{1}} \rrbracket ; \llbracket
  \mathbf{e}_{\mathbf{u}_1}  \rrbracket
  \right)_{\Gamma^{\mathrm{int}}}
  + \left(\frac{\eta_{p}}{h_{\Upsilon}} \{\!\!\{\mu^{-1} k_1\}\!\!\}
  \llbracket \eta_{p_{1}} \rrbracket ; \llbracket e_{p_1} \rrbracket
  \right)_{\Gamma^{\mathrm{int}}}
   \nonumber \\
  &
  + \left(\eta_{u} h_{\Upsilon} \{\!\!\{\mu k_2^{-1}\}\!\!\}
  \llbracket \boldsymbol{\eta}_{\mathbf{u}_{2}} \rrbracket ; \llbracket
  \mathbf{e}_{\mathbf{u}_2}  \rrbracket\right)_{\Gamma^{\mathrm{int}}}
  + \left(\frac{\eta_{p}}{h_{\Upsilon}} \{\!\!\{\mu^{-1} k_2\}\!\!\}
  \llbracket \eta_{p_{2}} \rrbracket ; \llbracket e_{p_2}
  \rrbracket\right)_{\Gamma^{\mathrm{int}}}
    \end{align}
  By employing Cauchy-Schwarz and
  Peter-Paul inequalities, we obtain
  the following bound\footnote{For 
    convenience of the reader, we
    color-coded the terms. (See the
    online version for the colored text.)
    The {\color{red} red-colored} terms 
    contain interpolation errors and
    contribute to
    $\|\mathbf{H}\|_{\mathrm{stab}}^{\mathrm{DG}}$.
    The {\color{blue} blue-colored} terms 
    contain approximation errors and
    contribute to
    $\|\mathbf{E}\|_{\mathrm{stab}}^{\mathrm{DG}}$.
    We employ Lemma \ref{Lemma:DG_approx_error_Gamma_int}
    on the {\color{magenta} magenta-colored}
    terms and employ Lemma \ref{Lemma:DG_int_error_Gamma_int}
    on the {\color{ForestGreen} green-colored} terms. 
  }: 
  {\footnotesize{  
  \begin{align}
    2 \mathcal{B}_{\mathrm{stab}}^{\mathrm{DG}}
    (\mathbf{H};\mathbf{E})
    &\leq 
    {\color{red}
    	\frac{\varepsilon_1}{2} \left\| \sqrt{\frac{\mu} {k_1}}
    \boldsymbol{\eta}_{\mathbf{u}_{1}} \right\|^{2} }
    {\color{blue} + \frac{1}{2\varepsilon_1}
      \left\| \sqrt{\frac{\mu} {k_1}}
      \mathbf{e}_{\mathbf{u}_{1}} \right\|^{2}}
    {\color{red}
    + \frac{\varepsilon_2}{2} \left\| \sqrt{\frac{\mu} {k_1}}
    \boldsymbol{\eta}_{\mathbf{u}_{1}} \right\|^{2} }
    {\color{blue}+ \frac{1}{2\varepsilon_2} \left\|
    \sqrt{\frac{k_1} {\mu}}
    \mathrm{grad}[e_{p_{1}}] \right\|^{2}} \nonumber \\
    & {\color{red}
    + \frac{\varepsilon_3}{2} \left\| \sqrt{\frac{k_1} {\mu}} \mathrm{grad}[\eta_{p_{1}}] \right\|^2 }
    {\color{blue}+ \frac{1}{2\varepsilon_3}
      \left\| \sqrt{\frac{\mu} {k_1}}
    \mathbf{e}_{\mathbf{u}_{1}} \right\|^2} 
    {\color{red}
    + \frac{\varepsilon_{4}}{2} \left\| \sqrt{\frac{k_1} {\mu}} \mathrm{grad}[\eta_{p_{1}}] \right\|^2  }
    {\color{blue}
    + \frac{1}{2 \varepsilon_{4}} \left\| \sqrt{\frac{k_1} {\mu}} \mathrm{grad}[e_{p_{1}}] \right\|^2 
    }
    \nonumber \\
    &
    {\color{red} + \frac{\varepsilon_5}{2} \left\| \sqrt{\frac{\mu} {k_2}}
    \boldsymbol{\eta}_{\mathbf{u}_{2}} \right\|^{2} }
    {\color{blue}
    + \frac{1}{2\varepsilon_5} \left\| \sqrt{\frac{\mu} {k_2}}
    \mathbf{e}_{\mathbf{u}_{2}} \right\|^{2}  
    }
    {\color{red}
    + \frac{\varepsilon_6}{2} \left\| \sqrt{\frac{\mu} {k_2}}
     \boldsymbol{\eta}_{\mathbf{u}_{2}} \right\|^{2} }
    {\color{blue}
    + \frac{1}{2\varepsilon_6}\left\| \sqrt{\frac{k_2} {\mu}}
     \mathrm{grad}[e_{p_{2}}] \right\|^{2} 
     }
    \nonumber \\
    &{\color{red}
    + \frac{\varepsilon_7}{2} \left\| \sqrt{\frac{k_2} {\mu}}
    \mathrm{grad}[\eta_{p_{2}}] \right\|^2 }
    {\color{blue}
    + \frac{1}{2\varepsilon_7} \left\| \sqrt{\frac{\mu} {k_2}}
    \mathbf{e}_{\mathbf{u}_{2}} \right\|^2 
     }
    {\color{red}+ \frac{\varepsilon_{8}}{2} \left\|
    \sqrt{\frac{k_2} {\mu}}
    \mathrm{grad}[\eta_{p_{2}}] \right\|^2  }
    {\color{blue}
    + \frac{1}{2 \varepsilon_{8}} \left\| \sqrt{\frac{k_2} {\mu}}
     \mathrm{grad}[e_{p_{2}}] \right\|^2
    }
    \nonumber \\
    & {\color{red}
    + \varepsilon_{9} \left\|\sqrt{\frac{\beta}{\mu}}
    (\eta_{p_{1}}-\eta_{p_{2}})\right\|^2 }
    {\color{blue}
    + \frac{1}{\varepsilon_{9}} \left\|
    \sqrt{\frac{\beta}{\mu}}
    (e_{p_{1}}-e_{p_{2}}) \right\|^2 
    }
     \nonumber \\
     &{\color{ForestGreen}+ \varepsilon_{10} \left\|
        \sqrt{\frac{h_{\Upsilon}}{\eta_p}
    \{\!\!\{ \mu k_1^{-1} \}\!\!\}} \; 
    \{\!\!\{ \boldsymbol{\eta}_{\mathbf{u}_1} \}\!\!\}
    \right\|^2_{\Gamma^{\mathrm{int}}}}
    {\color{blue}
    + \frac{1}{\varepsilon_{10}} \left\| \sqrt{\frac{\eta_p}{h_{\Upsilon}}
      \{\!\!\{\mu^{-1} k_1\}\!\!\}} \; 
    \llbracket e_{p_1} \rrbracket \right\|^2_{\Gamma^{\mathrm{int}}} 
    }
     \nonumber \\
     &{\color{ForestGreen}
        + \varepsilon_{11} \left\|
        \sqrt{h_{\Upsilon}^{-1}
        \{\!\!\{\mu^{-1} k_1\}\!\!\}} \;
        \llbracket \eta_{p_1} \rrbracket
        \right\|^2_{\Gamma^{\mathrm{int}}}
    }
    {\color{magenta}
     + \frac{1}{\varepsilon_{11}} \left\| \sqrt{h_{\Upsilon}
     \{\!\!\{ \mu k_1^{-1} \}\!\!\}} \;
     \{\!\!\{ \mathbf{e}_{\mathbf{u}_1} \}\!\!\}
     \right\|^2_{\Gamma^{\mathrm{int}}}
  }
     \nonumber \\
    &{\color{ForestGreen}+ \varepsilon_{12} \left\|
    \sqrt{\frac{h_{\Upsilon}}{\eta_p}
    \{\!\!\{ \mu k_2^{-1}\}\!\!\}} \;
    \{\!\!\{ \boldsymbol{\eta}_{\mathbf{u}_2} \}\!\!\}
    \right\|^2_{\Gamma^{\mathrm{int}}}}
    {\color{blue}
    + \frac{1}{\varepsilon_{12}} \left\|
    \sqrt{\frac{\eta_p}{h_{\Upsilon}}
    \{\!\!\{\mu^{-1} k_2\}\!\!\}} \; 
    \llbracket e_{p_2} \rrbracket
    \right\|^2_{\Gamma^{\mathrm{int}}} 
    }
    \nonumber \\
    &{\color{ForestGreen}+ \varepsilon_{13}
       \left\|\sqrt{h_{\Upsilon}^{-1}
         \{\!\!\{\mu^{-1} k_2\}\!\!\}} \; 
       \llbracket \eta_{p_2} \rrbracket
       \right\|^2_{\Gamma^{\mathrm{int}}}
     }
    {\color{magenta}
      + \frac{1}{\varepsilon_{13}}
      \left\|\sqrt{h_{\Upsilon}
    \{\!\!\{ \mu k_2^{-1} \}\!\!\}} \; 
    \{\!\!\{ \mathbf{e}_{\mathbf{u}_2}\}\!\!\}
    \right\|^2_{\Gamma^{\mathrm{int}}}
     }
     \nonumber \\
    &{\color{red}+\varepsilon_{14} \left\| \sqrt{\eta_u h_{\Upsilon}
    \{\!\!\{ \mu k_1^{-1} \}\!\!\}} \; 
    \llbracket \boldsymbol{\eta}_{\mathbf{u}_1}
    \rrbracket \right\|^2 _{\Gamma^{\mathrm{int}}}}
    {\color{blue}
      + \frac{1}{\varepsilon_{14}} \left\|\sqrt{\eta_u h_{\Upsilon}
        \{\!\!\{ \mu k_1^{-1} \}\!\!\}} \; \llbracket
      \mathbf{e}_{\mathbf{u}_1} \rrbracket
      \right\|^2_{\Gamma^{\mathrm{int}}}
    }
     \nonumber \\
     &{\color{red}+ \varepsilon_{15} \left\|
        \sqrt{\frac{\eta_p}{h_{\Upsilon}}
          \{\!\!\{\mu^{-1} k_1\}\!\!\}}\;
        \llbracket \eta_{p_1} \rrbracket
        \right\|^2_{\Gamma^{\mathrm{int}}}}
    {\color{blue}
    + \frac{1}{\varepsilon_{15}} \left\|
    \sqrt{\frac{\eta_p}{h_{\Upsilon}}
    \{\!\!\{\mu^{-1} k_1\}\!\!\}} \;
    \llbracket e_{p_1} \rrbracket
    \right\|^2_{\Gamma^{\mathrm{int}}} 
    }
    \nonumber \\
    &{\color{red}+ \varepsilon_{16} \left\| \sqrt{\eta_u
      h_{\Upsilon} \{\!\!\{ \mu k_2^{-1} \}\!\!\}}\;
    \llbracket \boldsymbol{\eta}_{\mathbf{u}_2}
    \rrbracket \right\|^2_{\Gamma^{\mathrm{int}}}}
    {\color{blue}
    + \frac{1}{\varepsilon_{16}} \left\|
    \sqrt{\eta_u h_{\Upsilon} \{\!\!\{
      \mu k_2^{-1} \}\!\!\}} \; \llbracket
    \mathbf{e}_{\mathbf{u}_2} \rrbracket
    \right\|^2_{\Gamma^{\mathrm{int}}}
     }
     \nonumber \\
     &{\color{red}+ \varepsilon_{17} \left\|
        \sqrt{\frac{\eta_p}{h_{\Upsilon}}
          \{\!\!\{\mu^{-1} k_2\}\!\!\}}\;
        \llbracket \eta_{p_2} \rrbracket
        \right\|^2_{\Gamma^{\mathrm{int}}}}
    {\color{blue}
    + \frac{1}{\varepsilon_{17}} \left\|
    \sqrt{\frac{\eta_p}{h_{\Upsilon}}
    \{\!\!\{\mu^{-1} k_2\}\!\!\}}
    \; \llbracket e_{p_2} \rrbracket
    \right\|^2_{\Gamma^{\mathrm{int}}}
   }
    \label{Eqn:Error_inequality}
  \end{align}
  }}
  with $\varepsilon_{i}$ ($i=1,\cdots, 17$)
  are arbitrary positive constants. 
  After employing Lemma
  \ref{Lemma:DG_approx_error_Gamma_int},
  the above inequality can be grouped
  as follows:
{\footnotesize{  
    \begin{align}
      2 \mathcal{B}_{\mathrm{stab}}^{\mathrm{DG}}
      (\mathbf{H};\mathbf{E})
      & \leq {\color{blue}
        \left(\frac{1}{2\varepsilon_1}
        + \frac{1}{2\varepsilon_3}
        + \frac{\mathcal{C}_{\mathbf{e}_{\mathbf{u}_1}}}{\varepsilon_{11}}
        \right) \left\| \sqrt{\frac{\mu}{k_1}}
        \mathbf{e}_{\mathbf{u}_{1}} \right\|^{2}
      } 
      {\color{blue} + \left(\frac{1}{2\varepsilon_2}
        + \frac{1}{2\varepsilon_4} \right) \left\| 
        \sqrt{\frac{k_1}{\mu}}
        \mathrm{grad}[e_{p_{1}}] \right\|^{2}
      } \nonumber \\
      & \quad  {\color{blue}
        + \left(\frac{1}{2\varepsilon_5} + \frac{1}{2\varepsilon_7}
        + \frac{\mathcal{C}_{\mathbf{e}_{\mathbf{u}_2}}}{\varepsilon_{13}}
        \right) \left\| \sqrt{\frac{\mu}{k_2}} \mathbf{e}_{\mathbf{u}_{2}}
        \right\|^{2}
      }
      {\color{blue} + \left(\frac{1}{2\varepsilon_6}
        + \frac{1}{2\varepsilon_8}\right) \left\|
        \sqrt{\frac{k_2}{\mu}} \mathrm{grad}[e_{p_{2}}]
        \right\|^{2}
      } \nonumber \\
      & {\color{blue}+ \frac{1}{\varepsilon_{9}} \left\|
          \sqrt{\frac{\beta}{\mu}} (e_{p_{1}}-e_{p_{2}})
          \right\|^2
        } \nonumber \\
      &\quad {\color{blue}+
        \left(\frac{1}{\varepsilon_{10}} + \frac{1}{\varepsilon_{15}}
        \right) \left\| \sqrt{\frac{\eta_p}{h_{\Upsilon}}
          \{\!\!\{\mu^{-1} k_1\}\!\!\}} \;
        \llbracket e_{p_1} \rrbracket \right\|^2_{\Gamma^{\mathrm{int}}}
        }
      {\color{blue}+ \frac{1}{\varepsilon_{14}} \left\|
        \sqrt{\eta_u h_{\Upsilon} \{\!\!\{ \mu k_1^{-1} \}\!\!\}}
        \; \llbracket \mathbf{e}_{\mathbf{u}_1} \rrbracket
        \right\|^2_{\Gamma^{\mathrm{int}}}
      } \nonumber \\
      &\quad {\color{blue}+
        \left(\frac{1}{\varepsilon_{12}} +
        \frac{1}{\varepsilon_{17}}\right)
        \left\| \sqrt{\frac{\eta_p}{h_{\Upsilon}}
          \{\!\!\{\mu^{-1} k_2\}\!\!\}} \;
        \llbracket e_{p_2} \rrbracket \right\|^2_{\Gamma^{\mathrm{int}}}
      }
      {\color{blue}+
        \frac{1}{\varepsilon_{16}} \left\|
        \sqrt{\eta_u h_{\Upsilon} \{\!\!\{ \mu k_2^{-1} \}\!\!\}}
        \; \llbracket \mathbf{e}_{\mathbf{u}_2} \rrbracket
        \right\|^2_{\Gamma^{\mathrm{int}}}
      } \nonumber \\
      &\quad{\color{red}+ \left(\frac{\varepsilon_1}{2} +
        \frac{\varepsilon_2}{2} \right) \left\|
        \sqrt{\frac{\mu}{k_1}}\boldsymbol{\eta}_{\mathbf{u}_{1}}
        \right\|^{2}
        + \left( \frac{\varepsilon_3}{2} + \frac{\varepsilon_4}{2}
        \right) \left\| \sqrt{\frac{k_1}{\mu}} 
	\mathrm{grad}[\eta_{p_{1}}] \right\|^2
      }	\nonumber \\
      &\quad {\color{red}+
        \left( \frac{\varepsilon_5}{2} + \frac{\varepsilon_6}{2}
        \right) \left\| \sqrt{\frac{\mu}{k_2}}
        \boldsymbol{\eta}_{\mathbf{u}_{2}} \right\|^{2}
        + \left( \frac{\varepsilon_7}{2} + \frac{\varepsilon_8}{2}
        \right) \left\| \sqrt{\frac{k_2}{\mu}}
        \mathrm{grad}[\eta_{p_{2}}] \right\|^2
      } \nonumber \\
      &\quad {\color{red}+
        \varepsilon_{9} \left\|\sqrt{\frac{\beta}{\mu}}
        (\eta_{p_{1}}-\eta_{p_{2}})\right\|^2
        } \nonumber \\
      & \quad {\color{red}+
        \varepsilon_{14} \left\| \sqrt{\eta_u h_{\Upsilon}
          \{\!\!\{\mu k_1^{-1} \}\!\!\}} \;
        \llbracket \boldsymbol{\eta}_{\mathbf{u}_1} \rrbracket
        \right\|^2_{\Gamma^{\mathrm{int}}}
      }
      {\color{red}+ \varepsilon_{15} \left\|
        \sqrt{\frac{\eta_p}{h_{\Upsilon}}
          \{\!\!\{\mu^{-1} k_1\}\!\!\}} \;
        \llbracket \eta_{p_1} \rrbracket
        \right\|^2_{\Gamma^{\mathrm{int}}}
      } \nonumber \\
      &\quad {\color{red}+
        \varepsilon_{16} \left\| \sqrt{\eta_u h_{\Upsilon}
          \{\!\!\{ \mu k_2^{-1} \}\!\!\}} \;
        \llbracket \boldsymbol{\eta}_{\mathbf{u}_2} \rrbracket
        \right\| ^2 _{\Gamma^{\mathrm{int}}}
      }
      {\color{red}+ \varepsilon_{17}
        \left\| \sqrt{\frac{\eta_p}{h_{\Upsilon}}
          \{\!\!\{\mu^{-1} k_2\}\!\!\}} \;
        \llbracket \eta_{p_2} \rrbracket
        \right\|^2_{\Gamma^{\mathrm{int}}}
      } \nonumber \\
      &\quad {\color{ForestGreen}+
        \varepsilon_{10} \left\|\sqrt{\frac{h_{\Upsilon}}{\eta_p}
          \{\!\!\{ \mu k_1^{-1} \}\!\!\}} \;
        \{\!\!\{ \boldsymbol{\eta}_{\mathbf{u}_1}
        \}\!\!\} \right\|^2_{\Gamma^{\mathrm{int}}}
      }
              {\color{ForestGreen}+ \varepsilon_{11}
                \left\| \sqrt{h_{\Upsilon}^{-1}
                  \{\!\!\{\mu^{-1} k_1\}\!\!\}}
                \llbracket \eta_{p_1} \rrbracket
                \right\|^2_{\Gamma^{\mathrm{int}}}
              } \nonumber \\
              &\quad{\color{ForestGreen}+\varepsilon_{12}
                \left\| \sqrt{\frac{h_{\Upsilon}}{\eta_p}
                  \{\!\!\{ \mu k_2^{-1} \}\!\!\}} \;
                          \{\!\!\{ \boldsymbol{\eta}_{\mathbf{u}_2}
                          \}\!\!\} \right\|^2_{\Gamma^{\mathrm{int}}}}
                        {\color{ForestGreen}+ \varepsilon_{13} \left\|
                        \sqrt{h_{\Upsilon}^{-1}
                          \{\!\!\{\mu^{-1} k_2\}\!\!\}} \;
                        \llbracket \eta_{p_2} \rrbracket
                        \right\|^2_{\Gamma^{\mathrm{int}}}
                        }
			\label{Eqn:Error_inequality2}
    \end{align}
}}    
We choose the coefficients of the
first nine terms (i.e., blue-colored
terms) in such a way that these nine
terms add up to the square of
$\|\mathbf{E}\|_{\mathrm{stab}}^{\mathrm{DG}}$.
This can be achieved by choosing
these coefficients as follows:
\begin{align}
  \frac{1}{2\varepsilon_1} + \frac{1}{2\varepsilon_3}
  + \frac{\mathcal{C}_{\mathbf{e}_{\mathbf{u}_1}}}{\varepsilon_{11}}
  = \frac{1}{2\varepsilon_5} + \frac{1}{2\varepsilon_7}
  + \frac{\mathcal{C}_{\mathbf{e}_{\mathbf{u}_2}}}{\varepsilon_{13}}
  = \frac{1}{2}, \; 
  \left( \frac{1}{2\varepsilon_2} +
  \frac{1}{2\varepsilon_4} \right)
  = \left( \frac{1}{2\varepsilon_6} +
  \frac{1}{2\varepsilon_8} \right)
  = \frac{1}{2}, \nonumber \\
 \frac{1}{\varepsilon_{9}} = \frac{1}{\varepsilon_{14}}  = \frac{1}{\varepsilon_{16}} = 1, \quad \mathrm{and} \quad 
\left(\frac{1}{\varepsilon_{10}} + \frac{1}{\varepsilon_{15}}\right) = \left(\frac{1}{\varepsilon_{12}} + \frac{1}{\varepsilon_{17}}\right) = 1
\end{align}
One way to satisfy the above constraints
is to make the following choices for the
individual constants:
\begin{align}
  \varepsilon_1 = \varepsilon_3 =
  \varepsilon_5 = \varepsilon_7 = 4, \; 
  \varepsilon_2 = \varepsilon_4
  = \varepsilon_6 =  \varepsilon_{8} =
  \varepsilon_{10} = \varepsilon_{12} =
  \varepsilon_{15} =  \varepsilon_{17} = 2, \;
  \varepsilon_{9} = \varepsilon_{14} =
  \varepsilon_{16} = 1, \nonumber \\ 
  \varepsilon_{11} =
  4\mathcal{C}_{\mathbf{e}_{\mathbf{u}_1}}
  \quad \mathrm{and} \quad
  \varepsilon_{13} =
  4\mathcal{C}_{\mathbf{e}_{\mathbf{u}_2}}
  \label{Eqn:varepsilon_choices}
\end{align}
By incorporating the above choices into
inequality \eqref{Eqn:Error_inequality2},
we obtain the following:
\begin{align}
  2 \left( \| \mathbf{E}\|_{\mathrm{stab}}^{\mathrm{DG}}\right)^2 & \leq
  {\color{blue}
    \left( \left\|\mathbf{E}
    \right\|_{\mathrm{stab}}^{\mathrm{DG}}\right)^2
  }
  {\color{red}
    +3 \left\| \sqrt{\frac{\mu}{k_1}}
    \boldsymbol{\eta}_{\mathbf{u}_{1}} \right\|^{2}
    +3 \left\| \sqrt{\frac{k_1}{\mu}} 
    \mathrm{grad}[\eta_{p_{1}}] \right\|^2
  } \nonumber \\
  &\quad {\color{red}
    +3 \left\| \sqrt{\frac{\mu}{k_2}}
    \boldsymbol{\eta}_{\mathbf{u}_{2}} \right\|^{2}
    + 3 \left\| \frac{k_2}{\sqrt{\mu}} 
    \mathrm{grad}[\eta_{p_{2}}] \right\|^2
    +  \left\|\sqrt{\frac{\beta}{\mu}}
    (\eta_{p_{1}}-\eta_{p_{2}})\right\|^2
  } \nonumber \\
  & \quad {\color{red}
    + 2 \left\| \sqrt{\frac{\eta_p}{h_{\Upsilon}}
      \{\!\!\{\mu^{-1} k_1\}\!\!\}} \;
    \llbracket \eta_{p_1} \rrbracket
    \right\|^2_{\Gamma^{\mathrm{int}}}
    +2 \left\|\sqrt{\frac{\eta_p}{h_{\Upsilon}}}
    \{\!\!\{\mu^{-1} k_2\}\!\!\}^{1/2} \;
    \llbracket \eta_{p_2} \rrbracket
    \right\|^2_{\Gamma^{\mathrm{int}}}
  } \nonumber \\
  &\quad {\color{red}
    + \left\| \sqrt{\eta_u h_{\Upsilon}
      \{\!\!\{ \mu k_1^{-1} \}\!\!\}} \; 
    \llbracket \boldsymbol{\eta}_{\mathbf{u}_1} \rrbracket
    \right\|^2_{\Gamma^{\mathrm{int}}}
    + \left\| \sqrt{\eta_u h_{\Upsilon}
      \{\!\!\{\mu k_2^{-1} \}\!\!\}} \;
    \llbracket \boldsymbol{\eta}_{\mathbf{u}_2} \rrbracket
    \right\|^2_{\Gamma^{\mathrm{int}}}
  } \nonumber \\
        &{\color{ForestGreen}
           \quad +4\mathcal{C}_{\mathbf{e}_{\mathbf{u}_1}}
           \left\| \sqrt{h_{\Upsilon}^{-1}
             \{\!\!\{\mu^{-1} k_1\}\!\!\}} \;
           \llbracket \eta_{p_1} \rrbracket 
           \right\|^2_{\Gamma^{\mathrm{int}}}
           +4\mathcal{C}_{\mathbf{e}_{\mathbf{u}_2}}
           \left\|\sqrt{h_{\Upsilon}^{-1}
           \{\!\!\{\mu^{-1} k_2\}\!\!\}} \;
           \llbracket \eta_{p_2} \rrbracket
           \right\|^2_{\Gamma^{\mathrm{int}}} 
	 } \nonumber \\
	&{\color{ForestGreen}
           \quad + 2 \left\| \sqrt{\frac{h_{\Upsilon}}{\eta_p}
             \{\!\!\{ \mu k_1^{-1} \}\!\!\}} \;
           \{\!\!\{ \boldsymbol{\eta}_{\mathbf{u}_1}
           \}\!\!\}\right\|^2_{\Gamma^{\mathrm{int}}}
	   + 2 \left\| \sqrt{\frac{h_{\Upsilon}}{\eta_p}
             \{\!\!\{ \mu k_2^{-1} \}\!\!\}} \;
           \{\!\!\{ \boldsymbol{\eta}_{\mathbf{u}_2} \}\!\!\}
           \right\|^2_{\Gamma^{\mathrm{int}}}  
	 } \nonumber \\
  &\leq
  {\color{blue}\left(\|\mathbf{E}
    \|_{\mathrm{stab}}^{\mathrm{DG}}\right)^2}
  {\color{red}
    +6 \left(\| \mathbf{H}\|_{\mathrm{stab}}^{\mathrm{DG}}\right)^2
  } \nonumber \\
  &\quad{\color{ForestGreen}
    + 4\mathcal{C}_{\mathbf{e}_{\mathbf{u}_1}}
     \left\| \sqrt{h^{-1}_{\Upsilon}
       \{\!\!\{\mu^{-1} k_1\}\!\!\}} \;
     \llbracket \eta_{p_1} \rrbracket
     \right\|^2_{\Gamma^{\mathrm{int}}}
  }
  {\color{ForestGreen}
     + 4 \mathcal{C}_{\mathbf{e}_{\mathbf{u}_2}}
     \left\| \sqrt{h^{-1}_{\Upsilon}
       \{\!\!\{\mu^{-1} k_2\}\!\!\}} \;
     \llbracket \eta_{p_2} \rrbracket
     \right\|^2_{\Gamma^{\mathrm{int}}}
     } \nonumber \\
  &\quad{\color{ForestGreen}
     +2 \left\| \sqrt{\frac{h_{\Upsilon}}{\eta_p}
       \{\!\!\{ \mu k_1^{-1} \}\!\!\}} \; 
     \{\!\!\{ \boldsymbol{\eta}_{\mathbf{u}_1}
     \}\!\!\} \right\|^2_{\Gamma^{\mathrm{int}}}
     + 2 \left\| \sqrt{\frac{h_{\Upsilon}}{\eta_p}
       \{\!\!\{ \mu k_2^{-1} \}\!\!\}} \;
     \{\!\!\{ \boldsymbol{\eta}_{\mathbf{u}_2} \}\!\!\}
     \right\|^2_{\Gamma^{\mathrm{int}}}
   }
\end{align}
Lemma \ref{Lemma:DG_int_error_Gamma_int}
implies the following:
\begin{align}
  \left( \| \mathbf{E}\|_{\mathrm{stab}}^{\mathrm{DG}}\right)^2
  \lesssim 6 \left( \| \mathbf{H}\|_{\mathrm{stab}}^{\mathrm{DG}}\right)^2
  &+ \sum_{\omega \in \mathcal{T}_h} \left(
  h_{\omega}^{2(p+1)} |\mathbf{u}_1|^2_{H^{p+1}(\omega)} 
  + h_{\omega}^{2(q+1)} |\mathbf{u}_2|^2_{H^{q+1}(\omega)}
  \right. \nonumber \\
  &\qquad \qquad \left. + h_{\omega}^{2r}|p_1|^2_{H^{r+1}(\omega)} 
  + h_{\omega}^{2s}|p_2|^2_{H^{s+1}(\omega)} \right) 
\end{align}
As $h \rightarrow 0$, $h_{\omega} \rightarrow 0 \;
\forall \omega \in \mathcal{T}_{h}$, which
in turn implies that $\|\mathbf{H}\|_{\mathrm{stab}}^{\mathrm{DG}}
\rightarrow 0$ (using Lemma
\ref{Lemma:DG_estimate_for_H_in_terms_of_stab_norm})
and all other terms on the right hand side tend
to zero (using Lemma \ref{Lemma:DG_int_error_Gamma_int}).
Thus, $\|\mathbf{E}_{\mathrm{stab}}^{\mathrm{DG}}\|
\rightarrow 0$ as $h \rightarrow 0$. Since
$\|\cdot\|_{\mathrm{stab}}^{\mathrm{DG}}$ is a norm (i.e.,
Lemma \ref{Lemma:DD2_stab_norm}), one can conclude
that $\mathbf{U}^{h} \rightarrow \mathbf{U}$ almost
everywhere as $h \rightarrow 0$.
\end{proof}

\begin{remark}
  The selection of constants $\varepsilon_{i}
  \; (i = 1, \cdots, 17)$ in equation
  \eqref{Eqn:varepsilon_choices} is arbitrary.
  We do not claim that this selection provides
  an optimal bound, which is not the aim of our
  paper. However, the selection is sufficient
  to establish the convergence of the proposed
  formulation.
\end{remark}

Lemmas \ref{Lemma:DG_int_error_Gamma_int} and
\ref{Lemma:DG_estimate_for_H_in_terms_of_stab_norm}
immediately give the following result:
\begin{corollary}{(Rates of convergence.)}
  \label{Corollary:DG_rates_of_convergence}
  Let $p$, $q$, $r$ and $s$ be the
  polynomial orders for approximating
  the fields $\mathbf{u}_1$, $\mathbf{u}_2$, 
  $p_1$ and $p_2$. Let the orders of regularity
  in terms of the Sobolev semi-norm for
  these solution fields be $\widehat{p}$,
  $\widehat{q}$, $\widehat{r}$ and
  $\widehat{s}$. Then the rates of convergence
  for these fields will be, respectively,
  $\min[p+1,\widehat{p}]$, $\min[q+1,\widehat{q}]$,
  $\min[r,\widehat{r}]$ and
  $\min[s,\widehat{s}]$. 
\end{corollary}

\begin{remark}
  In order for Lemma
  \ref{Lemma:DG_int_error_Gamma_int}
  to hold, $\eta_p \neq 0$, as $\eta_p$
  is in the denominator of the estimates
  \eqref{Eqn:int_error_avg_eta_u1} and
  \eqref{Eqn:int_error_avg_eta_u2}.
  Since the convergence theorem utilizes
  Lemma \ref{Lemma:DG_int_error_Gamma_int},
  the convergence of the proposed DG
  formulation is thus established for the
  case $\eta_p \neq 0$. However, numerical
  simulations suggest that the parameters
  $\eta_u$ and $\eta_p$ do not seem to have
  a noticeable effect on the results for
  problems involving conforming meshes
  and conforming interpolations. 
\end{remark}

\section{PATCH TESTS}
\label{Sec:S5_DG_Patch_tests}
Patch tests are generally used to indicate the quality of a finite element. Despite some debated mathematical controversies regarding the patch test, ``the patch test is the most practically useful technique for assessing element behavior''  as nicely pinpointed by \citep{hughes2012finite}. 
In this section, different constant flow patch tests are used to showcase various features of the proposed stabilized mixed DG formulation. First, the capability of the proposed formulation for modeling flow in a highly heterogeneous, layered porous domain with abrupt changes in macro- and micro-permeabilities, is shown. Then, the ability of the proposed stabilized mixed DG formulation for supporting non-conforming discretization, in the form of non-conforming order refinement and non-conforming element refinement, is assessed. Finally, the proposed stabilized mixed DG formulation is employed on meshes with non-constant Jacobian elements. 
For the case of non-conforming order refinement, a parametric study is performed to assess
the sensitivity of the solutions with respect to the stabilization parameters $\eta_{u}$ and $\eta_{p}$.

\subsection{Velocity-driven patch test}
\label{Sec:Vel_patchTest}
In reality, heterogeneity of the material properties
is indispensable when it comes to porous domains. In
many geological systems, medium properties can vary
by many orders of magnitude and rapid changes may
occur over small spatial scales.
The aim of this boundary value problem is
to show that the proposed stabilized mixed
DG formulation can perform satisfactorily
when the medium properties are heterogeneous. 

The heterogeneous domain consists of five horizontal
layers with different macro- and micro-permeabilities
in each layer.
As shown in \textbf{Fig.~\ref{Fig:Patch_test_velocitydriven}}, on the left side of each layer, a constant normal velocity ($\mathbf{u}_{i} \cdot \widehat{\mathbf{n}} = -\frac{k_i^{\mathrm{\#~layer}}}{\mu}$) is applied and on the right side, $\mathbf{u}_{i} \cdot \widehat{\mathbf{n}} = \frac{k_i^{\mathrm{\#~layer}}}{\mu}$ is prescribed. On the top and bottom of the domain, normal components of macro- and micro-velocities are prescribed to be zero. 
For uniqueness of the solution, pressure
is prescribed on one corner of the domain.
Table \ref{Tb3:2D_velocity_driven5HL_data}
provides the model parameters for this problem. 
{\small
  \begin{table}[!h]
    \caption{Model parameters 
        for velocity-driven patch test.}
    \centering
    \begin{tabular}{|c|c|} \hline
      Parameter & Value \\
      \hline
      $\gamma \mathbf{b}$ & $\{0.0,0.0\}$\\
      $L_x$ & $5.0$ \\
      $L_y$ & $4.0$\\
      $\mu $ & $1.0$ \\
      $\beta $ & $1.0$ \\
      $k$&  $0.2$ \\
      $\eta_u$& $100.0$\\
      $\eta_p$& $100.0$\\
      $h$&structured T3 mesh of size $0.04$ used\\
      \hline 
    \end{tabular}
    \label{Tb3:2D_velocity_driven5HL_data}
  \end{table}
}

As can be seen in \textbf{Fig.~\ref{Fig:Patch_test_velocitydriven_profiles}},
  velocities are constant and pressures are linearly varying in the horizontal direction in each layer, which are in agreement with the exact solution of this problem as remarked by \citep{Hughes_Masud_Wan_2006}.
This problem is also solved using the stabilized continuous Galerkin (CG) formulation of the DPP model developed by \citep{Nakshatrala_Joodat_Ballarini_P2} and the x-components of velocity profiles are compared under both DG and CG formulations at x=2.5 throughout the domain as shown in \textbf{Fig.~\ref{Fig:Patch_test_velocitydriven_VMS_DG2}}. As can be seen, spurious oscillations are observed along the interfaces of the layers under the CG formulation. Under the DG formulation, however, such oscillations are completely eliminated and the physical jumps in the velocity profiles are accurately captured across the interfaces.

\begin{figure}
\includegraphics[clip,width=0.65\linewidth]{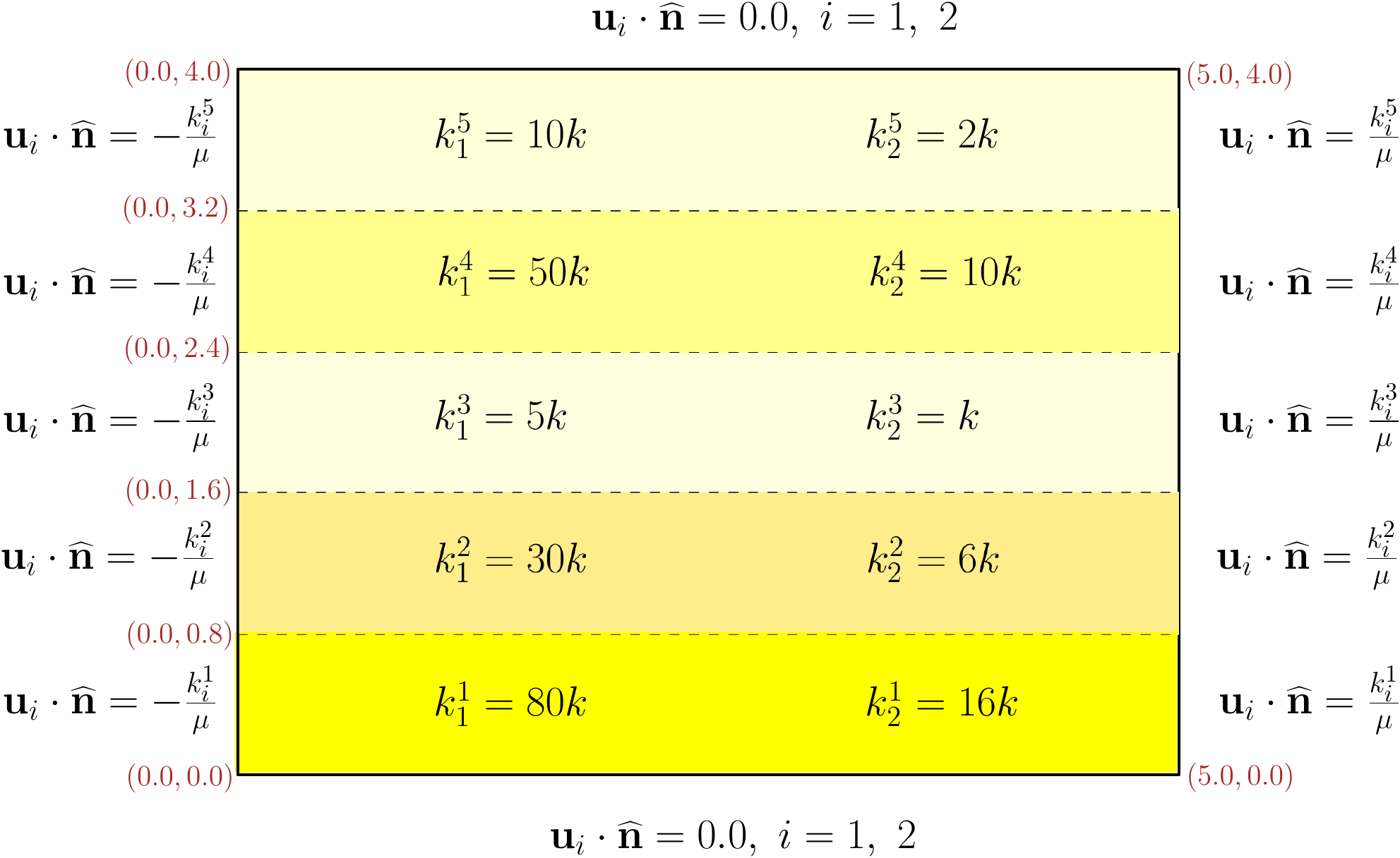}
\caption{\textsf{Velocity-driven patch test:}~This
  figure shows the computational domain, boundary
  conditions, and macro- and micro-permeabilities
  in each layer. \label{Fig:Patch_test_velocitydriven}}
\end{figure}

\begin{figure}[h]
\subfigure[Macro-pressure]{
\includegraphics[clip,width=0.42\linewidth]{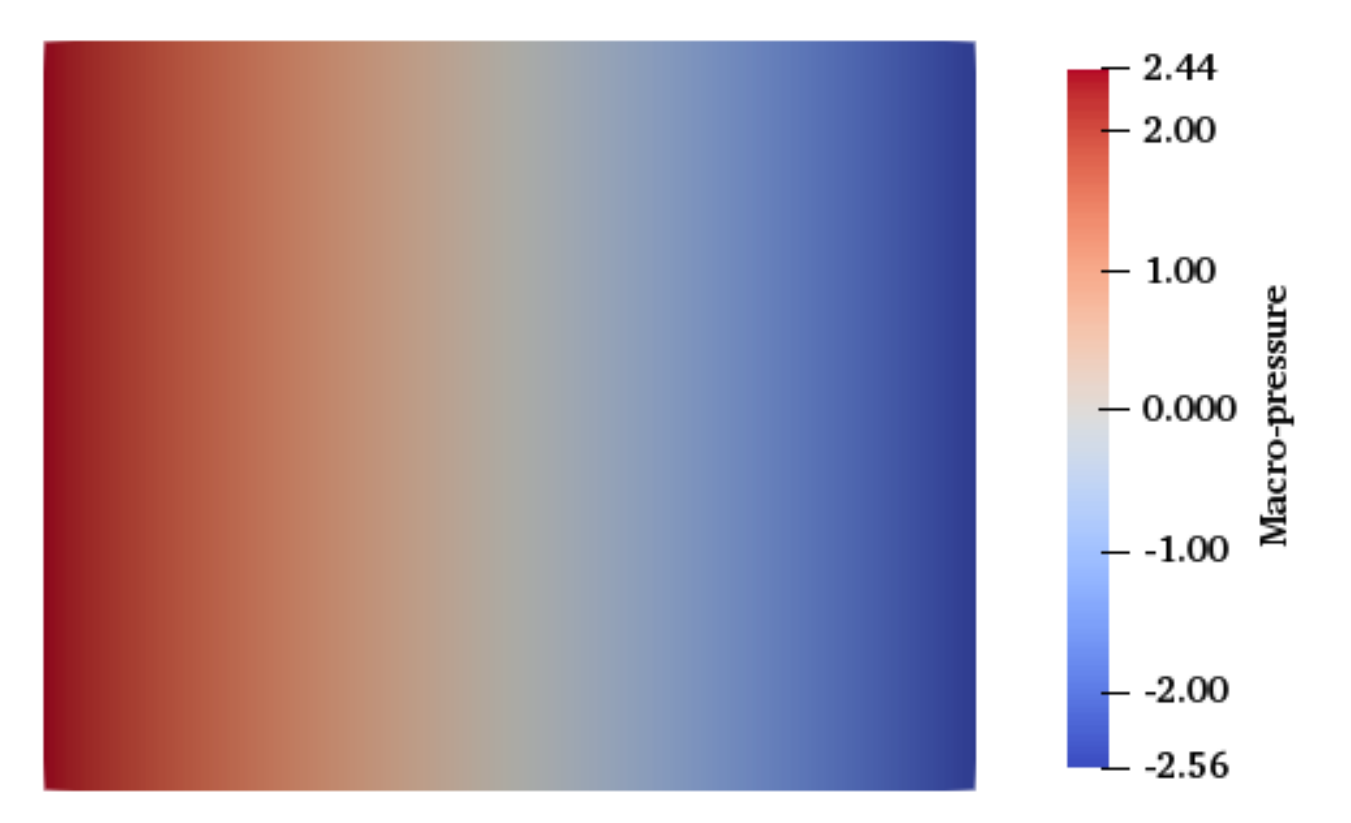}}
\subfigure[Micro-pressure]{
\includegraphics[clip,width=0.39\linewidth]{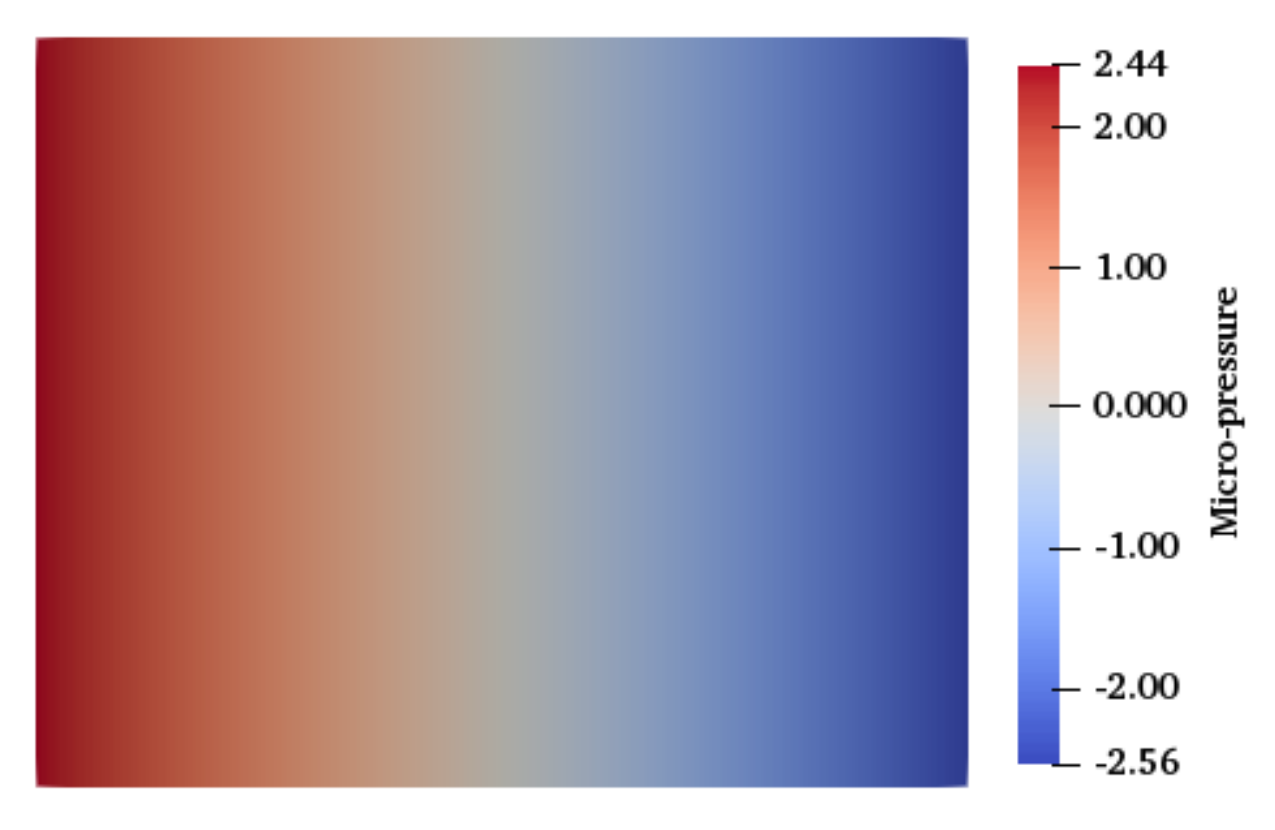}}
\subfigure[Macro-velocity]{
\includegraphics[clip,width=0.4\linewidth]{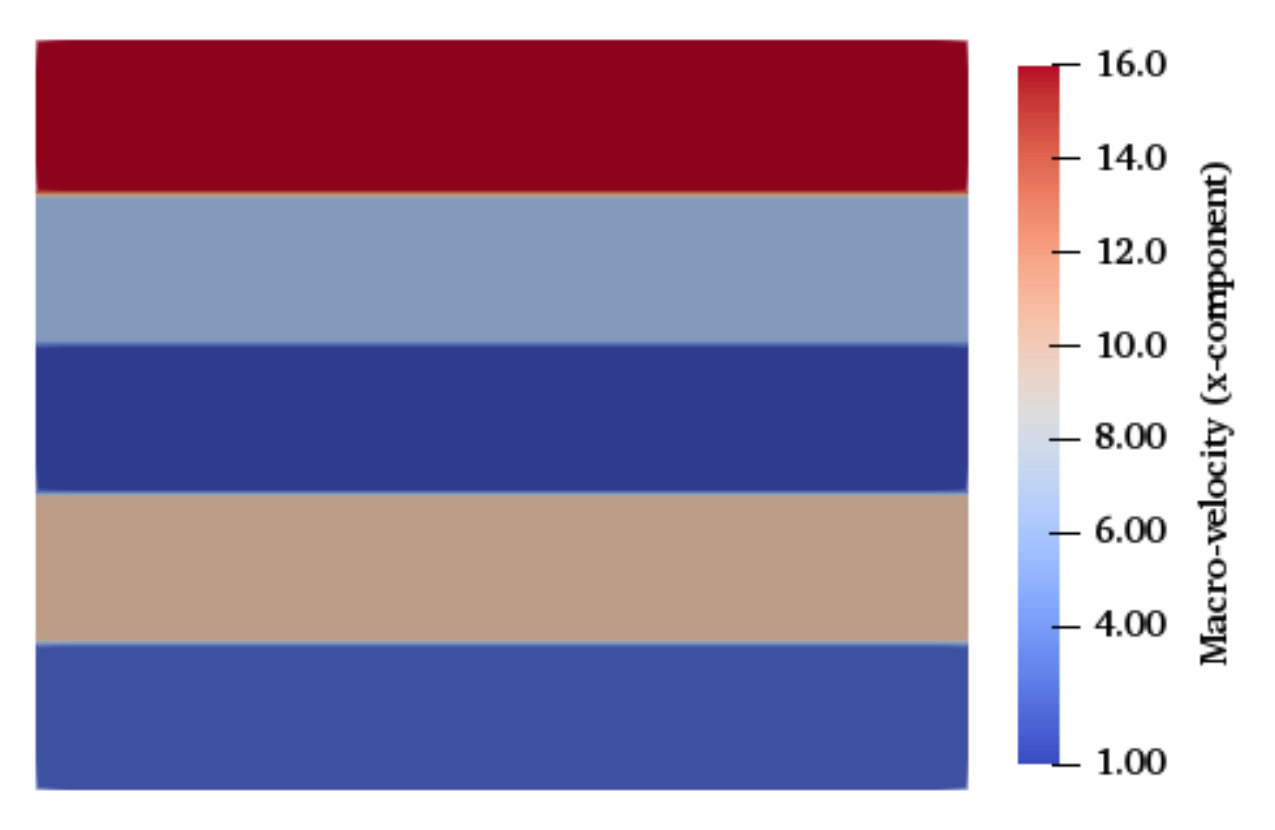}}
\subfigure[Micro-velocity]{
\includegraphics[clip,width=0.4\linewidth]{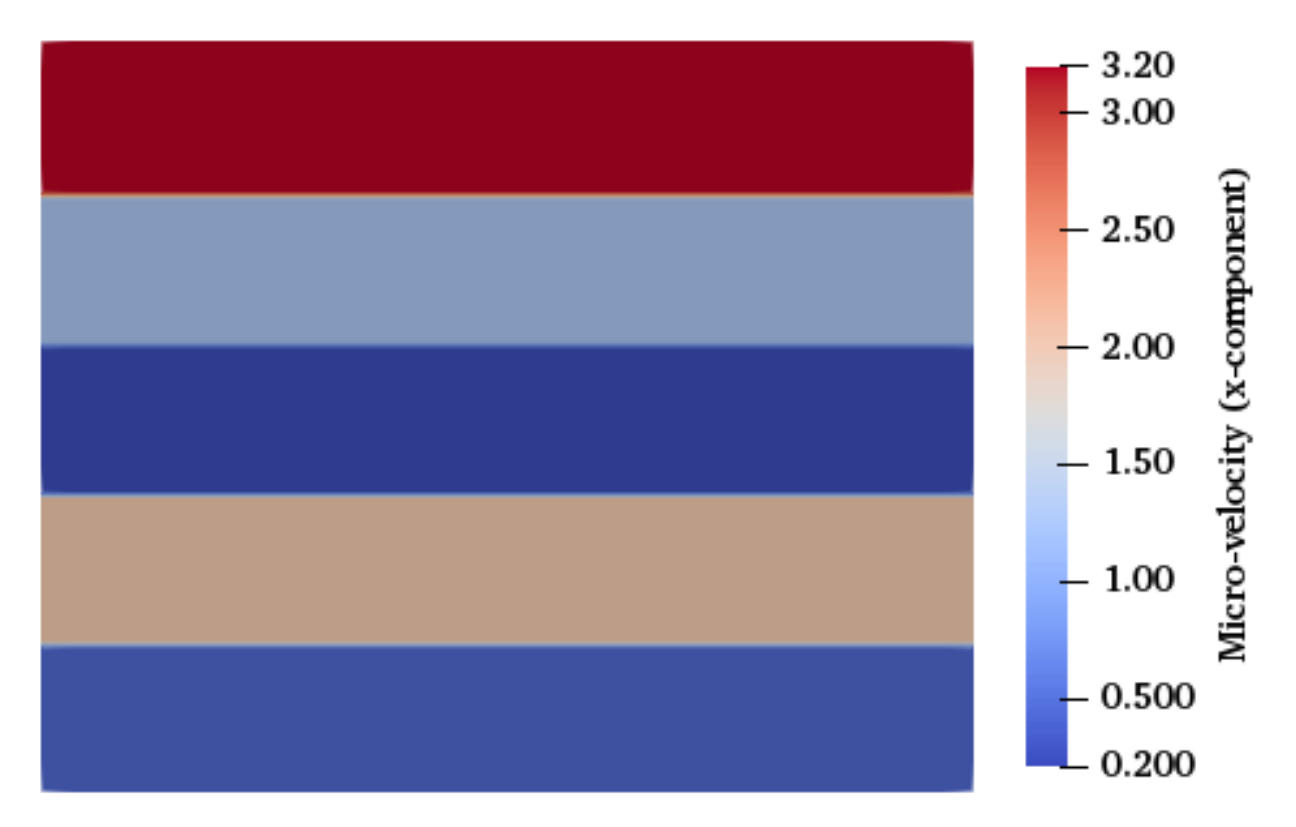}}
\caption{\textsf{Velocity-driven patch test:}~Velocities are constant within each layer and pressures are linearly varying in the horizontal direction which are in agreement with the exact solution of this problem. These results imply that the proposed formulation has successfully passed the velocity-driven patch test. \label{Fig:Patch_test_velocitydriven_profiles}}
\end{figure}

\begin{figure}[h]
	\subfigure{
		\includegraphics[clip,width=0.5\linewidth]{Figures/Fig5_a_Macro_vel_modified.pdf}}
	\subfigure{
		\includegraphics[clip,width=0.5\linewidth]{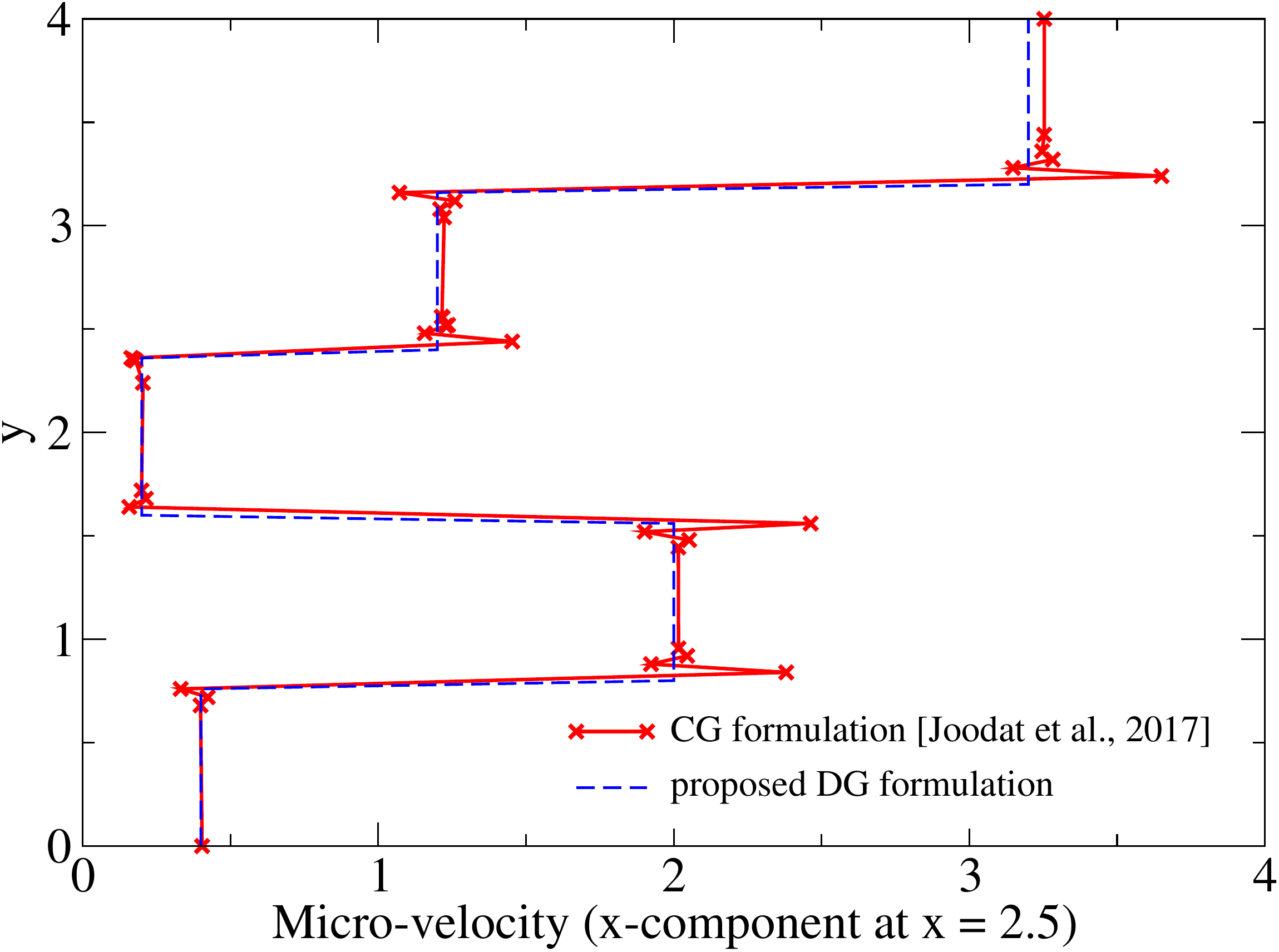}}
	\caption{\textsf{Velocity-driven patch test:}~This
          figure compares the velocities profiles obtained
          under the
          stabilized mixed CG
          formulation and the proposed DG formulation.
          The x-components of the macro-velocity (top)
          and micro-velocity (bottom)  
          at $x = 2.5$ are plotted. Under the CG formulation,
          overshoots and undershoots are observed along the
          interfaces of the layers. On the other hand, the
          proposed DG formulation is able to capture the
          physical jumps across the interfaces.
	\label{Fig:Patch_test_velocitydriven_VMS_DG2}}
\end{figure}
%
\subsection{Non-conforming discretization}
\label{Sec:Non_conforming}
One of the features of DG formulations is that the global error of the computation can be controlled by adjusting the numerical resolution in a selected set of the elements. Such a \emph{non-conforming discretization} can be obtained in two ways \citep{Hesthaven_Warburton_2007nodal}: One can either modify the local order of the interpolation, or locally change the element size in parts of the computational domain. \citep{babuvska2001finite,babuvska1981error} have discussed that the former method, also known as \emph{non-conforming order refinement} or \emph{non-conforming polynomial orders}, is more preferred for smooth problems. However, for the non-smooth case, which is due to the geometric features, sources, or boundary conditions, \emph{non-conforming element refinement} is the best choice. In the following, we show the application of non-conforming discretization under the proposed stabilized DG formulation using simply designed boundary value problems.
\subsubsection{Non-conforming polynomial orders}
\label{Sec5:2D_square}
Since the element communication under the DG formulations takes place through fluxes, each element can independently possess a desired order of interpolation. Hence, the DG methods can easily support the non-conforming polynomial orders (see \citep{remacle2003adaptive,canouet2005discontinuous,hesthaven2004high}).

In order to investigate the performance of our proposed stabilized mixed DG formulation under non-conforming polynomial orders,
a problem taken from
\citep{Nakshatrala_Joodat_Ballarini_P2} is used. 
The domain is considered to be a unit
square, with pressures being prescribed
on the entire boundary of both pore-networks as shown in \textbf{Fig.~ \ref{Fig:Dual_Problem_2D_domain}}. 
Prescribed pressure values on the respective boundary edges are obtained using the analytical solutions of this problem.
%
\begin{figure}[!h]
	\includegraphics[clip,scale=0.8]{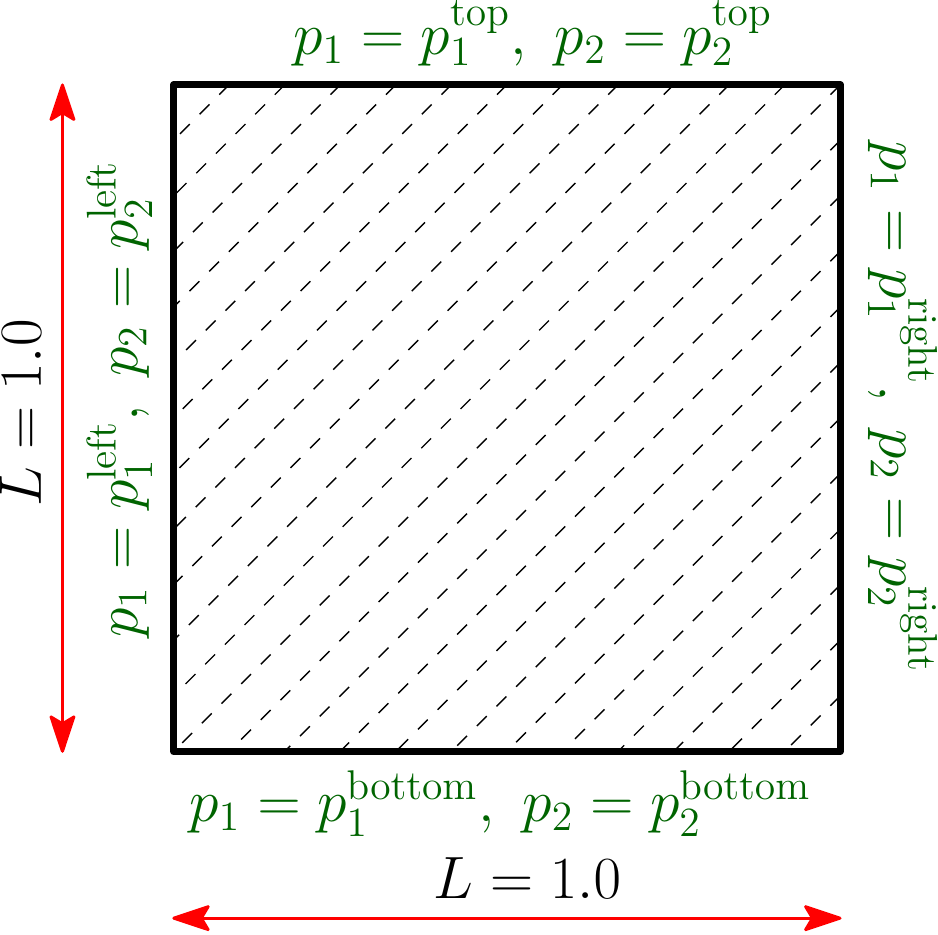}
	\caption{\textsf{Non-conforming polynomial orders:}~The computational domain in the 2D setting is a unit square. Pressures are prescribed on the entire boundary of both pore-networks. Prescribed pressure values on the respective boundary edges are obtained using the analytical solutions \eqref{Eqn:2D_Convergence_Analytical_p1} and \eqref{Eqn:2D_Convergence_Analytical_p2}.}
	\label{Fig:Dual_Problem_2D_domain}
	\vspace{1cm}
\end{figure}
The analytical solution for the pressure and velocity fields can be written as 
\begin{align}
	p_1(x,y) &= \frac{\mu}{\pi} \exp(\pi x) \sin(\pi y) - \frac{\mu}{\beta k_1} \exp(\eta y) \label{Eqn:2D_Convergence_Analytical_p1} \\
	p_2(x,y) &= \frac{\mu}{\pi} \exp(\pi x) \sin(\pi y) + \frac{\mu}{\beta k_2} \exp(\eta y)
	\label{Eqn:2D_Convergence_Analytical_p2}\\ 
	\mathbf{u}_1(x,y) &= -k_1\left(\begin{array}{c}
		\exp(\pi x) \sin(\pi y) \\
		\exp(\pi x) \cos(\pi y)
	\end{array}\right) 
	+ \left(\begin{array}{c}
		0 \\
		\frac{\eta}{\beta} \exp(\eta y) 
	\end{array}\right) 
	\\
	\mathbf{u}_2(x,y) &= -k_2\left(\begin{array}{c}
		\exp(\pi x) \sin(\pi y) \\
		\exp(\pi x) \cos(\pi y)
	\end{array}\right) 
	- \left(\begin{array}{c}
		0 \\
		\frac{\eta}{\beta} \exp(\eta y)
	\end{array}\right) 
	\label{Eqn:2D_square_Exact}
\end{align}
where
\begin{align}
	\label{Eqn:flow_characterization_param}
	\eta := \sqrt{\beta \frac{k_1 + k_2}{k_1 k_2}}
\end{align}
$\eta$ is a useful parameter to
characterize the flow of fluids
through porous media with double
porosity/permeability
\citep{Nakshatrala_Joodat_Ballarini}. 

{\small
  \begin{table}[!h]
    \caption{Model parameters for non-conforming polynomial orders, element-wise mass balance study, and 2D numerical convergence analysis.}
    \centering
    \begin{tabular}{|c|c|} \hline
      Parameter & Value \\
			\hline
			$\gamma \mathbf{b}$ & $\{0.0,0.0\}$\\
			$L $ & $1.0$ \\
			$\mu $ & $1.0$ \\
			$\beta $ & $1.0$ \\
			$k_1$&  $1.0$ \\
			$k_2$&  $0.1$\\
			$\eta$ & $\sqrt{11} \simeq 3.3166$\\
			$\eta_{u}$ & $10.0$\\
			$\eta_{p}$ & $1.0$\\
			$h$&structured T3 mesh of size $0.1$ used\\
			\hline
			$p_i^{\mathrm{left}},~i=1,2$&  Obtained by evaluating   \\
			$p_i^{\mathrm{right}},~i=1,2$&  the analytical solution \\
			$p_i^{\mathrm{top}},~i=1,2$&   (equations \eqref{Eqn:2D_Convergence_Analytical_p1} and \eqref{Eqn:2D_Convergence_Analytical_p2} ) \\
			$p_i^{\mathrm{bottom}},~i=1,2$&  on the respective boundaries. \\
			\hline 
		\end{tabular}
		\label{Tb4:2D_convergence_analysis_data}
	\end{table}
}
%
\begin{figure}[!h]
	\includegraphics[clip,scale=0.83]{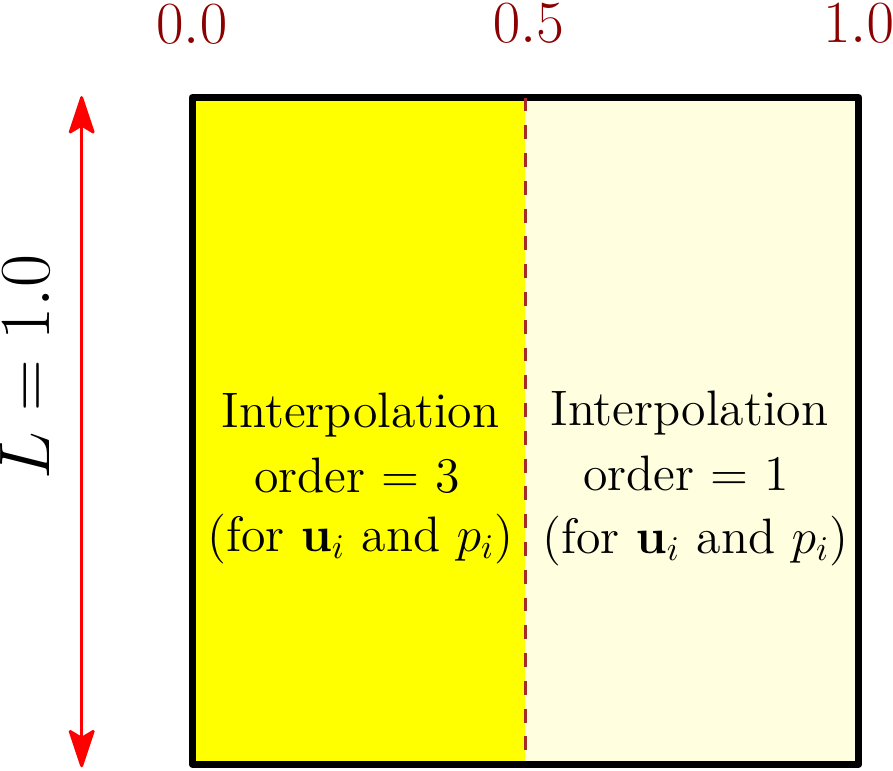}
	\caption{\textsf{Non-conforming polynomial orders:}~Different sets of equal-order interpolations used for this problem are shown in this figure. In the left part of the domain, third order interpolation polynomials are used for velocities and pressures, while in the right part, first order interpolation polynomials are used.
		\label{Fig:Dual_Problem_2D_domain2}}
\end{figure}
Table \ref{Tb4:2D_convergence_analysis_data} provides the parameter values for this problem.
In the left and right parts of the domain, two different sets of equal-order interpolation are employed for velocities and pressures as shown in \textbf{Fig.~\ref{Fig:Dual_Problem_2D_domain2}}. In the left half, third order interpolation polynomials are employed for velocities and pressures in each pore-network while in the right half, first-order interpolation polynomials are used. 

Smooth velocity profiles along the non-conforming edge ($x = 0.5$) are not achievable for a coarse
mesh (e.g., of size 10 x 10 elements mesh) without using extra stabilization terms (i.e., $\eta_{u} = \eta_{p} = 0$). 
One can either apply exhaustive mesh refinement, which in turn leads to a much higher computational cost, or can circumvent the unnecessary refinements by alternatively taking advantage of non-zero $\eta_{u}$ and $\eta_{p}$. 
\textbf{Figs.~\ref{Fig:Order_refinement_parametric-etap}}--\textbf{\ref{Fig:Order_refinement_parametric-etaup}} illustrate
the sensitivity of x-component of velocities along the non-conforming edge with respect to $\eta_{u}$, $\eta_{p}$ and their combined effect.
According to \textbf{Figs.~\ref{Fig:EtaP_macro}} and \textbf{\ref{Fig:EtaP_micro}}, the increase in $\eta_{p}$ per se in the absence of $\eta_{u}$ slightly improves the results. However, for the case of $\eta_{p} = 0$ and non-zero
$\eta_{u}$, a drastic enhancement is captured with $\eta_{u}$ of order one as shown in \textbf{Figs.~
\ref{Fig:EtaU_macro}} and \textbf{\ref{Fig:EtaU_micro}}. \textbf{Figs.~\ref{Fig:Eta_macro}} and \textbf{\ref{Fig:Eta_micro}} 
show the combined effect of $\eta_{u}$ and $\eta_{p}$ along the non-conforming edge in minimizing the drifts of macro and micro-velocity fields.
\begin{figure}
  \subfigure{
    \label{Fig:EtaP_macro}
    \includegraphics[trim={0 1cm 0 1cm},clip,width=0.7\linewidth]{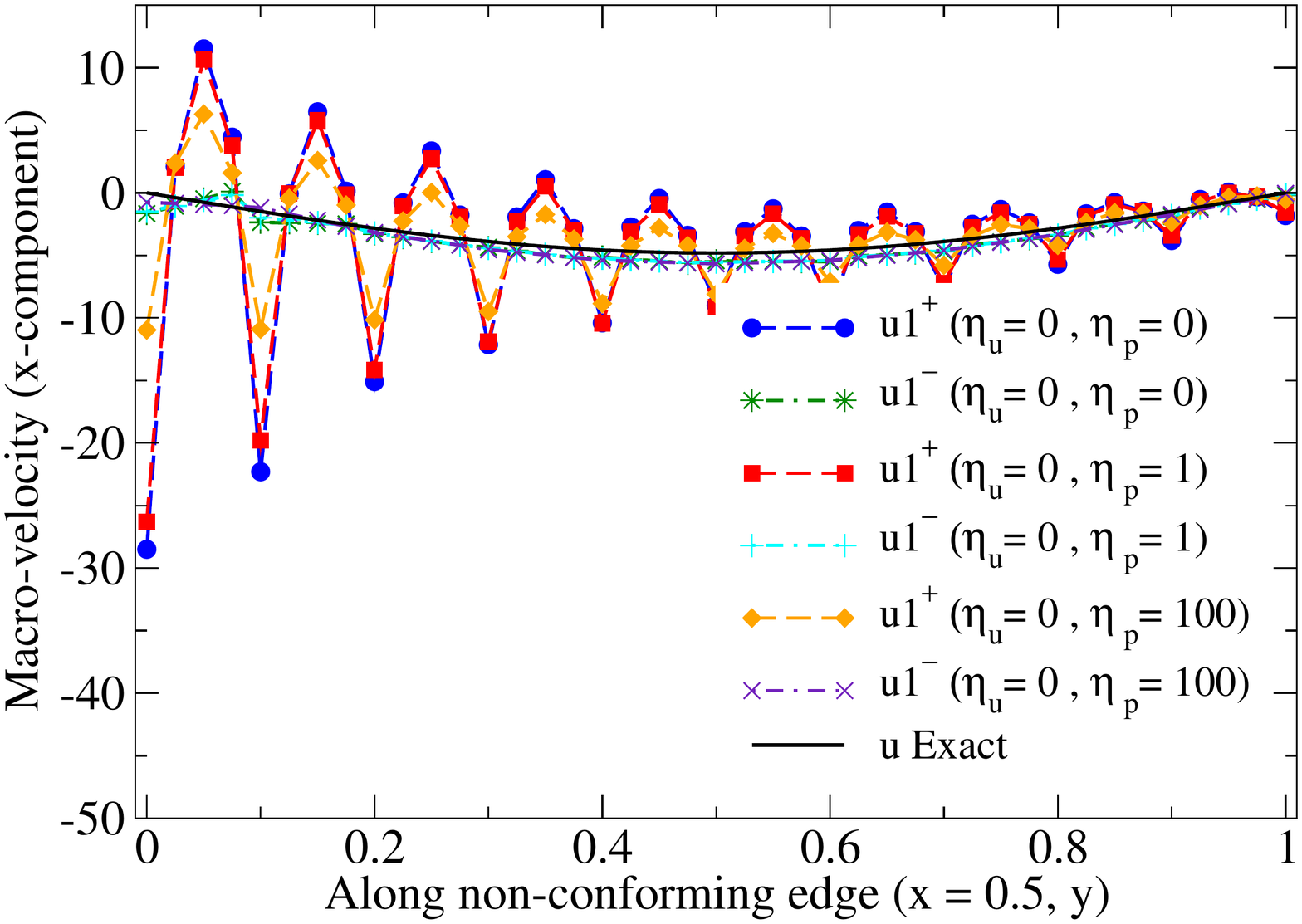}}
  \subfigure{
    \label{Fig:EtaP_micro}
    \includegraphics[trim={0 1.3cm 0 2.5cm},clip,width=0.7\linewidth]{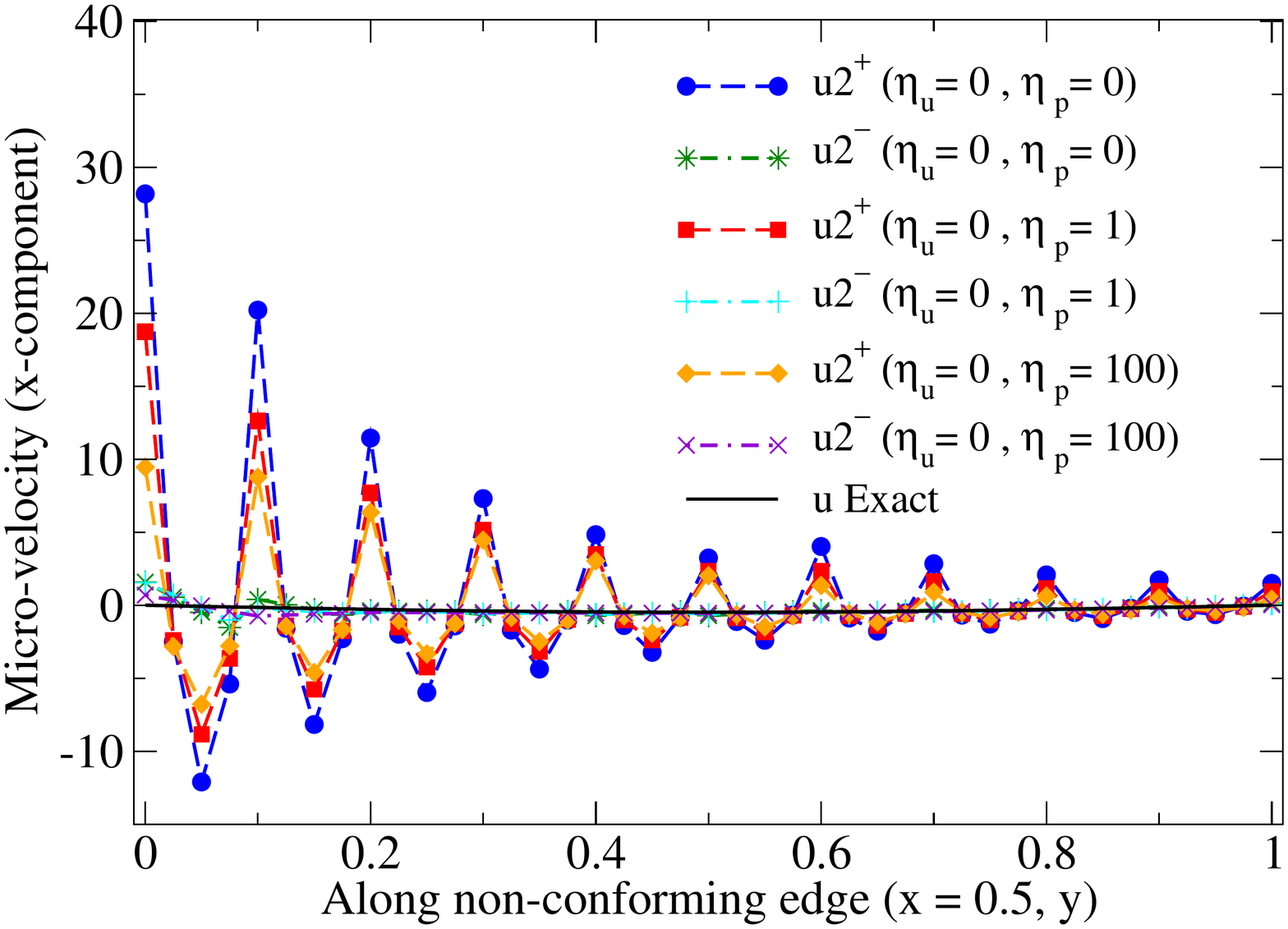}}
  \caption{\textsf{Non-conforming polynomial orders:}~This
    parametric study demonstrates that an increase in $\eta_{p}$
    in the absence of $\eta_{u}$ slightly improves the accuracy
    in capturing the jumps of the (macro- and micro-) velocities across
    a non-conforming edge. }
  \label{Fig:Order_refinement_parametric-etap}
\end{figure}
%
\begin{figure}
  \subfigure{
    \label{Fig:EtaU_macro}
    \includegraphics[trim={0 1cm 0 1cm},clip,width=0.7\linewidth]{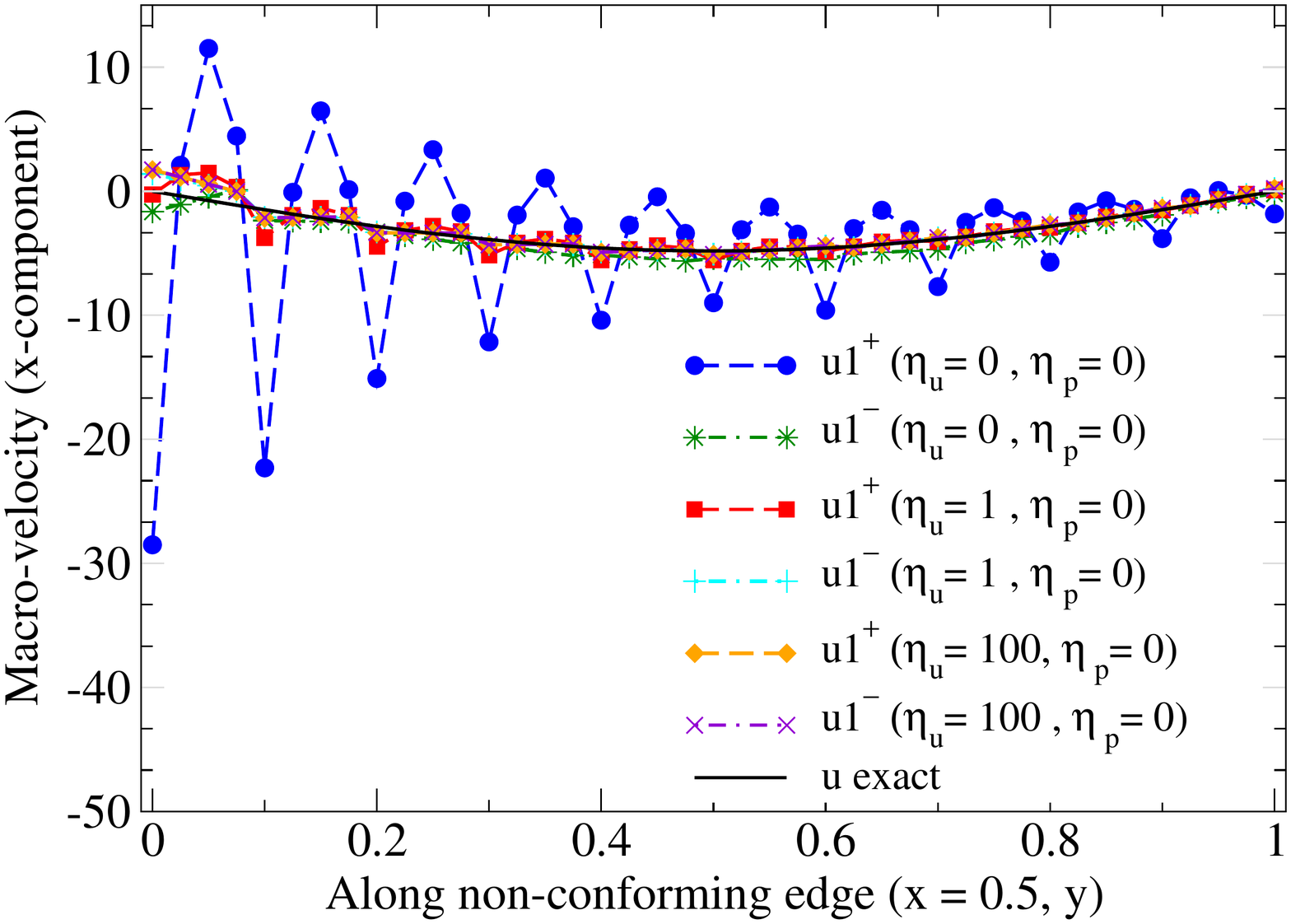}}
  \hspace{1cm}
  \subfigure{
    \label{Fig:EtaU_micro}
    \includegraphics[trim={0 1.3cm 0 2.5cm},clip,width=0.7\linewidth]{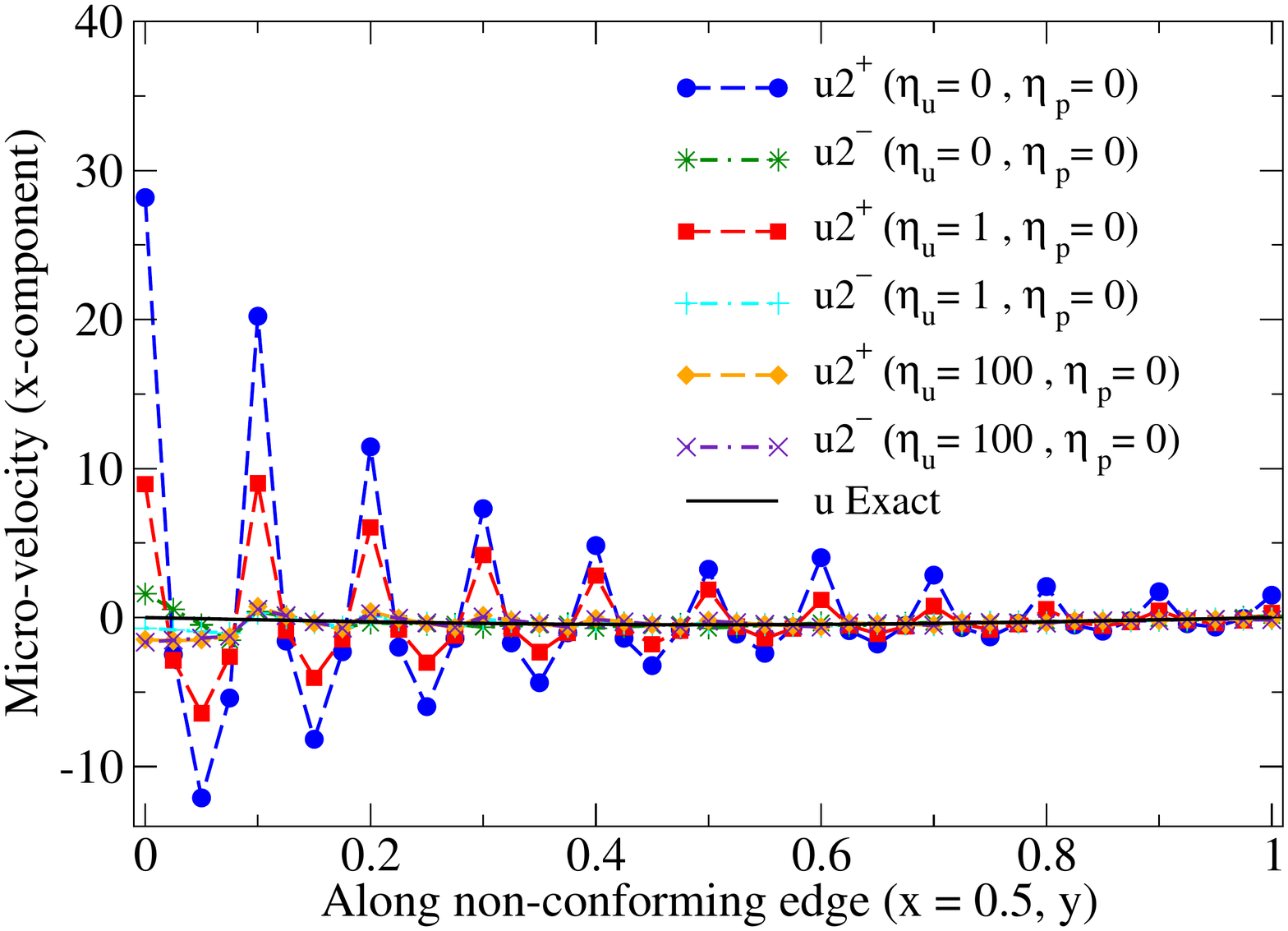}}
  \caption{\textsf{Non-conforming polynomial orders:}~This figure shows a parametric study performed on the effect of $\eta_u$ on velocity profiles. For the case of $\eta_{p} = 0$ and non-zero $\eta_{u}$, a drastic enhancement is captured with $\eta_{u}$ of order $1$. }
	\label{Fig:Order_refinement_parametric-etau}
\end{figure}
%
\begin{figure}
  \subfigure{
    \label{Fig:Eta_macro}
    \includegraphics[trim={0 1cm 0 1cm},clip,width=0.7\linewidth]{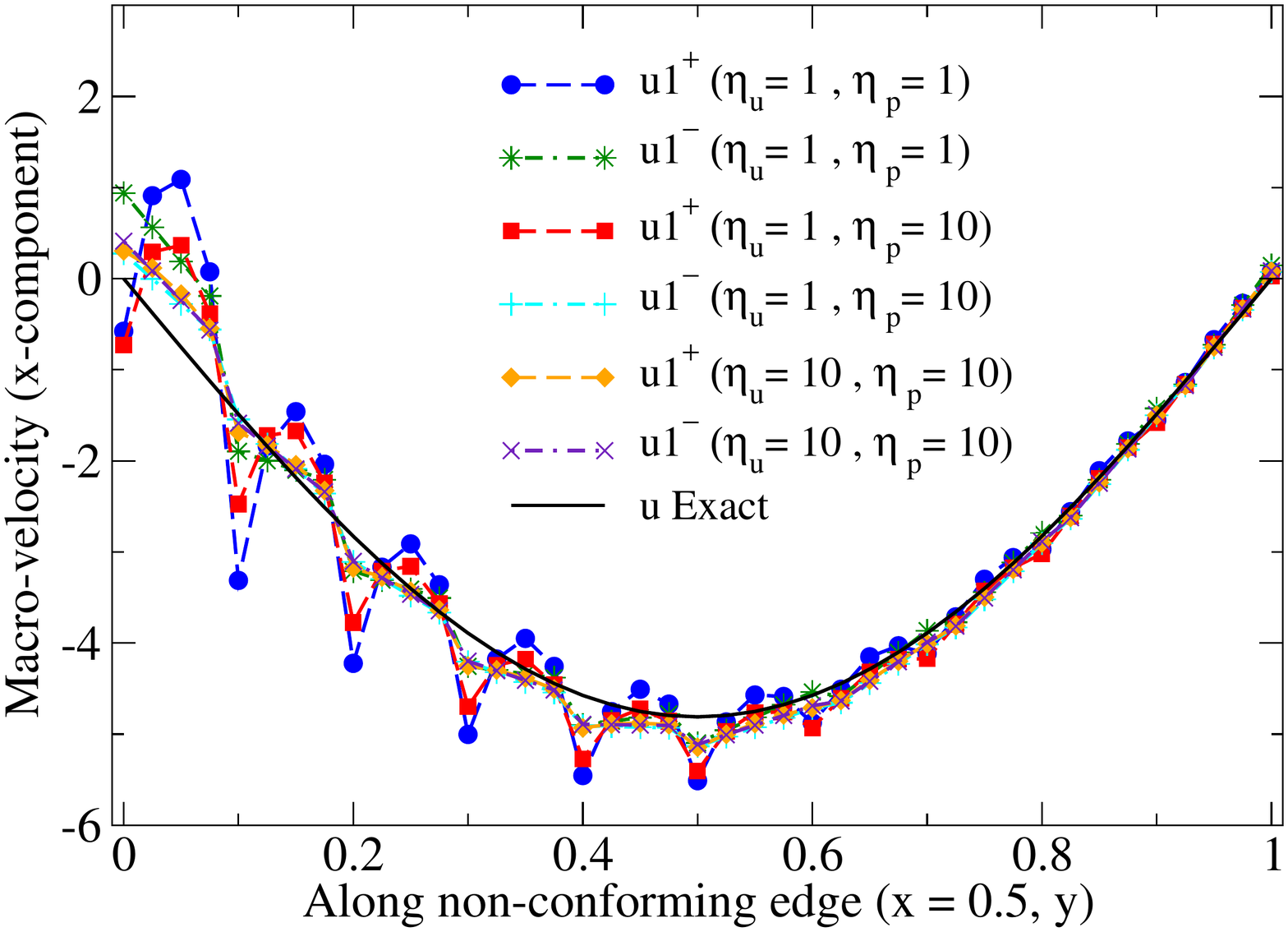}}
  \subfigure{
    \label{Fig:Eta_micro}
    \includegraphics[trim={0 1.3cm 0 2.5cm},clip,width=0.7\linewidth]{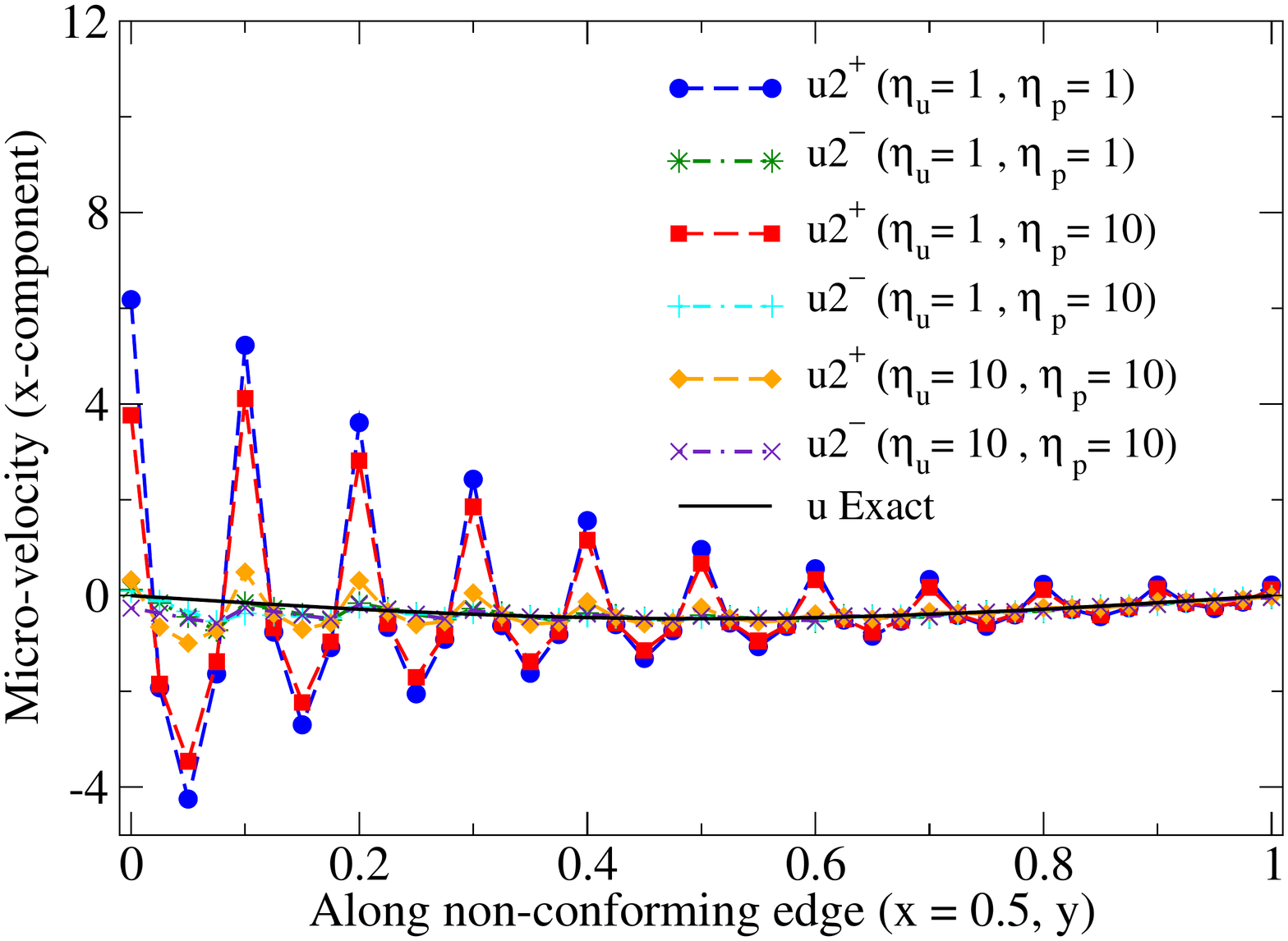}}
  \caption{\textsf{Non-conforming polynomial orders:}~This figure shows a parametric study performed on the combined effect of $\eta_u$ and $\eta_p$ on minimizing the drifts of 
  	macro and micro-velocity fields.}
	\label{Fig:Order_refinement_parametric-etaup}
\end{figure}

\textbf{Figs.~\ref{Fig:mismatching_polynomials_p}} and \textbf{\ref{Fig:mismatching_polynomials_v}} compare the exact and numerical solutions for the pressure and velocity fields by taking $\eta_{u} = 10$ and $\eta_{p} = 1$. As can be seen, the numerical and the exact solutions match, which implies that the proposed mixed DG formulation can nicely handle non-conforming polynomial orders.
\citep{badia2010stabilized} suggests the need for
such additional stabilization terms for modeling flow under Darcy equations. However, to the best of
the authors' knowledge,
no numerical simulation has been reported to quantify the effect of these stabilization parameters on the accuracy of results under the DPP model for the problems exhibiting mismatching interpolation order. 

\begin{figure}[!h]
	\subfigure[Macro-pressure (Numerical solution)]{
		\includegraphics[clip,width=0.4\linewidth]{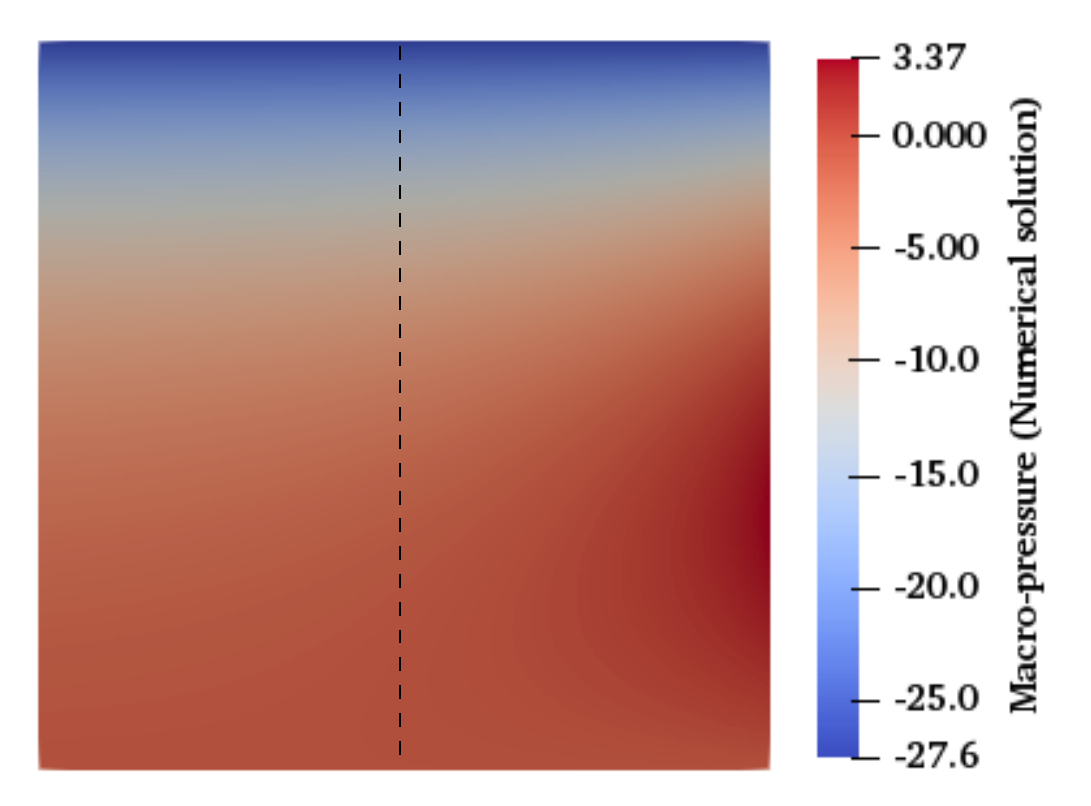}}
	\hspace{1.5cm}
	\subfigure[Macro-pressure (Exact solution)]{
		\includegraphics[clip,width=0.41\linewidth]{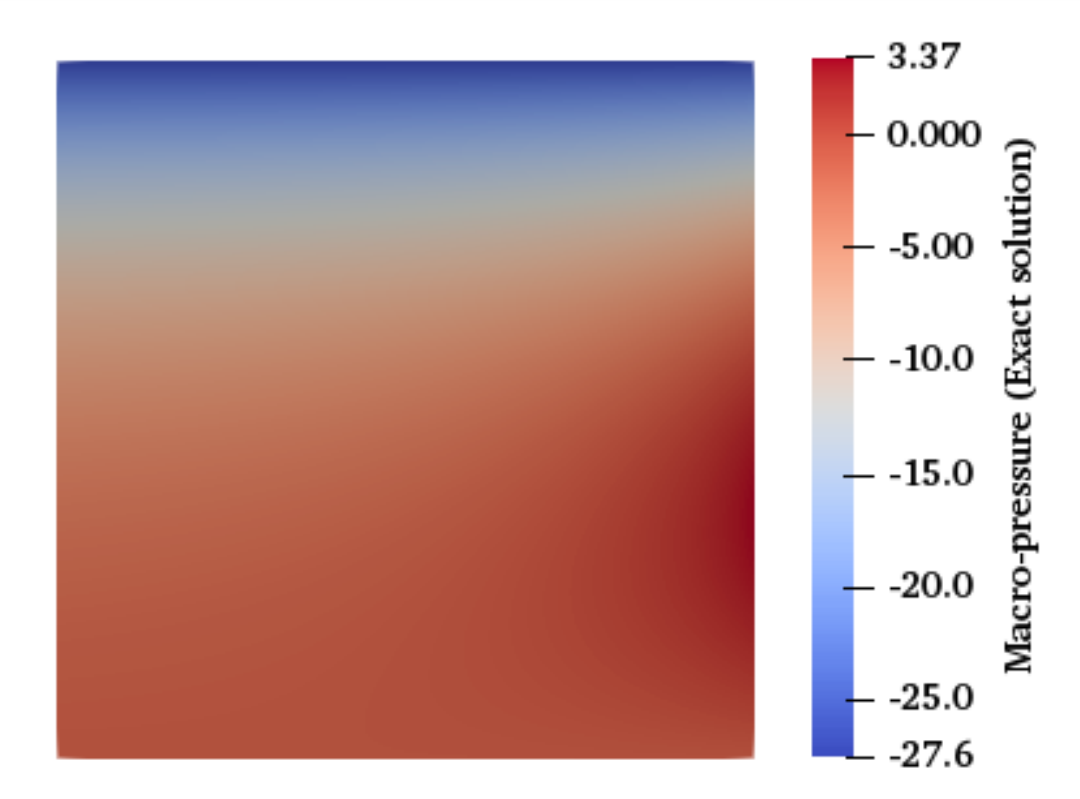}}
	\subfigure[Micro-pressure (Numerical solution)]{
		\includegraphics[clip,width=0.4\linewidth]{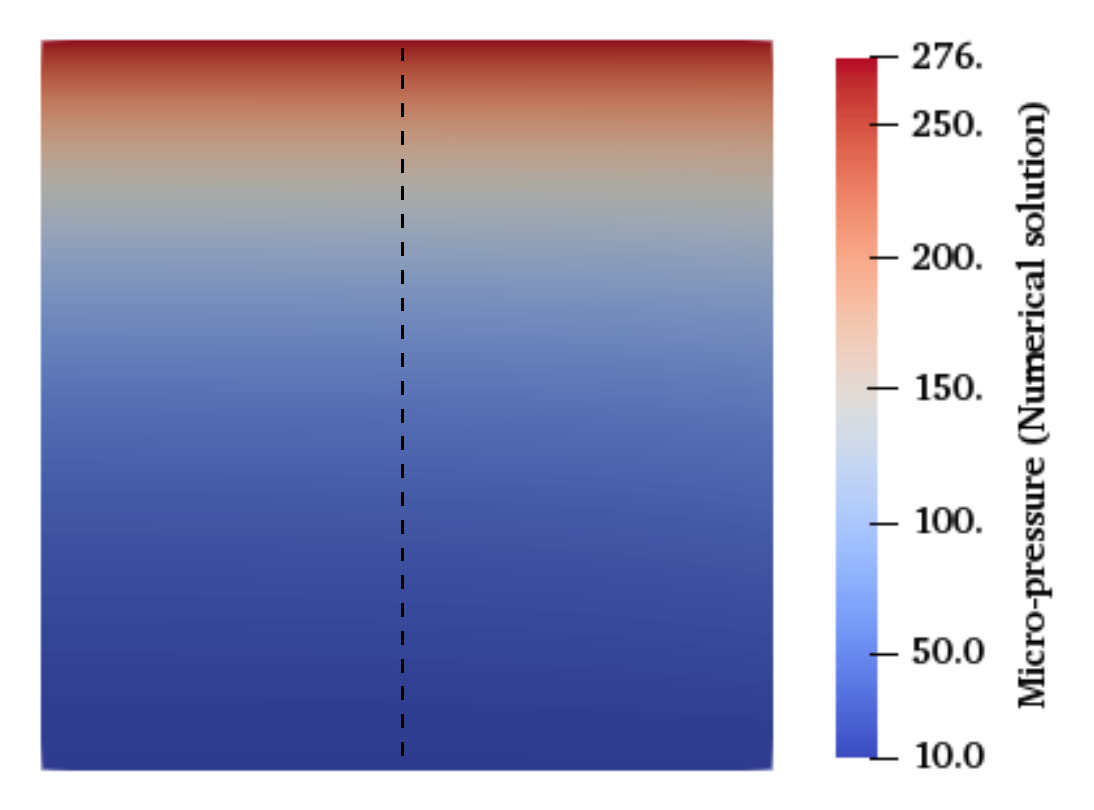}}
	\hspace{1.5cm}
	\subfigure[Micro-pressure (Exact solution)]{
		\includegraphics[clip,width=0.4\linewidth]{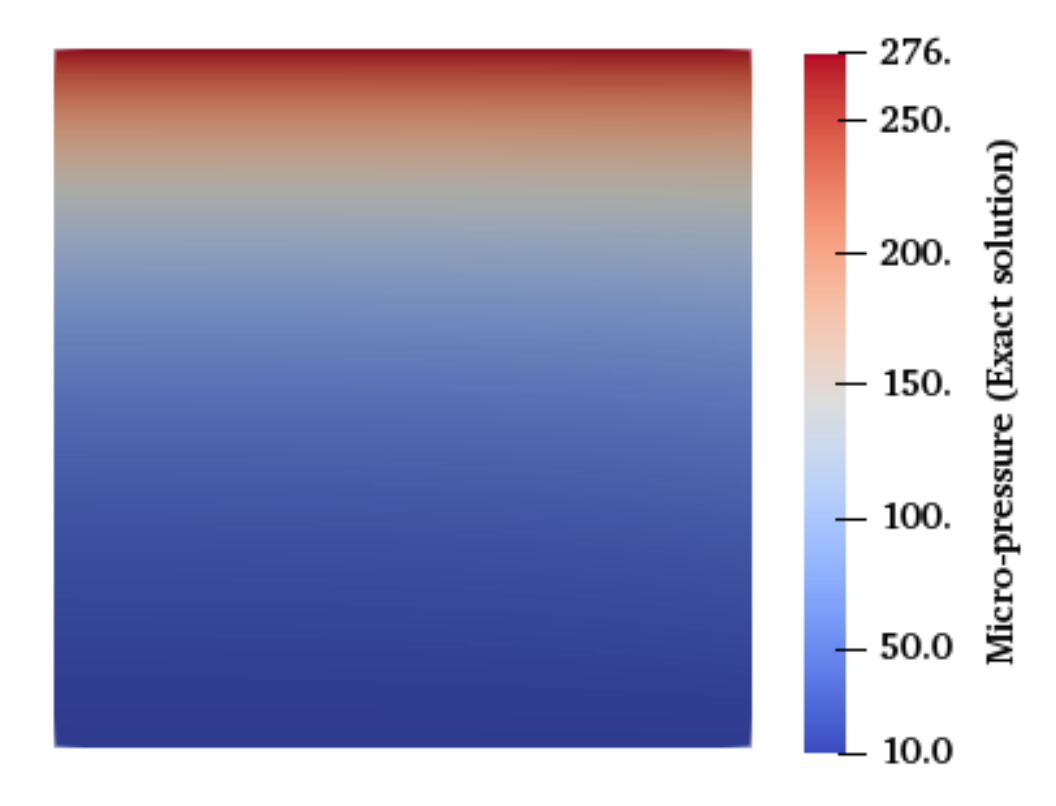}}
	\caption{\textsf{Non-conforming polynomial orders:}~This figure shows the exact and numerical solutions for the pressure profiles within the domain. In the left half of the domain, third order interpolation polynomials are used for velocities and pressures, while in the right half, first order interpolation polynomials are used. The exact and numerical solutions match which shows that the proposed stabilized DG formulation supports non-conforming order refinement. \label{Fig:mismatching_polynomials_p}}
\end{figure}
%
\begin{figure}[!h]
	\subfigure[Macro-velocity (Numerical solution)]{
		\includegraphics[clip,width=0.4\linewidth]{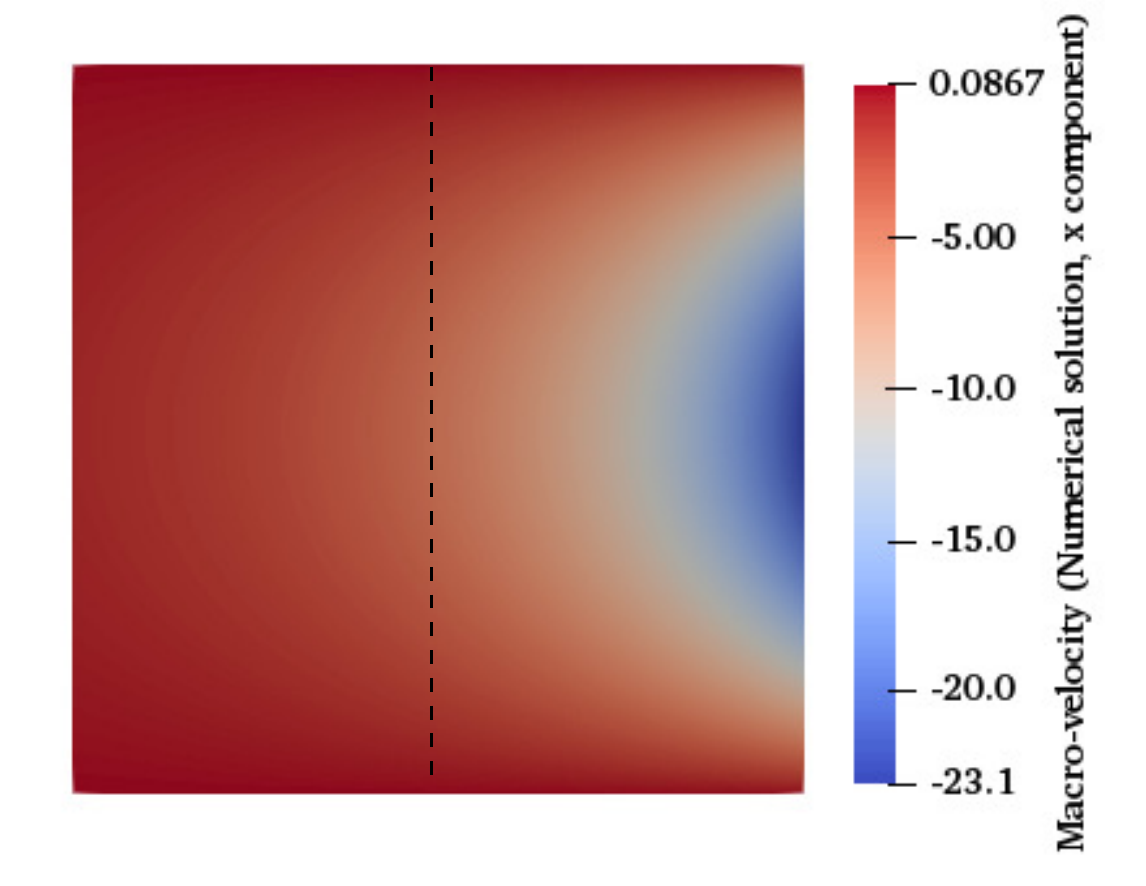}}
	\hspace{1.5cm}
	\subfigure[Macro-velocity (Exact solution)]{
		\includegraphics[clip,width=0.4\linewidth]{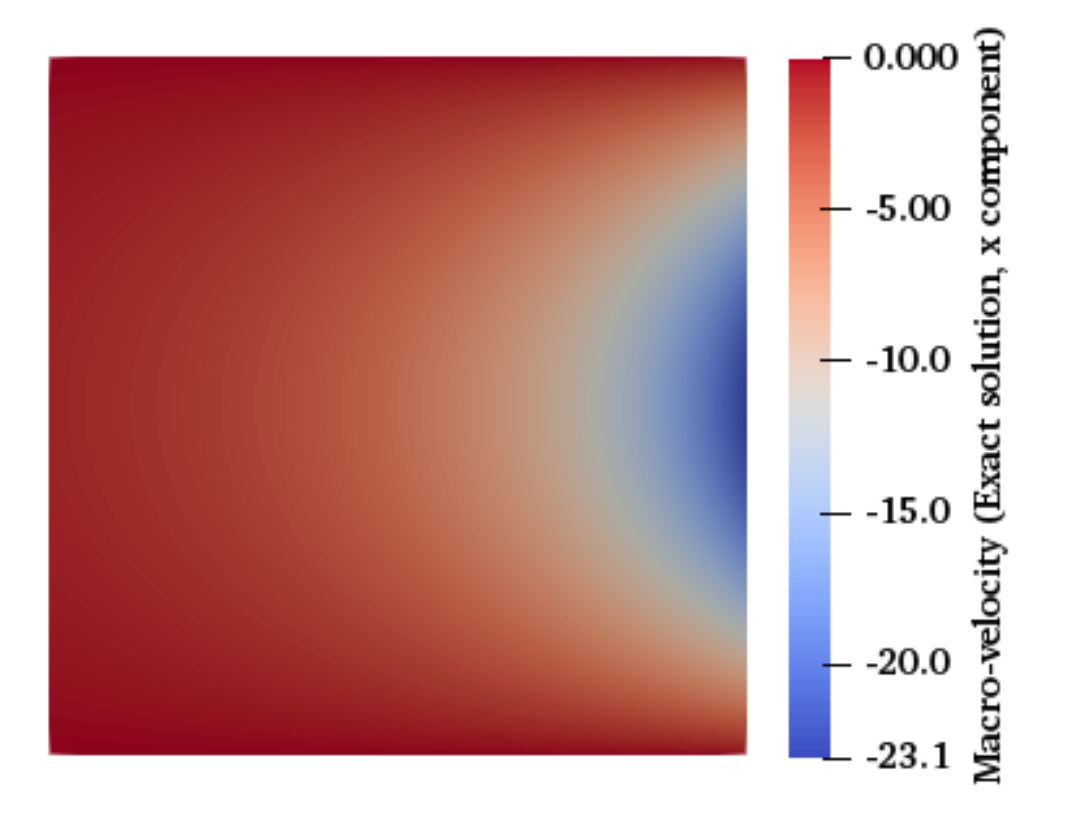}}
	\subfigure[Micro-velocity (Numerical solution)]{
		\includegraphics[clip,width=0.4\linewidth]{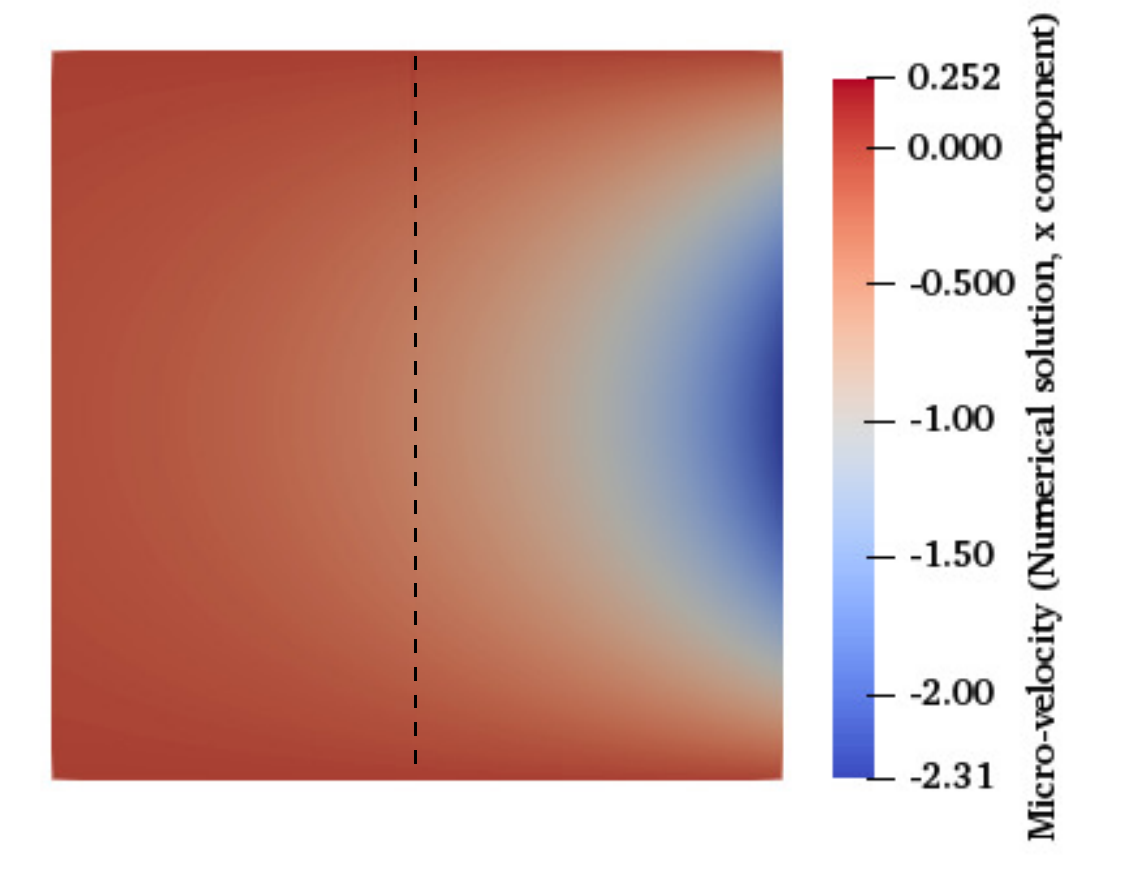}}
	\hspace{1.5cm}
	\subfigure[Micro-velocity (Exact solution)]{
		\includegraphics[clip,width=0.4\linewidth]{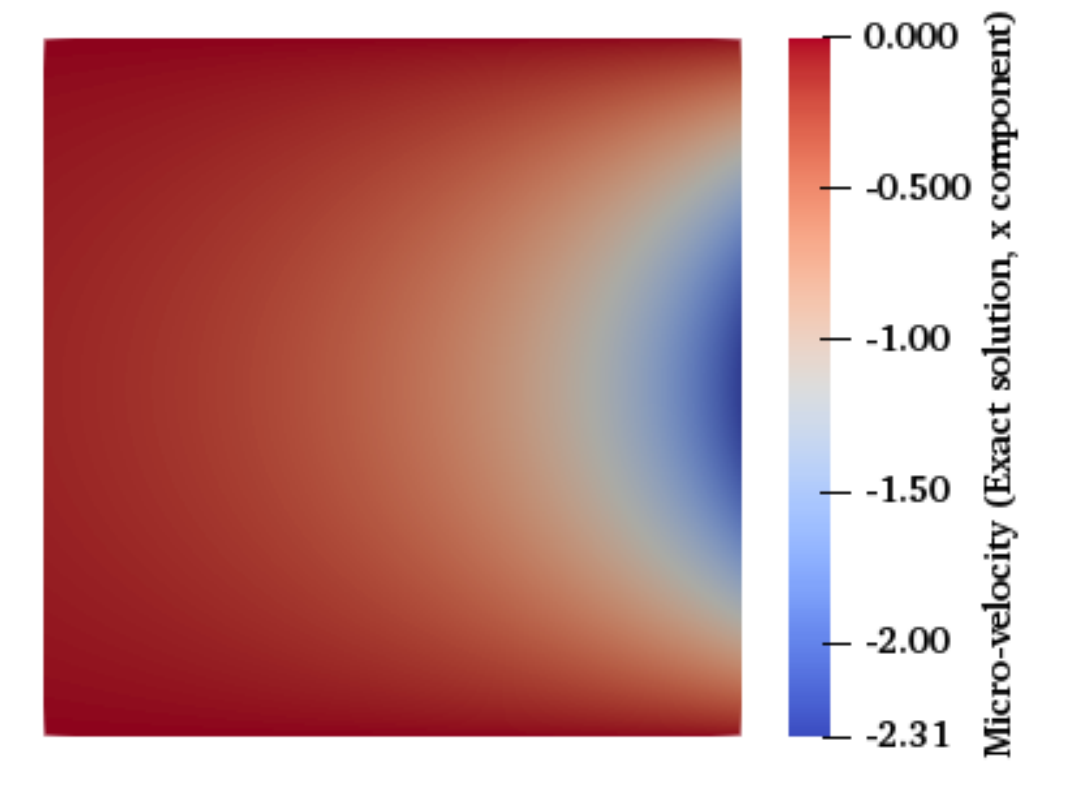}}
	\caption{\textsf{Non-conforming polynomial orders:}~This figure shows the exact and numerical solutions for the velocity profiles within the domain. In the left half of the domain, third order interpolation polynomials are used for velocities and pressures, while in the right half, first order interpolation polynomials are used. The exact and numerical solutions match which shows the proposed stabilized DG formulation supports non-conforming order refinement. \label{Fig:mismatching_polynomials_v}}
\end{figure}
%
\subsubsection{Non-conforming element refinement}
In mesh refinement procedures, one can either uphold the conformity of the mesh or produce irregular (non-conforming) meshes. The ability of DG formulations to support non-conforming elements obviates the user from propagating refinements beyond the desired elements
\citep{Hesthaven_Warburton_2007nodal}. The non-conforming meshes introduce \emph{hanging nodes} on the edge of neighboring elements. In general,
there are two strategies for handling non-matching interface discretization.
In the first approach, extra degrees of freedom are assigned to the hanging nodes; hence the shape functions are generated 
on both regular and hanging nodes in such a way that both Kronecker delta and partition of unity properties are satisfied.
Constructing these special shape functions for two- and three-dimensional problems is discussed in 
\citep{gupta1978finite,morton1995new}.
In the second approach, which is known as constrained approximation, the shape functions 
are generated only on the corner nodes of each element and the stiffness matrix is assembled via conventional algorithms. The constraints at hanging nodes
are then designed to be the average of their neighboring corner nodes. This can be enforced either through 
Lagrange multipliers or multiplication by the connectivity matrix \citep{ainsworth1997aspects,bank1983some}. This method is 
a classical standard procedure in treatment of mismatching girds and hanging nodes. For further details refer to \citep{oden1989toward}.
Herein, we resort to the second approach by introducing \textit{virtual nodes}, as the refinement algorithm is more straightforward 
compared to the first approach \citep{fries2011hanging}.

Applications of mesh refinement in the light of DG formulations are provided by \citep{burstedde2008scalable, kopera2014analysis, hartmann2002adaptive}, where the numerical fluxes on the non-conforming meshes are
incorporated in the DG solver.
In the following problem, the capability of our proposed stabilized mixed DG formulation for supporting the non-conforming element refinement is investigated. 
The domain is homogeneous with pressures being prescribed 
on the left and right boundaries of both pore-networks. The normal components of velocities 
are zero on top and bottom of the domain.
The model parameters for this
problem can be found in Table 
\ref{Tb4:nonconforming_element_refinement}. The refinement provided is based on physical considerations and
takes place on the right half of the domain, where the mismatching 
edge is shared by more than two elements, as can be seen in 
\textbf{Fig.~\ref{Fig:Non_conforming_domain}}.

The virtual nodes laid down on the 
non-conforming boundary face (nodes $13$ and $14$ in \textbf{Fig.~\ref{Fig:Non_conforming_mesh}}), each store a linear 
interpolation of nodes $2$ and $3$. These
nodes (similar to hanging nodes $8$ and $9$) do not initially impose any additional degrees 
of freedom and are merely auxiliary nodes on the 
edge of element $1$ for programming convenience.
The usual DG algorithm for the assembly of the global stiffness matrix 
is followed. Then, we enforce constraints for degrees of freedom 
 corresponding to hanging nodes (and virtual nodes) by
 Lagrange multiplier's approach as described in details in \citep{karniadakis2013spectral,fries2011hanging}.
 At this stage, the interactions of node $8$ with nodes $2$ and $3$ was facilitated via virtual node $13$, and similarly, the interaction of node $9$ with nodes $2$ and $3$ was assisted via virtual node $14$.
{\small
  \begin{table}[!h]
    \caption{Model parameters
      for non-conforming element refinement problem.}
    \centering
    \begin{tabular}{|c|c|} \hline
			Parameter & Value \\
			\hline
			$\gamma \mathbf{b}$ & $\{0.0,0.0\}$\\
			$L_x$ & $2.0$ \\
			$L_y$ & $1.0$\\
			$\mu $ & $1.0$ \\
			$\beta $ & $1.0$ \\
			$k_1$&  $1.0$  \\
			$k_2$&  $0.1$\\
			$\eta_u$& $0.0$\\
			$\eta_p$& $0.0$\\
			\hline 
		\end{tabular}
		\label{Tb4:nonconforming_element_refinement}
	\end{table}
}
%
\begin{figure}[!h]
	\subfigure[Computational domain \label{Fig:Non_conforming_domain}]{
		\includegraphics[clip,width=0.6\linewidth]{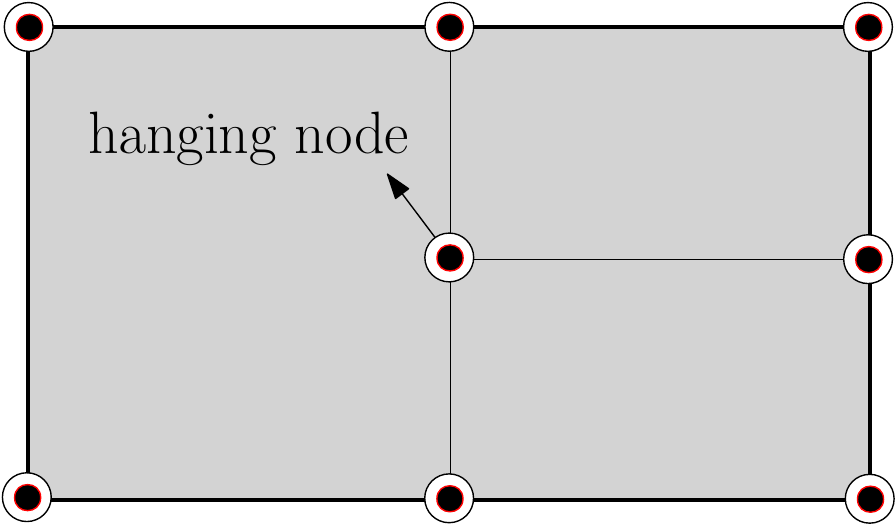}}

	\vspace{-1.3cm}
	\subfigure[Mesh discretization \label{Fig:Non_conforming_mesh}]{
		\includegraphics[clip,width=0.65\linewidth]{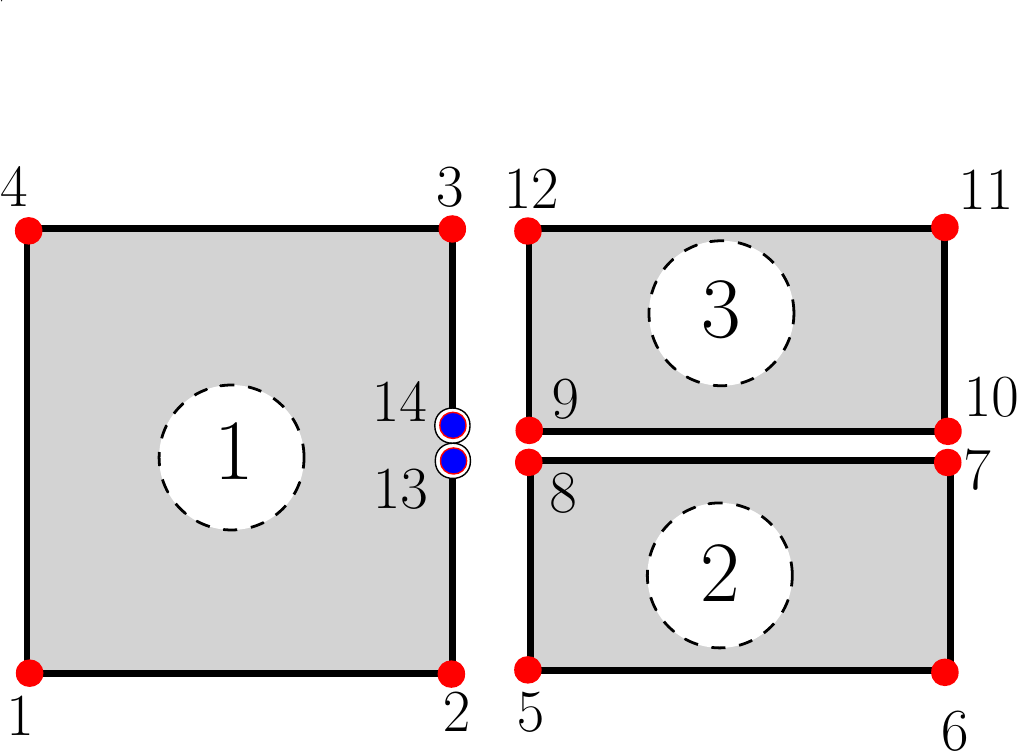}}
	\caption{\textsf{Non-conforming element refinement:}~ The top figure shows the representative computational domain with non-conforming element refinement (the hanging node on the non-conforming boundary is shown). The bottom figure shows the DG discretization of this domain. The blue nodes are the virtual nodes, each of which are a linear interpolation of nodes $2$ and $3$. They do not jack up the order of stiffness matrix as no degree of freedom is assigned to them. 
	}
\end{figure}
\textbf{Fig.~\ref{Fig:Hanging_node_results}} shows the velocity and pressure profiles within the domain. Pressures in both pore-networks are varying linearly and velocities are constant throughout the domain. These results show that the proposed stabilized DG formulation is capable of handling non-conforming element refinement (with hanging nodes in the mesh).
\begin{figure}[!h]
	\subfigure[Macro-pressure]{
		\includegraphics[clip,width=0.4\linewidth]{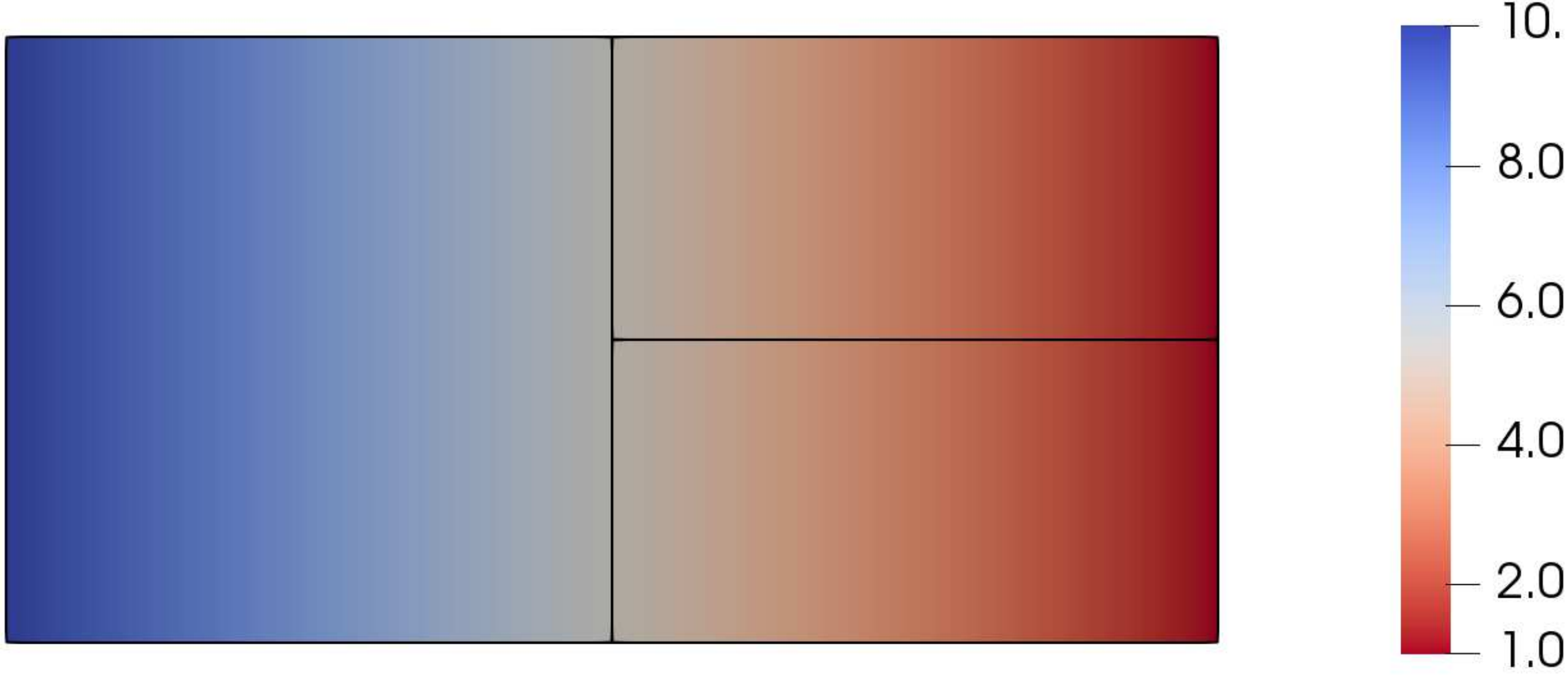}}
	\hspace{1.5cm}
	\subfigure[Micro-pressure]{
		\includegraphics[clip,width=0.4\linewidth]{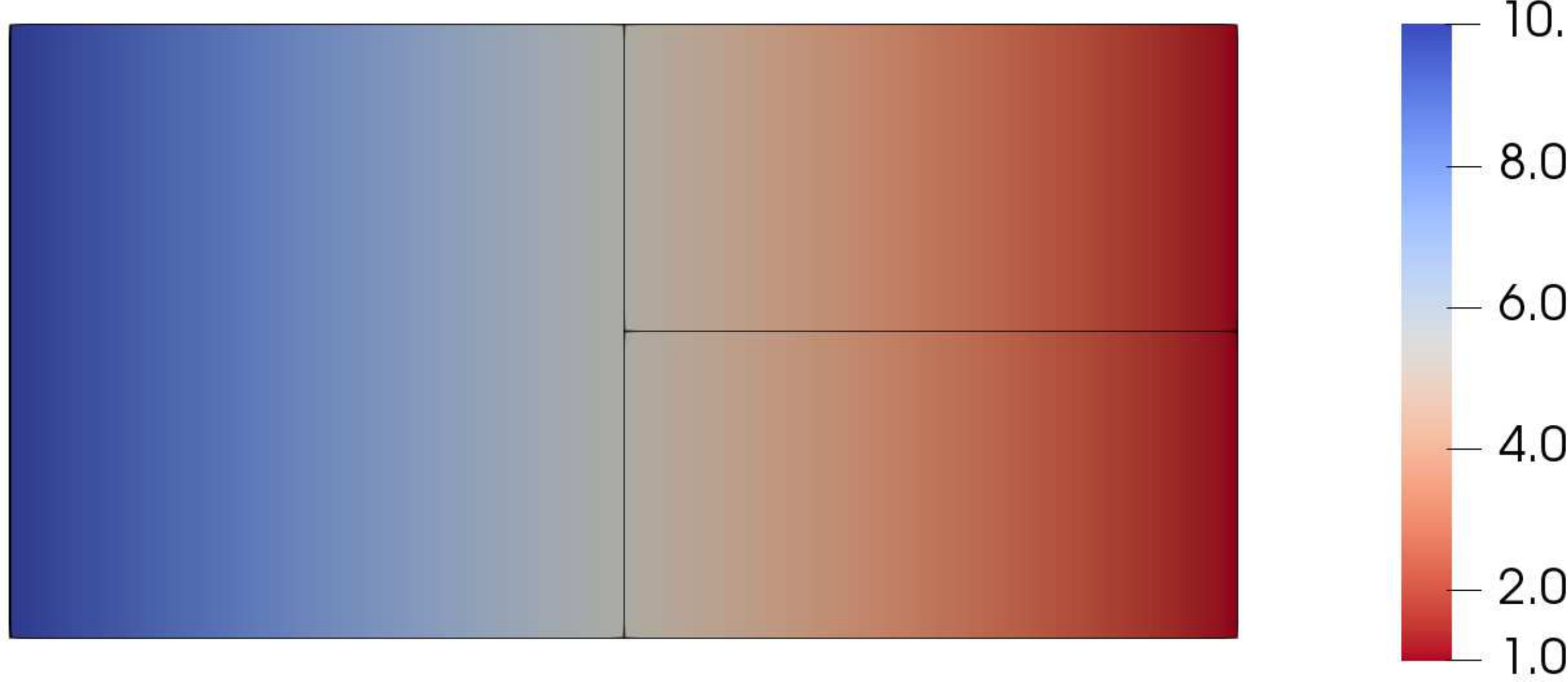}}
	\subfigure[Macro-velocity]{
		\includegraphics[clip,width=0.4\linewidth]{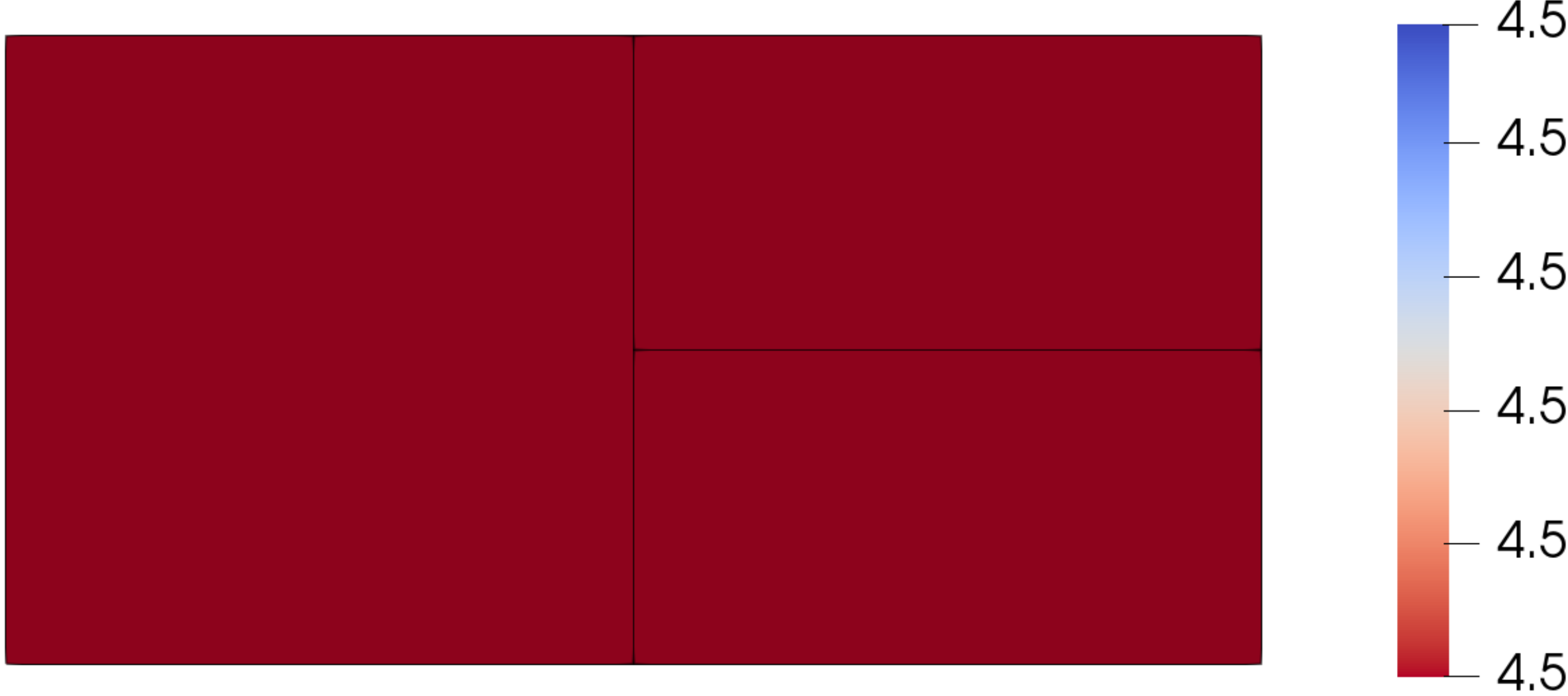}}
	\hspace{1.5cm}
	\subfigure[Micro-velocity]{
		\includegraphics[clip,width=0.4\linewidth]{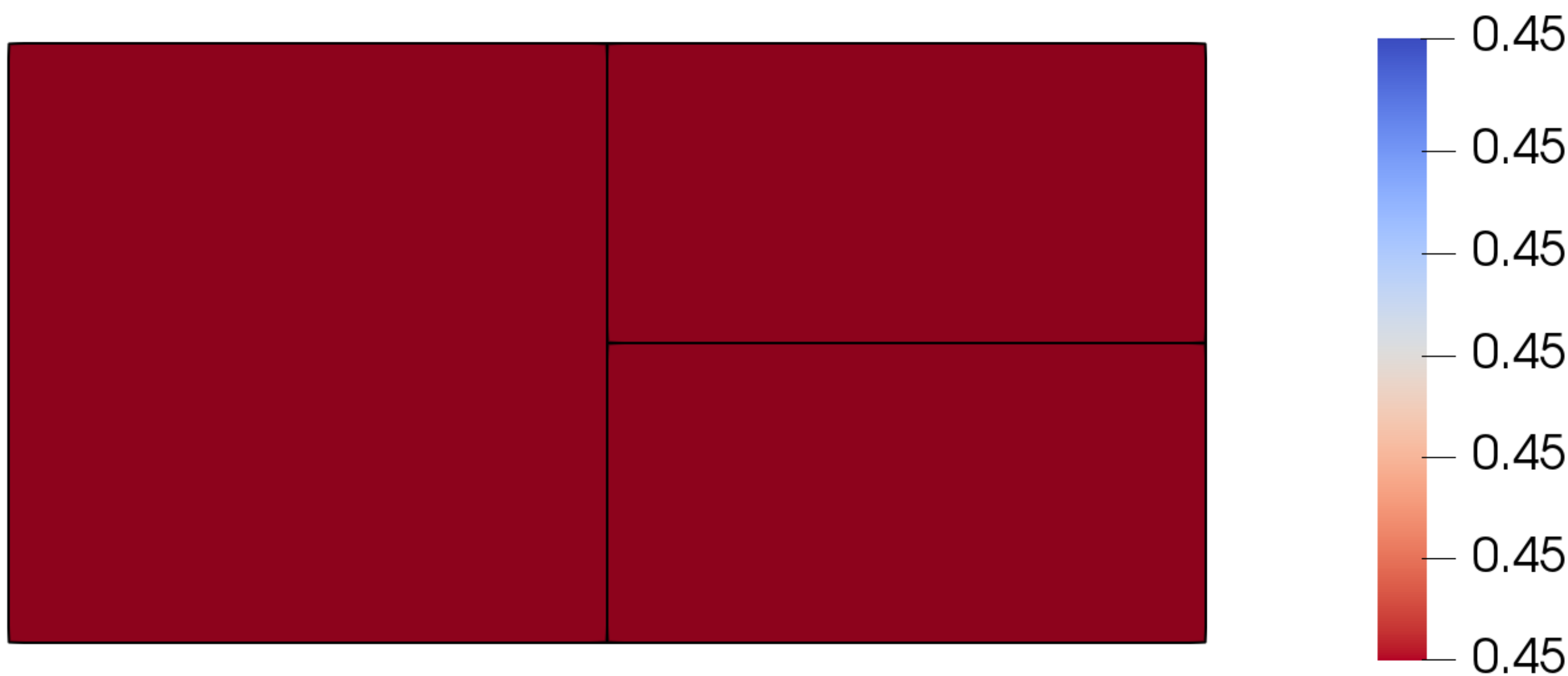}}
	\caption{\textsf{Non-conforming element refinement:}~Pressures in both pore-networks are varying linearly and velocities are constant throughout the domain.  These results show that the proposed stabilized DG formulation is capable of handling non-conforming element refinement (with hanging nodes in the mesh).  \label{Fig:Hanging_node_results}}
\end{figure}
\subsection{Non-constant Jacobian elements}
\label{Sec:Non_Constant_Jacob}
In practice, many hydrogeological systems have complex shapes and modeling of such domains, especially in the 3D settings, requires using of elements with irregular shapes. Divergent boundaries in such elements result in non-constant Jacobian determinants.
Herein, the aim is to show that the proposed stabilized mixed DG formulation can perform satisfactorily to model flow through computational domains composed of non-constant Jacobian elements.
It will be shown that under the equal-order interpolation for the field variables, our proposed formulation is still able to pass the constant flow patch test with irregular elements.
Two different computational domains with sample meshes having non-constant Jacobian brick elements are depicted in \textbf{Fig.~\ref{Fig:3D_patch_nonConstant_Domain}} and model parameters are provided in Table \ref{Tb4:nonconstant_Jacobian}.
%
\begin{figure}[!h]
	\subfigure[Mesh $\#$1 \label{Fig:Mesh1}]{
		\includegraphics[clip,width=0.3\linewidth]{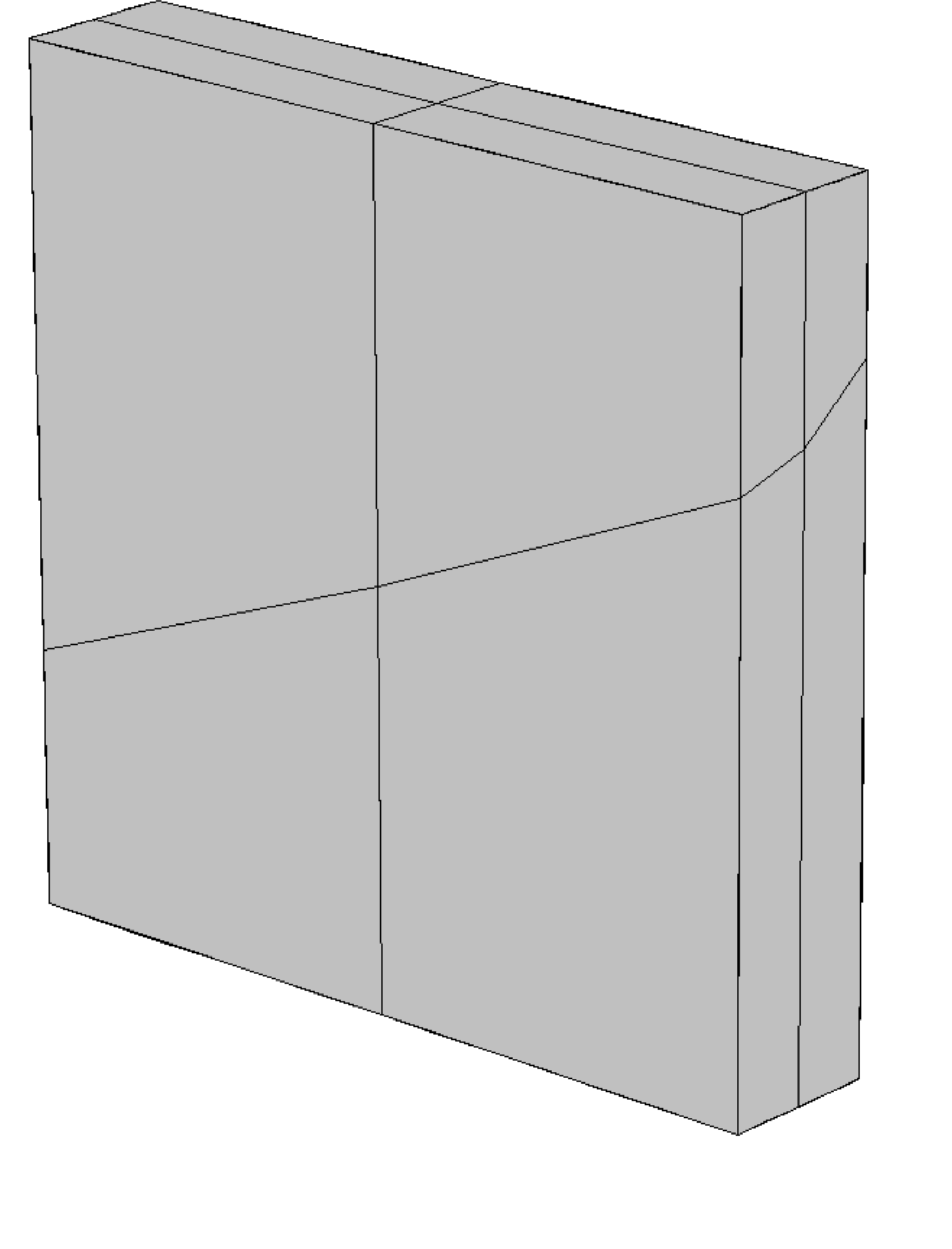}}
	\hspace{1.5cm}
	\subfigure[Mesh $\#$2 \label{Fig:Mesh2}]{
		\includegraphics[clip,width=0.4\linewidth]{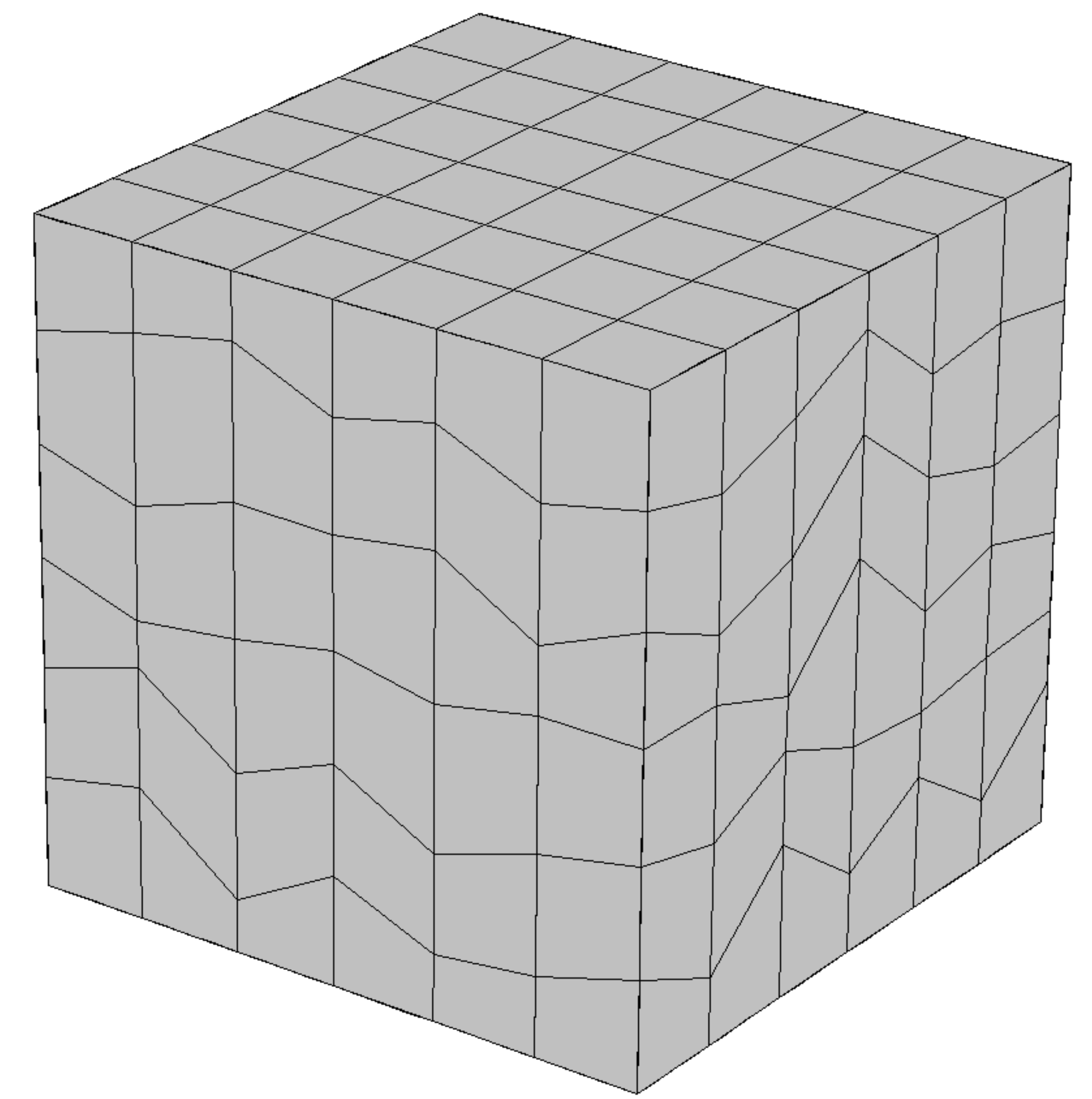}}
	\caption{\textsf{Non-constant Jacobian elements:}~This figure shows two different computational domains and their corresponding meshes for the constant flow patch test. For this problem, non-constant Jacobian brick elements are used.  \label{Fig:3D_patch_nonConstant_Domain}}
\end{figure}
{\small
	\begin{table}[!h]
		\caption{Model parameters for 3D computational domains with non-constant Jacobian elements.}
		\centering
		\begin{tabular}{|c|c|c|} \hline
			Parameter & Mesh \#1 & Mesh \#2 \\
			\hline
			$\gamma \mathbf{b}$ & $\{0.0,0.0,0.0\}$& $\{0.0,0.0,0.0\}$\\
			$L_x$ & $1.0$ &$1.0$\\
			$L_y$ & $1.0$&$0.2$\\
			$L_z$ & $1.0$&$1.0$\\
			$\mu $ & $1.0$&$1.0$ \\
			$\beta $ & $1.0$& $1.0$\\
			$k_1$&  $1.0$ &$1.0$ \\
			$k_2$&  $0.1$&$0.1$\\
			$\eta_u$& $0.0$&$0.0$\\
			$\eta_p$& $0.0$&$0.0$\\
			\hline 
		\end{tabular}
		\label{Tb4:nonconstant_Jacobian}
	\end{table}}
Pressures are prescribed at both left and right faces of the two pore-networks ($p_1(x=0,y,z) = p_2(x=0,y,z) = p^L$ and $p_1(x=1,y,z) = p_2(x=1,y,z) = p^R$). On the other faces, the normal component of velocity in both pore-networks is assumed to be zero (i.e., $\mathbf{u}_1 \cdot \widehat{\mathbf{n}} =
\mathbf{u}_{2} \cdot \widehat{\mathbf{n}} = 0$).
The pressure and velocity profiles for both domains are shown in \textbf{Figs.~\ref{Three_D_Domain_Non_1}} and \textbf{\ref{Three_D_Domain_Non_2}}. In both domains, pressures are varying linearly from the left face to the right one and velocities are constant throughout the domain as expected. These results show that the proposed mixed DG formulation is capable of providing accurate results using non-constant Jacobian elements.
%
%
\begin{figure}[!h]
	\subfigure[Macro-pressure]{
		\includegraphics[clip,width=0.4\linewidth]{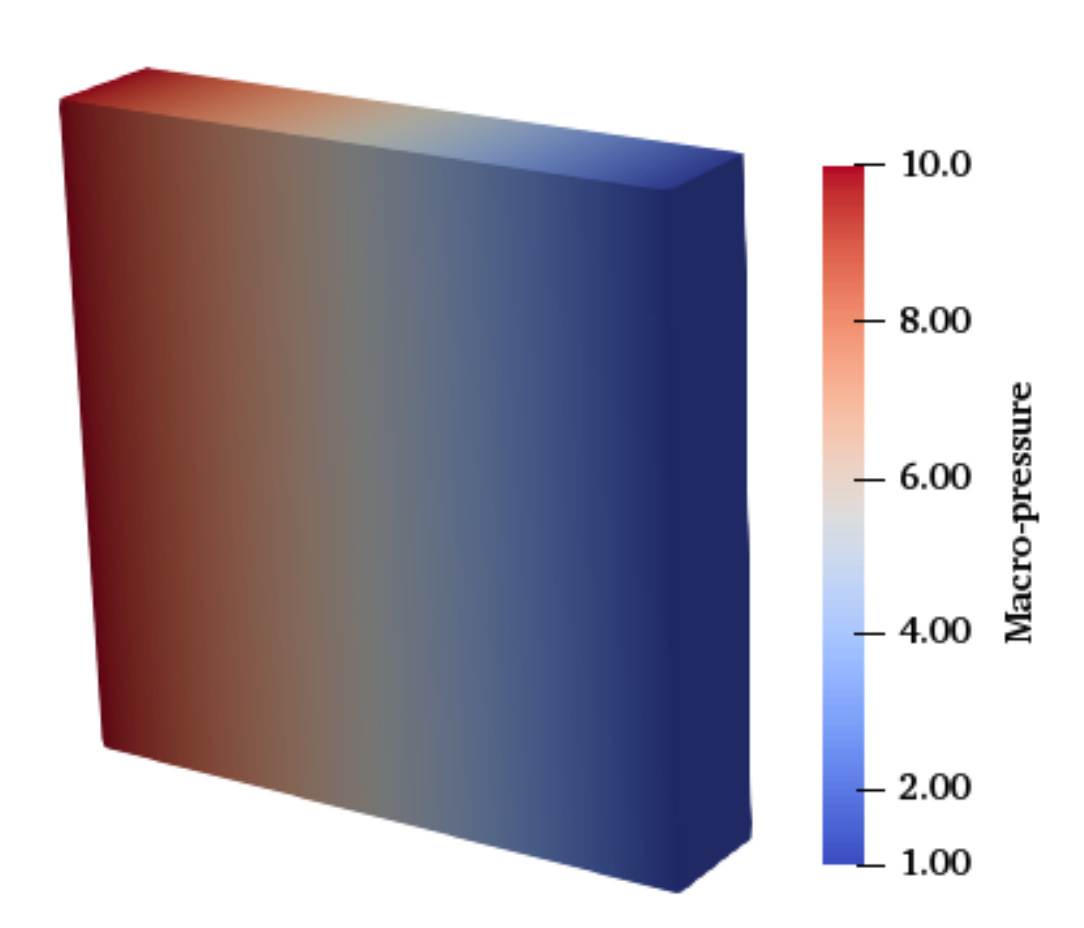}}
	\hspace{1.5cm}
	\subfigure[Micro-pressure ]{
		\includegraphics[clip,width=0.4\linewidth]{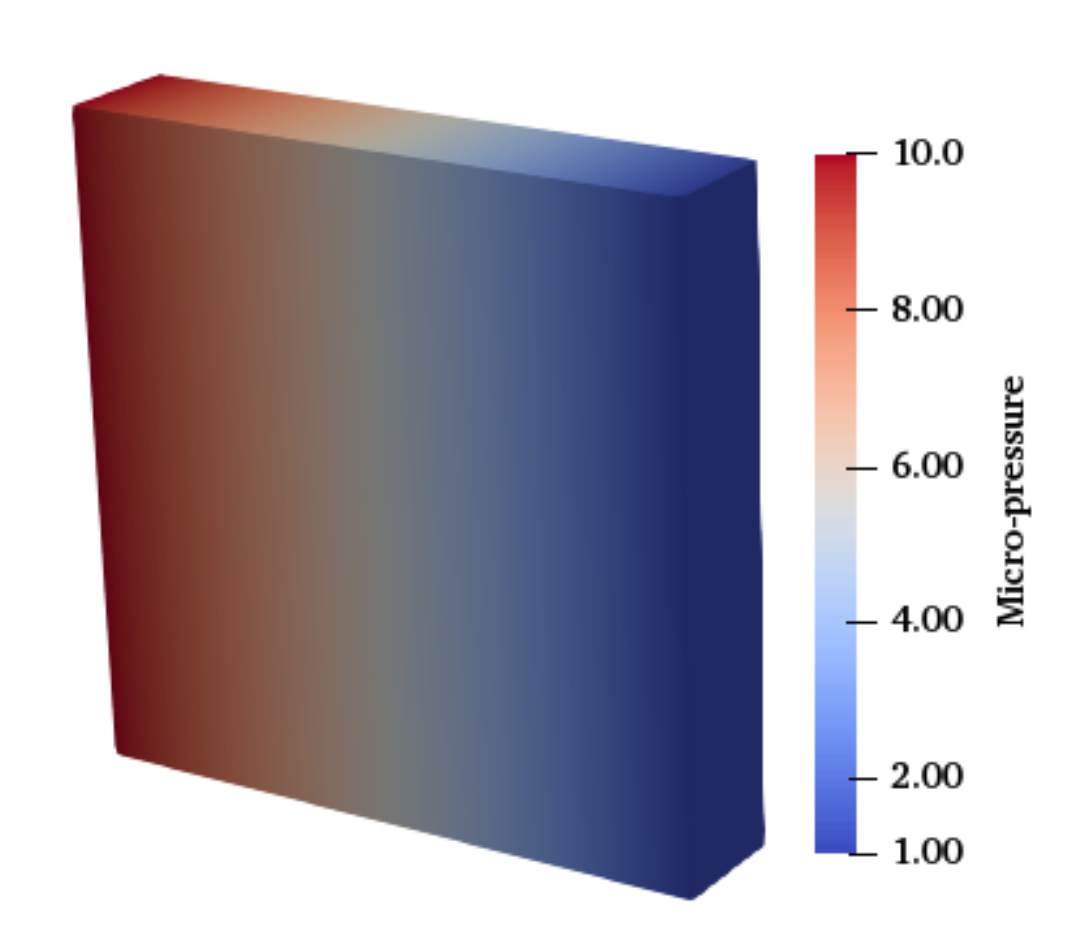}}
	\subfigure[Macro-velocity (x component) ]{
		\includegraphics[clip,width=0.4\linewidth]{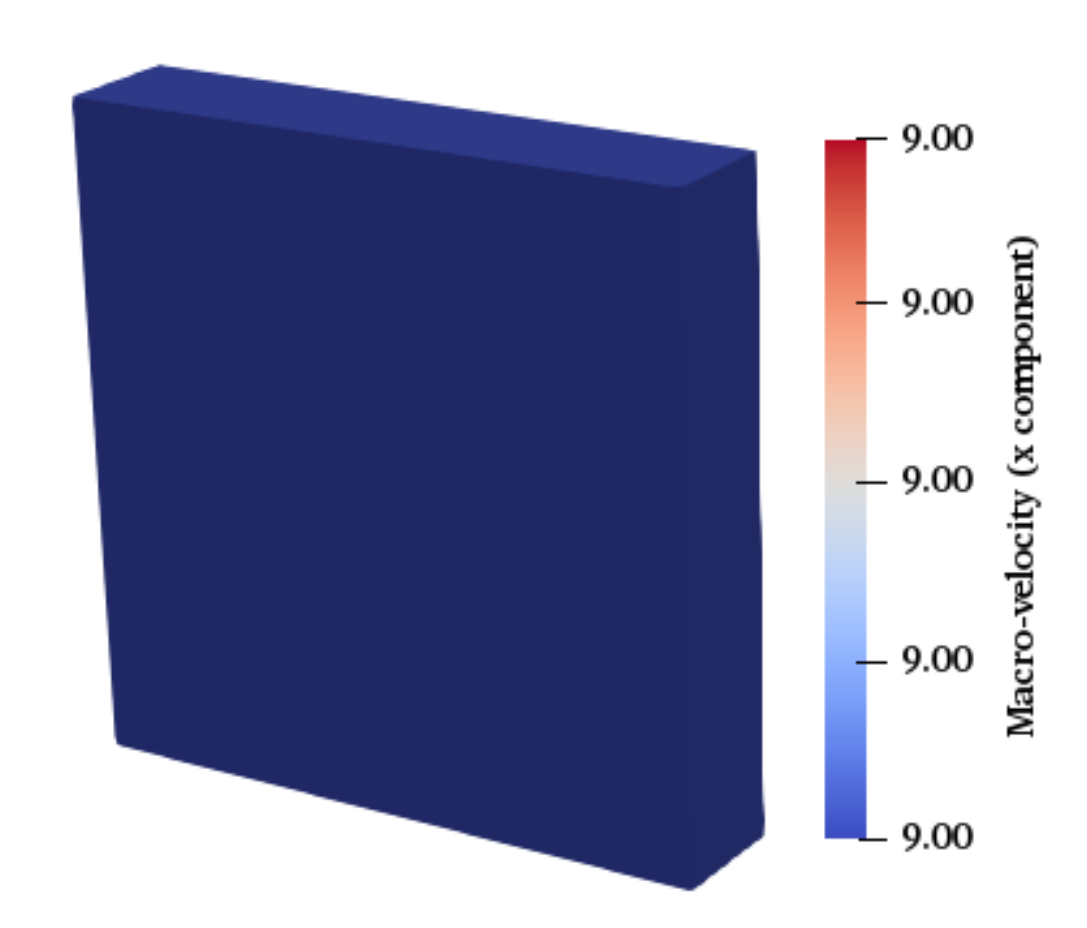}}
	\hspace{1.5cm}
	\subfigure[Micro-velocity (x component)]{
		\includegraphics[clip,width=0.4\linewidth]{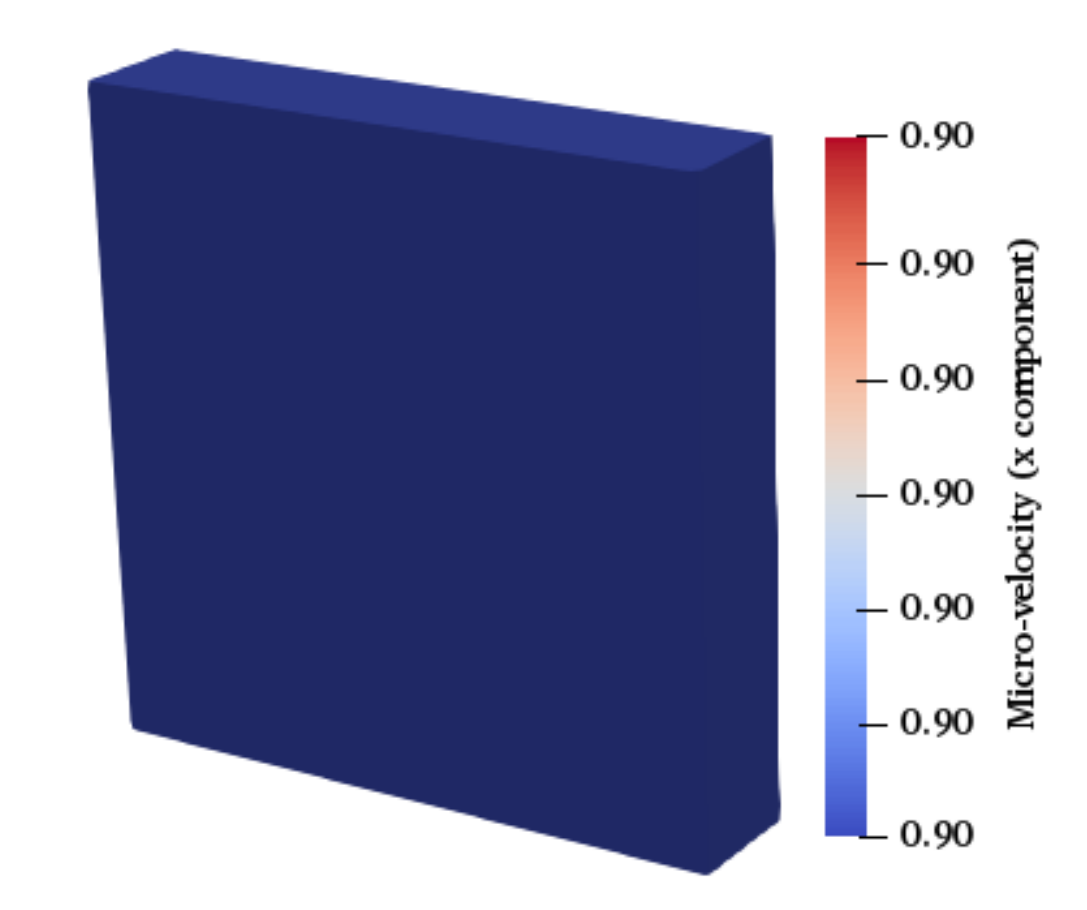}}
	\caption{\textsf{Non-constant Jacobian elements:}~Pressure and velocity profiles are shown for Mesh $\#$1 (\textbf{Fig.~\ref{Fig:Mesh1}}) with non-constant Jacobian elements. Pressures are varying linearly from the left face to the right one and velocities are constant throughout the domain as expected. These results show that the proposed mixed DG formulation is capable of providing accurate results using non-constant Jacobian elements. \label{Three_D_Domain_Non_1}}
\end{figure}
%
\begin{figure}[!h]
	\subfigure[Macro-pressure]{
		\includegraphics[clip,width=0.4\linewidth]{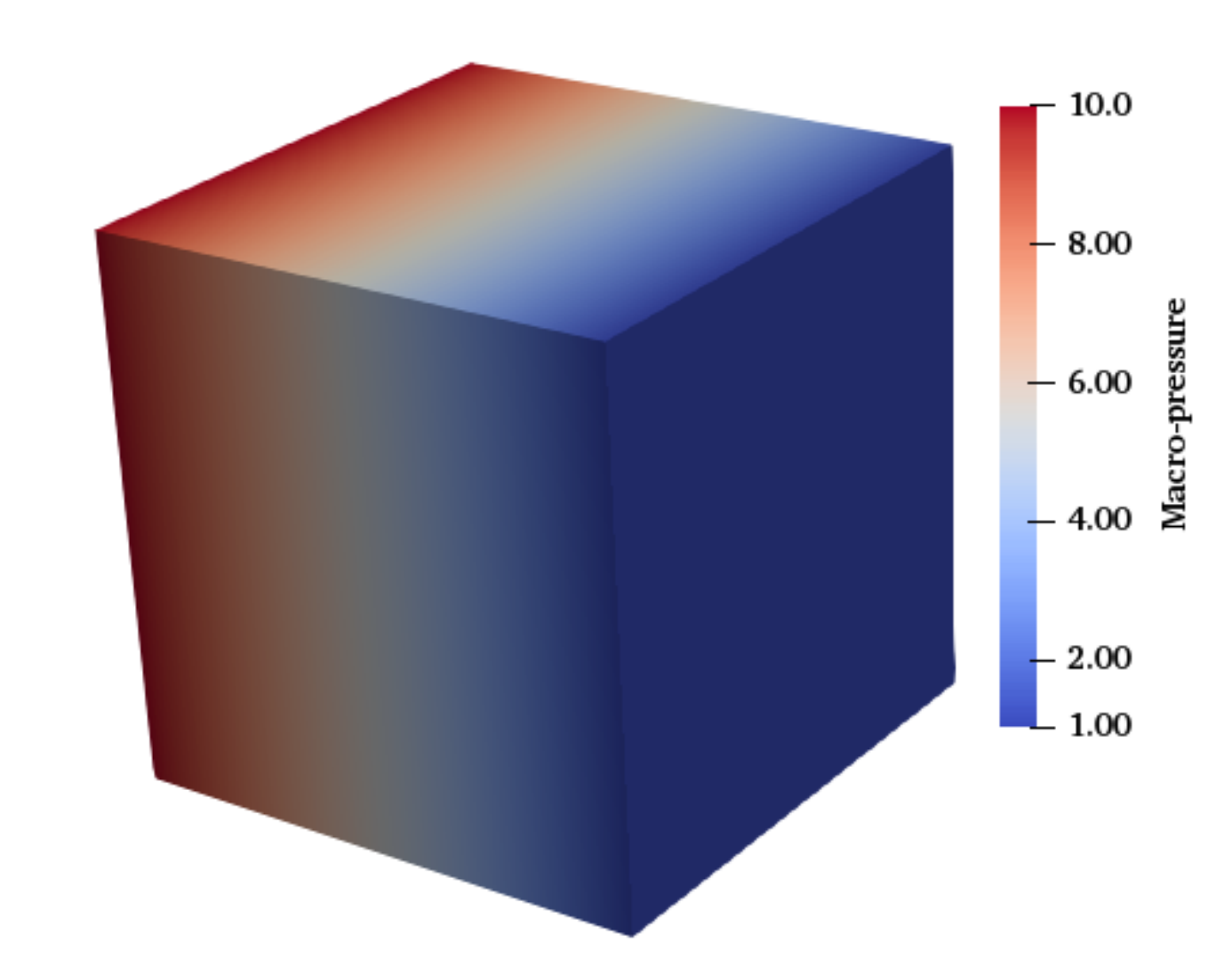}}
	\hspace{1.5cm}
	\subfigure[Micro-pressure ]{
		\includegraphics[clip,width=0.4\linewidth]{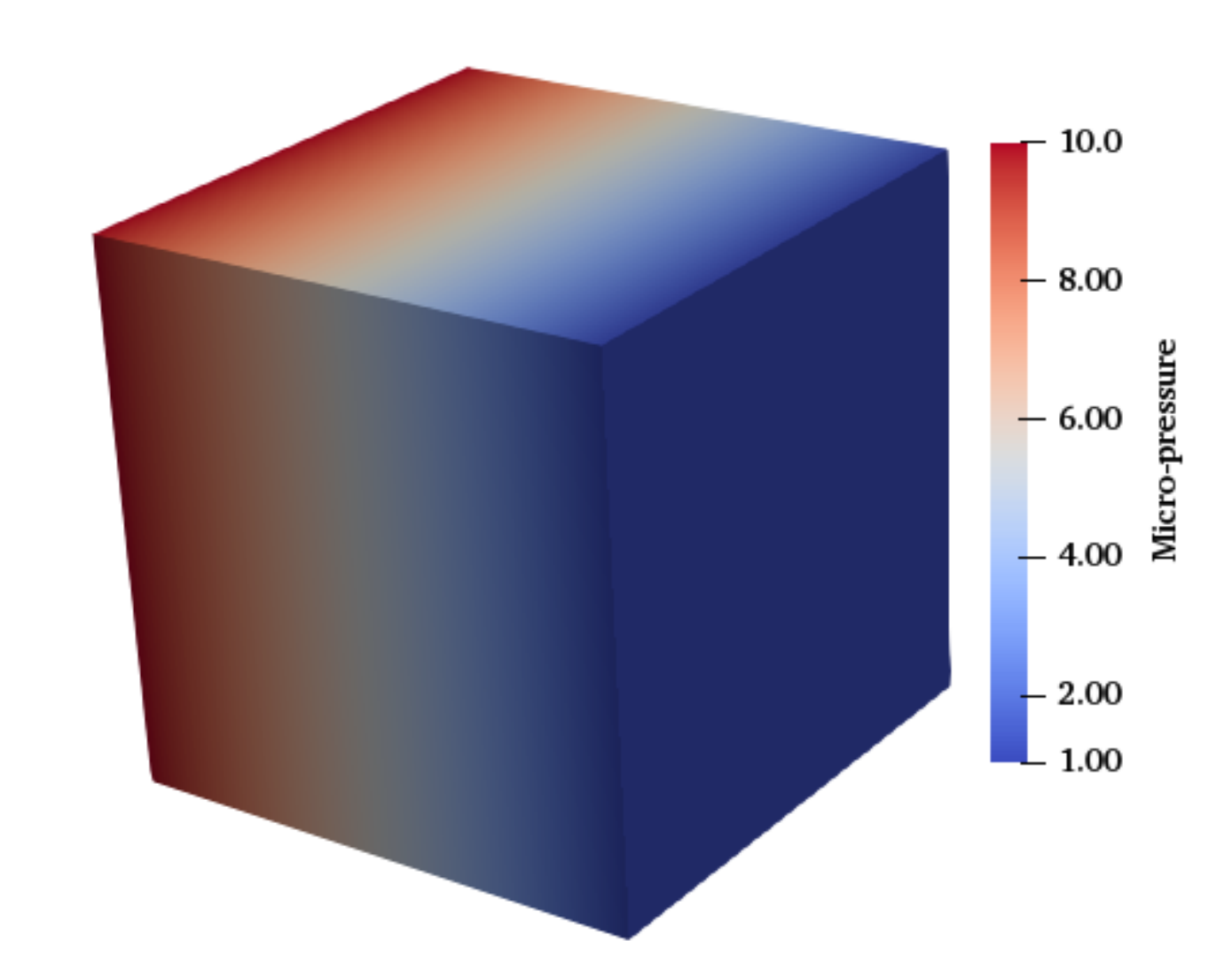}}
	\subfigure[Macro-velocity (x component) ]{
		\includegraphics[clip,width=0.4\linewidth]{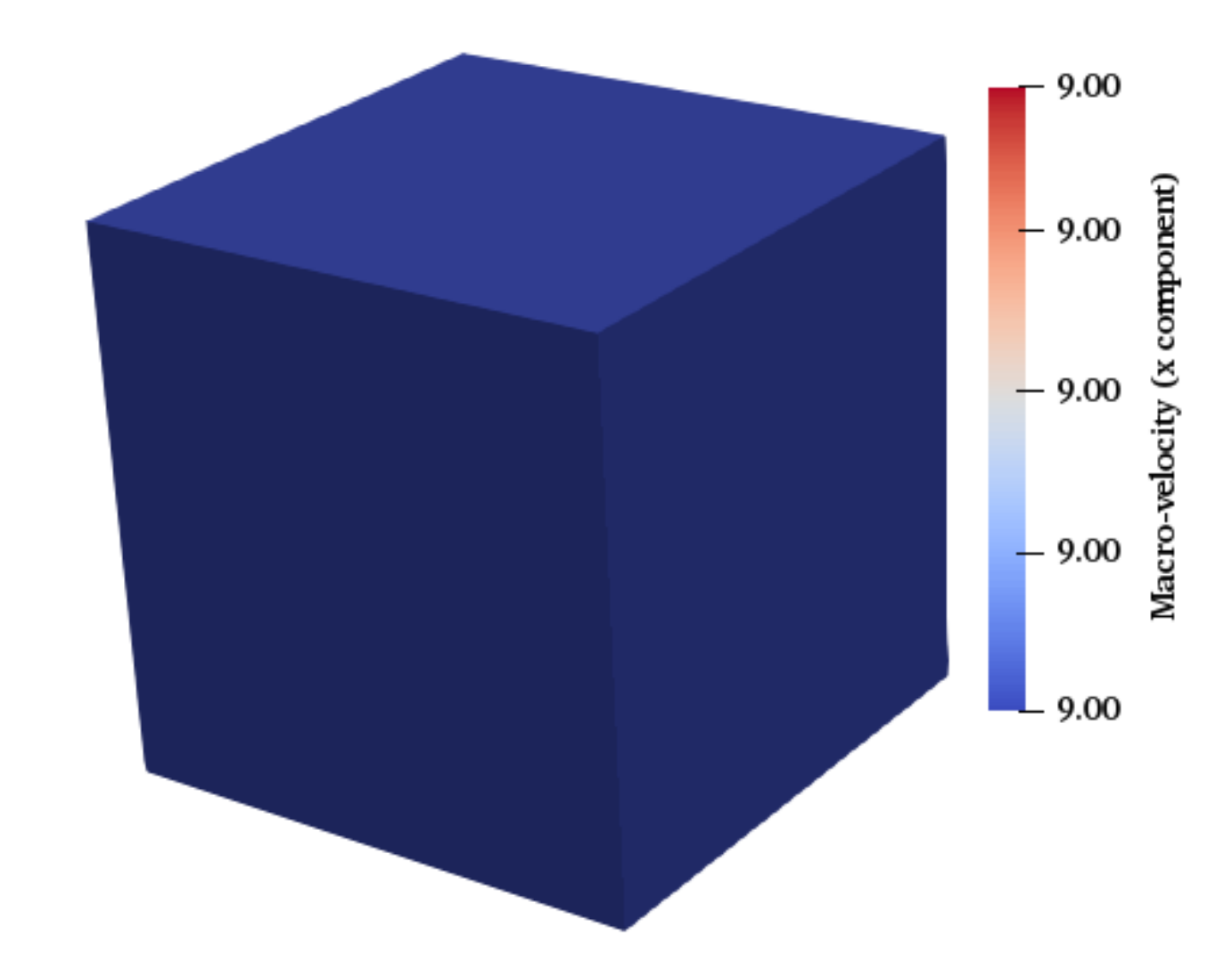}}
	\hspace{1.5cm}
	\subfigure[Micro-velocity (x component)]{
		\includegraphics[clip,width=0.4\linewidth]{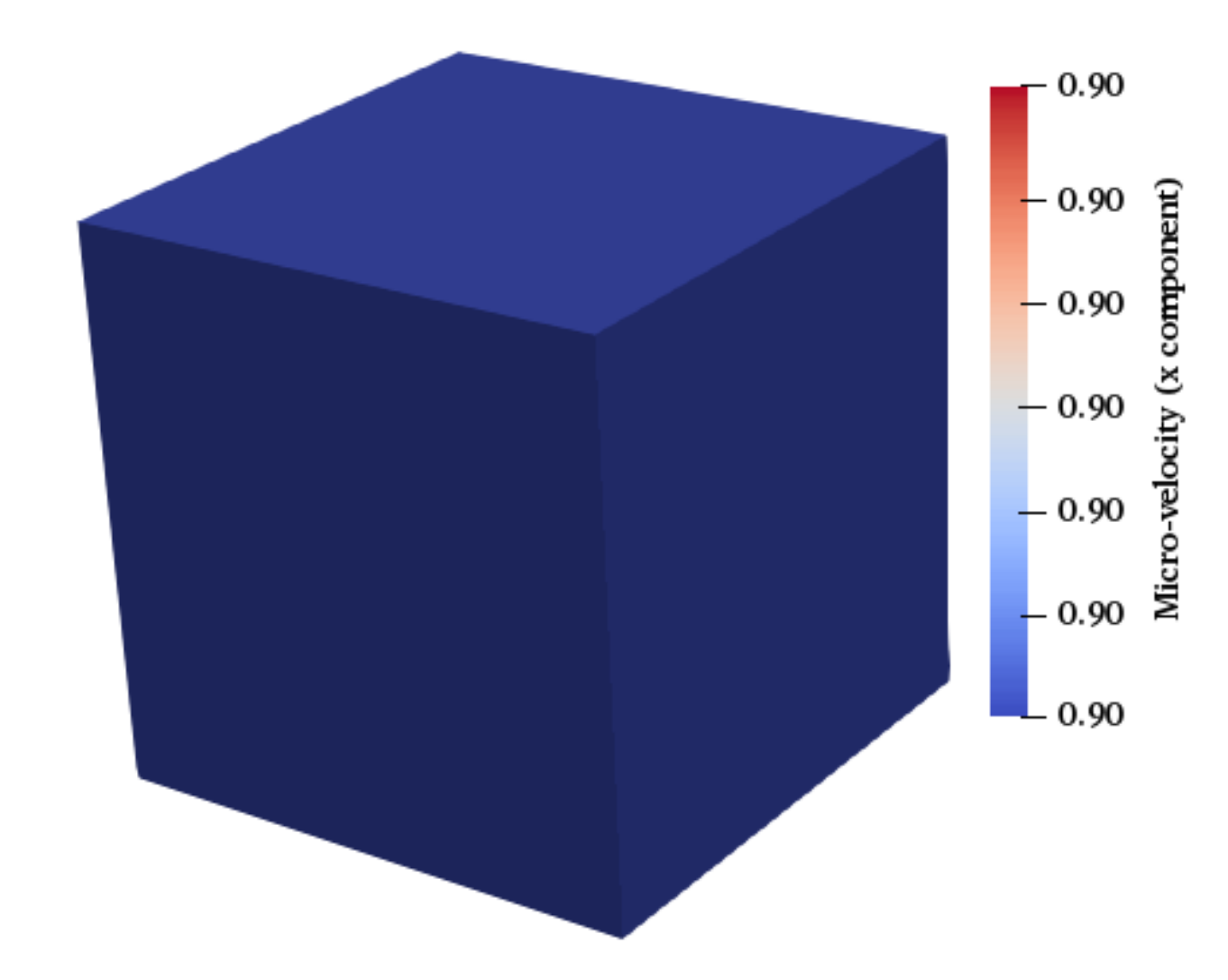}}
	\caption{\textsf{Non-constant Jacobian elements:}~Pressure and velocity contours are shown for Mesh $\#$2 (\textbf{Fig.~\ref{Fig:Mesh2}}) with non-constant Jacobian elements. Pressures are varying linearly from the left face to the right one and velocities are constant throughout the domain as expected. These results show that the proposed mixed DG formulation is capable of providing accurate results using non-constant Jacobian elements. \label{Three_D_Domain_Non_2}}
\end{figure}

\section{NUMERICAL CONVERGENCE ANALYSIS}
\label{Sec:S6_DG_Numerical_convergence}

In this section, we perform numerical convergence
analysis of the proposed stabilized DG formulation 
with respect to both $h$- and $p$-refinements.

\subsection{2D numerical convergence analysis:}
Convergence analysis in the 2D setting is performed on the boundary value problem described in Section \ref{Sec5:2D_square}. This problem was also employed by
\citep{Nakshatrala_Joodat_Ballarini_P2} for the convergence analysis
of the stabilized mixed \emph{continuous} Galerkin (CG)
formulation of the DPP model. The exact solutions for the pressures and velocities are provided by equations \eqref{Eqn:2D_Convergence_Analytical_p1} -- \eqref{Eqn:2D_square_Exact}. The domain for this problem is homogeneous (macro- and micro-permeabilities are constant within the domain) and same equal-order interpolations are used throughout the domain. The computational domain is shown in \textbf{Fig.~\ref{Fig:Dual_Problem_2D_domain}} and the parameter values are provided in Table \ref{Tb4:2D_convergence_analysis_data}. The three-node triangular element (T3), which is a simplicial finite element, is employed in the numerical simulation and the convergence is obtained under both \emph{$h$-refinement} and \emph{$p$-refinement}.
In \textbf{Figs.~\ref{Fig:DG_Problem_2D_h_refinement}} and \textbf{\ref{Fig:DG_Problem_2D_p_refinement}}, the convergence rates under $h$-refinement and $p$-refinement are provided for the $L_2$-norm and the $H^1$-norm of the pressure fields in the two pore-networks. The rates of convergence under $h$- and $p$-refinements are observed to be polynomial and exponential, respectively.
These results are in accordance with the theory
(viz. Corollary~\ref{Corollary:DG_rates_of_convergence}).
%
\begin{figure}[!h]
	\subfigure{
		\includegraphics[clip,scale=0.28]{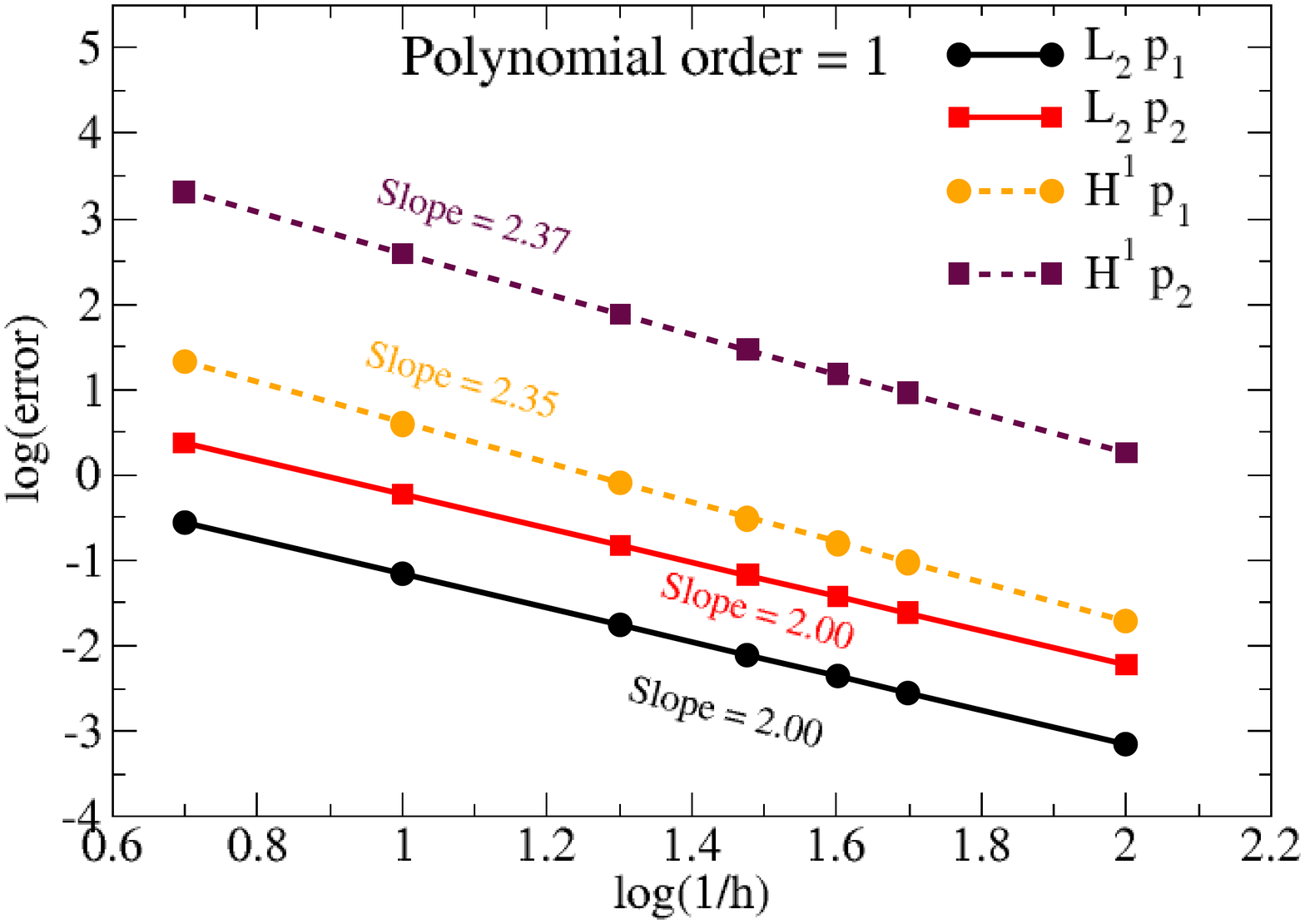}}
	\hspace{0.3cm}
	\subfigure{
		\includegraphics[clip,scale=0.28]{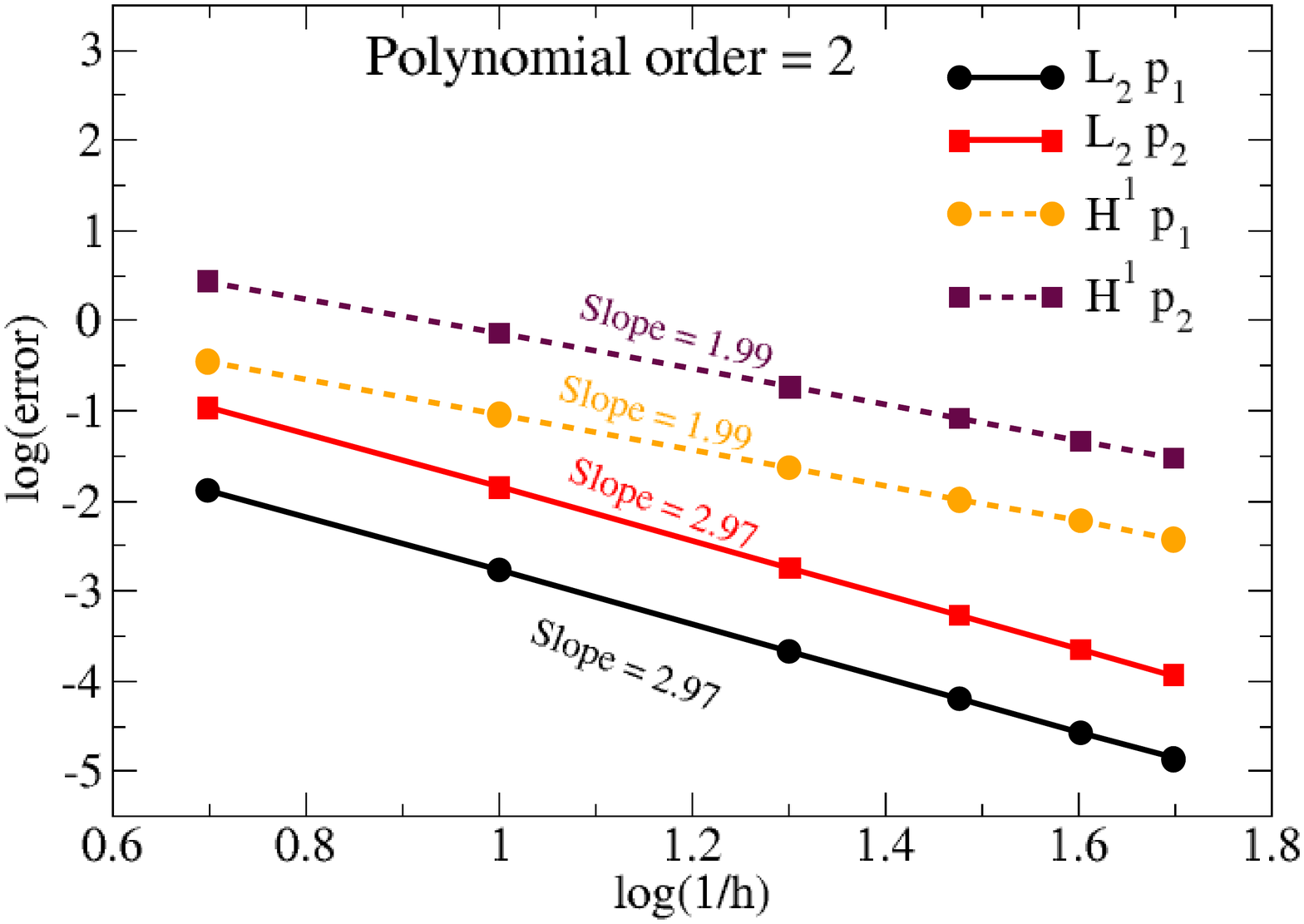}}
	\subfigure{
		\includegraphics[clip,scale=0.28]{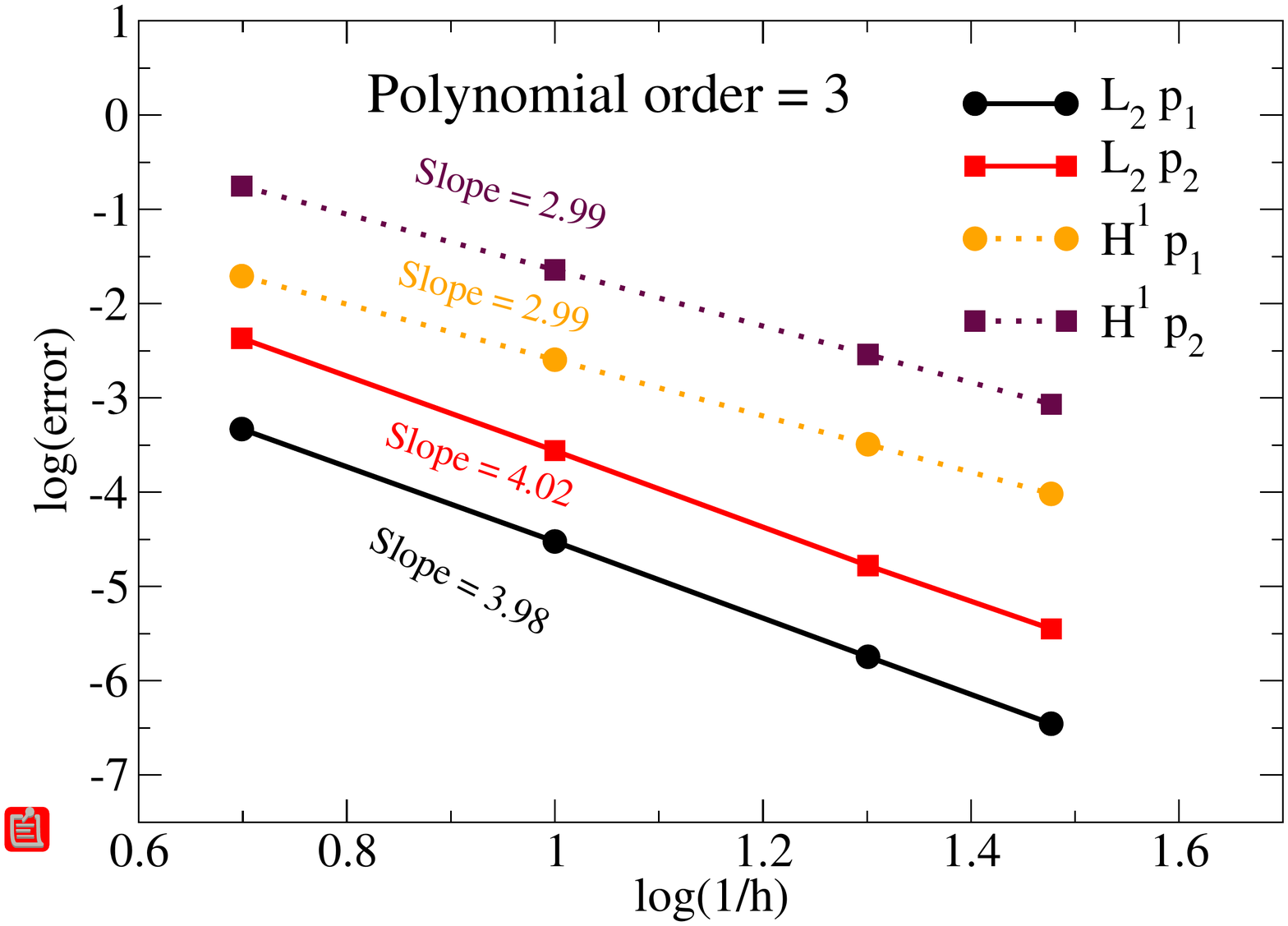}}
	\hspace{0.3cm}
	\subfigure{
		\includegraphics[clip,scale=0.28]{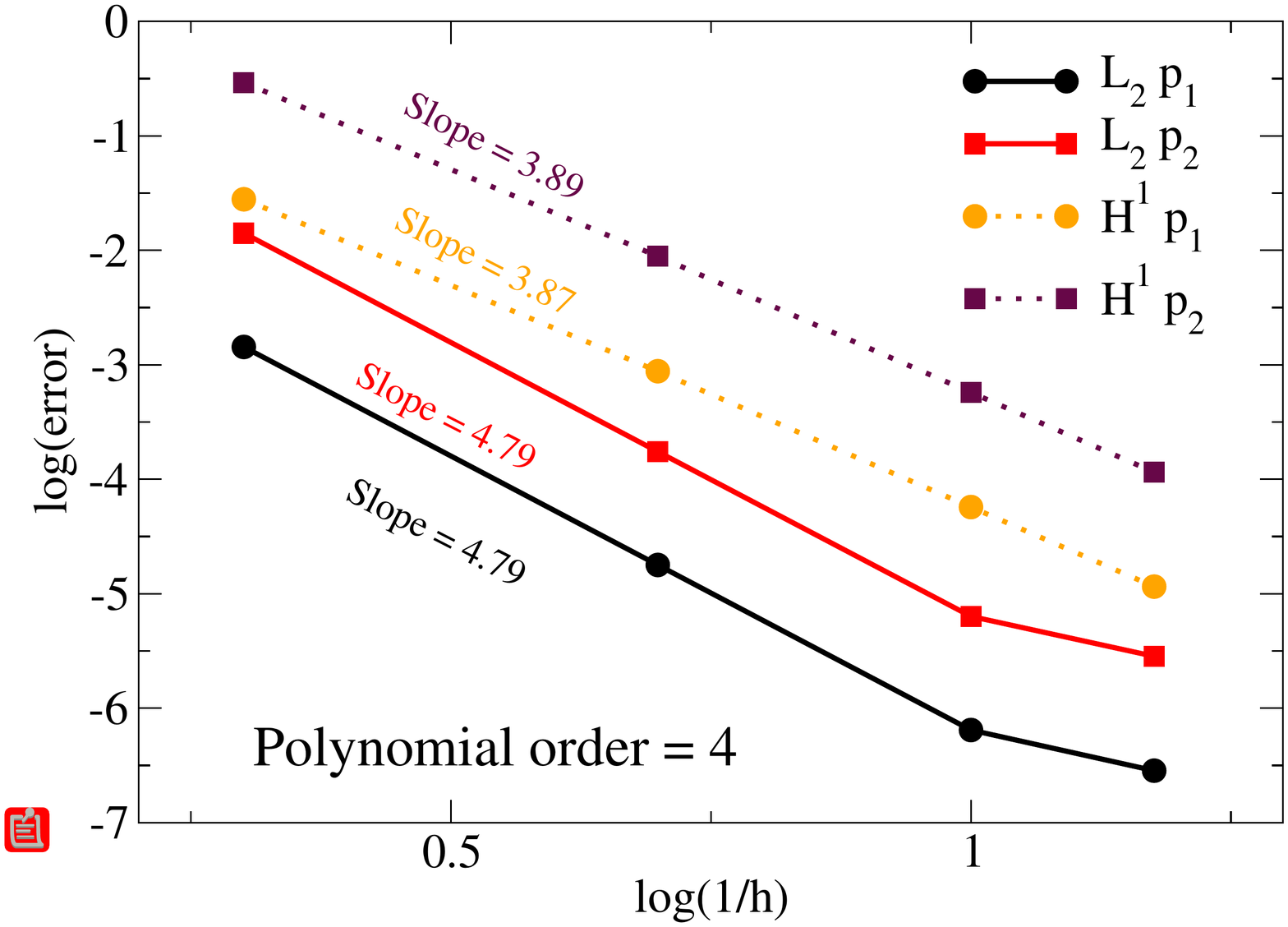}}
	\caption{\textsf{2D numerical convergence analysis:}~This figure provides the convergence rates under $h$-refinement for various polynomial orders. 
		The rate of convergence is polynomial, which is in accordance with the theory (viz. Corollary~\ref{Corollary:DG_rates_of_convergence}).}
	\label{Fig:DG_Problem_2D_h_refinement}
	\vspace{0.25cm}
\end{figure}
%
\begin{figure}
	\includegraphics[clip,scale=0.32]{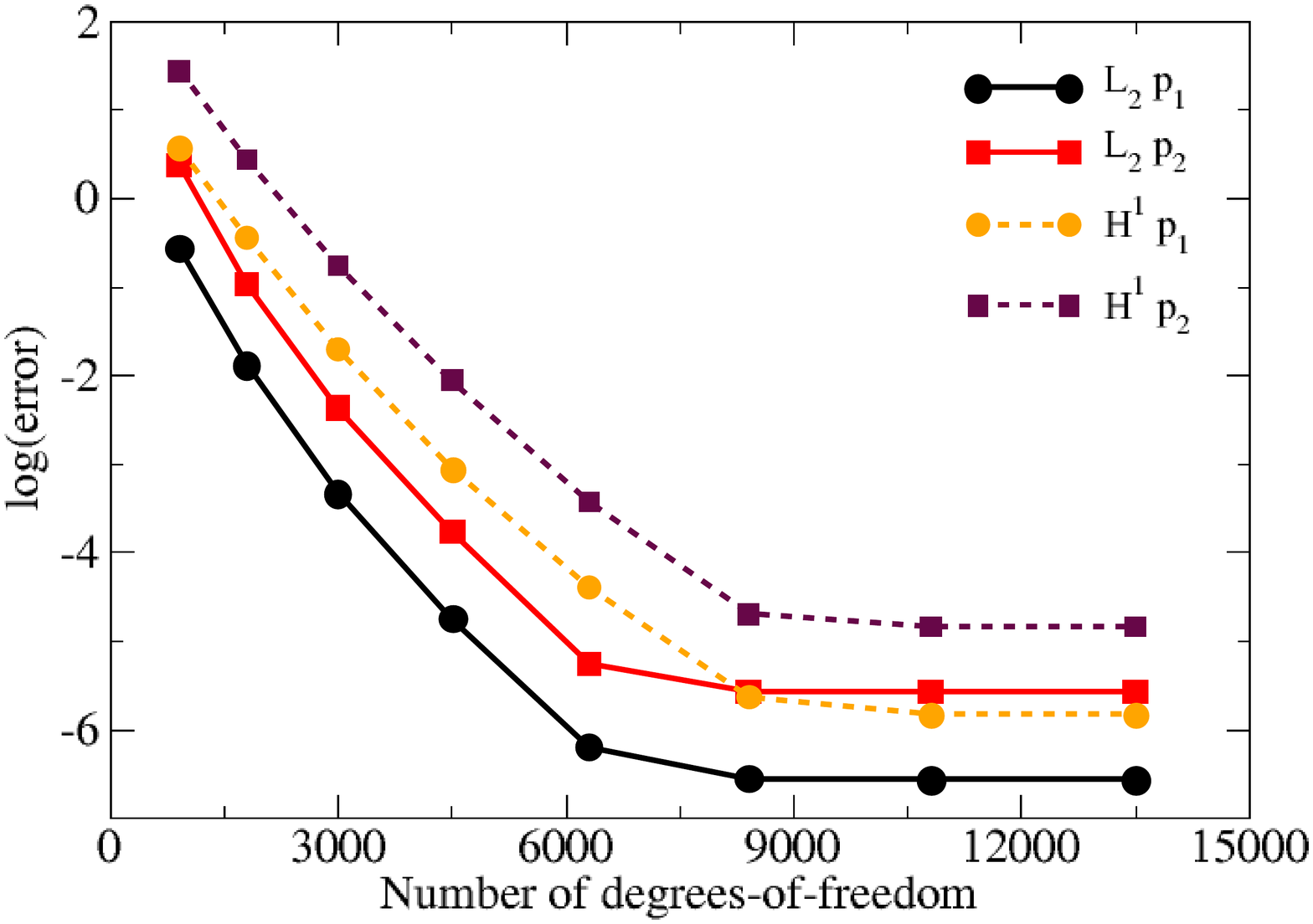}
	\caption{\textsf{2D numerical convergence analysis:}~This figure shows the results of numerical convergence
		under $p$-refinement for a fixed mesh size ($h = 0.2$).
		The number of degrees-of-freedom corresponds to $p
		= 1$ to $8$. The rate of convergence is exponential,
		which is in accordance with the theory (viz. Corollary~\ref{Corollary:DG_rates_of_convergence}).}
	\label{Fig:DG_Problem_2D_p_refinement}
\end{figure}
%
\subsection{3D numerical convergence analysis}
\label{Sec:3D_convergence}
The computational domain of this problem
is a unit cube with
pressure being prescribed on the entire
boundary of the two pore-networks. The
analytical solution takes the following
form: 
\begin{align}
  p_1(x,y,z) &= \frac{\mu}{\pi} \exp(\pi x)
  \left(\sin(\pi y) + \sin(\pi z)\right)
  - \frac{\mu}{\beta k_1} \left(\exp(\eta y)
  + \exp(\eta z) \right) \label{Eqn:3D_Convergence_Analytical_p1}\\
  p_2(x,y,z) &= \frac{\mu}{\pi} \exp(\pi x)
  \left(\sin(\pi y) + \sin(\pi z)\right)
  + \frac{\mu}{\beta k_2} \left(\exp(\eta y)
  + \exp(\eta z) \right) \label{Eqn:3D_Convergence_Analytical_p2}\\
  \mathbf{u}_1(x,y,z) &= -k_1 \exp(\pi x)
  \left(\begin{array}{c}
    \sin(\pi y) + \sin(\pi z) \\
    \cos(\pi y) \\
    \cos(\pi z) 
  \end{array}\right)
  + \frac{\eta}{\beta} \left(\begin{array}{c}
    0 \\
    \exp(\eta y) \\
    \exp(\eta z) 
  \end{array}\right) \\
  \mathbf{u}_2(x,y,z) &= -k_2 \exp(\pi x)
  \left(\begin{array}{c}
    \sin(\pi y) + \sin(\pi z) \\
    \cos(\pi y) \\
    \cos(\pi z) 
  \end{array}\right)
  - \frac{\eta}{\beta} \left(\begin{array}{c}
    0 \\
    \exp(\eta y) \\
    \exp(\eta z) 
  \end{array}\right) 
\end{align}
Pressure boundary conditions on each face are obtained by evaluating the analytical solution 
on the corresponding boundary of each pore-network. 
Table
\ref{Tb5:3D_convergence_analysis_data} provides
the parameter values employed in the numerical
simulation. 
{\small
  \begin{table}[!h]
    \caption{Model parameters for 3D numerical convergence analysis.}
    \centering
		\begin{tabular}{|c|c|} \hline
			Parameter & Value \\
			\hline
			$\gamma \mathbf{b}$ & $\{0.0,0.0,0.0\}$\\
			$L_x $ & $1.0$ \\
			$L_y $ & $1.0$ \\
			$\mu $ & $1.0$ \\
			$\beta $ & $1.0$ \\
			$k_1$&  $1.0$ \\
			$k_2$&  $0.1$\\
			$\eta$ & $\sqrt{11} \simeq 3.3166$\\
			$\eta_{u}$ & $100.0$\\
			$\eta_{p}$ & $0.0$\\
			\hline
			$p_i^{\mathrm{left}},~i=1,2$&  Obtained by evaluating   \\
			$p_i^{\mathrm{right}},~i=1,2$&  the analytical solution \\
			$p_i^{\mathrm{top}},~i=1,2$&   (equations \eqref{Eqn:3D_Convergence_Analytical_p1} and \eqref{Eqn:3D_Convergence_Analytical_p2} ) \\
			$p_i^{\mathrm{bottom}},~i=1,2$&  on the respective boundaries. \\
			\hline 
		\end{tabular}
		\label{Tb5:3D_convergence_analysis_data}
	\end{table}
}

The eight-node brick element (B8), which is a non-simplicial
element, is employed in this numerical simulation. \textbf{Figs.~\ref{Fig:DG_Problem_3D_h_refinement}} and \textbf{\ref{Fig:DG_Problem_3D_p_refinement}}
respectively provide the convergence rates under $h$-refinement and $p$-refinement for the
$L_2$-norm and the $H^1$-norm of the pressure fields in the two pore-networks.  As can be seen, the rates of convergence under the $h$- and $p$-refinements are polynomial and exponential, respectively; which are in accordance with the theory (viz. Corollary~\ref{Corollary:DG_rates_of_convergence}).

\begin{figure}
	\subfigure{
		\includegraphics[clip,scale=0.3]{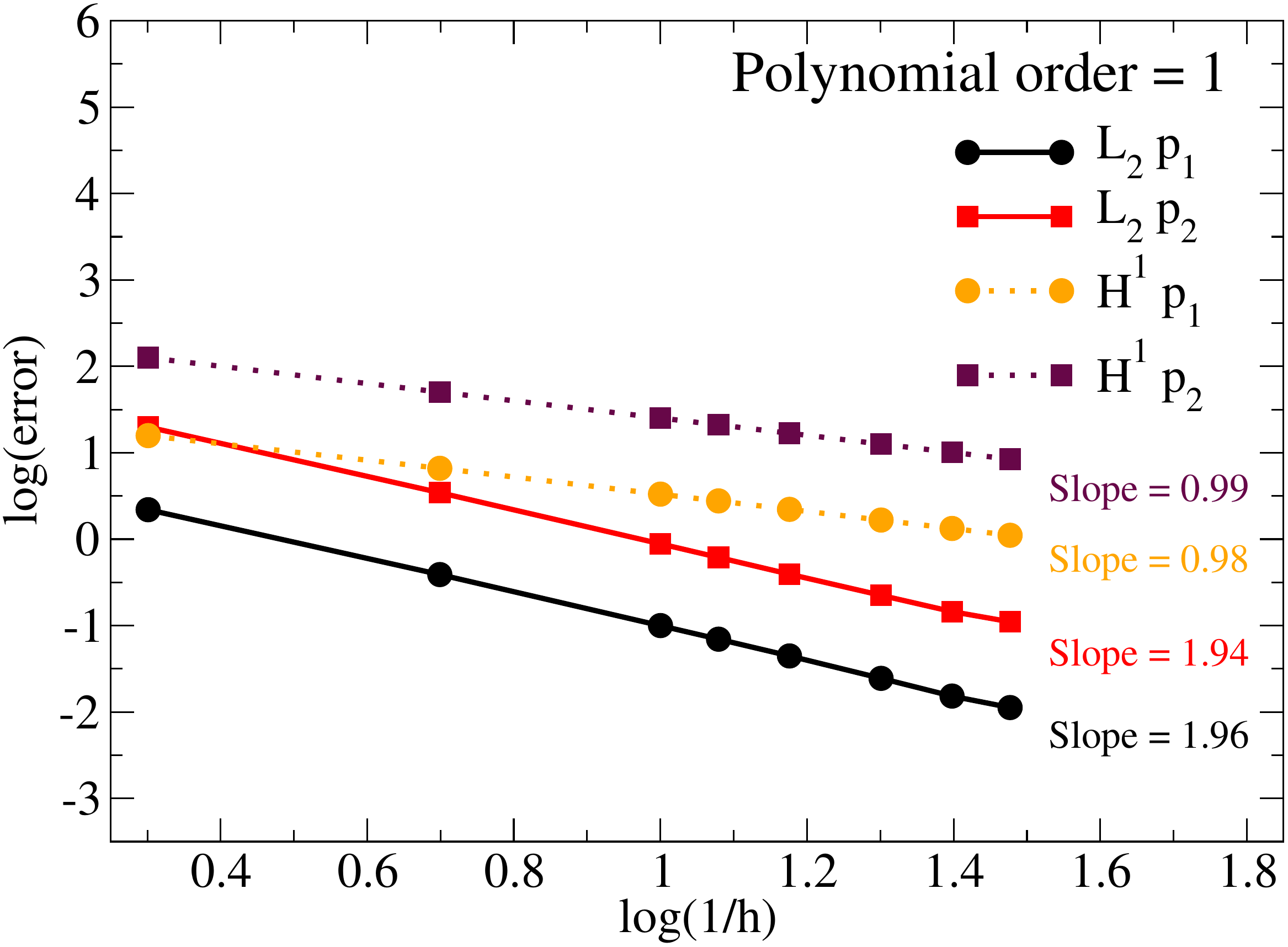}}
	\subfigure{
		\includegraphics[clip,scale=0.3]{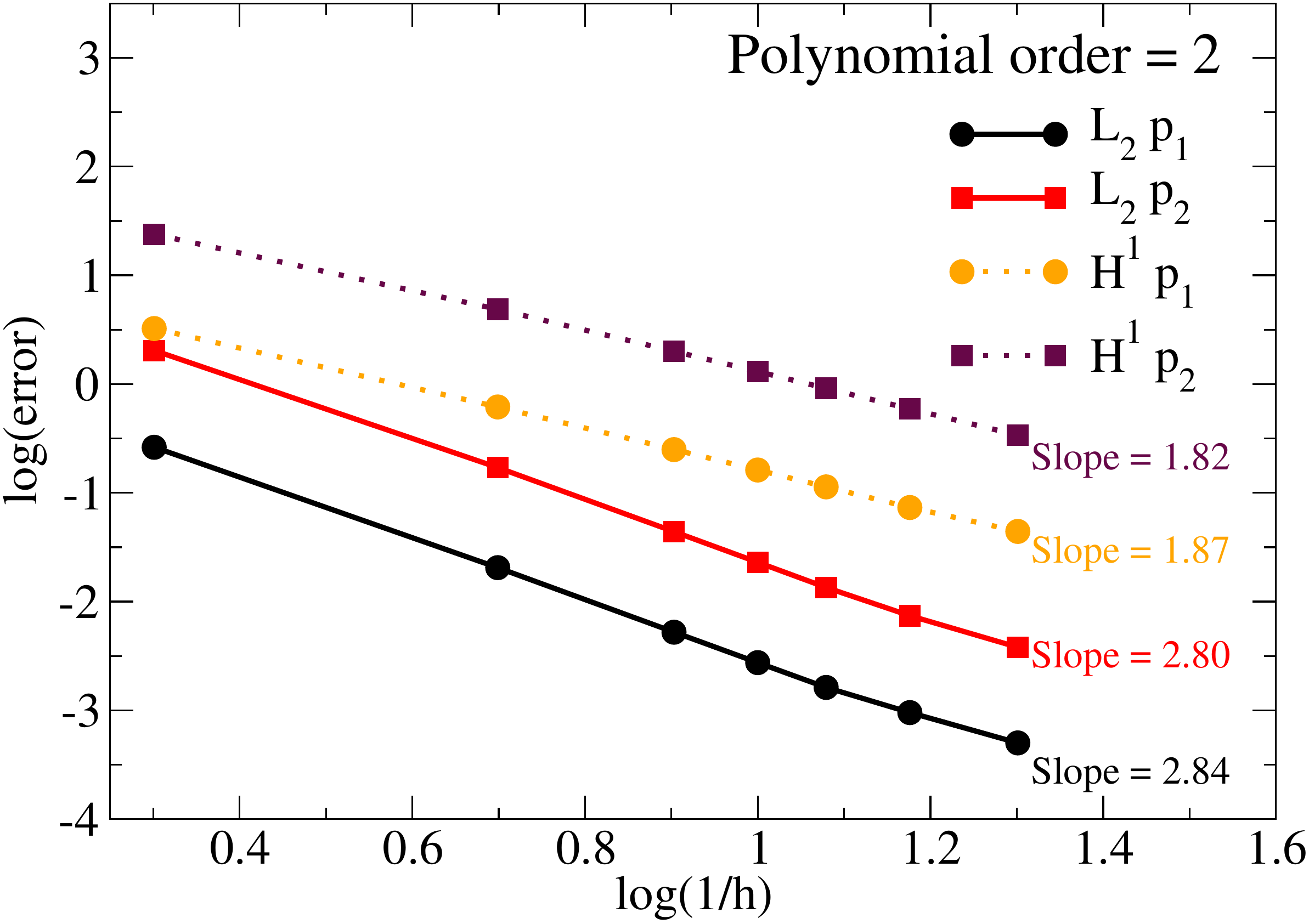}}
	\subfigure{
		\includegraphics[clip,scale=0.3]{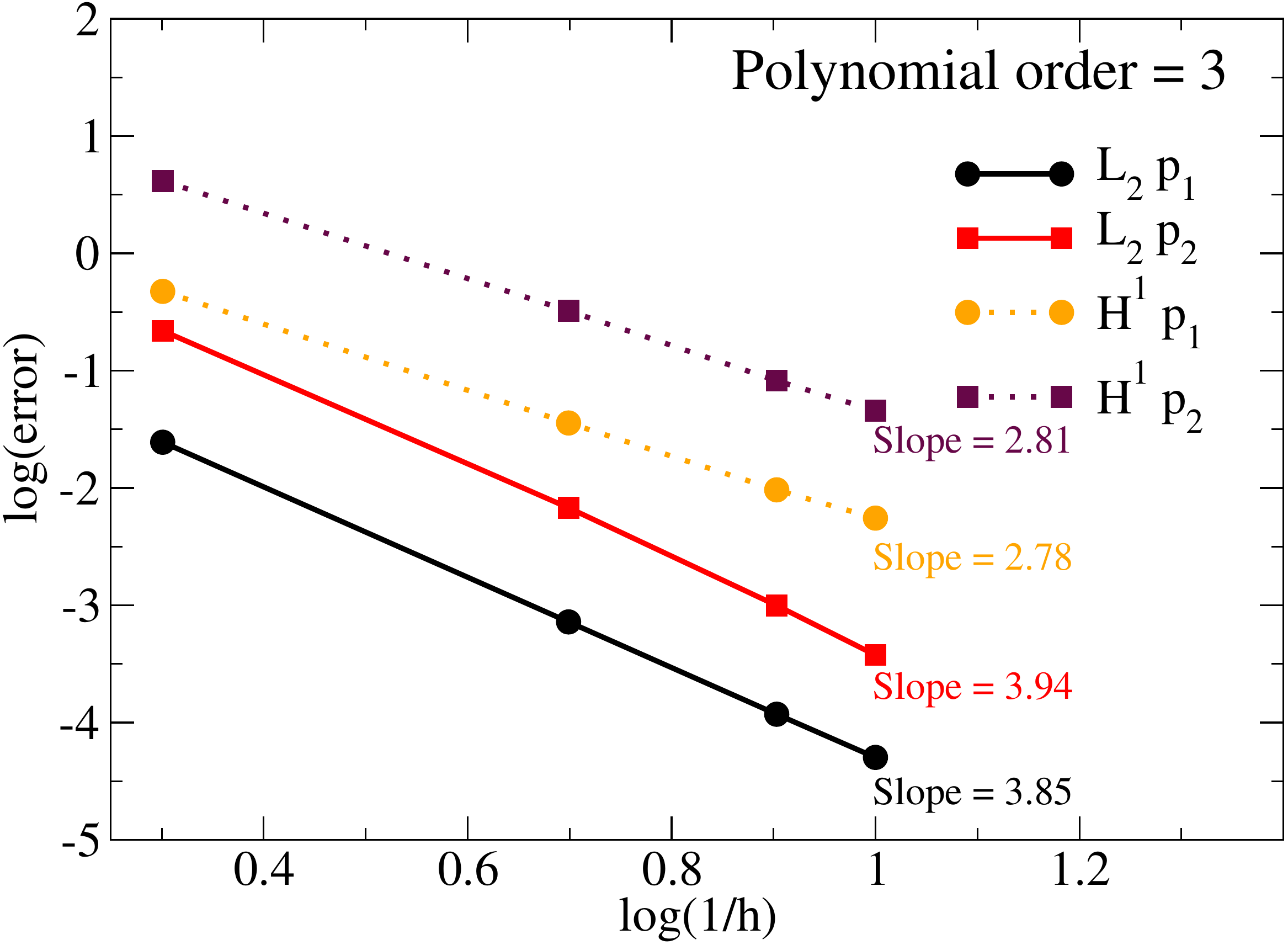}}
	\caption{\textsf{3D numerical convergence analysis:}~This figure provides the convergence rates under $h$-refinement for various polynomial orders. 
		The rate of convergence is polynomial, which is in accordance with the theory (viz. Corollary~\ref{Corollary:DG_rates_of_convergence}).}
	\label{Fig:DG_Problem_3D_h_refinement}
\end{figure}
%
\begin{figure}[!h]
	\includegraphics[clip,scale=0.32]{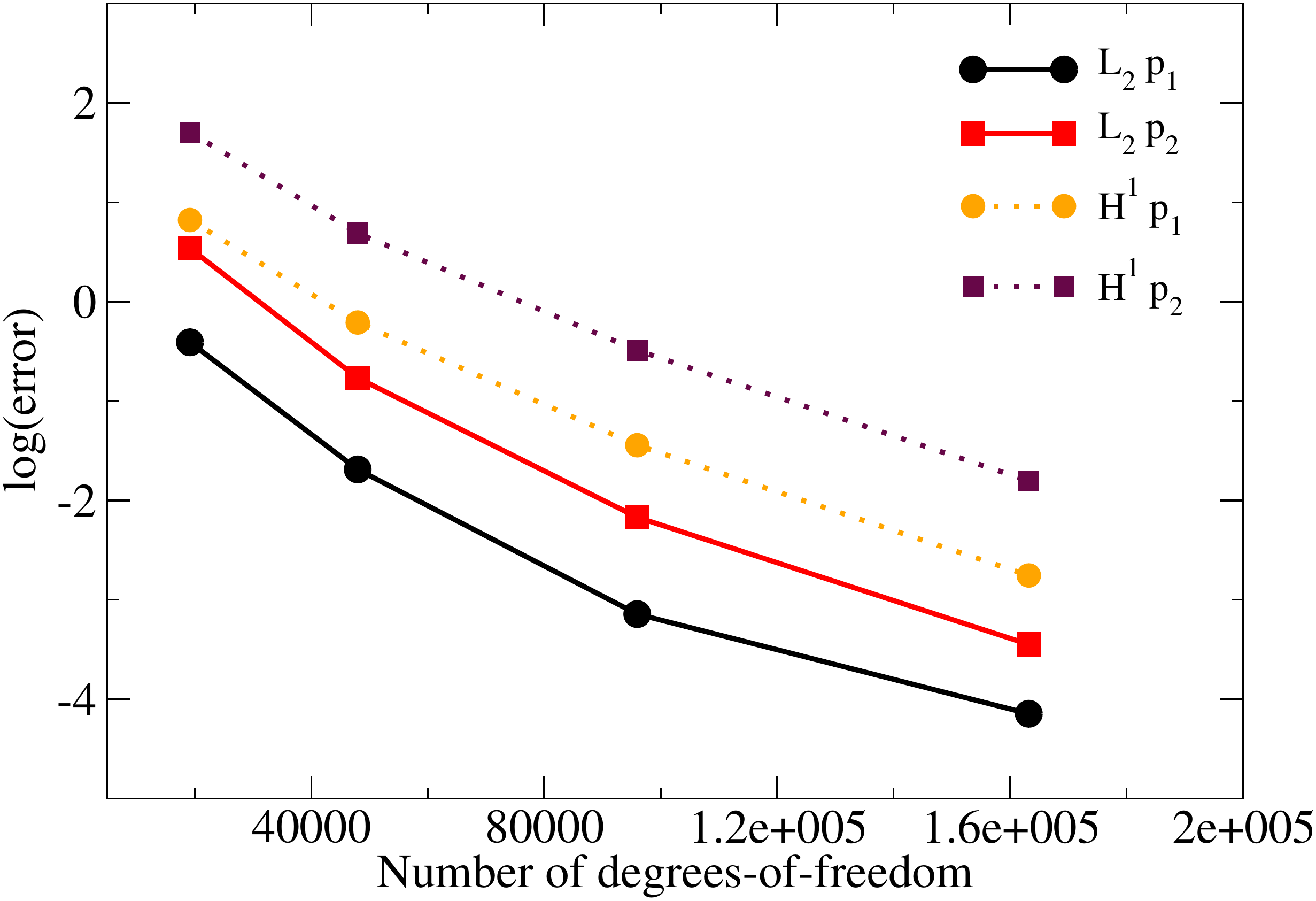}
	\caption{\textsf{3D numerical convergence analysis:}~This figure shows the results of numerical convergence
		under $p$-refinement for a fixed mesh size ($h = 0.2$).
		The number of degrees-of-freedom corresponds to $p
		= 1$ to $4$. The rate of convergence is exponential,
		which is in accordance with the theory (viz. Corollary~\ref{Corollary:DG_rates_of_convergence}).}
	\label{Fig:DG_Problem_3D_p_refinement}
\end{figure}
%

\section{CANONICAL PROBLEM AND STRUCTURE PRESERVING PROPERTIES}
\label{Sec:S7_DG_Structure_preserving}
In this section, first, robustness of the proposed stabilized mixed DG formulation is assessed using a standard test problem, with abrupt changes in material properties and elliptic singularities. In the literature, this problem is typically referred to as the quarter five-spot checkerboard problem. Second, the element-wise mass balance property associated with the CG and DG formulations is compared.

\subsection{Quarter five-spot checkerboard problem}
\label{Sec:fiv_spot}
The original form of this problem, known as ``five-spot problem'' with homogeneous properties, has been firstly designed for the Darcy equations. Herein, we extend this problem to the DPP model with modified boundary conditions and heterogeneous medium properties. 
\textbf{Fig.~\ref{Fig:DG_Checkerboard_Domain}} shows the computational domain and the boundary conditions for the five-spot problem. An injection well surrounded by four production wells placed at four corners of a square domain form a typical setting in the enhanced oil recovery applications. The underlying symmetry allows for solving the problem only in the top right quadrant, which is referred to as a ``quarter'' five-spot problem. In the well-known ``checkerboard problem'', such a computational domain is divided into four sub-regions I, II, III, and IV with abrupt changes in the permeability. 

In this problem, elliptic singularities are observed near the injection and production wells which are located at the opposite corners of the diagonals (denoted by $C_{\mathrm{inj}}$ and $C_{\mathrm{prod}}$, respectively). The normal component of velocity is prescribed to be zero on the entire boundary of the micro-pore network. In the macro-network, however, velocity at the injection and production wells is prescribed by applying a source/sink term while zero normal velocity is assumed on the rest of the boundary. 
It is worth mentioning that the prescribed source and sink strengths at injection and production wells are, respectively, equal to +1 and -1. However, instead of applying a pointwise sink/source at the location of wells, the normal component of velocity is applied along the external edges of the corner element in $x$- and $y$-directions with an equivalent distribution as shown in \textbf{Fig.~\ref{Fig:DG_Checkerboard_Domain}}.
%
\begin{figure}
  \includegraphics[clip,scale=0.44]{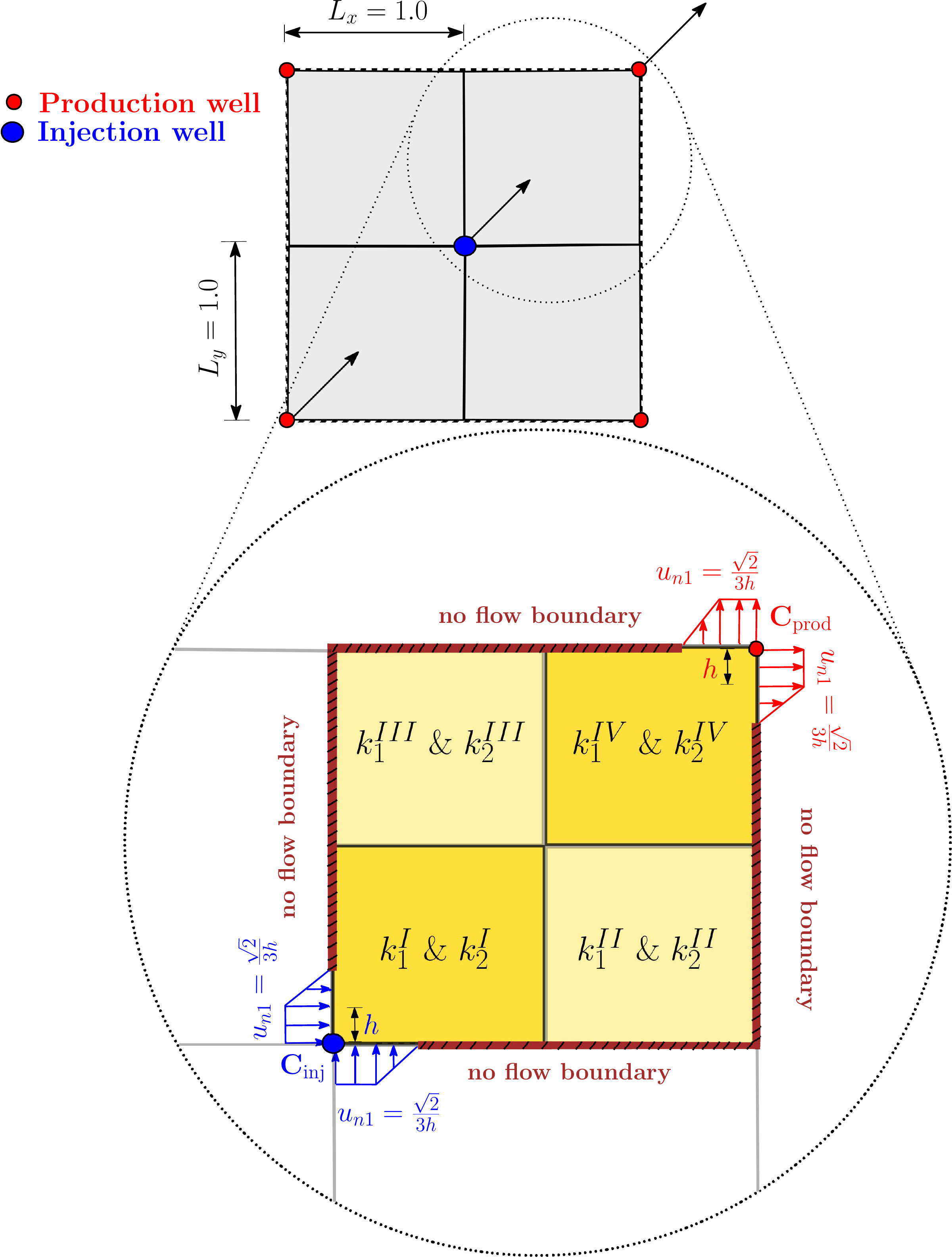}
  \caption{\textsf{Quarter five-spot checkerboard problem:}~This figure shows the computational domain and boundary conditions for the quarter five-spot checkerboard problem. The heterogeneous domain is divided into four sub-regions with permeabilities shown in equation \eqref{Eqn:Perm_case1_checkerboard}. The normal component of micro-velocity is equal to zero on the entire boundary. In the macro-network, however, source/sink strengths are prescribed in form of equivalent normal velocity distributions at the production ($C_{\mathrm{prod}}$) and injection ($C_{\mathrm{inj}}$) wells. On the rest of the boundary, the normal component of macro-velocity is assumed to be zero.
    \label{Fig:DG_Checkerboard_Domain}}
\end{figure}

Table \ref{Tb6:Five_Spot_data} provides the parameter values for this problem.
The permeability parameters in sub-regions I and IV and the ones in sub-regions II and III are mutually equal. Herein, we assume that sub-regions I and IV are more permeable compared to sub-regions II and III with the following drag coefficients:
\begin{align}
  &\left(\frac{\mu}{k_1}\right)_I
  = \left(\frac{\mu}{k_1}\right)_{IV} = 1,
  \quad \left(\frac{\mu}{k_1}\right)_{II}
  = \left(\frac{\mu}{k_1}\right)_{III} = 100,
  \nonumber \\
  &\left(\frac{\mu}{k_2}\right)_I
  = \left(\frac{\mu}{k_2}\right)_{IV} = 10,
  \quad \mathrm{and} \quad
  \left(\frac{\mu}{k_2}\right)_{II}
  = \left(\frac{\mu}{k_2}\right)_{III} = 1000
	\label{Eqn:Perm_case1_checkerboard}
\end{align}
{\small
  \begin{table}[!h]
    \caption{Model parameters for the quarter five-spot checkerboard problem.}
		\centering
		\begin{tabular}{|c|c|} \hline
			Parameter & Value \\
			\hline
			$\gamma \mathbf{b}$ & $\{0.0,0.0\}$\\
			$L_x $ & $1.0$ \\
			$L_y $ & $1.0$ \\
			$\mu $ & $1.0$ \\
			$\beta $ & $1.0$ \\
			$k_1, ~k_2$&  refer to Eqn. \eqref{Eqn:Perm_case1_checkerboard} \\
			$u_{n1}$ & $0.0~\mathrm{On}~ \partial \Omega - \{\mathrm{C}_{\mathrm{prod}}~ \&~ \mathrm{C}_{\mathrm{inj}}\}$ \\
			$u_{n2}$ & $0.0~\mathrm{On}~\partial \Omega$ \\ 
			\hline
			source and sink  & $-1~ \mathrm{at}~ \mathrm{C}_{\mathrm{prod}}$ \\
			strength & $+1~ \mathrm{at}~\mathrm{C}_{\mathrm{inj}}$ \\
			\hline
			$\eta_u$& $0, 10, 100$\\
			$\eta_p$& $0, 10, 100$\\
			$h$&structured T3 mesh of size $0.01$ used\\
			\hline 
		\end{tabular}
		\label{Tb6:Five_Spot_data}
	\end{table}
}

\textbf{Fig.~\ref{Fig:fiveSpot_problem_p}} shows the macro- and micro pressure profiles for this problem. Steep gradients near the injection and
production wells with no spurious oscillation in the pressure fields are observed under the proposed DG formulation which confirm the robustness of the numerical formulation. 
%
%
In order to further explore the effect of stabilization parameters on the solution profiles, this problem has been solved for different combinations of $\eta_u$ and $\eta_p$ as shown in \textbf{Fig.~\ref{Fig:fiveSpot_problem_v_1}}. As can be seen, $\eta_u$ and $\eta_p$ have no noticeable effect on x-component of velocities under the DG formulation. However, spurious oscillations are observed under the CG formulation at the interface of sub-regions with different permeability values which implies that CG formulations fall short in capturing material discontinuities. 
\begin{figure}
	\subfigure[Macro-pressure \label{Fig:fiveSpot_problem_p1}]{
		\includegraphics[scale=0.35]{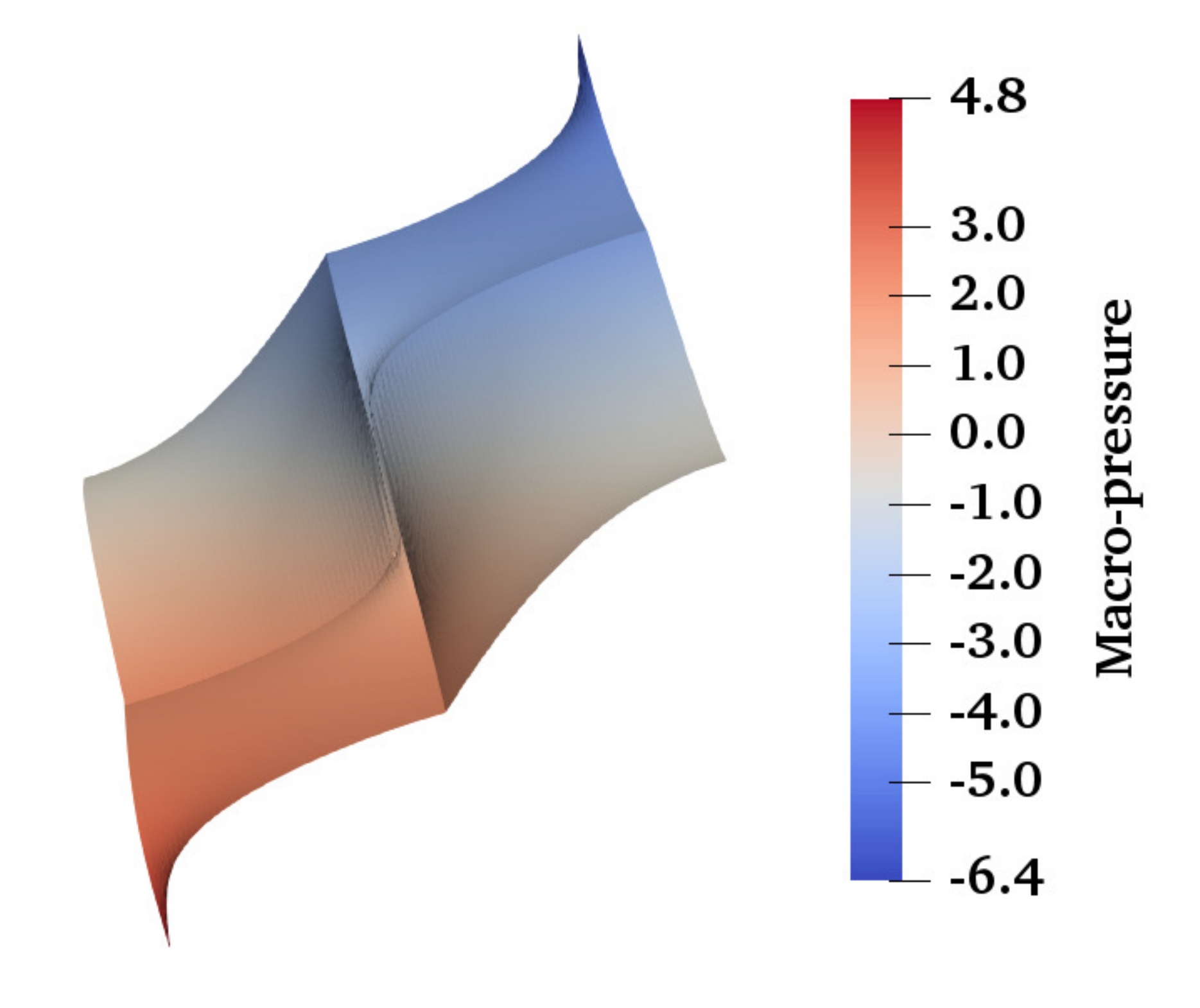}}
	\hspace{1cm}
	\subfigure[Micro-pressure \label{Fig:fiveSpot_problem_p2}]{
		\includegraphics[scale=0.33]{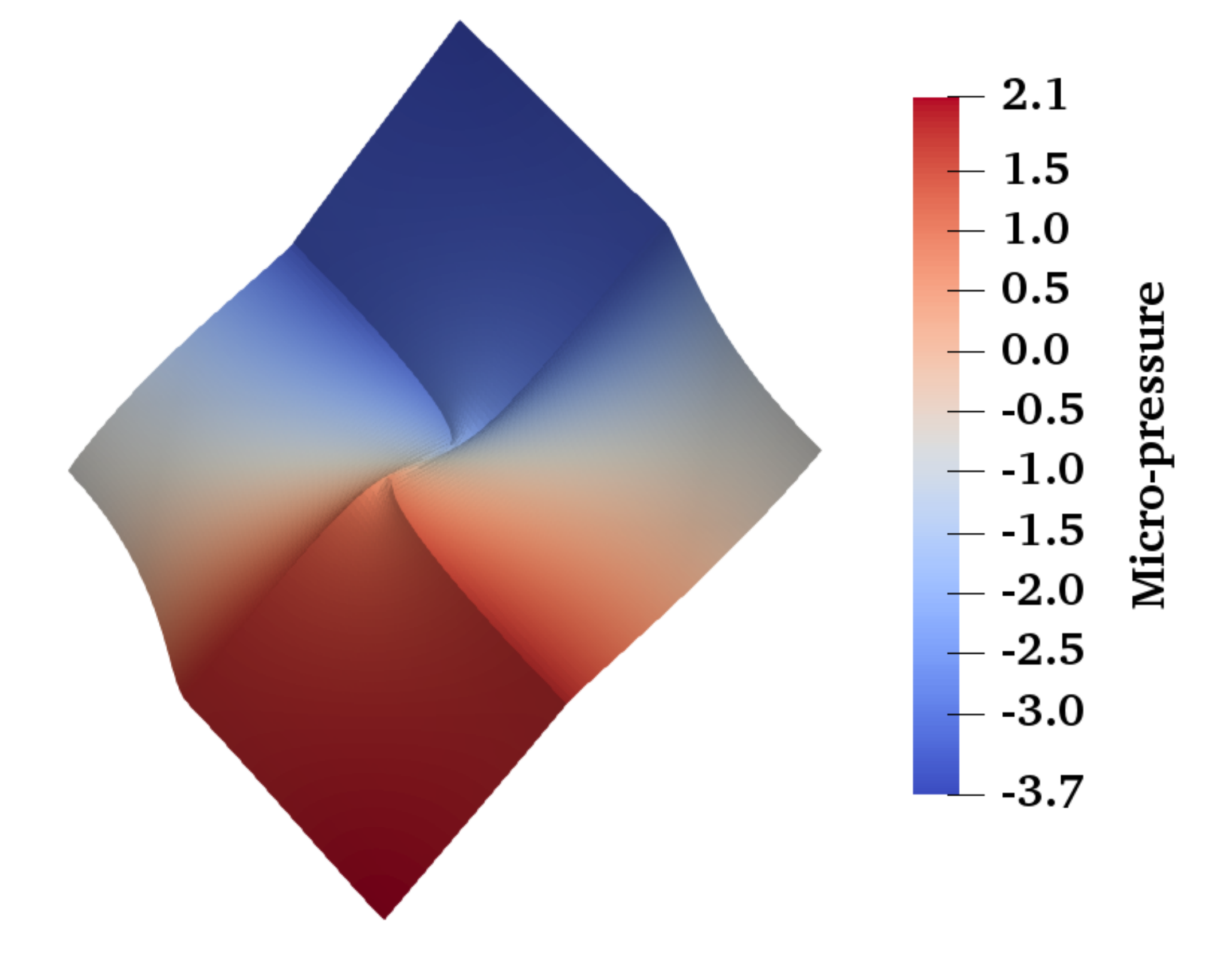}}
	\caption{\textsf{Quarter five-spot checkerboard problem:}~This figure shows that steep pressure gradients near the injection and production wells are correctly captured under the proposed DG formulation. However, no spurious oscillation are observed in the pressure fields which shows the robustness of our numerical formulation. These results are obtained for $\eta_u = \eta_p = 0$. \label{Fig:fiveSpot_problem_p}}
\end{figure}
\begin{figure}
	\subfigure{
		\label{Fig:fiveSpot_problem_v1_1}
		\includegraphics[trim=0cm 1cm 0cm 2cm,scale=0.45,clip]{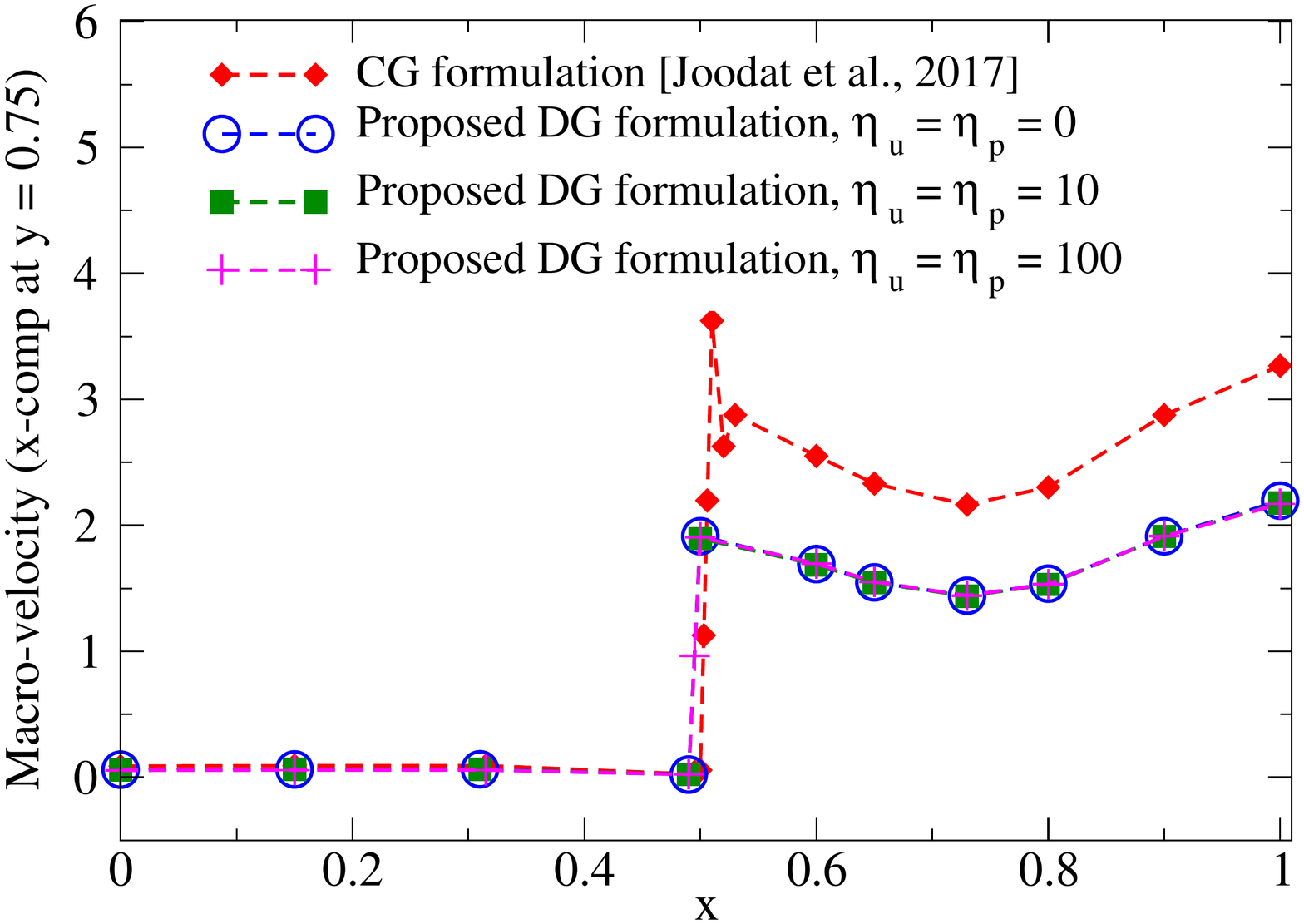}}
	\subfigure{
		\label{Fig:fiveSpot_problem_v2_1}
		\includegraphics[trim=0cm 1cm 0cm 3cm,scale=0.45,clip]{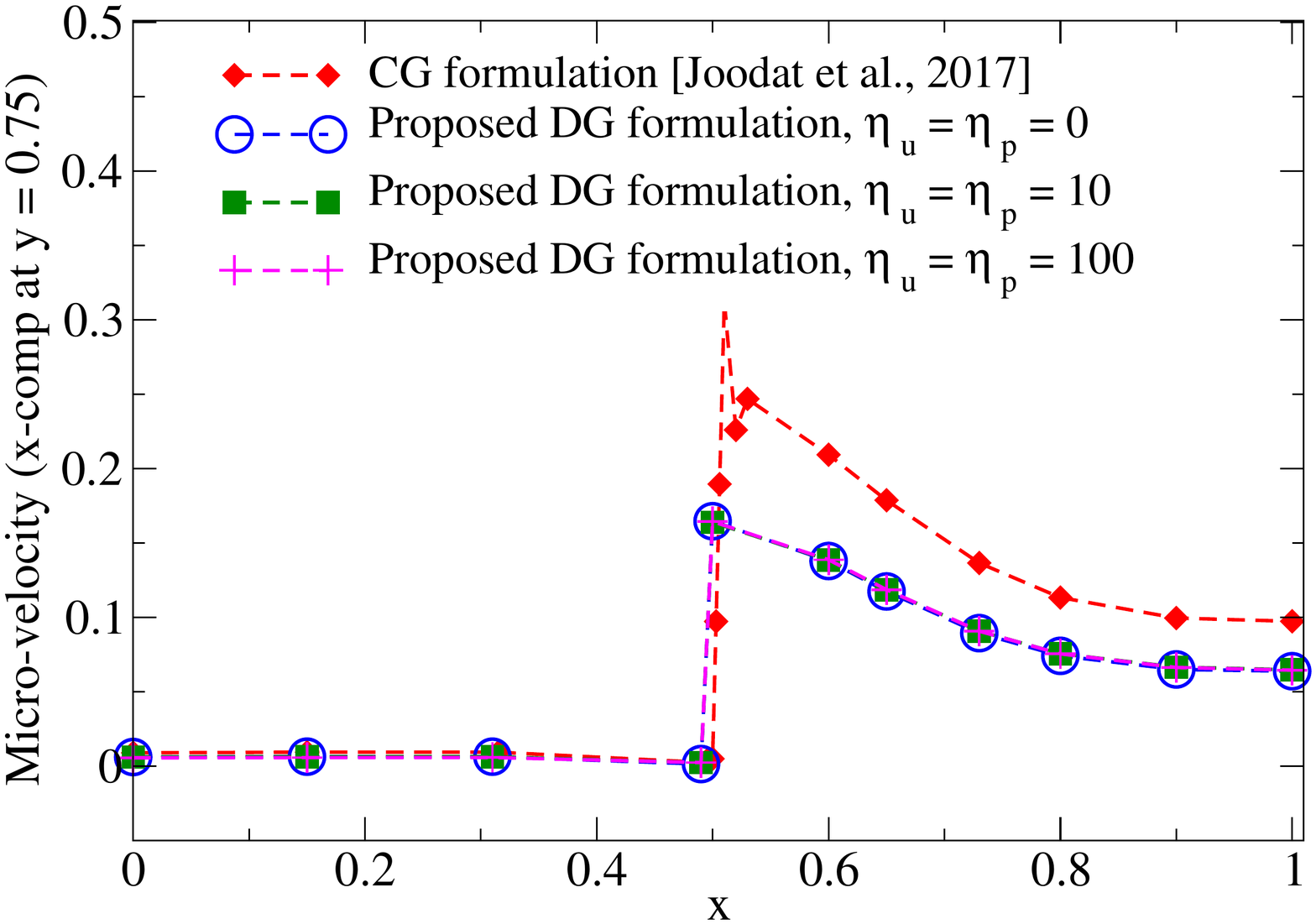}}
	\vspace{-0.25in}
	\caption{\textsf{Quarter five-spot
			checkerboard problem:}~This figure
		compares the x-component of the
		macro-velocity (top) and
		micro-velocity (bottom) profiles
		under the CG formulation and the
		proposed DG formulation with different values of stabilization parameters $\eta_u$ and $\eta_p$. As can be seen, under CG formulation, spurious oscillations are observed at the interface of sub-regions with different permeabilities. Moreover, $\eta_u$ and $\eta_p$ have no noticeable effect on solutions obtained under the DG formulation. \label{Fig:fiveSpot_problem_v_1}}
\end{figure}

\subsection{Element-wise mass balance} 
\label{Sec:Element_wise_MB}
A DG method, when designed properly, can exhibit
superior element-wise properties compared to its
continuous counterpart. CG formulations may suffer from
poor element-wise conservation;
however, they satisfy a global mass balance
\citep{Hughes_Engel2000continuous}.
The importance of element-wise mass
balance in subsurface modeling is
discussed in \citep{turner2012modeling}, which is
particularly true when the
flow is coupled with transport and/or
chemical reactions.

In this section, element-wise mass
balance error is investigated under
the proposed stabilized mixed DG formulation 
for the DPP model, and the results
are compared with its continuous
counterpart.
In the context of DPP, the net rate
of volumetric flux from both pore-networks can be obtained
as follows for an element $\omega
\in \mathcal{T}_h$: 
\begin{align}
  m(\omega):=\int_{\partial \omega}
  (\mathbf{u}_1 + \mathbf{u}_2)
  \cdot \widehat{\mathbf{n}} \; \mathrm{d}\Gamma
\label{Eqn:Local_mass_Flux}
\end{align}
After calculation, this equation should result in a zero value.
The maximum element-wise mass
inflow/outflow flux can be
obtained as follows:
\begin{subequations}
\begin{align}
  &m^{\mathrm{out}}_{\mathrm{max}} :=
  \max_{\omega \in \mathcal{T}_h}
  \left[\max[m(\omega),0]\right]
  \label{Eqn:Max_mass_outFlux} \\
  &m^{\mathrm{in}}_{\mathrm{max}} :=
  \max_{\omega \in \mathcal{T}_h}
  \left[\max[-m(\omega),0]\right]
  \label{Eqn:Max_mass_inFlux} 
\end{align}
\end{subequations}
It should be noted that the definition of the local mass flux presented in equation \eqref{Eqn:Local_mass_Flux} is different from the corresponding one under the Darcy equations. For the case of single porosity and under Darcy equations, the net flux is zero for the velocity. However, under the DPP model
the net flux need not be zero for the individual
velocities and it is shown to be zero for
the summation of $\mathbf{u}_1$ and $\mathbf{u}_2$.
The domain is discretized with structured T3 mesh of size $0.2$. 
We employ the same boundary value problem as stated in subsection \ref{Sec5:2D_square} with parameter values provided in Table \ref{Tb4:2D_convergence_analysis_data}.
Pressures are prescribed on the whole boundary in both pore-networks.

Comparisons of maximum local mass inflow/outflow with respect to different combinations of equal-order interpolation are illustrated in \textbf{Fig.~\ref{Fig:Local_mass_balance}} for both DG and CG formulations.
\textbf{Fig.~\ref{Fig:Local_mass_balance_field}} shows the local mass balance error in each element for cubic equal-order polynomials. The error values obtained under CG and DG formulations suggest that the DG formulation returns smaller errors.
%
%
\begin{figure}[!h]
  \includegraphics[clip,width=0.85
    \linewidth]{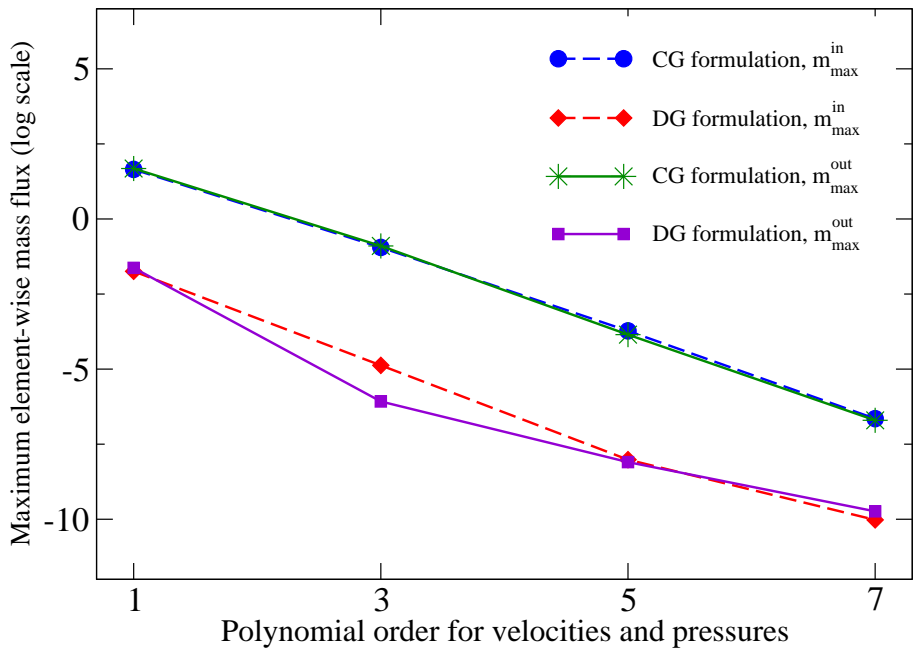}
  \caption{\textsf{Element-wise mass
      balance:}~This figure shows the variation
    of the maximum element-wise inflow/outflow
    flux with interpolation polynomial orders.}
	\label{Fig:Local_mass_balance}
\end{figure}

%
%
\begin{figure}[!h]
	\subfigure[CG formulation]{
		\includegraphics[clip,width=0.41\linewidth]{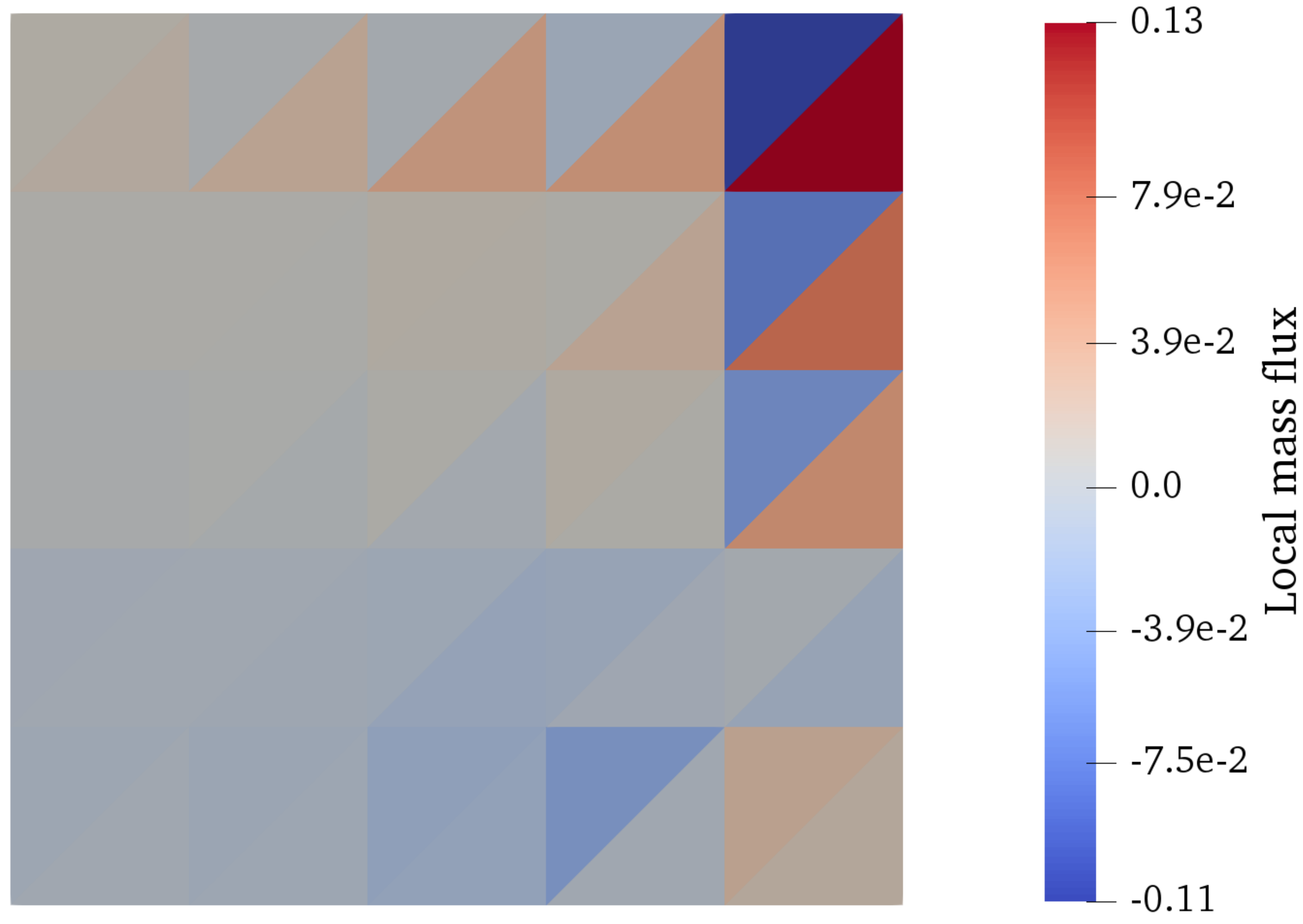}}
	\vspace{1cm}
	\subfigure[DG formulation]{
		\includegraphics[clip,width=0.41\linewidth]{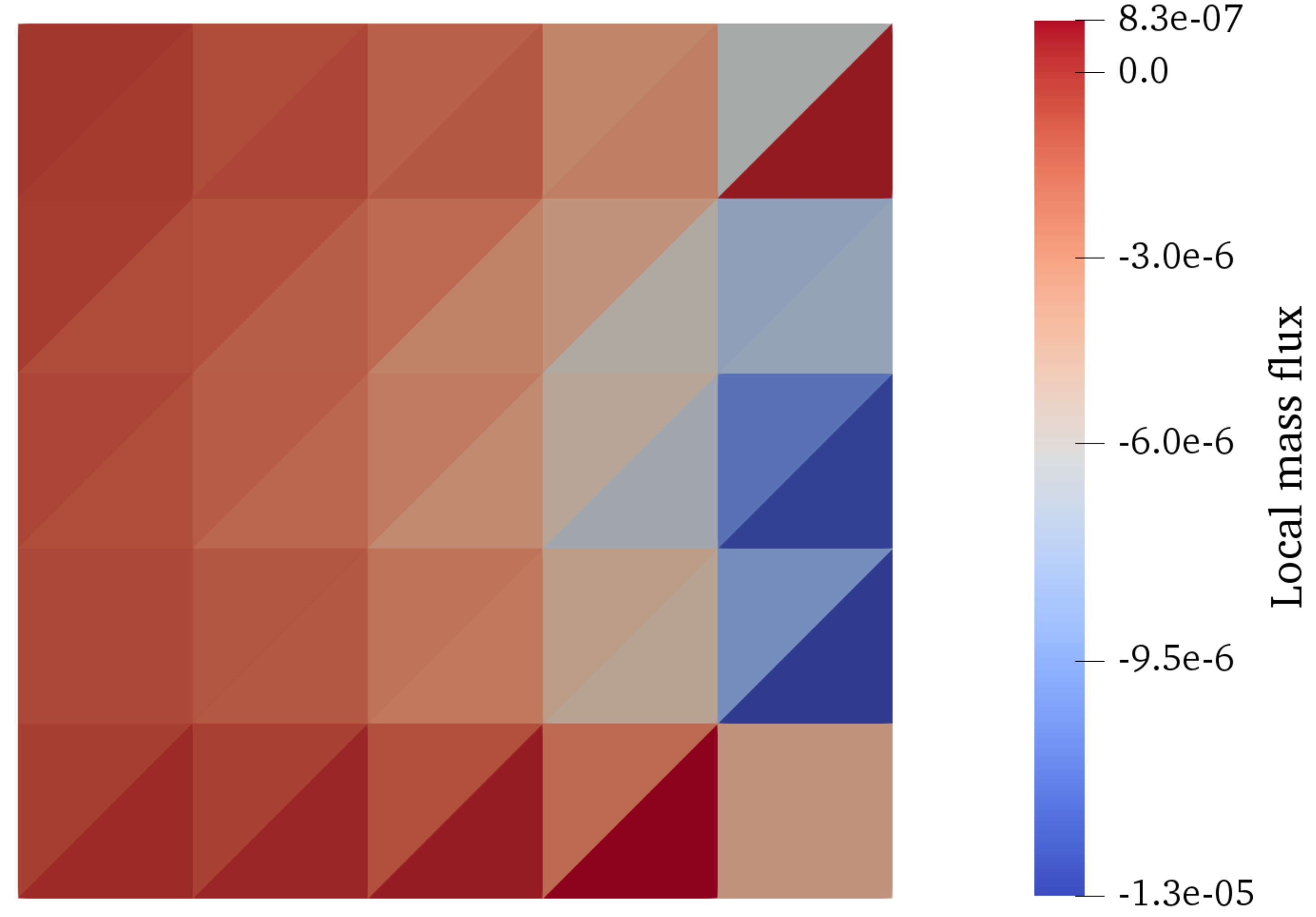}}
	
	\caption{\textsf{Element-wise mass
            balance:}~This figure shows the
          local mass balance error under both CG and DG formulations for cubic equal-order interpolation for all the field variables. As can be seen, the DG formulation returns smaller errors. }
	\label{Fig:Local_mass_balance_field}
\end{figure}

\section{COUPLED PROBLEM WITH HETEROGENEOUS MEDIUM PROPERTIES}
\label{Sec:S8_DG_NR}

In the previous sections, we used
patch tests and canonical problems
to demonstrate that the proposed
stabilized mixed DG formulation
can accurately capture the jumps
in the solution fields across
material interfaces.
We will further illustrate the
performance of this formulation
using a representative problem
pertaining to viscous fingering
in heterogeneous porous media.

Viscous fingering is a coupled phenomenon which
involves both flow and transport \citep{Drazin}.
In the flow of two immiscible fluids in a thin cell,
typically called the Hele-Shaw cell, a more viscous
fluid (with viscosity $\mu_{\mathrm{H}}$) is invaded
by a less viscous one (with viscosity $\mu_{\mathrm{L}}
< \mu_{\mathrm{H}}$), resulting in the creation of physical
(displacement) instabilities \citep{Homsy_1987}.
The classical viscous fingering in porous
media with a single pore-network (i.e.,
under Darcy equations) has been studied
by \citep{Saffman_1958}, and therefore,
this instability is sometimes referred
to as the Saffman-Taylor instability in
the literature \citep{Drazin}. 
Recently, \citep{Nakshatrala_Joodat_Ballarini_P2}
have numerically shown that viscous-fingering-type
instabilities can also occur in homogeneous porous
media with double pore-networks. They employed
the continuous Galerkin (CG) formulation of the DPP model, as their studies were
restricted to homogeneous porous media.

Herein, we will employ the proposed
DG formulation to study the effect
of \emph{heterogeneity} on the
appearance and growth of viscous-fingering-type
physical instabilities in porous media with two
pore-networks.
The governing equations for this two-way coupled flow and transport problem consist of two parts. Flow under the DPP model is governed by equations \eqref{Eqn:DG_GE_Darcy_BLM_1}--\eqref{Eqn:DG_GE_Darcy_pBC_2}
and the transient advection-diffusion problem is governed by the following set of equations:
\begin{subequations}
	\begin{align}
		\label{Eqn:Transport_GE}
		&\frac{\partial c(\mathbf{x},t)}{\partial t} + \mathrm{div}\left[\mathbf{u}(\mathbf{x},t) c(\mathbf{x},t)- D(\mathbf{x},t) \mathrm{grad}[c(\mathbf{x},t)]\right] = f(\mathbf{x},t)
		&&\quad \mathrm{in} \; \Omega \times \left(0,T\right) \\
		\label{Eqn:Transport_Dirichlet_BC}
		&c(\mathbf{x},t) = c^{p}(\mathbf{x},t)
		&&\quad \mathrm{on} \; \Gamma^{D} \times \left(0,T\right) \\
		\label{Eqn:Transport_Neumann_BC}
		&\widehat{\mathbf{n}} (\mathbf{x}) \cdot \left(\mathbf{u}(\mathbf{x},t) c(\mathbf{x},t) - D(\mathbf{x},t) \mathrm{grad}[c(\mathbf{x},t)]\right) = q^{p} (\mathbf{x},t)
		&&\quad \mathrm{on} \; \Gamma^{N} \times \left(0,T\right) \\
		\label{Eqn:Transport_IC}
		&c(\mathbf{x},t=0) = c_{0}(\mathbf{x})
		&&\quad \mathrm{in} \; \Omega
	\end{align}
\end{subequations}
where $c(\mathbf{x},t)$ denotes the concentration,
$D(\mathbf{x},t)$ is the diffusivity, 
  and the advection velocity $\mathbf{u}(\mathbf{x},t)$
  is sum of the macro- and micro-velocity fields
  (which are obtained from the flow problem).
  That is, 
  \begin{align}
    \mathbf{u}(\mathbf{x},t) = \mathbf{u}_1(\mathbf{x},t)
    + \mathbf{u}_2(\mathbf{x},t) 
  \end{align}
The concentration for the more viscous
fluid is assumed to be zero and for the
less viscous fluid is considered to be
equal to 1.
In order to complete the coupling of the flow and transport equations and upon introducing $\mu_{0}$ as the base viscosity of the less viscous fluid and
$R_{c} = \mathrm{log}\left(\mu_{\mathrm{H}}/\mu_{\mathrm{L}}\right)$ as the log-mobility ratio, the
viscosity of the fluid is assumed to exponentially depend on the
concentration of the diffusant as follows:
\begin{align}
  \mu (c(\mathbf{x} , t))
  = \mu_{0} \; \mathrm{exp}[R_{c}(1 - c(\mathbf{x}, t))]
\end{align}

We consider a domain consisting of two horizontal layers with different permeabilities. 
The pictorial description of the problem is provided
in \textbf{Fig.~\ref{Fig:DG_VF_Domain}}. The values of macro- and micro-permeabilities in the bottom layer are assumed to be higher than those of the upper layer. Such heterogeneity in the permeability imposes a perturbation on the interface of the two fluids which causes the appearance of unstable finger-like patterns throughout the domain at the fluid-fluid interface. Moreover, a random function is used for defining the initial condition for the transport problem within the domain. Parameter values for this coupled flow and transport problem are provided in Table \ref{Tb:Coupled_flow_transport_Hetero_data}. For the advection-diffusion model given by equations \eqref{Eqn:Transport_GE}--\eqref{Eqn:Transport_IC}, we have utilized Streamline Upwind Petrov-Galerkin (SUPG) formulation, as described in \citep{brooks1982streamline}. Also, see the computer
code provided in Appendix \ref{App:code}.

\begin{figure}
  \includegraphics[clip,width=0.6\linewidth]{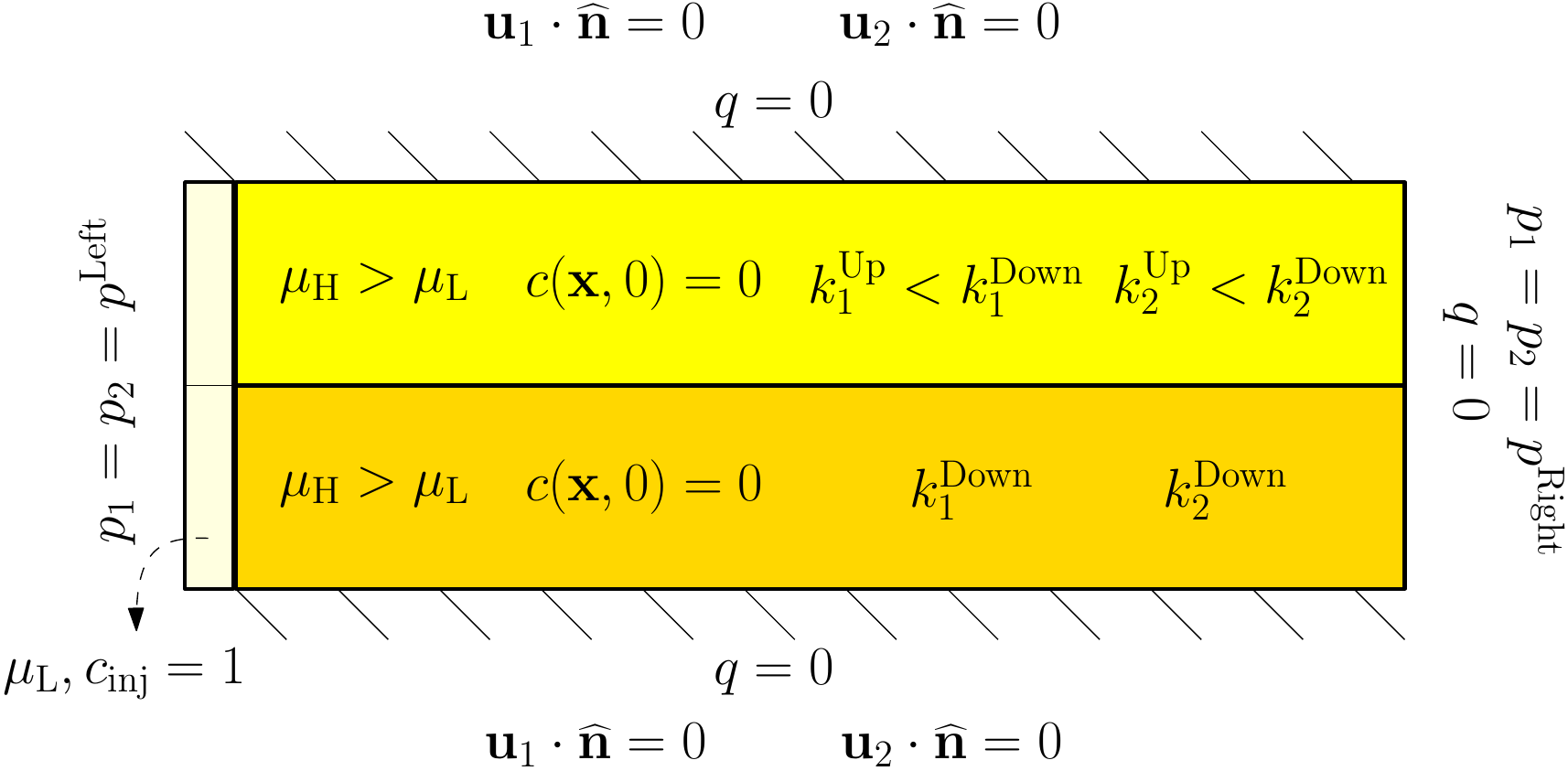}
  \caption{\textsf{Coupled flow and transport problem:}~This figure shows the pictorial description of coupled flow-transport problem with heterogeneous medium properties along with initial and boundary conditions.}
  \label{Fig:DG_VF_Domain}
\end{figure}

{\small
  \begin{table}[!h]
    \caption{Model parameters for coupled flow and
      transport problem in the heterogeneous
      domain.\label{Tb:Coupled_flow_transport_Hetero_data}}
    \centering
    \begin{tabular}{|c|c|} \hline
      Parameter & Value \\
      \hline
      $\gamma \mathbf{b}$ & $\{0.0,0.0\}$\\
      $f$ & $0.0$\\
      $L_x$,~$L_y$ & $1.0$,~$0.4$ \\
      $\mu_0 $ & $1 \times 10^{-3}$ \\
      $R_c$ & $3.0$ \\
      $D$ & $2 \times 10^{-6}$\\
      $\beta $ & $1.0$ \\
      $k_1^{\mathrm{Down}}$ & $1.1$\\
      $k_1^{\mathrm{Up}}$ & $0.9$\\
      $h$&structured T3 mesh\\
      &of size $0.01$ used\\
      \hline 
    \end{tabular}
    \begin{tabular}{|c|c|} \hline
      Parameter & Value \\
      \hline
      $k_2^{\mathrm{Down}}$ & $0.011$\\
      $k_2^{\mathrm{Up}}$ & $0.009$\\
      $c_0$ & $0.0$ \\
      $c_{\mathrm{inj}}$ & $1.0$ \\
      $p^{\mathrm{Left}}$ & $10.0$ \\
      $p^{\mathrm{Right}}$ & $1.0$ \\  
      $q$ & $0.0$\\          
      $\Delta t$ & $5 \times 10^{-5}$ \\
      $T$ & $1.5 \times 10^{-3}$\\
      $\eta_{u}$& $0$\\
      $\eta_P$ & $0$\\
      \hline 
    \end{tabular}
  \end{table}
}

\textbf{Fig.~\ref{Fig:DG_VF_Concentration_heterogeneous}}
shows the concentration profile at different time steps
throughout the heterogeneous domain. The more viscous fluid is
shown in dark blue and the less viscous fluid is shown
in dark red. As can be seen, physical instabilities in
form of separate finger-like intrusions are created at
the fluid-fluid interface. These intrusions are similar
to the viscous-fingering-type instabilities. At the early time
steps, we have a larger number of fingers compared to
the later time steps. These smaller fingers merge and
form fewer but much larger fingers as time goes by. 
It should be noted that finger-like physical
instabilities grow at a higher rate in the
bottom layer due to its higher permeability,
as can be seen in
\textbf{Fig.~\ref{Fig:DG_VF_Concentration_heterogeneous}}.
Moreover, at the later time steps, the fingers formed
in the bottom layer tend to move towards the interface
and enter the top layer.
The proposed DG formulation eliminated the numerical
instabilities (like Gibbs phenomenon and spurious
node-to-node oscillations) but yet accurately
captured the physical instabilities. 
It is worth mentioning that in
our numerical simulations, the
parameters $\eta_u$ and $\eta_p$
had no noticeable effect on the
generation of fingers.

\begin{figure}[!h]
  \subfigure[t=5 $\Delta t$]{
    \includegraphics[clip,width=0.4\linewidth]{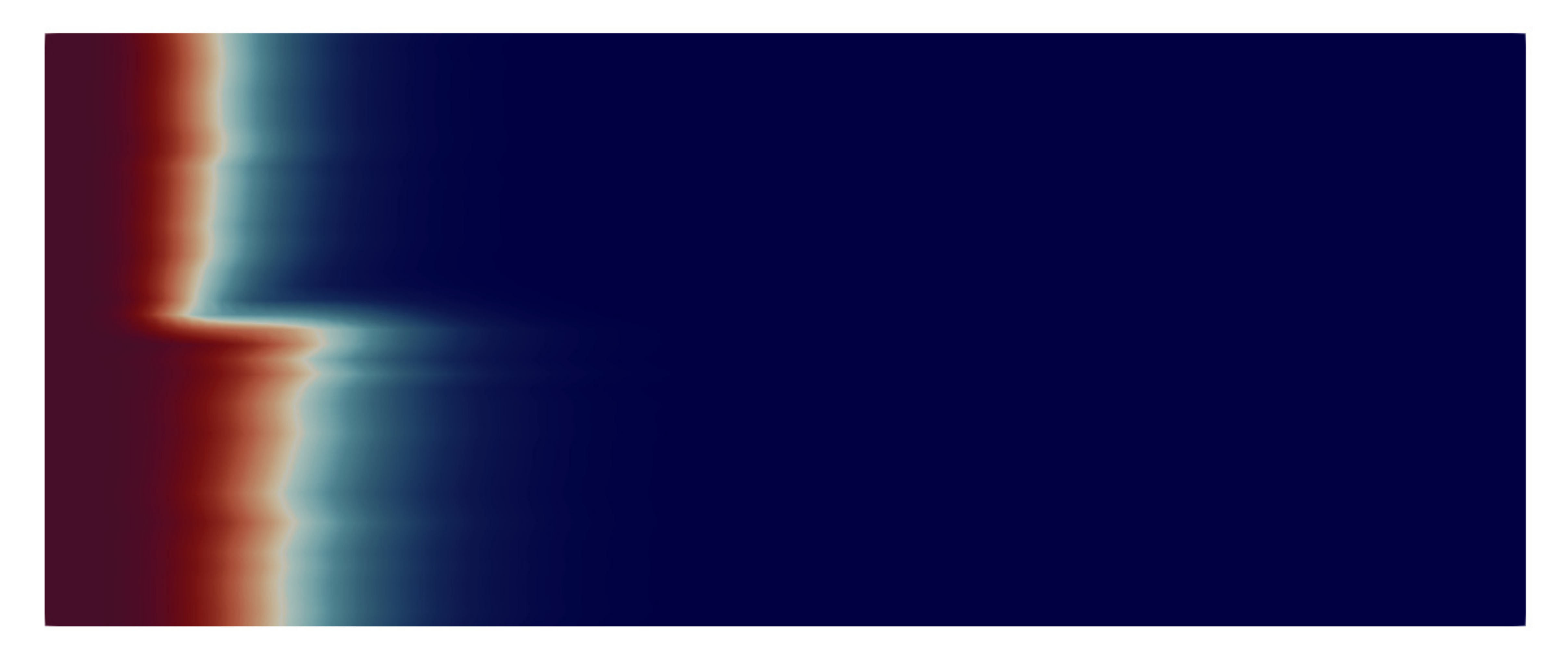}}
  \hspace{0.8cm}
  \subfigure[t=10 $\Delta t$]{
    \includegraphics[clip,width=0.4\linewidth]{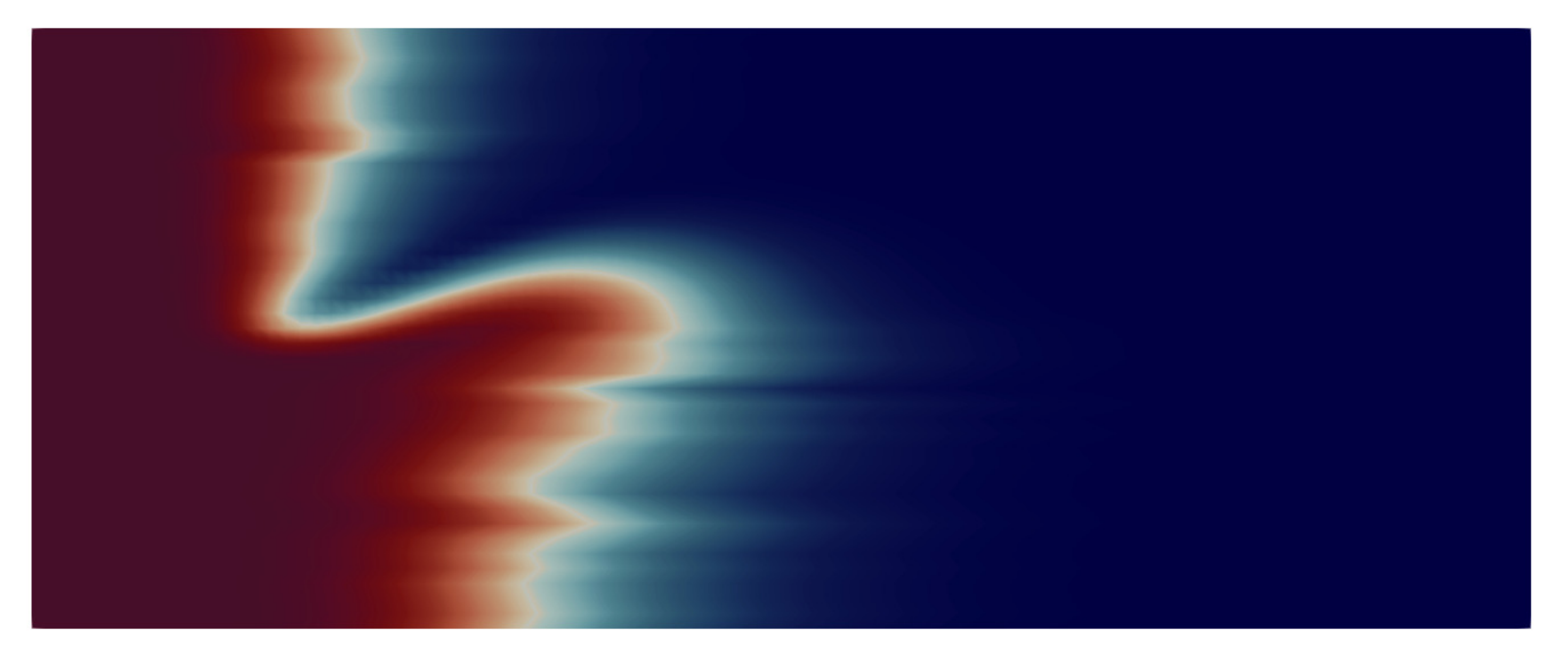}}
  \subfigure[t=15 $\Delta t$]{
    \includegraphics[clip,width=0.39\linewidth]{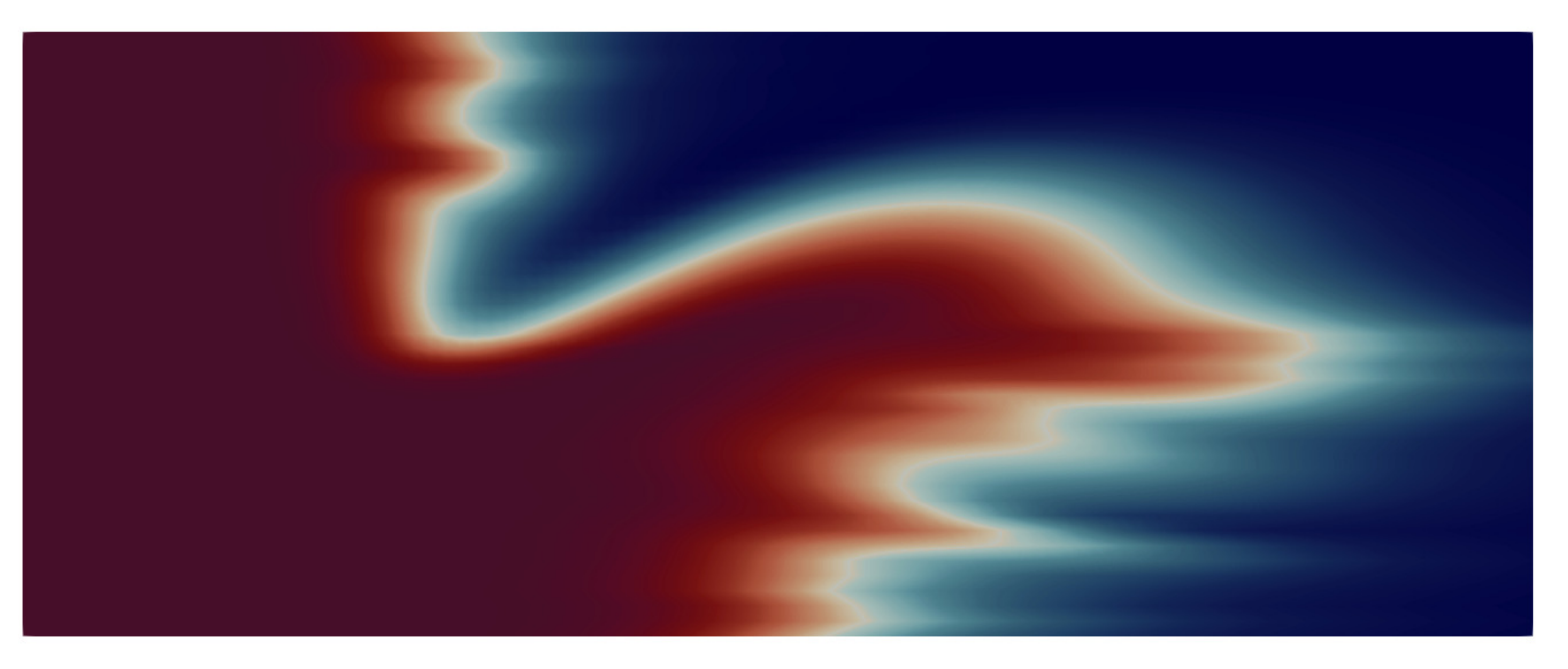}}
  \hspace{0.8cm}
  \subfigure[t=20 $\Delta t$]{
    \includegraphics[clip,width=0.4\linewidth]{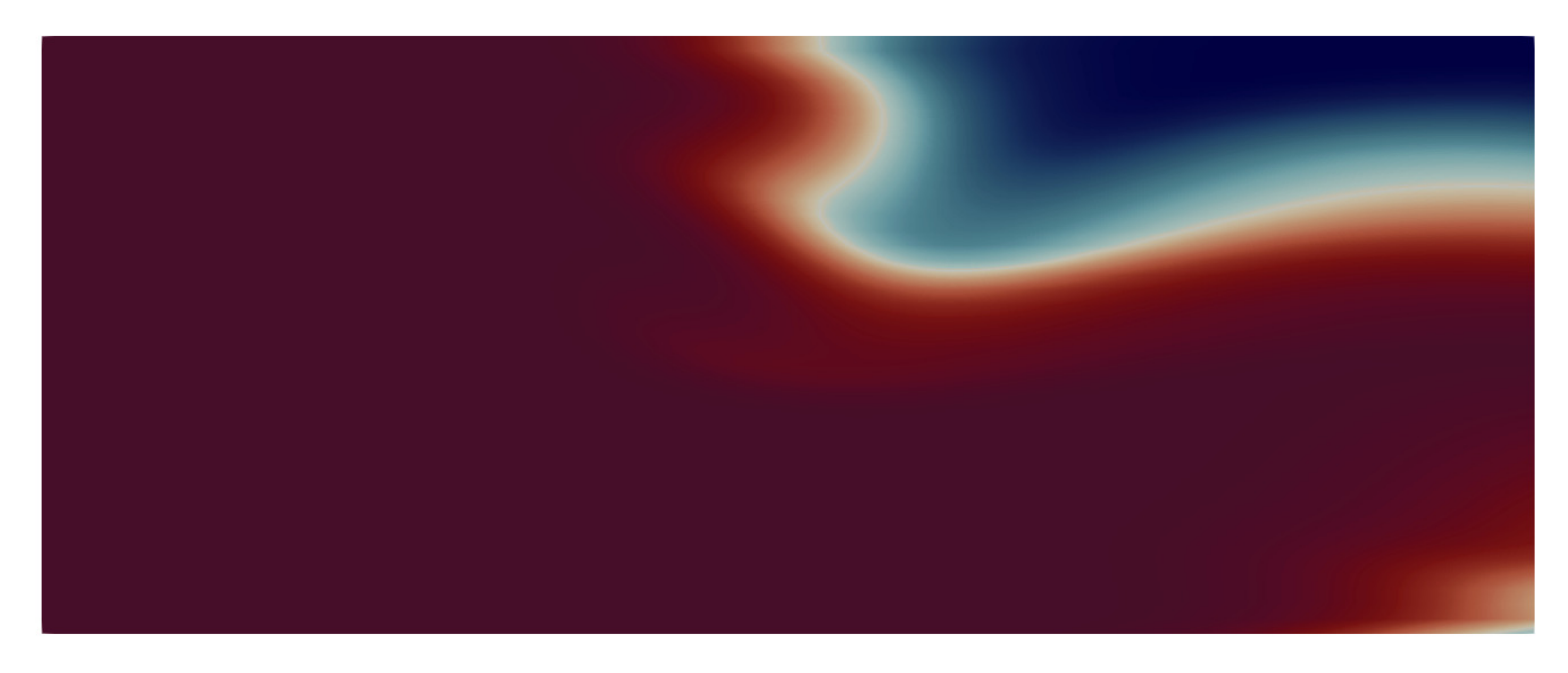}}
  \caption{\textsf{Coupled flow and transport problem:}~This figure illustrates that the proposed stabilized mixed DG formulation is capable of capturing well-known instabilities in fluid mechanics, similar to viscous-fingering instability, in a heterogeneous, layered porous domain with abrupt changes in permeabilities. Fingers are propagating at a higher rate in the bottom layer with higher macro- and micro-permeabilities.
    In our numerical simulations, the parameters
    $\eta_u$ and $\eta_p$ had no noticeable effect
    on the generation of fingers. We obtained similar
    results under $\eta_u = \eta_p = 10$ and $\eta_{u}
    = \eta_p = 100$; which are not shown here.}
  \label{Fig:DG_VF_Concentration_heterogeneous}
\end{figure}

\section{CONCLUDING REMARKS}
\label{Sec:S9_DG_CR}
A new stabilized mixed DG formulation has been
presented for the DPP mathematical model, which
describes the flow of a single-phase incompressible
fluid through a porous medium with two dominant pore-networks. 
Some of the main findings of this paper
on the \emph{computational front} and
the \emph{nature of flow} through porous
media with double pore-networks can be
summarized as follows:
\begin{enumerate}[(i)]
\item Arbitrary combinations of interpolation
  functions for the field variables are stable
  under the proposed DG formulation.
  Unlike the classical mixed DG formulation,
  which violates the LBB \emph{inf-sup}
  stability condition under the equal-order
  interpolation for all the field variables,
  the proposed DG formulation circumvents the
  LBB condition. This implies that the proposed DG
  formulation does not suffer from node-to-node
  spurious oscillations when the computationally
  convenient equal-order interpolation for all
  the field variables is employed.
\item Due to a careful selection of numerical
  fluxes, the proposed DG formulation does not
  suffer from the inherent instabilities that
  DG methods typically suffer from; for example,
  the Bassi-Rebay DG method.
\item The stabilization terms inside the domain
  are of adjoint-type and residual-based, and
  the corresponding stabilization parameters do
  not contain any mesh-dependent parameters.
  \item The proposed stabilized formulation performs
  remarkably well, in comparison with its continuous
  counterpart, in the presence of heterogeneity in material properties. In other words, under the proposed DG formulation no unphysical numerical instabilities are generated at the vicinity of discontinuities in material properties due to Gibbs phenomenon.
\item The formulation passes patch tests, even on
  meshes with non-constant Jacobian elements, in
  2D and 3D settings. 
\item The proposed DG formulation can support non-conforming discretization in form of non-conforming polynomial orders or non-conforming element refinement, thus allowing
efficient $h$-, $p$-, and $hp$-adaptivities.
\item A sensitivity study revealed the importance
  of $\eta_u$ and $\eta_{p}$ (i.e., jump terms with respect to
  the normal components of the velocities and pressures, respectively) to
  reduce the drift along the interior edges for
  the case of non-conforming polynomial orders.
\item It is shown, theoretically, that the proposed
  formulation is convergent. The convergence rates obtained under
  both $h$- and $p$-refinement methods in
  several numerical experiments are in
  accordance with the theory.
\item It is shown that the proposed DG formulation
  can be employed to solve coupled flow-transport
  problems in porous media with double pore-networks. 
  In particular, the effect of heterogeneity of medium properties is studied
  on the appearance and growth of fingers under viscous-fingering-type instability.
  The proposed formulation is capable of suppressing
  the non-physical numerical instabilities (like Gibbs phenomenon and spurious
  node-to-node oscillations), yet capturing
  the underlying physical ones. 
  
\end{enumerate}

\appendix
\section{COMPUTER IMPLEMENTATION}
\label{App:code}
The numerical results pertaining to the non-conforming 
discretization (Section \ref{Sec:Non_conforming}) and
non-constant Jacobian elements (Section \ref{Sec:Non_Constant_Jacob}), 
have been obtained using \textsf{COMSOL Java API} \citep{COMSOL_Java_API}.
The numerical simulations for the 3D numerical convergence analysis (Section \ref{Sec:3D_convergence}) and the coupled problem
(Section \ref{Sec:S8_DG_NR}) were carried out using the
\textsf{Firedrake} Project \citep{rathgeber2017firedrake,
  luporini2015cross}.
All the remaining numerical results were
generated using the \textsf{FEniCS} Project
\citep{logg2012automated,alnaes2015fenics}.

The \textsf{FEniCS} and \textsf{Firedrake} Projects
are built upon several scientific packages and provide
automated frameworks to solve partial differential
equations in serial and parallel environments. Both
provide an easy-to-use Python-based interface to
develop computer codes, to access the scientific
packages on which they are built upon, and to
generate the output in various formats which are
compatible with popular visualization software
packages such as \textsf{ParaView} \citep{paraview}
and \textsf{VisIt} \citep{VisIt}. 
Under both these projects, mesh generation can
be performed either within the code or using
the third party mesh generators such as \textsf{GMSH} 
\citep{Geuzaine_IJNME_2009}.

Among the various components available in \textsf{FEniCS},
we have used the Unified Form
Language (\textsf{UFL}) \citep{Alnaes_UFL_2014} 
and the \textsf{DOLFIN} library
\citep{logg2010dolfin,logg2012dolfin} in 
our implementations.
The former enables the user to declare the finite element
discretization of variational forms and the latter is used 
for the automated assembly of the finite element discrete formulations.
The \textsf{Firedrake} Project employs the \textsf{UFL} from the \textsf{FEniCS} 
Project. However, the main difference between the \textsf{FEniCS}
and \textsf{Firedrake} Projects is that all data structures,
linear solvers and non-linear solvers for the
former are provided by \textsf{DOLFIN} library and for
the latter are provided entirely by the \textsf{PETSc}
library \citep{petsc-user-ref}. Another notable difference
is that the \textsf{FEniCS} Project offers only simplicial
element (e.g., triangular and tetrahedron elements),
whereas the \textsf{Firedrake} Project offers
non-simplicial elements in addition to the simplicial
ones. 

In our numerical simulations, MUMPS
\citep{MUMPS:1} direct solver and
the sparse LU decomposition direct
solver from the UMFPACK
\citep{davis2004algorithm}
were, respectively, employed
with default settings under
the \textsf{COMSOL Java API}
and the \textsf{FEniCS} Project.
The GMRES iterative solver with ``bjacobi''
preconditioner and the relative convergence
tolerance of $10^{-7}$ was employed under
the \textsf{Firedrake} Project.

Below, we have provided a \textsf{Firedrake}-based
computer code, which can be used to generate the
results for the coupled problem, which is discussed
in Section \ref{Sec:S8_DG_NR}.

\begin{lstlisting}[language=Python,caption=Firedrake code
  for solving the coupled problem in the heterogeneous
  porous medium,label=Code:ex1,
  frame=single]
from firedrake import *
import numpy
import random
try:
  import matplotlib.pyplot as plt
except:
  warning("Matplotlib not imported")

#== Create mesh == 
nx, ny = 100, 40
Lx, Ly = 1.0, 0.4
mesh = RectangleMesh(nx,ny,Lx,Ly)

#== Function spaces == 
#---Double porosity/permeability flow problem---
velSpace = VectorFunctionSpace(mesh,"DG",2)
pSpace = FunctionSpace(mesh,"DG",2)
wSpace = MixedFunctionSpace([velSpace,pSpace,velSpace,pSpace])

#---Advection-diffusion problem---
uSpace = FunctionSpace(mesh,"CG",1)

#---Permeability---
kSpace = FunctionSpace(mesh, "DG", 0)

#== Material properties and parameters ==
mu0, Rc, D = Constant(1e-3), Constant(3.0), Constant(2e-6)
k1_0 = 1.1
k1_1 = 0.9
tol = 1E-14

class myk1(Expression): #Macro-permeability
    def eval(self, values, x):
        if x[1] < Ly/2 + tol:
            values[0] = k1_0
        else: 
            values[0] = k1_1

k1 = interpolate(myk1(),kSpace)

k2_0 = 0.01 * 1.1
k2_1 = 0.01 * 0.9

class myk2(Expression): #Micro-permeability
    def eval(self, values, x):
        if x[1] < Ly/2 + tol:
            values[0] = k2_0
        else: 
            values[0] = k2_1

k2 = interpolate(myk2(),kSpace)

#---Drag coefficients---

def alpha1(c): 
    return mu0 * exp(Rc * (1.0 - c))/k1

def invalpha1(c): 
    return 1/alpha1(c)

def alpha2(c): 
    return mu0 * exp(Rc * (1.0 - c))/k2

def invalpha2(c): 
    return 1/alpha2(c)

#== Boundary and initial conditions ==
v_topbottom = Constant(0.0)
p_L = Constant(10.0) 
p_R = Constant(1.0) 
c_inj = Constant(1.0)

#== Perturbation function for initial concentration ==
#---Needed to trigger the instability---
class c_0(Expression):
    def eval(self, values, x):
        if x[0] < 0.010*Lx:
            values[0] = abs(.10*exp(-x[0]*x[0]) * random.random())
        else: 
            values[0] = 0.0

#== Define trial and test functions ==  
#---DPP flow problem---
(v1,p1,v2,p2) = TrialFunctions(wSpace)
(w1,q1,w2,q2) = TestFunctions(wSpace)
DPP_solution = Function(wSpace) 

#---AD problem---
c1 = TrialFunction(uSpace)
u = TestFunction(uSpace)
conc = Function(uSpace) 
conc_k = interpolate(c_0(),uSpace)

#== Time parameters ==
T = 0.0015	# Total simulation time
dt = 0.00005	# Time step

#== Boundary conditions ==
#---DPP velocity BCs---
bcDPP = []

#---AD concentration BCs---  
bcleft_c = DirichletBC(uSpace,c_inj,1,method = "geometric")

bcAD = [bcleft_c]

#== Define source terms ==
#---DPP model---
rhob1, rhob2 = Constant((0.0,0.0)), Constant((0.0,0.0))

#---AD problem---
f = Constant(0.0)

#== Normal vectors and mesh size ==
n = FacetNormal(mesh)
h = CellSize(mesh)
h_avg = (h('+') + h('-'))/2

#== Penalty parameters ==
eta_p, eta_u = Constant(0.0), Constant(0.0)

#== Define variational forms ==  

#---DPP stabilized mixed DG formulation---
aDPP = dot(w1, alpha1(conc_k) * v1) * dx +\
      dot(w2, alpha2(conc_k) * v2) * dx -\
      div(w1) * p1 * dx -\
      div(w2) * p2 * dx +\
      q1 * div(v1) * dx +\
      q2 * div(v2) * dx +\
      q1 * (p1 - p2) * dx -\
      q2 * (p1 - p2) * dx +\
      jump(w1,n) * avg(p1) * dS +\
      jump(w2,n) * avg(p2) * dS -\
      avg(q1) * jump(v1,n) * dS -\
      avg(q2) * jump(v2,n) * dS +\
      dot(w1,n) * p1 * ds(3) +\
      dot(w2,n) * p2 * ds(3) -\
      q1 * dot(v1,n) * ds(3) -\
      q2 * dot(v2,n) * ds(3) +\
      dot(w1,n) * p1 * ds(4) +\
      dot(w2,n) * p2 * ds(4) -\
      q1 * dot(v1,n) * ds(4) -\
      q2 * dot(v2,n) * ds(4) -\
      0.5 * dot( alpha1(conc_k) * w1 - grad(q1), \
                 invalpha1(conc_k) * (alpha1(conc_k) * v1 + grad(p1)) ) * dx -\
      0.5 * dot( alpha2(conc_k) * w2 - grad(q2), \
                 invalpha2(conc_k) * (alpha2(conc_k) * v2 + grad(p2)) ) * dx +\
      (eta_u * h_avg)  * avg(alpha1(conc_k)) * (jump(v1,n) * jump(w1,n)) * dS +\
      (eta_u * h_avg)  * avg(alpha2(conc_k)) * (jump(v2,n) * jump(w2,n)) * dS +\
      (eta_p / h_avg)  * avg(1 / alpha1(conc_k)) * dot(jump(q1,n),jump(p1,n)) * dS +\
      (eta_p / h_avg)  * avg(1 / alpha2(conc_k)) * dot(jump(q2,n),jump(p2,n)) * dS   

LDPP = dot(w1,rhob1) * dx +\
      dot(w2,rhob2) * dx -\
      dot(w1,n) * p_L * ds(1) -\
      dot(w2,n) * p_L * ds(1) -\
      dot(w1,n) * p_R * ds(2) -\
      dot(w2,n) * p_R * ds(2) -\
      0.5 * dot( alpha1(conc_k) * w1 - grad(q1), \
                 invalpha1(conc_k) * rhob1 ) * dx -\
      0.5 * dot( alpha2(conc_k) * w2 - grad(q2), \
                 invalpha2(conc_k) * rhob2 ) * dx 


#---AD formulation with SUPG Stabilization---
vnorm = sqrt(dot((DPP_solution.sub(0)+DPP_solution.sub(2)),\
                 (DPP_solution.sub(0)+DPP_solution.sub(2))))

taw = h/(2*vnorm)*dot((DPP_solution.sub(0)+DPP_solution.sub(2)),\
                      grad(u))

a_r = taw*(c1 + dt*(dot((DPP_solution.sub(0)+DPP_solution.sub(2)),\
                        grad(c1)) - div(D*grad(c1))))*dx

L_r = taw*(conc_k + dt*f)*dx

#---Weak form (GL + SUPG)---
aAD = a_r + u*c1*dx + dt*(u*dot((DPP_solution.sub(0)+DPP_solution.sub(2)),\
                                grad(c1))*dx + dot(grad(u),D*grad(c1))*dx)

LAD = L_r + u*conc_k*dx + dt*u*f*dx

#---Create files for storing solution---
cfile = File("Concentration.pvd")
v1file = File("Macro_Velocity.pvd")
p1file = File("Macro_Pressure.pvd")
v2file = File("Micro_Velocity.pvd")
p2file = File("Micro_Pressure.pvd")

#== Solver for flow problem ==
solver_parameters = { # Default solver -- medium sized problems
  'ksp_type': 'gmres',
  'pc_type': 'bjacobi',
  'mat_type': 'aij',
  'ksp_rtol': 1e-7,
  'ksp_monitor': True
}

problem_flow = LinearVariationalProblem(aDPP, LDPP, DPP_solution, bcs=bcDPP,
  constant_jacobian=False)
solver_flow = LinearVariationalSolver(problem_flow, options_prefix="flow_",
  solver_parameters=solver_parameters)

#== March the solution over time  ==
t = dt
while t <= T:
    print '=============================='
    print '          time =', t
    print '=============================='
    c_0.t = t
    
    #---Compute DPP model---
    solver_flow.solve()

    #---Compute AD problem---
    solve(aAD == LAD,conc,bcs=bcAD)	
    conc_k.assign(conc)   # update for next iteration
    
    #---Dump solutions for each time step---
    cfile.write(conc, time = t)
    v1file.write(DPP_solution.sub(0), time = t)
    p1file.write(DPP_solution.sub(1), time = t)
    v2file.write(DPP_solution.sub(2), time = t)
    p2file.write(DPP_solution.sub(3), time = t)
    t += dt

print "total time = ", t

v1sol, p1sol, v2sol, p2sol = DPP_solution.split()

#==  Dump solution fields to file in VTK format ==
file = File("Concentration.pvd")
file.write(conc)

file = File('Macro_Velocity.pvd')
file.write(v1sol)

file = File('Macro_Pressure.pvd')
file.write(p1sol)

file = File('Micro_Velocity.pvd')
file.write(v2sol)

file = File('Micro_Pressure.pvd')
file.write(p2soll)

\end{lstlisting}

\bibliographystyle{plainnat}
\bibliography{Master_References/Books,Master_References/Master_References}
%
\end{document}